\documentclass[12pt, oneside]{book}

\usepackage[letterpaper, margin=1in, headheight=14pt, footskip=0.5in]{geometry}
\usepackage{setspace}
\usepackage{tocloft}
\usepackage{titlesec}
\usepackage{chngcntr}
\usepackage{tikz}
\usepackage{tabularx}

\usepackage[T1]{fontenc}
\usepackage[utf8]{inputenc}
\usepackage{lmodern}
\usepackage{amsmath, amssymb, amsthm, mathtools}
\usepackage{dsfont}

\usepackage{graphicx}
\usepackage{subcaption} 
\usepackage{xcolor}
\usepackage{booktabs}
\usepackage{multirow}
\usepackage{enumitem}

\setlist{noitemsep}          
\setlist{nosep}              

\graphicspath{{Figures/}{./Figures/}} 

\usepackage[numbers, square, comma, sort&compress]{natbib}
\usepackage{algorithm,algorithmic}

\newtheorem{lemma}{Lemma}

\newtheorem{proposition}{Proposition}
\theoremstyle{definition}

\newtheorem*{assumption*}{\assumptionnumber}
\providecommand{\assumptionnumber}{}


\DeclareMathOperator*{\argmin}{argmin}

\DeclareMathOperator*{\minimize}{minimize}
\DeclareMathOperator*{\maximize}{maximize}

\DeclareMathOperator{\Cor}{Cor}
\DeclareMathOperator{\Cov}{Cov}
\DeclareMathOperator{\Var}{Var}

\DeclareMathOperator{\row}{row}
\DeclareMathOperator{\nul}{null}

\DeclareMathOperator{\sign}{sign}

\DeclareMathOperator{\diag}{diag}

\DeclareMathOperator{\bias}{Bias}

\DeclareMathOperator{\st}{subject\,\,to}

\def\hbeta{\hat{\beta}}
\def\htheta{\hat{\theta}}

\newcommand{\E}{\mathbb{E}}
\newcommand{\Pprob}{\mathbb{P}} 
\newcommand{\R}{\mathbb{R}}

\newcommand{\T}{\mathsf{T}}
\newcommand{\given}{\, \vert \,}
\renewcommand{\hat}{\widehat}
\newcommand{\US}{U.S.}

\def\nth{^{\textnormal{th}}}

\makeatletter
\newcommand*\rel@kern[1]{\kern#1\dimexpr\macc@kerna}
\newcommand*\widebar[1]{%
  \begingroup
  \def\mathaccent##1##2{%
    \rel@kern{0.8}%
    \overline{\rel@kern{-0.8}\macc@nucleus\rel@kern{0.2}}%
    \rel@kern{-0.2}%
  }%
  \macc@depth\@ne
  \let\math@bgroup\@empty \let\math@egroup\macc@set@skewchar
  \mathsurround\z@ \frozen@everymath{\mathgroup\macc@group\relax}%
  \macc@set@skewchar\relax
  \let\mathaccentV\macc@nested@a
  \macc@nested@a\relax111{#1}%
  \endgroup
}
\makeatother
\mathtoolsset{showonlyrefs}

\counterwithout{chapter}{part} 
\counterwithin{figure}{chapter}
\counterwithin{table}{chapter}
\counterwithin{algorithm}{chapter}

\titleformat{\chapter}[display]
  {\normalfont\large\bfseries}{\chaptertitlename\ \thechapter}{12pt}{\Large}

\usepackage{hyperref}
\hypersetup{
  colorlinks=true, linkcolor=black, citecolor=blue, urlcolor=blue,
  pdftitle={Estimating Epidemic Rate Parameters: Adaptivity, Bias, and Convolution},
  pdfauthor={Jeremy Goldwasser},
  pdflang={en-US}
}
\usepackage[capitalise, noabbrev]{cleveref}
\usepackage{bookmark}

\newcommand{\thesistitle}{Estimating Epidemic Rate Parameters: Adaptivity, Bias, and Convolution}
\newcommand{\theauthor}{Jeremy Goldwasser}
\newcommand{\department}{Statistics}
\newcommand{\gradterm}{Summer 2026}
\newcommand{\gradyear}{2026}

\newcommand{\chairname}{Professor Ryan Tibshirani}
\newcommand{\committeememberone}{Assistant Professor Ryan Giordano}
\newcommand{\committeemembertwo}{Associate Professor Will Fithian}
\newcommand{\committeememberthree}{Adjunct Professor Chris Paciorek}
\crefname{appendix}{Appendix}{Appendices}

\usepackage{mdframed}

\newmdenv[
  topline=false,
  bottomline=false,
  rightline=false,
  leftline=true,
  linewidth=2pt,
  linecolor=black,
  leftmargin=40pt,
  rightmargin=40pt,
  innerleftmargin=15pt,
  innerrightmargin=0pt,
  innertopmargin=5pt,
  innerbottommargin=5pt,
  skipabove=\baselineskip,
  skipbelow=\baselineskip
]{leftbar}

\begin{document}

\pagenumbering{roman}
\pagestyle{plain} 

\begin{titlepage}
\thispagestyle{empty} 
\begin{center}
\vspace*{1in}
{\large \thesistitle} \par
\vspace{0.5in}
by \par
\theauthor \par
\vspace{0.5in}
A dissertation submitted in partial satisfaction of the\\
requirements for the degree of\\
Doctor of Philosophy\\
in\\
\department\\
in the\\
Graduate Division\\
of the\\
University of California, Berkeley \par
\vspace{0.5in}
Committee in charge:\\
\chairname, Chair\\
\committeememberone\\
\committeemembertwo\\
\committeememberthree \par
\vspace{0.5in}
\gradterm
\end{center}
\end{titlepage}

\newpage
\thispagestyle{empty} 
\vspace*{\fill}
\begin{center}
\thesistitle\\[12pt]
\theauthor\ \copyright\ \gradyear\\
All rights reserved.
\end{center}
\vspace*{\fill}

\newpage
\pagenumbering{arabic}
\begin{center}
{\large Abstract}\\[12pt]
\thesistitle\\[12pt]
by\\[6pt]
\theauthor\\[6pt]
Doctor of Philosophy in \department\\[6pt]
University of California, Berkeley\\[6pt]
\chairname, Chair
\end{center}
\noindent
Metrics like the case-fatality rate and reproduction number are key descriptors of epidemics from the COVID-19 pandemic to the seasonal flu. 
In retrospect, these quantities enrich our understanding of infectious disease outbreaks; in real-time, they are absolutely critical to informing public health response. 
Thus, an important question in epidemiology is how best to estimate such metrics, especially in real-time. 
This question is complicated by practical considerations like data availability, as well as the fact that the metrics themselves may change as the epidemic unfolds. 

\vspace{\baselineskip}
\noindent We focus on two classes of time-varying metrics, the first of which we call ``severity rates.'' 
These encompass the case-fatality rate, hospitalization-fatality rate, and any metric tracking the rate at which a primary event develops into a more severe secondary event. 
Our first contribution is to identify weaknesses of the standard methods for estimating severity rates in real-time.
Strikingly, our empirical and mathematical analyses reveal significant bias for the most common method, a lagged ratio between primary and secondary counts.
We also show that a related ratio estimator is biased, making connections to $R_t$ estimation.

\vspace{\baselineskip}
\noindent Drawing on these insights, we propose methodology that estimates the entire time series of severity rates. 
Our approach maximizes an approximate likelihood function that we derived based on the convolutional nature of the problem. 
This likelihood is regularized by a trend filtering penalty that produces piecewise polynomial fits, and it flexibly adapts to local changes in the underlying signal.
We further regularize tail predictions in order to better suit our method for real-time estimation.
Experiments on real and semi-synthetic COVID-19 data demonstrate consistent performance improvements over the biased ratio baselines.

\vspace{\baselineskip}
\noindent The second metric we study is the reproduction number ($R_t$), which measures the average number of secondary transmissions from a single infection at time $t$.
We show that two competing technical definitions --- ``instantaneous'' and ``mechanistic'' $R_t$ --- are equivalent under homogeneous mixing, a standard assumption in compartmental models like SEIR.
We also derive the generation interval implied by SEIR, showing that this seemingly assumption-light input encodes specific compartmental structure.

\vspace{\baselineskip}
\noindent We introduce a new deconvolution method for $R_t$ estimation.
Instantaneous $R_t$ is a backward-looking quantity that avoids the forward-looking bias inherent to severity rate estimation, making it well-suited for real-time inference.
Typically, it is estimated from data such as case reports or hospitalizations, which are observed days after transmissions occur.
Many widely-used methods ignore this delay and are thus biased, especially around interventions.
More sophisticated tools account for latent infections, but are slow and rely on Bayesian priors.
Our frequentist approach, ConvRt, adapts ideas from the severity rate work --- convolutional modeling, polynomial smoothing, and tail regularization.
It predicts similarly to leading methods in a fraction of the computational cost on COVID-19 and influenza data.
Experiments on a range of synthetic benchmarks show favorable performance in point estimation and uncertainty quantification.

\newpage
\pagenumbering{roman}
\tableofcontents

\newpage
\phantomsection
\addcontentsline{toc}{chapter}{\listfigurename}
\listoffigures

\phantomsection
\addcontentsline{toc}{chapter}{\listtablename}
\listoftables


\newpage
\phantomsection
\addcontentsline{toc}{chapter}{Acknowledgements}
\chapter*{Acknowledgements}
\hspace{\parindent} This thesis would not have been possible without numerous individuals who supported me at every step of the way. First and foremost, I owe a deep debt of gratitude to my advisor, Ryan Tibshirani. Through his technical mentorship, Ryan has exemplified what it means to be statistician: to be rigorous, collaborative, and above all curious. And, just as importantly, he has demonstrated true leadership: thoughtful, decisive, hardworking, patient, and unfailingly kind. As I enter the working world --- again, thanks to his generosity --- I attribute my professional confidence to the time and care he invested in developing my abilities. It has been an extraordinary opportunity, and a real pleasure, to work with Ryan from the moment he came to Berkeley in 2022.

I am also extremely grateful for the support of my co-advisor, Giles Hooker. Despite our long-distance advisor-ship, Giles has been anything but distant throughout my PhD years. I have relished growing from his research vision, good humor, and attention to detail, including relearning Python while I was off at a summer internship. The papers we wrote together are not in this thesis, but I view them with pride and profound gratitude for his mentorship.

Several researchers outside Berkeley have helped me immensely on this journey. Alyssa Bilinski has been a tremendous support, particularly during our close collaboration this final semester. My excitement towards our rapid-fire project owes not only to her intellectual leadership, but also her incredible warmth and openness. Daniel McDonald has been an absolute workhorse; I have enjoyed all our conversations spanning statistics and classical music. I’m lucky to have worked with Patrick Kimes at Genentech, guiding me through an unforgettable summer internship. My research career began under the tutelage of Dragomir Radev, who tragically passed away in 2023. Drago is remembered as an AI pioneer, larger-than-life personality, and pillar of the Yale community.

Members of the Statistics Department have supported me in many ways throughout my PhD. I would like to thank Professors Ryan Giordano, Will Fithian, and Chris Paciorek for serving on my committees and offering meaningful perspectives. None of this would have happened without Bin Yu, who had faith in my abilities at the very beginning; she showed me how to express altruism, candor, and intellectual vitality in the statistical sciences and beyond. Addison Hu was a dedicated mentor on Part I of this thesis, and I’m fortunate to call him a friend. These years have been enriched by the enduring friendship of Will Torous, Omer Ronen, Erez Buchweitz, Seunghoon Paik, Eric Xia, Austin Tau, and many other graduate students. Lastly, without the amazing department staff --- especially Tanisha Robinson, Alex Coughlin, and La Shana Polaris --- Evans would surely collapse ahead of schedule. 

I am proud to join my father, sisters, aunts, uncles, cousins, and friends in the ranks of alumni of the world’s \# 1 public university (go bears!). Indeed, my brilliant friends in the broader Berkeley community are befitting of the title. While there are too many people to mention by name, I’d like to give a special shout-out to Ari Goldberg, Seul Ah Kim, Yarden Goraly, Eli Bronstein, Ron Boger, Steph Brener, Karna Mendonca, and Isaac Lipsky. It has also been an indescribable privilege to perform 60+ concerts with the UC Berkeley Symphony Orchestra, under the baton of David Milnes.

Finally, I’d like to thank my family: my parents Richard and Susan, and my sisters Lily and Sarah. Their unconditional love has sustained me through it all.

\newpage
\pagenumbering{arabic}
\pagestyle{plain}

\chapter{Introduction}
\label{ch:intro}

\section{Public Health Metrics}

\subsection*{Motivation}
Data analysis is a critical component of public health. 
Decision-making in an infectious disease outbreak is difficult in its own right, but nearly impossible without good data. 
Statistical analyses provide guidance for policy-relevant questions such as: 
\begin{itemize}
    \item How many people are likely to be infected in the short- and long-term?
    \item What is the risk of hospitalization, or death?
    \item Which geographic or demographic subpopulations need the most resources?
\end{itemize}

These questions were of critical importance during the COVID-19 pandemic. 
Seroprevalence studies estimated that nearly 200 million Americans had been infected by February 2022, after the Omicron wave \citep{clarke2022seroprevalence}.
In total, over a million Americans died, with an infection-fatality rate (IFR) around 0.5\% \citep{meyerowitzkatz2020systematic}. 
About 4.5\% of infections, totaling $\sim$3.5 million, were hospitalized during the pre-vaccination ancestral period.
Mortality rates were heavily age-stratified, below 0.01\% for children and roughly 15\% for elders above 85 \citep{levin2020assessing}. 
They also varied by demographics, affecting men and non-White ethnicities disproportionately.

Throughout the pandemic, governments imposed or lifted non-pharmaceutical interventions (NPIs) like school lockdowns and mask mandates based on data-driven forecasts. 
These numbers were also instrumental in shaping the public's behavior.
Unfortunately, the data in question was far from perfect, in part due to low case ascertainment rates.
This was especially true in the early months of 2020, when decisions (or their lack thereof) were the most consequential. 
Even with the lockdowns starting March 2020, over 50 million Americans had likely been infected by the end of the summer \citep{reese2021estimated}. 

Moreover, salient epidemiological metrics changed over time, reflecting dynamic public behavior and biological factors.
New variants rarely had the same infectiousness and severity as their predecessors; for example, the Omicron variant was more transmissible than Delta, but also had lower hospitalization and fatality rates \citep{cfrs_by_variant}. 
Vaccines and other therapeutics lowered risks of adverse outcomes, though coverage was far from uniform. 
Viral transmission also depended on human elements like seasonal trends, NPIs, and falling compliance with guidelines.
In response to changing conditions, public health officials needed to adapt accordingly. 

While we focus on COVID-19, the pandemic's lessons extend to other diseases.
Tens of thousands of Americans die each year from seasonal influenza, and hundreds of thousands globally \citep{iuliano2018estimates, reed2015estimating, thompson2003mortality}. 
Tracking disease burden helps guide real-time responses such as vaccination, testing, and communication campaigns.
No two flu seasons are exactly alike, due to shifts in circulating strains, vaccine uptake, antigenic drift, and other factors.
Surveillance metrics are also critical for managing RSV, malaria, monkeypox, HIV, and many other diseases \citep{lloydsmith2005superspreading}. 

While there are many quantities of interest in epidemiology, we focus on two categories in this work. 
The first, \textit{severity rates}, is a framework that encompasses several well-known metrics. 
The second, \textit{reproduction numbers}, measures how many people an average case infects. 

\subsection*{Severity Rates}

Several public health metrics of interest express the probability that a second,
often more severe outcome will follow a primary event. 
We refer to such metrics as ``severity rates.''
The most commonly-reported severity rate is the case-fatality rate, or CFR \citep{cfr_cao,
nishiuraEx1}. 
Other examples of interest describe different types of
events, for example hospitalizations \citep{HFR_linelist3} and wastewater
shedding \citep{wastewater_cases}.

The CFR is sometimes treated as a proxy for the infection-fatality
rate, or IFR. 
While the IFR more faithfully represents disease status, infection counts are
generally unknown, so it is harder to compute. It can be attempted using seroprevalence data to estimate the ascertainment rates that relate infections to observed cases \citep{timevar_ifr, lancet_ifr}. 

It is common to treat severity rates as stationary over time \citep{ghani,
  jewell2007nonparametric, reich2012estimating, lancet_controversial}. This
assumes that the population-level probability of developing a secondary
event, given a primary event, is constant throughout the entire epidemic. 
However, as explained above, severity rates often change for biological, therapeutic, and behavioral reasons. 
For example, \citet{cfrs_by_variant} estimates the  
original COVID-19 variant to have a 3.6\% CFR globally, compared to 2.0\% for
the Delta variant, and 0.7\% for Omicron.  
Therefore we define the severity rate at time $t$ as
\begin{equation}\label{eq:severity}
    p_t = \Pprob(\text{secondary event will occur} \given \text{primary event at $t$}).
\end{equation}

\subsection*{Reproduction Numbers}

Like severity rates, reproduction numbers also concern the rate at which secondary events occur.
The reproduction number is best understood as the average number of secondary transmissions over the lifetime of a single primary infection.
There are several types, the most central of which is the \textit{basic reproduction number} $R_0$. This is the average number of transmissions for a single infection in a fully susceptible population, with no deliberate interventions in place \citep{diekmann1990definition}.

In contrast, the \textit{effective reproduction number} $R_t$ tracks how this quantity changes over time as immunity builds and conditions shift.
There are multiple competing definitions of $R_t$.
\textit{Case $R_t$} is the average number of secondary transmissions amongst the cohort of individuals infected at time $t$ \citep{fraser2007}.
Similar to our definition of severity rates \eqref{eq:severity}, case $R_t$ is a forward-looking quantity that conditions on primary events at $t$.

\textit{Instantaneous $R_t$} instead describes the average number of transmissions that occur at $t$, from infections before $t$.
This backward-looking notion is the standard definition in practice: it engenders a formula called the ``renewal equation'' amenable to real-time estimation. 
In spite of their differences, instantaneous and case $R_t$ are equivalent if conditions remain locally stable.

A third formulation arises from compartmental epidemic models, which decompose a population into groups based on disease status \citep{anderson1991infectious, kermack1927contribution}.
Movement between compartments is governed by transmission and recovery parameters, from which the reproduction number can be expressed.
We explore the relationships among these definitions, and methods to estimate them, in \cref{pt2}.

\subsection*{Computing Epidemic Rates}
In an ideal setting, severity rates and reproduction numbers can be obtained directly from a
line-list data set containing individual patient outcomes. 
Such line-lists do not necessarily need to be comprehensive, or even representative of the target population of interest.
Data obtained from contact tracing has been used to calculate reproduction numbers for SARS, measles, monkeypox, and many other infectious diseases.
Various techniques can be used to adjust for biases in non-representative line-lists, such as post-stratification on characteristics like age \citep{nishiuraEx2}, race \citep{HFR_linelist2}, and seroprevalence \citep{lancet_ifr}.

However, real-time tracking of individuals in fast-moving epidemics
has in many settings been infeasible, even at smaller scales. 
For example, the CDC's line-list during the COVID-19 pandemic excluded several large geographic regions, was only updated monthly, and had missing death statuses for a high proportion of individuals \citep{CDC_line_list}. 
Techniques like importance reweighting are often insufficient to address these problems.
Therefore, while line-lists may help obtain a preliminary snapshot or retrospective analysis, their utility is limited for real-time estimation. 

When this is the case, severity rates and reproduction numbers
are routinely estimated from aggregate count data. 
We present a number of examples for COVID-19.
Case and death counts were frequently used to estimate and report CFRs in the academic literature \citep{yuan2020monitoring,
  horita2022global, LIU2023100350, germany}, as well as in major news outlets including
The Atlantic \citep{atlantic}, Wall Street Journal \citep{wsj}, and New York Times \citep{nyt}. 
They were also used by public health organizations like the CDC \citep{ahmad2023covid,mississippi}, 
and reported by the Trump and Biden administrations \citep{whitehouse2020mcenany, whitehouse2021briefing}. 
Similarly, reproduction numbers were routinely reported in journal publications, news articles, and government statements. 
Most estimated instantaneous $R_t$ from aggregate data, with methods described in Chapters \ref{ch:paper3} and \ref{ch4}.

\section{Our contributions}

Our work studies methods to estimate time-varying epidemic rate parameters, especially in real-time.
Part \ref{pt1} of this thesis focuses on severity rates.
In Chapter \ref{ch:paper1}, we analyze the standard real-time estimators, 
and expose that they may be highly biased.
In Chapter \ref{ch:paper2}, we propose a solution that estimates severity rates via adaptive deconvolution. 
Part \ref{pt2} turns to reproduction numbers. Chapter \ref{ch:paper3} draws new connections between different formulations of $R_t$, and Chapter \ref{ch4} develops a fast, frequentist estimation method.

\subsection*{Part \ref{pt1}: Estimating Time-Varying Severity Rates}
\begin{leftbar}
\itshape
This work was done in collaboration with Addison Hu, Alyssa Bilinski, Daniel McDonald, and Ryan Tibshirani. 
\end{leftbar}

\subsubsection*{Understanding Bias in Real-Time Estimation}
Time-varying severity rates are typically estimated with a ratio of primary and secondary counts.
In fact, ratio estimators are so common that CFR is often 
(mis)labeled the case-fatality \emph{ratio} \citep{timevar_ifr}. 
Our work shows that these ratio estimators are prone to nontrivial
statistical bias. 
We empirically validate our mathematical analysis, 
tracking the hospitalization-fatality rate (HFR) during COVID-19.

Bias arises as a consequence of changing severity
rates---
precisely what time-varying estimates are employed to detect.
On \US\ data, for example, the ratio estimators failed to quickly signal increased
risk in the onset of the Delta wave.
In addition, bias also arises due to misspecification of the delay distribution that relates primary and secondary events. 
This is particularly troublesome during various stages of a surge, especially for the popular lagged ratio.
Once again, these are periods in which it is critical to properly ascertain the severity rate.
While the HFRs dropped in the aftermath of the initial Omicron surge, the lagged ratio's estimates jumped 50\% above the true values.

We provide practical heuristics for when epidemiologists should expect this bias to occur. 
These may prevent poor decision-making based on biased estimates of severity rates. 
More broadly, our work exposes the limitations of existing approaches, motivating the use of alternatives.

\subsubsection*{Adaptive Deconvolution with Trend Filtering}

We propose an approach that overcomes the limitations of the
methods described in Chapter \ref{ch:paper1}. 
Our deconvolution method has less bias than the existing ratio-based estimators, as
it uses a statistical model that explicitly encodes severity rates to connect
primary and secondary events.
Since this model is underspecified, we use regularization---specifically, a form
of trend filtering---to smooth out the estimated severity rates, while
maintaining the ability to adapt to potentially abrupt changes in the underlying
signal. In the real-time
case, we use additional regularization at the right tail of the sequence.

We validate our methodology on a broad range of experiments, calculating the hospitalization-fatality rate (HFR) of COVID-19. 
These experiments use real hospitalization data,
and generate synthetic death counts using plausible HFR curves.
To evaluate our real-time method, we add noise by modeling overdispersion
in the versioned death counts.
Our approach achieves around 15\% lower MAE than the convolutional ratio estimator,
and a 55\% improvement over the lagged ratio (used widely in practice).
It continues to outperform ratio-based methods under various degrees of
model misspecification. 
Finally, the same qualitative benefits are observed on 
real COVID-19 deaths.

\subsection*{Part \ref{pt2}: Estimating the Reproduction Number}
\begin{leftbar}
\itshape
This work was done in collaboration with Alyssa Bilinski and Ryan Tibshirani. 
\end{leftbar}


\subsubsection*{Unifying Perspectives on $R_t$}

The effective reproduction number is typically estimated using the instantaneous $R_t$ formulation.
Compartmental models like SEIR yield a distinct definition, which we call mechanistic $R_t$. 
These depend on entirely different mathematical quantities. 
For example, instantaneous $R_t$ uses the generation interval distribution, while mechanistic $R_t$ uses the mean duration of infectiousness.
These two perspectives are usually treated as separate modeling traditions, each with its own assumptions and estimation strategies.

We prove that instantaneous and mechanistic $R_t$ are equivalent under homogeneous mixing, the standard assumption for compartmental models.
We also derive the generation interval distribution that SEIR models imply.
This equivalence makes explicit that generation intervals encode mechanistic structure, even when they appear to have few assumptions.

\subsubsection*{Estimating $R_t$ via Deconvolution}

We propose ConvRt, a frequentist $R_t$ estimator that conducts two stages of deconvolution.
Most $R_t$ methods use the renewal equation, a formula that relates instantaneous $R_t$ to infection counts.
They use data on case reports or hospitalizations as a proxy for infections, which are not observed.
Most existing methods either ignore this latent structure or handle it with computationally expensive MCMC sampling.

ConvRt models latent infections explicitly but replaces MCMC with penalized maximum-likelihood, and thus fits in seconds.
It performs two deconvolution steps with spline-based Poisson regression: first estimating infections, then $R_t$ itself.
We adapt strategies from the severity rate work to stabilize real-time predictions.
We also develop conformal inference bands that account for the bias introduced by tail regularization.

We benchmark ConvRt against five existing methods on six simulated datasets, as well as real influenza hospitalizations and COVID-19 case data.
ConvRt matches or exceeds the CDC's preferred method in point accuracy and calibration, while running orders of magnitude faster.
It significantly outperforms all other methods that run in comparable time.

\part{Estimating Time-Varying Severity Rates}\label{pt1}
\chapter{Understanding Bias in Real-Time Estimation}
\label{ch:paper1}

\subsection*{Notation}

Notation is summarized in \cref{tab:notation-ch1} and applies to both this chapter and \cref{ch:paper2}.

\begin{table}[!ht]
\centering
\caption{Severity rate notation (Chapters~2--3).}
\label{tab:notation-ch1}
\renewcommand{\arraystretch}{1.15}
\begin{tabularx}{\linewidth}{@{}l >{\raggedright\arraybackslash}X@{}}
\toprule
\textbf{Symbol} & \textbf{Meaning} \\
\midrule
\multicolumn{2}{@{}l}{\textsc{Event counts}} \\
\addlinespace[2pt]
$x_t$ & Primary events (e.g.\ cases, hospitalizations) at $t$ \\
$y_t$ & Secondary events (e.g.\ deaths) at $t$ \\
\midrule
\multicolumn{2}{@{}l}{\textsc{Severity rate and delay}} \\
\addlinespace[2pt]
$p_t$ & Severity rate: probability of secondary event given primary event at $t$ \\
$\pi_k^{(t)}$ & Delay distribution: probability secondary event occurs $k$ steps after primary event at $t$ \\
$d$ & Maximum delay (support of the delay distribution) \\
\midrule
\multicolumn{2}{@{}l}{\textsc{Estimators}} \\
\addlinespace[2pt]
$\hat{p}_t^\ell$ & Lagged ratio estimator with lag $\ell$ \\
$\hat{p}_t^\gamma$ & Convolutional ratio estimator with delay estimate $\gamma$ \\
\bottomrule
\end{tabularx}
\end{table}

\section{Ratio Estimators of Severity Rates}\label{sec:defs}

Severity rates convey the probability that a primary event will result in a
secondary event in the future. In the case of CFR, for example, a primary event
is a positive COVID-19 case and a secondary event is a death with a positive
test result. 
As introduced in \cref{ch:intro}, the 
severity rate $p_t$ is the probability of a future secondary event, given a primary event occuring at $t$.

Here, $t$ may represent a discrete interval of time, such as a given day or 
week. It also may be understood in a continuous-time fashion. Although this will    
not be our focus in this paper, the same general principles apply in the
continuous-time case. For simplicity, we will consider only the discrete-time 
setting, and we index time steps via integers, as in $t=0,1,2,\dots$.    

Throughout, we denote by $\{x_t\}$ and $\{y_t\}$ the aggregate time series of 
new primary and secondary events, respectively. These are often counts, and we
will generally refer to them as such. At time $t$, we assume data for all past
$s \leq t$ is available, but future data is not. Therefore real-time estimates of
$p_t$ can only rely on past counts \smash{$x_{\leq t}$} and
\smash{$y_{\leq t}$}. In practice, to stabilize estimates, smoothed
counts are often used in place of raw counts. This may be simply absorbed into
the notation for $x_t$ and $y_t$, and we do not address smoothing explicitly in
the formulation of the estimators, but refer back to this issue in
\cref{sec:setup} and \cref{sec:setup2}.

In the context of CFR, for example, $\{x_t\}$ and $\{y_t\}$ denote cases and deaths.
However, the following estimators and their biases are shared for any severity rate $p_t$, including CFR and HFR.

\paragraph{Lagged estimator.} 

The canonical estimator for time-varying severity rates is a ratio between
the counts of primary and secondary events, offset by a lag $\ell$. This
estimator is widely-used in epidemiology, both in the academic literature and in 
public health practice and communication (e.g., \citealp{wsj, atlantic,
  yuan2020monitoring, timevar_ifr, thomas2021estimating, horita2022global,
  LIU2023100350, germany}). For concreteness, we define the \emph{lagged ratio} 
at time $t$ as:   
\begin{equation}
\label{eq:lagged}
\hat{p}_t^\ell = \frac{y_t}{x_{t-\ell}},
\end{equation}
where $\ell \geq 0$ is a given parameter (often chosen to maximize
cross-correlation between $\{x_t\}$ and $\{y_t\}$). 
Note without loss of generality that the lag itself may be time-varying, adapting to shifts in the delay distribution.

\paragraph{Convolutional estimator.} 

Alternative methods for estimating severity rates utilize a delay distribution,
which relates the two time series. The delay distribution at time $t$ and lag
$k$ is defined as: 
\[
\pi_k^{(t)} = \Pprob(\text{secondary event at $t+k$} \given \text{primary event at
  $t$, secondary event occurs}).  
\]

Throughout this work, we assume that the delay distribution has a finite support
of $d$ time steps. 
For the sake of notational convenience, we also assume the delay distribution is stationary:
\smash{$\pi_k^{(t)} = \pi_k$} for all $k$ and $t$. 
However, our analyses hold in the case of time-varying delay distributions.
While the delay distribution is generally unknown,
several tools exist to estimate them from aggregate or line-list data; see
\citealp{delay_distrs} for a review. Given $\pi$, we can express the
expected number of secondary events at time $t$ as follows: 
\begin{align}
\E[y_t \given x_{\leq t}]
&= \sum_{k=0}^d x_{t-k} \Pprob(\text{secondary at $t$} \given \text{primary at
  $t-k$}) \nonumber \\   
&= \sum_{k=0}^d x_{t-k} \Pprob(\text{secondary after $k$ time steps} \given
  \text{secondary occurs, primary at $t-k$}) \nonumber \\
&\hspace{100pt} \cdot\Pprob(\text{secondary occurs} \given\text{primary at $t-k$})
  \nonumber \\  
\label{eq:model}
&= \sum_{k=0}^d x_{t-k} \pi_k p_{t-k}.
\end{align}

This is a convolution of the delay distribution against the product of primary
incidence and the severity rate. If the severity rate remains constant, $p_{t-k}
= p_t$ for all $k$ between $0$ and $d$, then the expression in \eqref{eq:model}
simplifies to \smash{$\E[y_t \given x_{\leq t}] = p_t \sum_{k=0}^d
  x_{t-k} \pi_k$}. As studied in \citet{UKpaper}, we can rearrange this
relationship in order to estimate the severity rate at $t$, after plugging-in an
estimate $\gamma$ of the delay distribution $\pi$:   
\begin{equation}
\label{eq:conv}
\hat{p}_t^\gamma = \frac{y_t}{\sum_{k=0}^d x_{t-k} \gamma_k}.
\end{equation}

We call this the \emph{convolutional ratio} estimator of the severity rate. The superscript in our notation \smash{$\hat{p}_t^\gamma$} emphasizes that
$\gamma$ is the distribution used in the definition of the estimator
\eqref{eq:conv}. To reiterate, in moving from \eqref{eq:model} to
\eqref{eq:conv}, we are implicitly assuming that the severity rate $p_t$ is
stationary over the interval of time from $t-d$ and $t$. Of course, this runs in
contradiction to the fact that we are trying to estimate a time-varying severity
rate in the first place. As we will see shortly, this can create significant
bias in the convolutional ratio.  

The convolutional ratio was originally proposed for estimating stationary severity rates. \citet{nishiura} developed the estimator in this setting and used it to analyze the CFR in the H1N1 influenza pandemic of 2009. (The main difference to \eqref{eq:conv} is that in the stationary case we aggregate both the numerator and denominator over all past data.) This estimator, sometimes called the \emph{delay-adjusted} estimator, has since become popular in the academic literature and public health practice (e.g., \citealp{nishiuraEx1, nishiuraEx2, Russell2020, Unnikrishnan2021}). Its extension to the time-varying setting is more recent and less widely adopted; the R package \texttt{cfr} \citep{cfr_package} provides an implementation for both cases.  

Furthermore, we note that \eqref{eq:conv} can be seen as a generalization of the  
lagged ratio estimator \eqref{eq:lagged}: when we take $\gamma$ to be a point
mass at lag $\ell$, i.e., $\gamma_\ell = 1$ and $\gamma_k = 0$ for all $k \not=
\ell$, then \eqref{eq:conv} reduces to \eqref{eq:lagged}.


\section{Bias Analysis}
\label{sec:methods}

In this section, we analyze the bias of the ratio estimators.

\subsection{Well-specified analysis}
\label{sec:wellspecified}

First we analyze the bias of the convolutional ratio \eqref{eq:conv} in what we  
call the \emph{well-specified} case, where the true delay distribution $\pi$ is 
known. Formally, for an estimator \smash{$\hat{p}_t$} of $p_t$, we define its
bias as:     
\[
\bias(\hat{p}_t) = \E[\hat{p}_t \given x_{\leq t}] - p_t. 
\]

\begin{proposition}
\label{prop:OracleBias}
Assume the true delay distribution $\pi$ is known.
The bias of the well-specified convolutional ratio \smash{$\hat{p}_t^\pi$} is 
\begin{equation}
\label{eq:OracleBias}
\bias(\hat{p}_t^\pi)  = \sum_{k=0}^d \Bigg[ \frac{x_{t-k}\pi_k}{\sum_{j=0}^d
  x_{t-j}\pi_j} (p_{t-k}-p_t) \Bigg]. 
\end{equation}
\end{proposition}

The proof of \cref{prop:OracleBias} is given in \cref{apx:OracleBias}.
Intuitively, the well-specified bias can be interpreted as a weighted average of the difference between trailing and current severity rates. 
The weights depend on the delay distribution and primary incidence curve.
\cref{fig:wellspecified} provides an accompanying illustration.     

\begin{enumerate}
\item \textbf{Changes in severity rate}. The central component of this bias
  expression is the difference $p_{t-k}-p_t$. When the severity rate is constant
  over the $d$ preceding time points, the convolutional ratio is unbiased 
  (because this difference is zero). This falls in line with the motivation used
  to derive this estimator, as explained in the last subsection. But when 
  severity rates change before $t$, these difference terms will be nonzero, in
  which case the estimator will be generally biased. 
  \cref{fig:toy_hfr} shows a simple example of this: the estimated
  severity rates are most inaccurate in periods where the true rate is changing
  quickly.        

  To make matters worse, the bias is in the opposite direction of the trend we 
  want to detect; for example, suppose the severity rate is monotonically
  falling, with $p_t < p_{t-1} < \dots < p_{t-d}$. The bias will then be
  positive, meaning the ratio estimates do not decline with the true rate. In
  fact, the estimated severity may even rise, not fall. Conversely, when true
  severity rates are rising, the estimates will be too low. 

\item \textbf{The delay distribution}. How much the changing severity rates
  impact the bias depends on the shape of the delay distribution $\pi$. In
  general, the bias will be greatest when the delay distribution has a long
  enough tail to upweight significant differences in severity rate. While the
  effect of the delay distribution's shape may appear subtle, \cref{sec:results} highlights its
  surprisingly large effects. The simple example in \cref{fig:toy_delay}
  also shows significant differences in bias between shorter and longer delay
  distributions.     

\item \textbf{The primary incidence curve.} Changing primary incidence $x_t$ 
  will also affect the bias, presuming the severity rate changes roughly 
  monotonically in the recent past. Intuitively, this up- or down-weights the
  terms $x_{t-k}\pi_k (p_{t-k} - p_t)$ at times further from the present, which  
  are likely to contribute the most bias. In general, falling primary incidences
  will amplify the bias, whereas rising events will minimize it. Figure
  \cref{fig:toy_primary} provides an illustration of this phenomenon.   
\end{enumerate}

\begin{figure}[htb]
\centering
\begin{subfigure}[b]{0.325\linewidth}
  \centering
  \includegraphics[width=\linewidth]{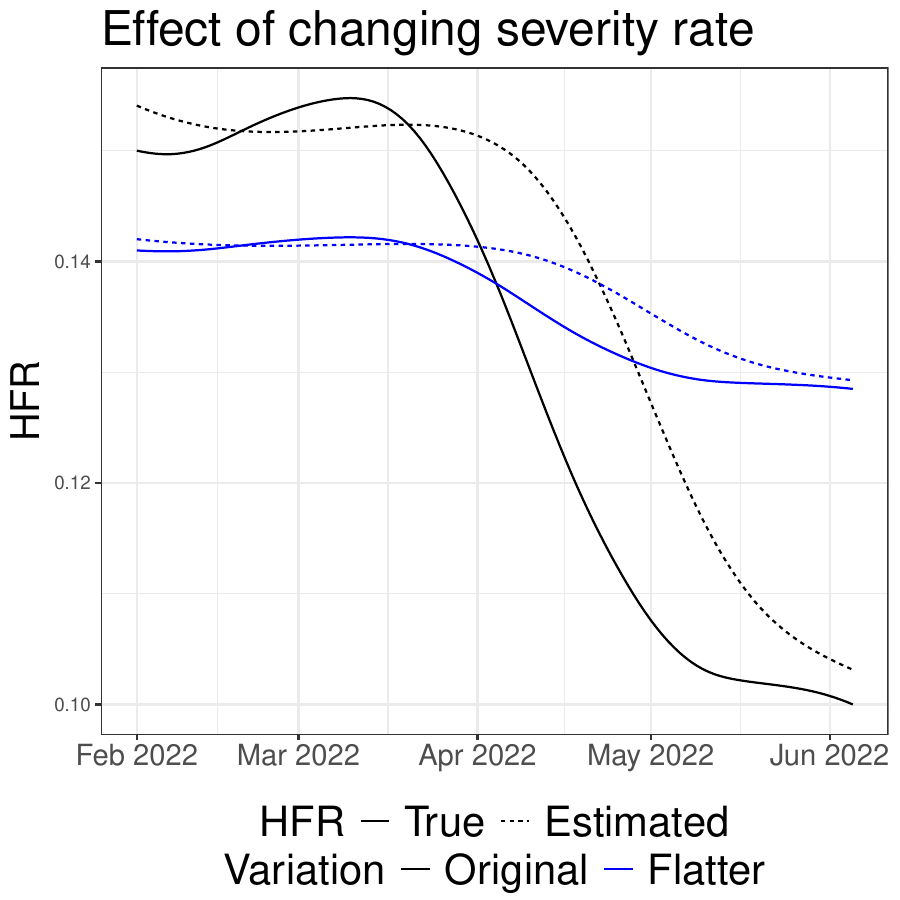} 
  \caption{}
  \label{fig:toy_hfr}
\end{subfigure}
\begin{subfigure}[b]{0.325\linewidth}
  \centering
  \includegraphics[width=\linewidth]{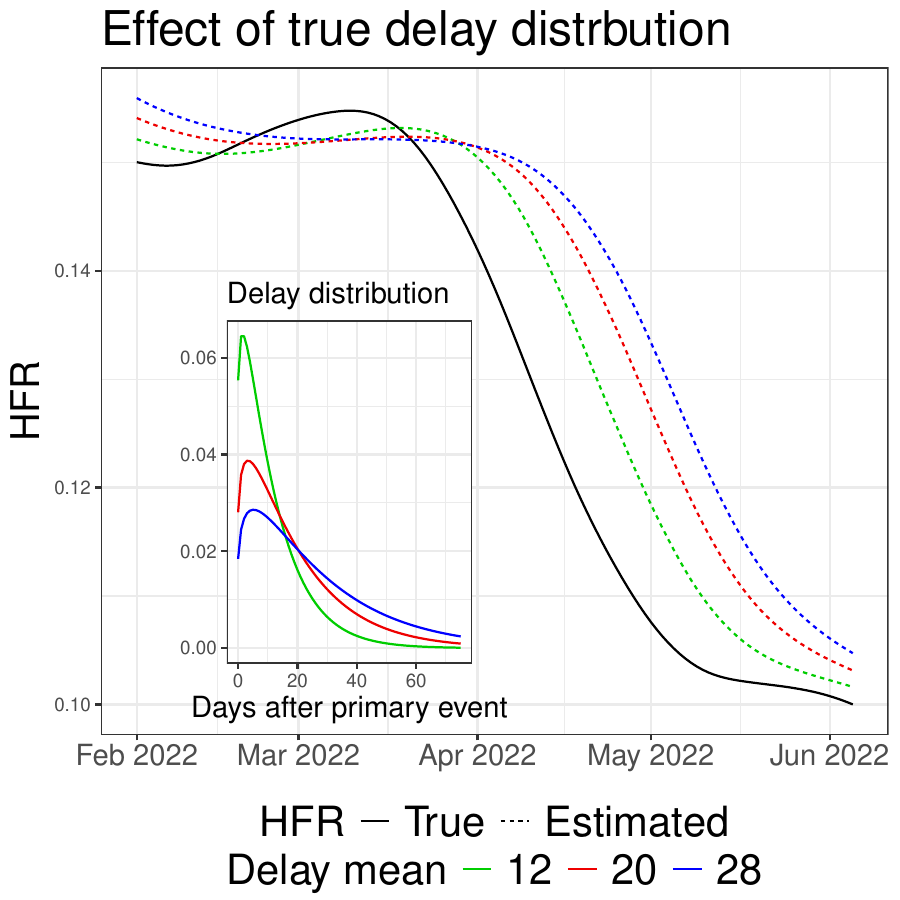}
  \caption{}
  \label{fig:toy_delay}
\end{subfigure}
\begin{subfigure}[b]{0.325\linewidth}
  \centering
  \includegraphics[width=\linewidth]{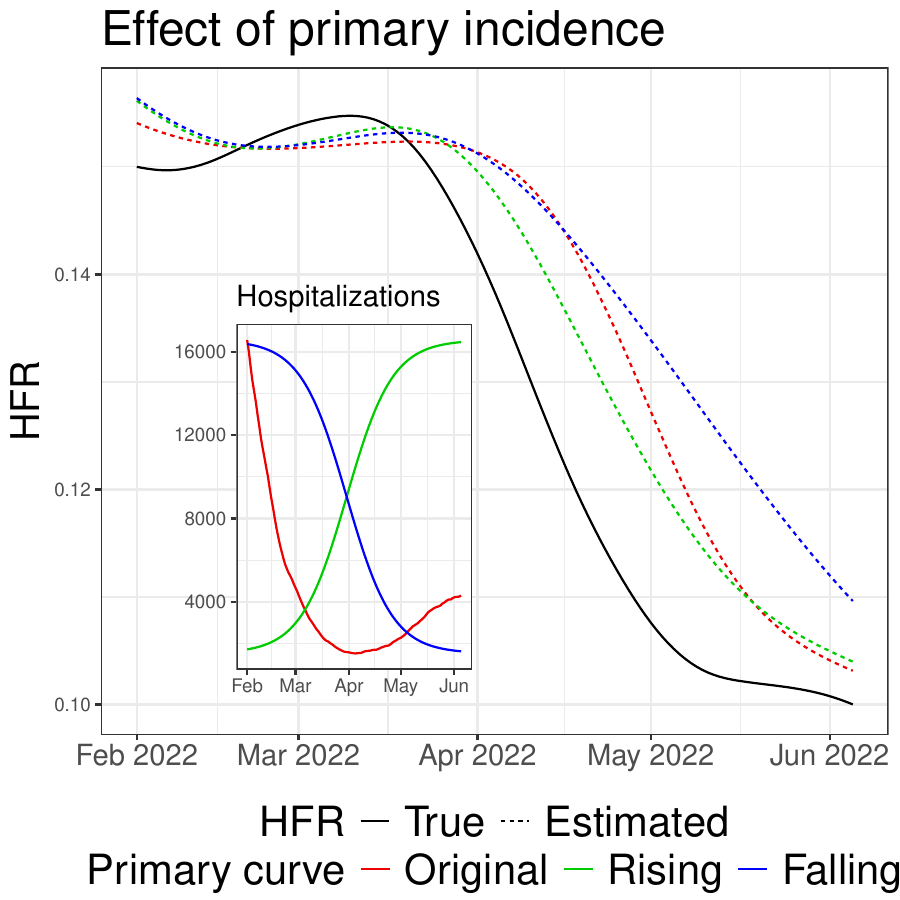} 
  \caption{}
  \label{fig:toy_primary}
\end{subfigure}
\caption[Three simulations explaining the bias of the well-specified convolutional ratio.]{Simple examples which illustrate the effects of the three factors
  explained above on 
  the bias of the well-specified convolutional ratio \eqref{eq:OracleBias}.
  In each figure, the colors correspond to different versions of the factor of interest.
  The primary incidence curve measures COVID-19 hospital admissions, as
  reported to the HHS in early 2022. We then simulate a secondary incidence
  curve, COVID-19 deaths, from \eqref{eq:model}, without noise. The underlying  
  HFR curve $p_t$ and delay distribution $\pi$ used (in simulating deaths) were 
  derived from external data sources, as explained in more detail in
  \cref{sec:setup}. }      
\label{fig:wellspecified}
\end{figure}

\cref{apx:analysis} provides further analysis, by discussing 
simplified settings in which the bias \eqref{eq:OracleBias} described in
\cref{prop:OracleBias} itself simplifies in elucidating ways.   

\subsection{Misspecified analysis}
\label{sec:misspecified}

We now analyze the bias of the convolutional ratio \eqref{eq:conv} for an
arbitrary distribution $\gamma$. Recall that $\pi$ denotes the true delay
distribution in \eqref{eq:model}. We refer to the present scenario as the
\emph{misspecified} case, as $\gamma$ may differ from $\pi$.    

\begin{proposition}
\label{prop:MispBias}
The bias of the convolutional ratio \smash{$\hat{p}_t^\gamma$} (where the true 
delay distribution $\pi$ is unknown, and the working delay distribution $\gamma$
is arbitrary, but also supported on $d$ time steps) is
\begin{equation}
\label{eq:MispBias}
\bias(\hat{p}_t^\gamma) = A_t^\gamma \bias(\hat{p}_t^\pi) + p_t (A_t^\gamma-1),  
\end{equation}
where \smash{$A_t^\gamma = \sum_{j=0}^d x_{t-j}\pi_j \,\big/\, \sum_{j=0}^d 
  x_{t-j} \gamma_j$}. This compares how the delay distributions convolve against  
the most recent primary incidence levels. 
\end{proposition}

The proof of \cref{prop:MispBias} is given in \cref{apx:MispBias}.
Before delving into the technical details, we first summarize its implications for the lagged ratio, which as a special case of \eqref{eq:conv} with a point-mass delay distribution is itself a misspecified convolutional ratio.
When primary incidence $\ell$ days ago is less than the average count weighted by the delay distribution, 
the lagged ratio is typically above the well-specified convolutional ratio. 
Otherwise, when counts $\ell$ days ago exceed the delay-weighted average, 
it tends to lie below.
This leads to the following heuristics, which we justify subsequently.
\begin{enumerate}
    \item \textbf{Heuristic \#1:} When primary incidence is climbing rapidly, the lagged ratio tends to overestimate severity, at least relative to the well-specified convolutional ratio.
    \item \textbf{Heuristic \#2:} As incidence falls swiftly, the lagged ratio may drop just as sharply, even if true severity has not changed.
    \item \textbf{Heuristic \#3:} Once incidence stabilizes after a decline, the lagged ratio can exhibit sudden, counterintuitive jumps.
\end{enumerate}

Under misspecification, the proposition gives an
additive decomposition \eqref{eq:MispBias} of the convolutional ratio bias,
based on the well-specified bias \smash{$\bias(\hat{p}_t^\pi)$} (as studied in
\cref{prop:OracleBias}), and a misspecification factor
\smash{$A_t^\gamma$}. At the outset, we note that if $\pi = \gamma$ (no  
misspecification), we have \smash{$A_t^\gamma = 1$}, and \eqref{eq:MispBias}
reduces to the well-specified bias. Generally, values of \smash{$A_t^\gamma >
  1$} amplify the well-specified bias (hereafter, the \emph{oracle bias}) and add positive misspecification bias \smash{$p_t
  (A_t^\gamma-1) > 0$}; meanwhile, values of \smash{$A_t^\gamma < 1$} shrink the
oracle bias and add negative misspecification bias \smash{$p_t (A_t^\gamma-1) <
  0$}. 

Whether or not misspecification contributes a larger magnitude of bias overall
hence depends on whether or not the sign of the misspecification term \smash{$p_t 
  (A_t^\gamma-1)$} agrees with the sign of the oracle bias
\smash{$\bias(\hat{p}_t^\pi)$}. This need not always be the case, though in our
experience, it is often true in both real and simulated experiments, as we will
see in \cref{sec:results}. Here, to gain more insight, we study the
behavior of the bias in three settings: 
\begin{itemize} 
\item Smooth $\gamma$ with a lighter tail and smaller mean than $\pi$ (more mass
  concentrated at recent time points).    
\item Smooth $\gamma$ with a heavier tail and larger mean than $\pi$ (less mass
  concentrated at recent time points).    
\item Nonsmooth $\gamma$, with a point mass at lag $\ell$; we reiterate that in 
  this case the convolutional ratio reduces to the lagged ratio
  \smash{$\hat{p}_t^\ell$} in \eqref{eq:lagged}. We also note that its
  misspecification factor \smash{$A_t^\gamma$} reduces to a quantity we
  similarly denote \smash{$A_t^\ell = \sum_{j=0}^d x_{t-j}\pi_j \,\big/\,
    x_{t-\ell}$}, and its bias \eqref{eq:MispBias} reduces to: 
  \begin{equation}
  \label{eq:LagBias}
  \bias(\hat{p}_t^\ell) = \frac{\sum_{j=0}^d x_{t-j}\pi_j}{x_{t-\ell}}
  \bias(\hat{p}_t^\pi) + p_t \Bigg(\frac{\sum_{k=0}^d x_{t-k}\pi_k}{x_{t-\ell}}
  - 1 \Bigg).  
\end{equation}
\end{itemize}

We discuss the behavior of the bias in these three settings, as a function of
the primary incidence curve (which drives bias through the misspecification
factor \smash{$A_t^\gamma$}). \cref{fig:misspecified} provides an
accompanying illustration.

\paragraph{Primary incidence rising.} 

Consider the case where primary events are rising---first slowly, then rapidly
before leveling off. A lighter-tailed $\gamma$ will place more weight
on recent time points, with higher counts, than $\pi$; thus \smash{$\sum_{j=0}^d  
  x_{t-j}\gamma_j > \sum_{j=0}^d x_{t-j}\pi_j$}, so \smash{$A_t^\gamma <
  1$}. The opposite occurs for a heavier-tailed $\gamma$: this
places more weight on distant low-count time points and less on the ongoing
surge, so \smash{$A_t^\gamma > 1$}. A point mass distribution $\gamma$ places
all of its mass on $x_{t-\ell}$, and during the steepest phase of the rise, this
is considerably less than $x_t$. The true delay $\pi$ distributes mass across
the  $d$ most recent time points, and a large fraction of its mass will be
convolved against the last $\ell-1$ time points, whose counts exceed
$x_{t-\ell}$. The times before $t-\ell$ have less of an offsetting effect, 
because incidence has risen at a  growing rate. Hence, during a surge we will 
see \smash{$x_{t-\ell} < \sum_{j=0}^d x_{t-j}\pi_j$}, and \smash{$A_t^\ell >
  1$}. However, the behavior of \smash{$A_t^\ell$} will be generally more
erratic than \smash{$A_t^\gamma$} for a smooth distribution $\gamma$, as the
denominator in the former is less smooth as $t$ varies. 

\cref{fig:misspecified} visualizes this as hopitalizations (primary
events) rise between December 2021 and mid-January 2022. Throughout this period
\smash{$A_t^\gamma$} is below/above 1 for the light-tailed/heavy-tailed
$\gamma$. Meanwhile, \smash{$A_t^\ell$} spikes to 1.25 in early January.  
Correspondly, the lagged HFR rises to 20\% when the true one drops to 15\%.    

\begin{figure}[p]
\centering
\includegraphics[width=0.9\linewidth]{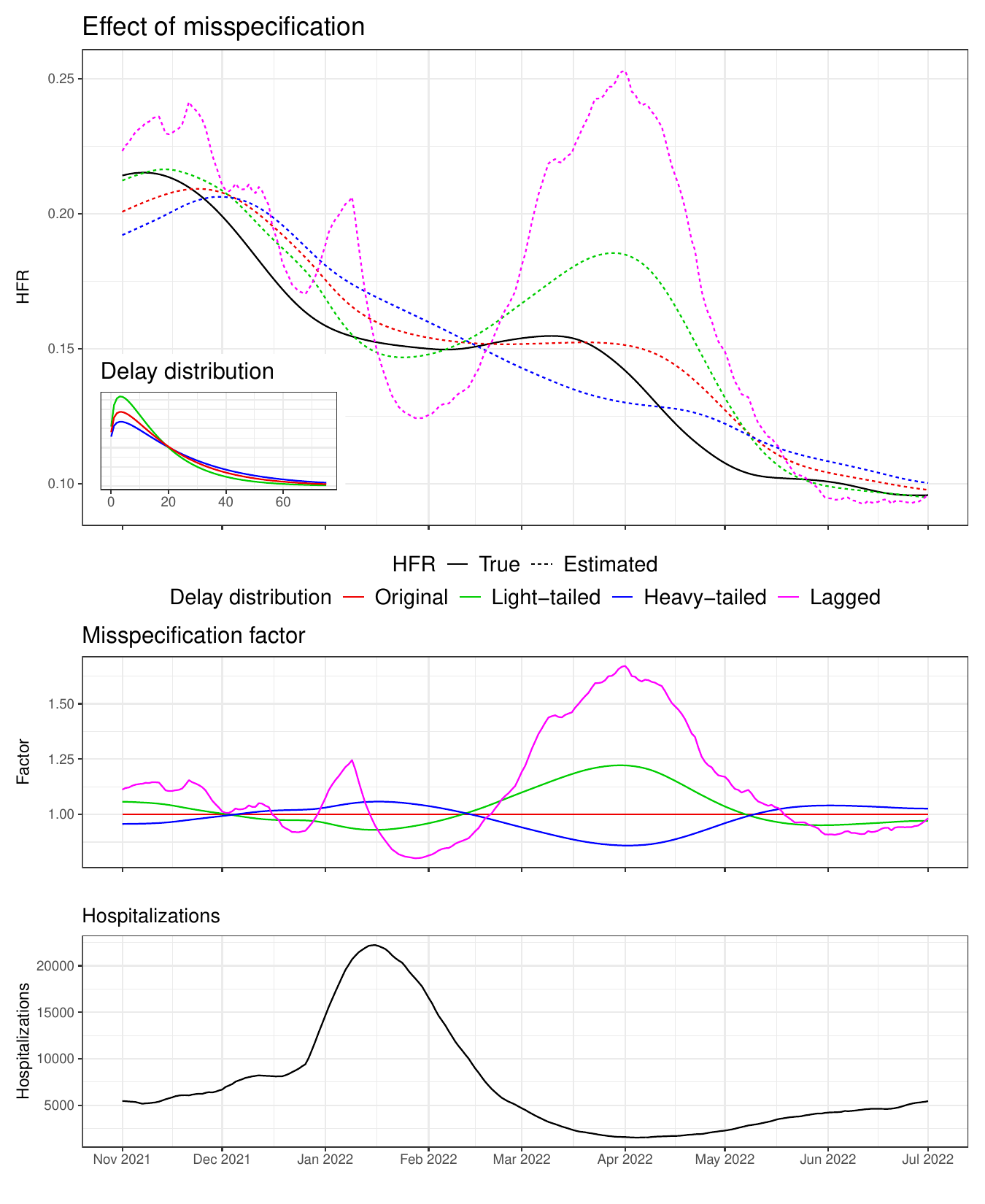}
\caption[Examples of convolutional ratio estimates under misspecification of the
  delay distribution.]{Examples of convolutional ratio estimates under misspecification of the
  delay distribution.
  As in \cref{fig:wellspecified}, the primary events are COVID-19
  hospitalizations, as reported to the HHS, and secondary events are deaths 
  simulated noiselessly from \eqref{eq:model}. The underlying HFR curve $p_t$
  and delay distribution $\pi$ used in the simulation were fit using external
  data sources detailed shortly in \cref{sec:setup}. The lagged ratio
  estimator used $\ell=16$, chosen to maximize cross-correlation between
  hospitalizations and deaths.} 
\label{fig:misspecified}
\end{figure}

\paragraph{Primary incidence falling.}

Next assume primary incidence reaches a maximum and begins to fall. The smooth
distributions behave much the same as when incidence was rising. The light-tailed
distribution has more mass around the peak than $\pi$, so \smash{$A_t^\gamma < 
1$}. Conversely, \smash{$A_t^\gamma > 1$} for the heavier-tail distribution
because it convolves more mass before the top of the rise. The lagged bias
changes its behavior in this period; while \smash{$A_t^\ell$} had exceeded 1
before the peak, it quickly plunges below 1. At exactly $\ell$ time points after  
the peak, the lagged estimator attains the smallest possible value of
\smash{$A_t^\ell$}, as $x_{t-\ell}$ maximizes its denominator. Again, the
lagged ratio is likely to have larger fluctuations of \smash{$A_t^\ell$},
since its denominator reaches extremes that are not witnessed in
\smash{$A_t^\gamma$} (the convolution in the denominator of 
\smash{$A_t^\gamma$} acts as a smoother).   

In \cref{fig:misspecified}, we can see \smash{$A_t^\ell$} drop below 0.8
near the start of February 2022; this happens precisely $\ell=16$ days after
daily new hospitalizations peak above 20,000 in mid-January. The lagged 
ratio falls from 20\% to 12.5\% accordingly, with the true HFR remains roughly
constant, hovering around 15\%. In the same period, the convolutional ratios
(with light- or heavy-tailed $\gamma$) stay quite close to the true HFR.   

\paragraph{Primary incidence levels out from a fall.}

The most jarring instance of misspecification bias occurs as primary incidence
levels out. The true delay distribution $\pi$ has a heavier tail than the
low-mean, light-tailed distribution $\gamma$. It also has a heavier tail than
the point mass distribution, which has no tail at all. This has important
implications as the peak of the surge fades into the past. Compared to $\pi$,
the light-tailed $\gamma$ and point mass distribution convolve little to no mass
with the high-count period of the wave. As a result, both \smash{$A_t^\gamma$}  
and \smash{$A_t^\ell$} rise above 1, and severity rate estimates spike. The
magnitude of this spike depends how quickly primary incidence is changing. 

\cref{fig:misspecified} displays this false spike. Around the start of April 2022,
we see \smash{$A_t^\gamma$} (for light-tailed $\gamma$) and $A_t^\ell$ reach 
maximums near 1.25 and 1.68, respectively. Their corresponding HFR estimates reach
18\% and 25\% while the true HFR has fallen below 14\%. This is of course highly 
problematic as it signals a rise in severity at a very counterintuitive time,
when hospitalizations are at their lowest. The heavy-tailed delay $\gamma$ has
the opposite trend and underestimates the true HFR at this time, but by a
smaller amount. 

\section{Empirical Setup}
\label{sec:setup}

Here we describe the data and general experimental setup used in
\cref{fig:wellspecified,fig:misspecified}, and in
\cref{sec:results}.    

\paragraph{Hospitalization-fatality rate.}

Our experiments primarily analyze the HFR throughout the COVID-19 pandemic. While HFR 
may be less common as an object of study compared to CFR, it has a few
advantages. 
Firstly, we were able to find a good ``ground truth proxy'' for the national HFR 
during the COVID-19 pandemic, as published by the National Hospital Care Survey
(NHCS). This helps guide our simulations and also serves as validation data for
us, as we describe in more detail below. 

A second advantage is that hospitalization reporting was much more complete
than case reporting throughout the pandemic. Hospitals were mandated to report
new daily COVID-19 admissions to the Department of Health and Human Services
(HHS) \citep{HHS2023}. Due to changes in case ascertainment over time (cases as
a fraction of infections), it is harder to interpret the CFR in a time-varying
fashion, i.e., harder to understand what precisely this is reflecting over the
course of the pandemic.  

Thirdly, hospitalization counts published by the HHS are
aligned by admission date. This makes it more meaningful to interpret the HFR as
a reflection of severity, especially in a time-varying fashion. In comparison,
case counts as aggregated by John Hopkins University (JHU) \citep{JHU} (the
central resource for comprehensive COVID-19 case data in the US) are aligned by
report date. Extreme reporting delays (sometimes cases were reported 45 days
after infections, see, e.g., \citealp{Jahja2022}) make the CFR less meaningful
to study as a time-varying quantity, even outside of ascertainment issues.  



To assess the generality of our results, \cref{apx:CFR} conducts the simulated analyses using CFR in place of HFR. 
In addition to \cref{fig:sims} (forthcoming in the results section), Figures \ref{fig:wellspecified} and \ref{fig:misspecified} are reproduced using case data. 
In all instances, the resulting trends mirror those for HFR, confirming that the bias mechanisms identified here apply broadly across severity metrics. 

\paragraph{Aggregate data streams.}

To estimate the real-time HFR, we use aggregate counts of daily COVID-19 
hospitalizations and deaths as made available in the Epidata API
\citep{Epidata}, developed by the Delphi Group. Like HHS for hospitalizations,
the JHU Center for Systems Science and Engineering (CSSE) provided the
definitive resource for real-time death counts during the pandemic. These counts
reflect times at which deaths were reported to health authorities, not
necessarily when they actually happened. Hence raw JHU death counts are highly
volatile due to reporting idiosyncrasies like day-of-week effects and data
dumps. Hospitalizations are also subject to strong day-of-week effects. We
thus smooth all data with a 7-day trailing average, for both hospitalizations
and deaths.   

Our real-time estimates of HFR actually use data that was available two days
after the date in question. This was done to account for a typical two-day
latency in the most recent data available. In this sense, one can actually view
our real-time estimates as a two-day backcast of the HFR. In the rare event that
counts were still unavailable at a two-day lag, we imputed their values with the
most recently observed data (this is a common scheme, called
last-observation-carried-forward or LOCF).        

\paragraph{Hyperparameters.}

The ratio estimators of the HFR require choices of the lag $\ell$ and delay
distribution $\gamma$. The experiments in \cref{sec:results} use a lag of
$\ell=20$ days, which roughly maximizes the cross-correlation between
hospitalizations and deaths over the entire pandemic. For $\gamma$, we use a
discrete gamma distribution, and set its support length to be $d=75$ days, a
conservative choice. For its mean, we use 20 again; this agrees nicely with a UK
analysis that finds a median hospitalization-to-death time of 11 days
\citep{UKdelay},\footnote{Of course, conditions in the UK may be quite different
  from the \US. However, we rely on the UK study because it provides the most
  comprehensive information on COVID-19 hospitalization-to-death delay 
  distributions.}   
and a CDC analysis that finds 63\% of COVID-19 deaths are reported in 10 days   
\citep{cdc_deaths_demographic_geographic_2023}. We set the standard deviation to 
18, because the delay distributions fit by the UK study had standard deviations
that were roughly 90\% of their means. 

\cref{apx:robustness} evaluates
the robustness of findings against different hyperparameter values. 
We show our findings do not change across a wide range of lags (\cref{fig:lag}) and delay distributions (\cref{fig:delays}).

\paragraph{Validation data.}

Due to a lack of ground truth, the bias cannot be rigorously evaluated on real data. 
However, there are sound ways to better approximate the true HFRs from external sources.
One way is to use estimates from the National Hospital Care Survey (NHCS)
\citep{NHCS2023}, which records weekly HFR based on a representative subset 
of 601 hospitals across the \US. These estimates end up being consistently
biased downwards because they are only based on deaths which occur in the
hospital. A CDC analysis \citep{ahmad2023covid} found that roughly 60\% of
COVID-19 deaths occurred in hospitals in 2022, down from nearly 70\% in 2021 and
2022. To account for non-inpatient deaths, we divide the NHCS estimates by these
percentages. Lastly, we smooth the resulting HFR estimates with a spline, via
the \texttt{smooth.spline} function in R (which chooses the smoothness
hyperparameter by minimizing generalized cross-validation error). This results 
in our proxy for the ground truth HFR curve $p_t$.  

This external estimate is a useful benchmark to judge the fidelity of our HFR
estimates. Of course, it is not perfect, and is derived from a relatively small
subset of hospitals. Results in \cref{sec:results} suggest it may be too
high in late 2022. 
\cref{apx:alt_gt} discusses alternative
approximations to the ground truth HFR. 
These alternatives, shown in \cref{fig:approxGT}, all exhibit similar qualitative behavior. 

\section{Empirical Results}
\label{sec:results}

We study the performance of the ratio estimators in greater depth, on real and 
simulated data. Throughout, we continue to use COVID-19 hospitalizations
reported to the HHS as the primary incidence curve. 
Code to reproduce our experimental results is available at \texttt{https://github.com/jeremy-goldwasser/Severity-Bias}.


\subsection{COVID-19 data}
\label{sec:results_real}

\cref{fig:basic_est_vs_gt_figs} displays real-time HFR estimates from
November 2020 to December 2022, which spans the major COVID-19 waves. These
estimates were fit to the data described in \cref{sec:setup}. The
real-time hospitalization and death counts exhibit a fair degree of instability,
and recall, these were preprocessed with a 7-day trailing average. Even then,
the convolutional \eqref{eq:conv} and lagged ratio \eqref{eq:lagged} estmators
each had wild spikes, so we further smoothed the estimates from both methods
with a 7-day trailing average. 
\cref{apx:robustness} studies the
sensitivity of the results to the choice of smoothing window in postprocessing.
\cref{fig:window} shows window size has very little qualitative effect on he main shape of the HFR estimates from either method. 

\begin{figure}[t!]
\centering
\begin{subfigure}[b]{\linewidth} 
  \centering
  \includegraphics[width=\linewidth]{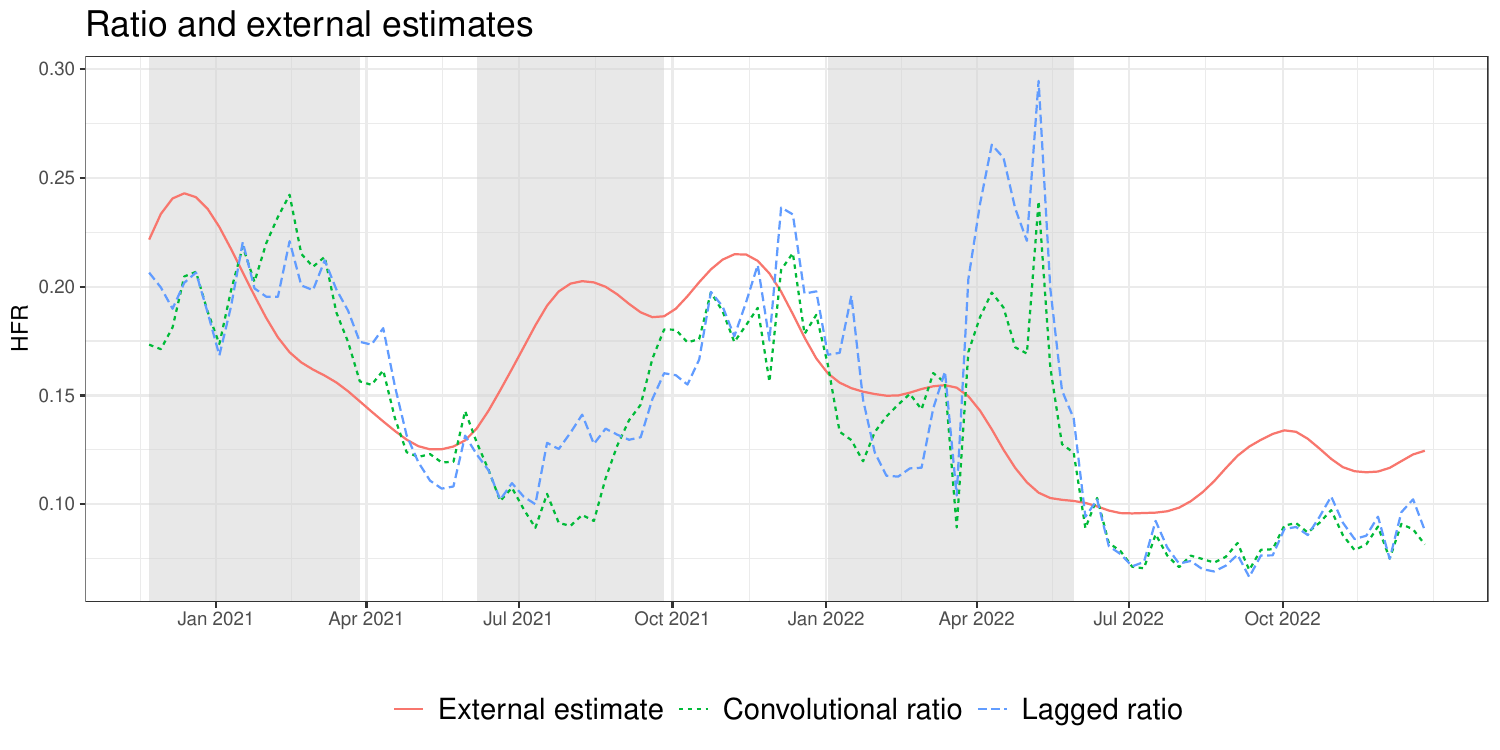}
  \caption{Convolutional and lagged ratios against external estimate of
    HFR. Waves corresponding to the original, Delta, and Omicron variants are highlighted in gray. } 
\end{subfigure}

\bigskip
\begin{subfigure}[b]{\linewidth}
  \centering
  \includegraphics[width=\linewidth]{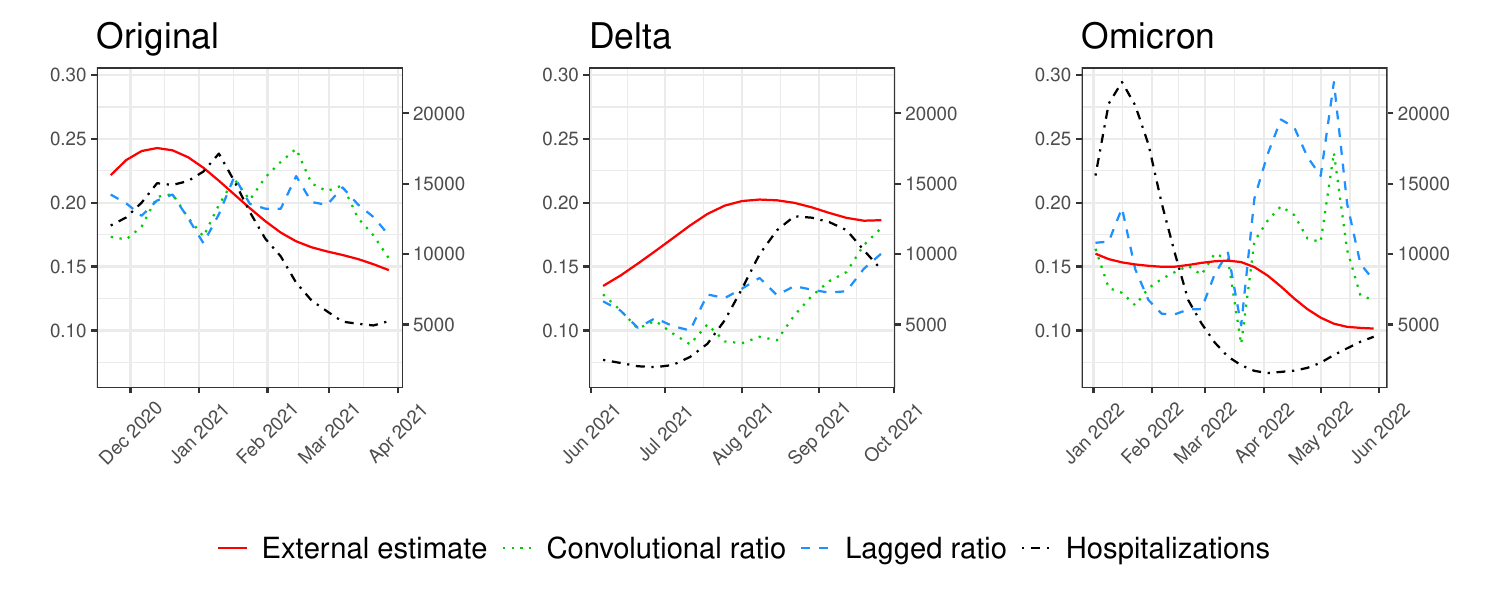}
  \caption{Zooming in to focus on major variants, with hospitalizations overlaid
    (right y-axis).}
\end{subfigure}

\caption[Comparing ratio and external (NHCS-based) estimates of HFR.]{Comparing ratio and external (NHCS-based) estimates of HFR. Computed using real-time
  COVID-19 counts in the \US, from November 2020 through December 2022.}
\label{fig:basic_est_vs_gt_figs}
\end{figure}

Overall, both ratio estimators perform poorly---their bias is consistent and 
nontrivial, especially for the lagged estimator. Both respond very slowly to
changes in the HFR. As the HFR declines following the wave in winter 2021, both   
ratios hover near 20\% for several months. More troublingly, they are too slow
to detect the rise in HFR in the early Delta period (summer 2021). If the
purpose of these estimators is to inform stakeholders of increased risks in real
time, they failed during the Delta surge.  

The most significant bias comes in the middle of the Omicron wave in spring
2022. In this period, the HFR remains around 15\% until April, then sharply
declines to 9\% two months later. The ratio estimates first fluctuate around the
true HFR, and then subsequently, both estimates surge as the true HFR nears
its nadir, with the lagged ratio approaching 30\%. This dramatic upswing signals
a serious false alarm.  The analysis in Sections \ref{sec:wellspecified} and
\ref{sec:misspecified} explain each of these failure cases. 

\paragraph{Well-specified analysis.}

We start by analyzing the convolutional ratio with respect to the well-specified
bias expression in \cref{prop:OracleBias}. While this expression
assumes that the true delay distribution is known, we found that different
choices of delay distribution generally yield similar bias; see \cref{fig:delays} in \cref{apx:robustness}.
This indicates that our convolutional ratio may 
not be far from the oracle (well-specified) ratio. 

\cref{prop:OracleBias} indicates that the bias moves in the opposite
direction of the true severity rate. This occurs during the Delta wave, when the
HFR rises well before the ratio estimates do. On the other hand, falling HFR
produces positive bias, as observed in the original and Omicron waves.  

The enormity of the bias during Omicron can partially be attributed to the
precipitous decline in hospitalizations, as falling primary incidence has been
shown to exacerbate the bias. Average daily hospitalizations declined from over
20,000 in mid-January to only 1,500 by April 1. Lastly, the delay distribution
is relatively heavy-tailed, because the aggregate deaths here (from JHU) are
aligned by report date. 
We find that this has a substantial impact on the bias, 
as analyzed in \cref{apx:NCHS_deaths}.
\cref{fig:jhu_vs_nchs} shows a large gap in bias between ratios computed on NCHS versus JHU deaths.

\paragraph{Misspecified analysis.}

The misspecification analysis explains central discrepancies between the
convolutional and lagged ratios. \cref{sec:misspecified} discusses why we 
expect \smash{$A_t^\ell < 1$} around the start of a decline in primary
incidence (Heuristic \#2). As a result, the lagged ratio will incur negative misspecification
bias, so it takes lower values than the convolutional ratio. 
We observe this when hospitalizations with the original variant
decline from their peak in mid-January 2021. Throughout February 2021, lagged
estimates are about 2\% below the positively biased convolutional ratios.   

As primary incidence rises, we expect \smash{$A_t^\ell > 1$}, contributing
positive bias relative to the well-specified
ratio (Heuristic \#1). Correspondingly, when hospitalizations surge due to the Delta variant in 
August 2021, the lagged ratio is less negatively biased. Lastly, recall we
expect \smash{$A_t^\ell > 1$} \emph{after} a fall in primary incidence (Heuristic \#3). This
accounts for the lagged ratio having higher bias in April 2022, when
hospitalizations level out from the Omicron surge. There, the lagged HFR (and 
presumably also the misspecification factor \smash{$A_t^\ell$}) hits its maximum
value.   

\paragraph{}

We performed several robustness checks to assess the stability of these
findings. 
\cref{apx:robustness} compares the ratio HFRs to multiple external estimates of the ground truth (\cref{fig:approxGT}). 
It also studies the effect of using finalized count data (\cref{fig:rt_and_final}) and different hyperparameter choices (Figures \ref{fig:window}-\ref{fig:delays}).
Finally, it examines the ratio estimators' behavior on six \US\ states (\cref{fig:state-level}). 
By and large, the ratio estimators yield roughly
the same type of bias throughout.

\subsection{Simulated data}
\label{sec:results_sim}

We further evaluate the ratio estimators in a variety of simulation settings. 
Keeping the primary incidence $x_t$ as hospitalizations reported to the HHS, we simulate deaths based on the convolutional model \eqref{eq:model} without noise. That is, we generate
\[
y_t = \E[y_t\given x_{\leq t}] = \sum_{k=0}^d x_{t-k} \Pprob(\text{die at $t$} \given \text{hospitalized at
  $t-k$}) = \sum_{k=0}^d x_{t-k} \pi_k p_{t-k},
\]

The simulations explore three different underlying HFR curves and two delay distributions. The delay distribution $\pi$ is a discrete gamma with
standard deviation 90\% of its mean. We consider means of 12 and 24 to
compare short and long distributions.
For the HFRs $p_t$, we first use the external estimates given by NHCS
In addition, we mimic the opposite trend by inverting and rescaling this curve. 
We do so with the following formula, preserving the minimum and maximum HFRs:
\[
p^\text{Inv} = \frac{1}{p}\cdot\frac{\text{min}(p)}{\text{max}(p)}
\]

The third HFR setting is a stationary $p=10\%$ over all time.
This case elucidates the quantity \smash{$A_t^\ell$} that drives the lagged ratio's bias. This value is \smash{$A_t^\ell = \hat{p}_t^\ell/p$}, the lagged ratio itself scaled by a constant. 
To see this, recall that the oracle convolutional ratio is unbiased under stationarity, so \cref{prop:MispBias} simplifies to
\smash{$\E[\hat{p}_t^\ell \given x_{\leq t}] = p A_t^\ell$} for the
lagged ratio. Furthermore, \smash{$\E[\hat{p}_t^\ell \given x_{\leq t}]  
  = \hat{p}_t^\ell$} in this noiseless simulation, so \smash{$A_t^\ell = \hat{p}_t^\ell / p$}. 
$A_t^\gamma$ also is $\hat{p}_t^\gamma / p$ for the well-specified convolutional ratio, since the estimator is unbiased and $A_t^\pi=1$ by definition. 

\begin{figure}[t!]
\centering
\includegraphics[width=\linewidth]{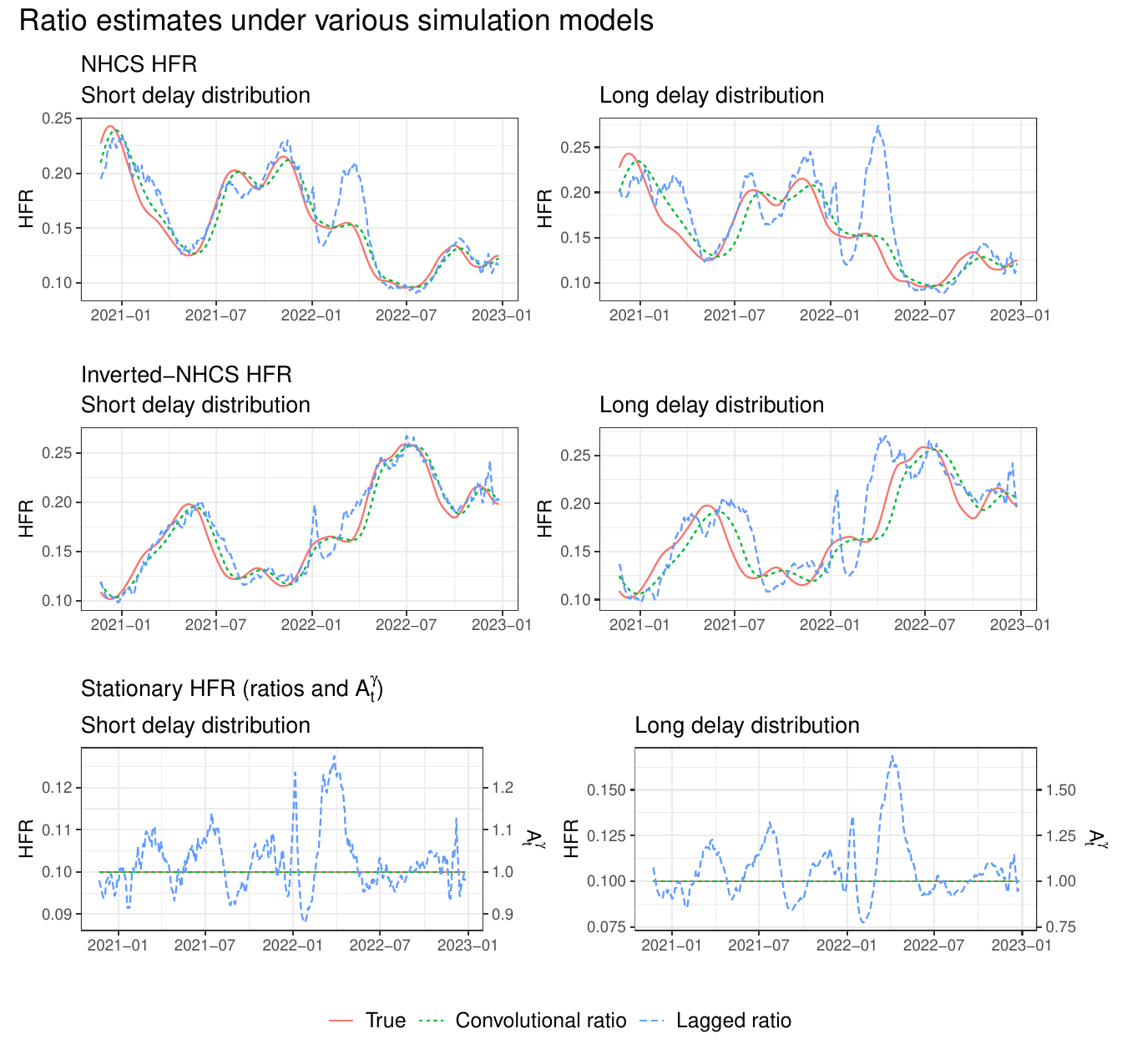}
\caption[Convolutional and lagged ratios on simulated data with various HFR curves and delay distributions.]{Convolutional and lagged ratios over various simulation settings, with
  three different underlying HFR curves (rows) and two delays (columns). The last row, for stationary HFR, also provides $A_t^\gamma = \hat{p}_t^\gamma/p$.}
\label{fig:sims}
\end{figure}

To estimate the severity rates, the convolutional ratio is well-specified with $\gamma=\pi$.
For the lagged ratio, we choose the lag $\ell$ by maximizing the
cross-correlation between hospitalizations and deaths. 
In both cases, we do not smooth the data in any capacity. (The same is true for Figures \ref{fig:wellspecified} and \ref{fig:misspecified}.)
 
\cref{fig:sims} shows the results across the six settings in total. As
expected, HFR estimates are significantly more biased when the underlying delay
distribution is longer. This bias is most pronounced with the lagged ratio. For example, even when the true HFR is a constant 10\%, the lagged ratio estimates
hit 13\% under the light-tailed delay, and 15\% under the heavy-tailed one. In
the NCHS HFR setting, it spikes (as we have seen before) as hospitalizations
level out in spring 2022, reaching over 20\% and 25\% under the light- and 
heavy-tailed delays, respectively. Note the convolutional ratio does not share
these dramatic oscillations. By and large, it tracks the general shape of the
true curve, albeit at a delay.
Its overall favorable performance rests on the fact that we have chosen in this
simulation to provide it with the benefit of the true delay distribution (no
misspecification).   

The results in \cref{sec:misspecified} explain the wide
gap in performance between these estimators. 
As anticipated, the lagged ratio is higher than the oracle convolutional
ratio when \smash{$A_t^\ell > 1$}, and lower when \smash{$A_t^\ell <
  1$}. Comparing \smash{$A_t^\ell$} to the estimated HFR curves, the bias moves
very similarly. 
For example, during the Delta and Omicron waves, rapid rises in
hospitalizations produced high values of \smash{$A_t^\ell$}. This accounts for
the spikes in August 2021 and January 2022 (Heuristic \#1). 
Heuristic \#2 dictates the lagged estimator should have lower bias than the well-specified ratio as primary events
fall. We observe this in Delta (September 2021) and Omicron (Februrary 2022).  
Lastly, when hospitalizations level out from
the Omicron surge, \smash{$A_t^\ell$} spikes to 1.2 and 1.5 for the short
and long distributions. This explains the positive bias in spring 2022, per Heuristic \#3.

The bias \eqref{eq:MispBias} rescales the oracle bias and adds a
misspecification term. Studying \cref{fig:sims}, we observe the
misspecification term tends to dominate when \smash{$A_t^\ell$} strays away from 
1. To understand this, consider periods in which the oracle bias is negative. As
introduced in \cref{sec:misspecified}, the oracle and misspecification
terms are at odds with each other when this is the case. Invariably, the lagged
ratio moves in the direction of the misspecification term
\smash{$p_t(A_t^\ell-1)$}. In the NHCS HFR setting, for example, the lagged 
estimates spike with $A_t^\ell$ in August 2021. In the inverted setting, the 
lagged bias tracks the down-up-down motion of \smash{$A_t^\ell$} during the
first five months of 2021. That the misspecification term wins out in these
conflicting settings indicates it comprises a disproportionate amount of the
bias. Indeed, the oracle bias is low enough that multiplicative rescaling must 
not have a large effect.  

\cref{apx:CFR} repeats these analyses for CFR, using case data in lieu of hospitalizations (Figures \ref{fig:wellspecified_cfr}-\ref{fig:sims_cfr}). It shows the same qualitative behavior, with the well-specified convolutional ratio remaining decently calibrated and the lagged ratio displaying inaccurate, erratic swings. This indicates that the underlying bias mechanisms generalize directly to CFR and other severity rates.

\section{Discussion}

Our analyses and experiments illustrate that practitioners should take caution
when using standard ratio estimators for time-varying severity rates. They
exhibit nontrivial bias as severity rates change, particularly the popular
lagged ratio. If a major purpose of such estimators is to inform stakeholders of
changing risks in real time, then such bias indicates they may fail to do so in
a reliable manner.
While our main analyses focus on HFR, we additionally replicate the simulated experiments for CFR, confirming that ratio estimates of other severity rates share the same bias patterns.

Analyzing the bias, as we have done in \cref{prop:OracleBias} and 
\cref{prop:MispBias}, allows us to form real-time heuristics about what to
expect in practice. For example, based solely on the primary incidence curve, we
generally expect the lagged ratio to make the following errors: 
\begin{itemize}
\item unreasonably high severity estimates when primary incidence is rising
  quickly (Heuristic \#1); 
\item rapid declines when primary incidence is falling quickly (Heuristic \#2); 
\item unexpected surges when primary incidence has leveled out after falling (Heuristic \#3).  
\end{itemize}
Practitioners may be able to adjust their reactions accordingly. For example, if 
the lagged CFR spikes shortly after hospitalizations have declined and reach a
stable low point, then a savvy epidemiologist can temper their alarm with the
knowledge it may well be spurious.  

While the lagged ratio seems ubiquitous in practice, the convolutional ratio
(when  a reasonable estimate of the delay distribution can be formed) can be
better behaved and should probably be favored. While it is still subject to
bias, this tends to be of a smaller magnitude. 

Going beyond, there is still room to improve upon the backward-looking
convolutional ratio. 
A promising forward-looking severity estimator was proposed by \citet{fusedlasso}. 
This method obtains all historical severity rates $p_t$ at once, by minimizing  
\[
\sum_t (y_t - \sum_{j=0}^d x_{t-j}\gamma_j p_{t-j})^2 + 
\lambda \sum_t |p_t - p_{t-1}|,
\]
\noindent where $\lambda \geq 0$ is a parameter that controls the level of
regularization. The regularizer above is a total variation penalty, also called the fused lasso; it produces a piecewise constant fit of the severity rates.
Unlike the convolutional ratio, this method models the relationship between events without assuming severity rates are locally stationary. 
Therefore, it may be a less biased alternative.
Since the severity rates are defined implicitly via optimization, their bias is analytically intractable and must be assessed empirically.

Unfortunately, while this method was introduced as a real-time tool, it struggles with instability at the most recent timesteps.
Improving its capabilities for real-time estimation, 
and extending it to fit smoother severity rates using
trend filtering penalties (see, e.g., \citealp{Tibshirani2014}), are interesting
directions for future work. \citet{Jahja2022} applied trend filtering to a
similar deconvolution problem, reconstructing latent infections from case
reports. Their insights on tail regularization may be useful to stabilize
real-time severity estimates.

Severity rates may be biased in ways beyond the statistical bias our work
focuses on. In \cref{sec:setup}, we mentioned that HFR estimation from
line-lists can be subject to ``survivorship bias'': the failure to account for 
deaths occurring outside the hospital \citep{lipsitch2015potential}.
Under-reporting is another central challenge, particularly for CFR. Not all
infections are reported, reporting rates change across time, and severe cases
are more likely to be reported than mild cases \citep{Tsang2021}. \citet{reich2012estimating}
proposed an estimator for a time-invariant \emph{relative} CFR---the ratio of
CFRs between groups---which learns latent reporting rates via the EM
algorithm. \citet{anastasios} applied this in the context of COVID-19. Their
work also identifies other sources of bias, like differences in case definition
and testing eligibility.  
Finally, examining how hyperparameter tuning strategies affect the bias poses an opportunity for future work.

\chapter{Adaptive Deconvolution with Trend Filtering}
\label{ch:paper2}



Public health authorities use severity rates to track the deadliness of a disease over the course of an epidemic. Typically, metrics like the case-fatality are estimated by dividing the number of new deaths by the number of recently reported infections. However, recent work reveals that this simple calculation can be seriously misleading, which poses large problems for decision-making during epidemics. When severity rates are underestimated, officials can miss genuine rises in risk and postpone critical interventions, and thus more people in harm’s way. Conversely, when these metrics are overestimated, authorities may enact needless restrictions and stoke undue public fear, wasting resources and eroding trust.

To address this, we propose methods to estimate time-varying severity rates, both in retrospect and real time. Our methods maximize a novel characterization of the likelihood, deconvolving the time series of severity rates that relate primary and secondary events (e.g. cases and deaths). The likelihood is regularized by a trend filtering penalty, which fits piecewise polynomials with adaptively selected knots in order to balance smoothness with strong local adaptivity. 

On a range of experiments, our methods consistently outperform the two leading benchmarks, reducing mean absolute error by roughly 15\% and 55\%. Moreover, they exhibit stronger qualitative behavior, with both smoother fits and less bias. These results support our class of estimators as a promising new approach to track public health risks as they unfold.

\section{Introduction}

\cref{ch:paper1} studied standard real-time methods for estimating severity rates, 
\begin{equation}\label{eq:severity2}
p_t = \Pprob(\text{secondary event will occur} \given \text{primary event at $t$}).
\end{equation}
We demonstrated
that these ratio-based methods can be subject to large, predictable statistical biases. 
In this chapter, we propose new methodology that seeks to overcome these challenges.

Severity rates can be estimated
\emph{retrospectively}, where data collected after time $t$ is used to estimate  
$p_t$. Retrospective analyses can be valuable, as they can provide insights into 
epidemic dynamics amidst changing conditions. Often of greater interest,
however, is estimating severity rates in \emph{real time}, which means only data
available up until $t$ may be used to estimate $p_t$. This is far more
challenging, and expanding \eqref{eq:severity2} helps explain why: we have 
\smash{$p_t = \sum_{k=0}^{\infty} \Pprob(\text{secondary event occurs at $t+k$}
  \given \text{primary event at $t$})$}, which depends on epidemic events (of 
potentially changing likelihood) in the future, yet real-time estimators for
$p_t$ may only use data through $t$, and none afterwards.

We introduce a new approach, tailored for both retrospective and real-time estimation. 
Our methods are centered around a generative model for secondary events (e.g. deaths),
that depends on the primary event counts (e.g. hospitalizations), delay distribution, and severity rates.
This Poisson binomial distribution is a likelihood model for the 
time series of severity rates. 
We estimate these latent variables via deconvolution,
maximizing an approximate likelihood for computational tractability.
The likelihood is regularized by a trend filtering penalty, which fits piecewise polynomials with adaptively selected knots in order to balance smoothness with strong local adaptivity. 
We add further regularization for the real-time case, lifting techniques from \citet{Jahja2022}.
Evaluating on real and semi-synthetic COVID-19 data, we find our method
consistently outperforms the ratio-based methods,
including under misspecification.

The rest of this article is structured as follows. After previewing and
motivating our methodology with a real data application comparing HFR estimates
in \cref{sec:real-data}, we develop in \cref{sec:stat-model} the
Poisson-binomial model, and corresponding approximations used for
deconvolution. In \cref{sec:estimators}, we introduce our approaches for
estimating severity rates, first in the retrospective setting and then in
real-time. In \cref{sec:setup2}, we describe our experimental setup, and
in \cref{sec:results2}, we analyze its results.
\cref{sec:discussion} concludes with a discussion.

\subsection*{Notation}

Notation for $x_t$, $y_t$, $p_t$, and $\pi_k^{(t)}$ is shared with \cref{ch:paper1} (see \cref{tab:notation-ch1}). Additional symbols introduced in this chapter are listed in \cref{tab:notation-ch2}.

\begin{table}[!ht]
\centering
\caption{Additional deconvolution notation (Chapter~3).}
\label{tab:notation-ch2}
\renewcommand{\arraystretch}{1.15}
\begin{tabularx}{\linewidth}{@{}l >{\raggedright\arraybackslash}X@{}}
\toprule
\textbf{Symbol} & \textbf{Meaning} \\
\midrule
\multicolumn{2}{@{}l}{\textsc{Distributional quantities}} \\
\addlinespace[2pt]
$\mu_t$ & Conditional mean $\E[y_t \given x_{\leq t}]$ \\
$\sigma_t^2$ & Conditional variance $\Var(y_t \given x_{\leq t})$ \\
\midrule
\multicolumn{2}{@{}l}{\textsc{Estimation}} \\
\addlinespace[2pt]
$D^{(m)}$ & Finite difference operator of order $m$ \\
$\lambda$ & Trend filtering regularization parameter \\
$\gamma$ & Tail-smoothness hyperparameter (real-time setting) \\
\bottomrule
\end{tabularx}
\end{table}

\subsection{Motivating application: time-varying HFR in COVID-19}
\label{sec:real-data}  

In this subsection, we compare our deconvolution approach to the existing
ratio-based methods in estimating the COVID-19 hospitalization-fatality rate (HFR) 
over the pandemic in the state of Pennsylvania. The HFR is an example of
a severity rate \eqref{eq:severity2} of key interest in public health, where primary
events are hopsitalizations and secondary events are deaths. Instead of
estimating a single number for the entire COVID-19 pandemic, we estimate the HFR
as a function of time---allowing the HFR to change in response to changing
underlying conditions, such as the introduction of new therapeutics, or the 
emergence of new variants.      

We consider two different settings in which we estimate the HFR curve: the
retrospective setting, where we use data that was only available in hindsight to
estimate the HFR at each time $t$ as best as possible; and the real-time
setting, where we limit ourselves to data available at each time $t$ in order to
estimate the HFR at $t$. Our proposed method, which performs regularized
deconvolution in a likelihood model, is described in \cref{sec:retro} for in the retrospective setting,
and \cref{sec:real-time} for real-time.
Existing ratio-based methods are described in
\cref{sec:methods2}.

The data used in our experiments in this subsection are from standard public
health reporting pipelines. For primary events, we use daily COVID-19 hospital
admissions, as coordinated by the National Healthcare Safety Network (NHSN). For
secondary events, we use different datasets for the retrospective and real-time
cases. In retrospect, we use daily COVID-19 death counts,\footnote{This data is 
  only available at the weekly level, hence we imputed daily deaths by sampling
  from a multinomial distribution with the weekly NCHS death total as the number
  of trials and uniform probabilities (all equal to $1/7$).}  
as collected by the National Center for Health Statistics (NCHS). In real-time,
we use daily COVID-19 deaths, as assimilated by Johns Hopkins University
(JHU). NCHS deaths were only available weeks after each event in question, while
JHU published provisional death counts in real-time. These data sources are
described in more detail in \cref{sec:data-generation}. 

The estimators compared here all depend on a delay distribution, representing
the (stochastic) transition from hospitalizations to deaths. We take this to be
a discretized gamma distribution (or a point mass, for the lagged ratio method),
whose mean is chosen separately in the retrospective and real-time cases by  
maximizing cross-correlation between the relevant observables, in the same
manner as described later in \cref{sec:data-generation}. 
All methods have hyperparameters, which we tune via cross-validation and the
one standard error (1se) rule, as described in \cref{sec:design}. 
We favor the 1se rule because the min rule qualitatively tends to undersmooth, and even with this smoothness-favoring tuning, the ratio-based methods still exhibit the instability highlighted below.
For the
deconvolution method, we choose quadratic order ($m=2$) regularization in the
retrospective case, because this yields a smooth HFR curve, and constant order 
($m=0$) in the real-time case, because this displays the most stable performance 
across our experimental suite to come later.  

\begin{figure}[t]
\centering
\includegraphics[width=0.95\textwidth]{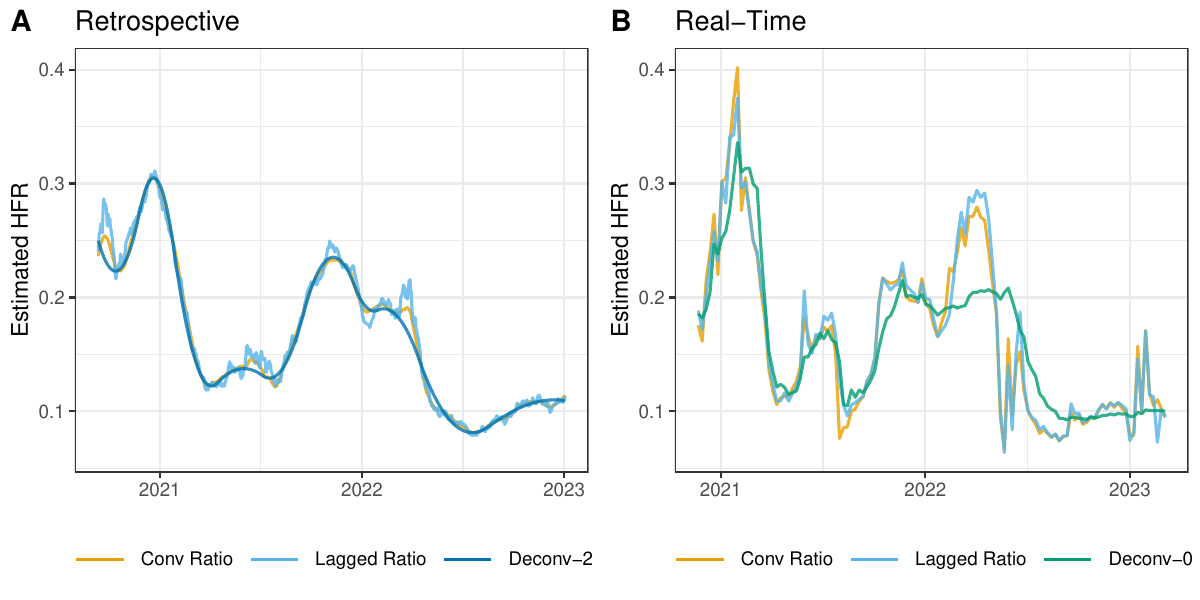}
\caption{HFR estimates on real COVID-19 data in Pennsylvania.}
\label{fig:real-data}
\end{figure}

\cref{fig:real-data} displays the retrospective and real-time severity
rate estimates. The high-level summary---across both settings (retrospective and 
real-time) and all methods---is as follows. We see HFRs peaking at around 30\%
during the surge around the beginning of 2021, then falling below 15\% as the
surge ended. Conditions remained relatively steady with the Alpha variant in
spring 2021, then climbed rapidly with the Delta variant later that year. They
dipped to around 20\% as the Omicron variant swept through the state in winter
2022. As the wave passed, severity rates fell sharply, bottoming out around
10\%. 

Comparing estimates to each another, the deconvolution estimates demonstrate
notable stability relative to the ratio-based estimates. This is true in both
retrospective and real-time cases, but especially pronounced in the real-time
case, displayed in panel B. In this case, the ratio estimates spike to about
40\% at the peak of the winter wave in early 2021, while the deconvolution
estimate stays near 30\% (consistent with all estimates in the retrospective
case). Moreover, the ratio estimates suspiciously \emph{increase} as the Omicron
wave passes in April 2022, whereas the deconvolution estimate holds steady over
the same period, before smoothly declining (again broadly consistent with the
retrospective estimates).

There is no ground truth HFR in these real data experiments against which we can
measure our estimates, to ultimately determine which is more accurate. However,
the suspicious behavior of the ratio-based methods in \cref{fig:real-data}
is consistent with the prior analysis in \cref{ch:paper1}.
For example, that work reveals how the real-time estimators are upwardly biased during periods of sharp decline in the primary counts, including early 2022. 
Furthermore, in Section \ref{sec:setup2} and \ref{sec:results2} of 
this paper, we carry out semi-synthetic experiments which start from real
primary events (COVID-19 hospitalizations), and then simulate secondary events 
(COVID-19 deaths) from custom models which are tailored to match both the trend
and noise levels in observed secondary events (NCHS and JHU deaths in the
retrospective and real-time settings, respectively). 
This allows us to recreate the qualitative differences between the estimators seen in \cref{fig:real-data} and quantify the degree to which the proposed deconvolution approach improves over ratio-based alternatives.

\section{Statistical model}\label{sec:stat-model}

In this section, we present a model which relates secondary to primary event
counts, via severity rates, and describe an approximate likelihood suitable for
inference. Terminology is shared with \cref{ch:paper1}.

\subsection{Exact likelihood}\label{sec:exact-lik}

At a time point $t$, let $x_t$ and $y_t$ be the aggregate numbers of
primary and secondary events, respectively. In the context of CFR, for example,
$x_t$ would be the number of new cases at time $t$, and $y_t$ the number of new 
deaths. Throughout we consider discrete integer-valued time points such as 
$t=1,2,3,\dots$, but allow negative time values, for notational simplicity (for  
example, when indexing the set of time points that precede a given time $t$, as 
we do below).   

In general, the number of secondary events at time $y_t$ can be expressed as a
sum of indicator functions, where each indicator represents whether a given
primary in the past event resulted in a secondary event at $t$. Using cases and
deaths as primary and secondary events, for concreteness, observe that conditional on the primary counts $x = \{x_s\}$,
\begin{equation}
\label{eq:secondary_incidence}
y_t \given x = \sum_{k=0}^{\infty} \sum_{i=1}^{x_{t-k}} \mathds{1}\{\text{case $i$
  at time $t-k$ dies at $t$}\}.
\end{equation}
The secondary incidence time series can be understood probabilistically by
noting that, conditional on $x$, the indicator functions are Bernoulli random variables. Applying the
definition of conditional probability,
\begin{align}\label{eq:expectation}
&\E[\mathds{1}\{\text{case $i$ at $t-k$ dies at $t$}\} \given x] \\ 
&\quad = \Pprob(\text{dies at $t$} \given \text{case at $t-k$})\nonumber \\
&\quad = \Pprob(\text{death occurs $\cap$ dies at $t$} 
  \given \text{case at $t-k$})\nonumber \\ 
&\quad = \Pprob(\text{death occurs} \given \text{case at $t-k$}) \cdot
 \Pprob(\text{dies at $t$} \given \text{case at $t-k$, death occurs}).\nonumber
\end{align}
The first term of the final line is the severity rate $p_{t-k}$, from
\eqref{eq:severity2}. Meanwhile, we define the second term to be the $k\nth$
element of the \emph{delay distribution} at time $t-k$. Formally, define  
\begin{equation}
\pi_k^{(t)} = \Pprob(\text{dies at $t+k$} \given \text{case at $t$, death occurs}) 
\end{equation}
as the $k\nth$ element of the delay distribution at time $t$, where
\smash{$\sum_{k=0}^\infty\pi_k^{(t)} = 1$}. 
Thus by definition the expectation \eqref{eq:expectation} is equal to $p_{t-k} \pi_k^{(t-k)}$, so we may write   
\begin{equation}
\label{eq:indicator_bernoulli}
\mathds{1}\{\text{case $i$ occurring at time $t-k$ dies at $t$}\} \sim
\mathrm{Bernoulli}(p_{t-k} \pi_k^{(t-k)}).
\end{equation}
It is reasonable to assume that these Bernoulli random variables are
independent across $i$. Each represents a different patient, whose outcomes should not affect
one another. 
An exception could arise if severity rates are calculated within a very small, resource-constrained group---for example, if one patient's death frees up hospital resources for others, marginally improving their prognosis \citep{nursing}.
In general, however, severity rates are
typically assessed over a broad population, at the resolution of counties or 
states or even countries---not single hospitals or nursing homes. Henceforth, we make the
assumption that individual patient outcomes are independent. 
This assumption of independent Bernoulli outcomes is shared by all other works on severity rates we compare to \citep{UKpaper,germany,ghani,timevar_ifr}.

Under this assumption, note that \smash{$y_t \given x_{\leq t}$} follows a
\emph{Poisson-binomial} distribution, where we abbreviate \smash{$x_{\leq t} =
  (x_1, \ldots, x_t)$}. The Poisson-binomial generalizes the binomial
distribution, in that it represents a sum of independent Bernoulli random variables that do not necessarily share the same success probability. 
In the context of severity rates, each Bernoulli trial is whether an individual patient has a secondary event time $t$, as expressed in \eqref{eq:secondary_incidence}. 
Note by \eqref{eq:indicator_bernoulli} that all individuals with primary events on the same dates have the same ``success'' probability; this is $\pi_{k}^{(t-k)}p_{t-k}$ for all $x_{t-k}$ trials at time $t-k$.
Thus we may concisely write  
\begin{equation}
\label{eq:pb}
y_t \given x_{\leq t} \sim \mathrm{PoissonBinomial} \Big( 
\underbrace{\pi_0^{(t)}p_{t},\dots,\pi_0^{(t)}p_{t}}_{x_{t}}, \, 
\underbrace{\pi_1^{(t-1)}p_{t-1},\dots,\pi_1^{(t-1)}p_{t-1}}_{x_{t-1}}, \,
\dots \Big),
\end{equation}
where the underbraces indicate the number of times each success probability is
repeated. A useful fact to record will be the mean and variance of this
distribution:
\begin{align}
\label{eq:pb_mean}
\mu_t &= \E[y_t \given x_{\leq t}] = \sum_{k=0}^\infty x_{t-k} \pi_k^{(t-k)}    
p_{t-k}, \\
\label{eq:pb_var}
\sigma_t^2 &= \Var(y_t\given x_{\leq t}) = \sum_{k=0}^\infty
  x_{t-k}\pi_k^{(t-k)} p_{t-k} (1-\pi_k^{(t-k)} p_{t-k}).  
\end{align}
which follow by linearity of expectation and independence.

\subsection{Approximate likelihood}\label{sec:approx}

\paragraph{Approximations to the Poisson-binomial.}

The Poisson-binomial likelihood is generally intractable for large counts. 
As with the binomial distribution, for the probability mass function at 
a realized count $y_t$, it is necessary to consider all \smash{$\sum_{k=0}^\infty x_{t-k}$}
choose $y_t$ combinations of patients that may yield this realized number of
secondary events. For the Poisson-binomial, however, these combinations do not
share the same probability of occurring, since the Bernoulli variates in
\eqref{eq:pb} have different success probabilities. As a result, evaluating the
probability mass function requires enumerating all products of $y_t$ Bernoulli 
probabilities. This combinatorial explosion is computationally prohibitive
unless the total number of primary events is very small, which is not the case
in our application, where we are concerned with COVID-19 cases or
hospitalizations for \US\ states, or across the whole nation.

Fortunately, different approximations exist for the Poisson-binomial
distribution. A natural choice is the Poisson distribution with rate $\mu_t$
from \eqref{eq:pb_mean}. This has nonnegative, albeit unbounded, support. It has
variance \smash{$\sum_{k=0}^\infty x_{t-k} \pi_k^{(t-k)} p_{t-k}$}, which is
larger than $\sigma_t^2$ in \eqref{eq:pb_var}, but this difference matters
little in practice, at least in our applications of interest. For example, using
numbers from the national experiments we run in \cref{sec:setup2}, the
Poisson variance is only bigger by about 1.28\%. More formally,
\citet{barbour1984poisson} showed that the Poisson approximation to the
Poisson-binomial incurs small error when the Bernoulli random variables have low
success rates. This is generally true for severity rates in epidemics, where the
probabilities \smash{$\pi_k^{(t-k)} p_{t-k}$} are small. By the results of
\citet{barbour1984poisson}, the total variation distance between the
Poisson-binomial in \eqref{eq:pb} and its corresponding Poisson approximation
(with rate $\mu_t$) is upper bounded by
\begin{equation}
\sum_{k=0}^{\infty}(\pi_k^{(t-k)}p_{t-k})^2.
\end{equation}
Again employing data from our \US\ simulation, the above bound is only 
$1.69 \times 10^{-5}$. The negligible TV distance justifies approximating the distribution of $y_t \mid x_{\leq t}$ as Poisson.

Alternatively, \citet{pb_theory} studies a normal approximation to the
Poisson-binomial, with mean $\mu_t$ and variance $\sigma_t^2$. Their empirical
and theoretical results show the normal approximation more faithfully captures
the upper tail of the cumulative density function than the Poisson. However, it 
struggles with the lower tail when the mean is low, since it includes negative
support.  

\paragraph{Weak dependence at successive time points.}

Conditional on past counts of primary events, we have shown that secondary event
counts follow a Poisson-binomial distribution. Strictly speaking, these
secondary event counts need not be conditionally independent at successive time
points, since they are defined over a common group of individuals. For example,
if case $i$ from $t-k$ dies at time $t$, then this same case cannot die at time
$t+1$. Fortunately, the next result demonstrates that the counts from successive
time steps are weakly dependent, under a simplifying assumption of equal
variance.  

\begin{proposition}\label{prop:corr}
Assume \smash{$\Var(y_t \given x_{\leq t}) = \Var(y_{t+1}\given x_{\leq
  t+1})$}. Under the Poisson-binomial model described in
\eqref{eq:secondary_incidence}, \eqref{eq:indicator_bernoulli}, \eqref{eq:pb},
it holds that
\begin{equation}
\Cor(y_t, y_{t+1} \given x_{\leq t+1}) \in \bigg[\frac{-\max_{k \geq 0}
  \pi_k^{(t-k)} p_{t-k}}{1-\max_{k \geq 0} \pi_k^{(t-k)} p_{t-k}}, \, 0 \bigg].    
\end{equation}
\end{proposition}

\cref{apx:corr} contains the proof of this result. Given reasonably
long-tailed delay distributions with low severity rates, each
\smash{$\pi_k^{(t-k)} p_{t-k}$} will be quite small. Plugging in numbers
corresponding to the national data considered in \cref{sec:setup2}, the 
lower bound on the correlations from \cref{prop:corr} is roughly
-0.018. Computing the actual correlations from our simulated data, the lowest
observed value is -0.012. As these numbers are so low, it seems reasonable to
approximate the distribution of \smash{$y_t \given x_{\leq t}$} as being
independent over different time points $t$, for the sake of estimation.   

While our likelihood characterization for severity rates is novel, several works in compartmental models and reproduction numbers also assume aggregate counts on different timesteps are independent \citep{cori2013new, rtestim, cauchemez2008likelihood, wallinga_teunis}. 


\paragraph{Joint distribution and MLE.}

We have established that \smash{$y_t \given
x_{\leq t}$} may be approximated by Poisson or Gaussian distributions, which are nearly conditionally independent 
as $t$ varies. Next, we characterize the joint distribution over $t$, leaving regularization to the next
section. Let $p$ denote the vector of severity rates, which has coordinate $p_t$
at time $t$.  By approximating the law of \smash{$y_t \given x_{\leq t}$} as
Poisson with rate $\mu_t = \mu_t(p)$ as in \eqref{eq:pb_mean}, and assuming 
independence of these conditional distributions over $t$, we arrive at the
following maximum likelihood problem: 

\begin{align}
&\maximize_{0 \preceq p \preceq 1} \; \prod_t \frac{\mu_t(p)^{y_t}e^{-\mu_t(p)}}  
  {y_t!} \\
\label{eq:mle-pois}
\iff &\minimize_{0 \preceq p \preceq 1} \; \sum_t \bigg[ \bigg(\sum_{k=0}^\infty     
  x_{t-k}\pi_k^{(t-k)}p_{t-k}\bigg) - y_t\log\bigg(\sum_{k=0}^\infty
  x_{t-k}\pi_k^{(t-k)}p_{t-k}\bigg)\bigg],  
\end{align}
where the notation $0 \preceq z \preceq 1$ denotes elementwise constraints.  
This optimization problem falls in the class of \emph{Poisson linear
  inverse} problems, well-studied in statistical physics and image processing 
\citep{richardson1972bayesian, lucy1974iterative, dupe2011linear,
  rond2016poisson}. For a Poisson linear inverse problem to be computationally
tractable, the linear operator being applied to the optimization variable must
have nonnegative entries. That is indeed the case here, as each of the terms 
\smash{$x_{t-k} \pi_k^{(t-k)}$} are nonnegative.      


We note that the identity-link parameterization used in \eqref{eq:mle-pois} contrasts with log-link Poisson regression, which is a generalized linear model that may be more familiar to many statisticians. In Poisson regression, the log of the rate is a linear function of the features; in this context, one would learn parameters $\beta$ such that $\log\mu_t = \sum_{k=0}^\infty x_{t-k}\pi_k^{(t-k)} \beta_{t-k}$. 
However, solving for $\beta$ is not useful for our purposes of estimating severity rates, as
$p$ cannot be recovered from a given $\beta$ in 
\begin{equation}
\log \bigg( \sum_{k=0}^\infty x_{t-k}\pi_k^{(t-k)}p_{t-k} \bigg) =
\sum_{k=0}^\infty x_{t-k}\pi_k^{(t-k)}\beta_{t-k}.
\end{equation}
By contrast, the identity link makes $\mu_t$ linear in $p$, so the likelihood \eqref{eq:mle-pois} can be optimized over $p$ directly. This problem is underspecified, but adding sufficient regularization renders the problem tractable.

In addition to the Poisson model described above, we also studied and 
implemented a severity rate estimator using the normal approximation. While both    
estimators performed roughly similarly well in our experiments, the Poisson
model was simpler (note that the severity rates $p$ appear in the variance in 
\eqref{eq:pb_var}, which complicates the Gaussian approximation), and had
marginally higher accuracy. For that reason, we only provide a high-level
account of the Gaussian methods in the main text, delegating further discussion
to \cref{apx:gauss}.

\paragraph{Estimating the delay distribution.}
In most applications, the oracle delay distribution is unknown. In lieu of
\smash{$\pi^{(t)}$}, a plug-in estimate \smash{$\hat\pi^{(t)}$} can be used
instead, which typically has a finite support, placing all of its mass on
$[0,d]$, where $d < \infty$. In other words, this assumes that no secondary
events occur after $d$ time steps (and turns all infinite sums in
\eqref{eq:mle-pois} from $k=0$ to $\infty$ into finite sums from $k=0$ to 
$d$). In any case, whether or not the support of the delay distribution is
finite, there will always be more severity rates than secondary event times, and
so the maximum likehood problem \eqref{eq:mle-pois} is underspecified. The  
estimators we propose in \cref{sec:estimators} will thus use
regularization.  

Estimating the delay distributions \smash{$\pi^{(t)}$} is an entire line of work
in and of itself. Many approaches use line lists which contain the times of
primary and secondary events, other approaches may be based on parametric
approximations
whose parameters are based on small observational studies or even chosen based
on epidemiological literature. 
Popular choices of parametrizations include discretized gamma, Weibull, and log-normal distributions.
For a thorough account of existing methods, we
refer the reader to \citet{delay_distrs}.

\section{Severity rate estimators}\label{sec:estimators}

We now present our proposed severity rate estimators for both the retrospective
and real-time settings. We focus on the Poisson likelihood, with
\cref{apx:gauss} covering the Gaussian case. Some form of regularization is
necessary to compute likelihood-based severity rates, due to the fact that the
maximum likelihood problem (either Poisson or Gaussian) is underspecified. While
many forms of regularization would suffice to make the problem well-specified,
we use a form which aligns with the smoothness considerations underlying
severity rates, discussed in detail below. We then present our retrospective and
real-time estimators. 

For our methods to work well, the approximate likelihood outlined in \cref{sec:approx} must match the true data-generating process reasonably well. 
To review, this model assumes that individuals have independent outcomes, the Poisson (or Gaussian) approximates the Poisson-binomial accurately, $y_{t}$ and $y_{t+1}$ are conditionally independent, and the delay distribution is estimated well. 
In addition, the following estimators assume the actual severity rates are smooth. The nature of the smoothness assumption depends on the order of trend filtering regularization, introduced next.

\subsection{Trend filtering regularization}\label{sec:tf}

Trend filtering is a nonparametric estimation technique, which models the mean
function of univariate data as a piecewise polynomial where the order (e.g.,
linear or cubic) is user-specified \citep{og_tf}. While this is reminiscent of
the smoothing spline, trend filtering can adjust more responsively to the local 
level of smoothness in the underlying signal \citep{Tibshirani2014} by
adaptively selecting the locations---also called knots---at which the fitted
polynomial changes. Local adaptivity is important for estimating severity rates,
as these may remain constant for long stretches of time when epidemic conditions
do not change substantially, before changing rapidly in response to new
conditions (e.g., the arrival of a new variant or a change in vaccine coverage).
That said, smoothing splines may perform comparably in practice when severity varies gradually; we find this to be the case for $R_t$ in \cref{ch4}, where spline- and trend-filtering-based estimates are quite similar.
The main cost of trend filtering relative to splines is the difficulty of uncertainty quantification, which we discuss in \cref{sec:discussion} and revisit in the Conclusion.     

For an integer order $m\geq 0$, trend filtering in its original form (for
nonparametric regression) models the mean vector $\theta \in \R^n$ over $n$ data
points as a piecewise polynomial, of degree $m$. It does so by minimizing a loss
term while penalizing the differences of $\theta$ of order $m+1$. These
differences can be compactly represented by multiplying $\theta$ by a difference
operator $D^{(m+1)}$. This can be viewed as the discrete analog of a derivative 
operator, and is defined recursively as
\begin{equation}
D^{(m+1)} = D^{(1)}\,D^{(m)}\in\mathbb{R}^{(n-m-1)\times n}
\end{equation}
where $D^{(1)}$ is the first difference operator
\begin{equation}
D^{(1)} =
\begin{bmatrix}
-1 & 1 & 0 & \cdots & 0 & 0\\
0 & -1 & 1 & \cdots & 0 & 0\\
\vdots & & \ddots & \ddots & & \vdots\\
0 & 0 & 0 & \cdots & -1 & 1
\end{bmatrix}
\,\in\,\R^{(n-m-1)\times (n-m)}.
\end{equation}
Trend filtering applies a penalty based on the $\ell_1$ norm, $\lambda
\|D^{(m+1)} \theta\|_1$, for a tuning parameter $\lambda \geq 0$. To give
examples, for the first three orders $m=0,1,2$, the trend filtering penalties
are
\[
\|D^{(1)} \theta\|_1 = \sum_{i=1}^{n-1} |\theta_{i+1} - \theta_i|, \;\;
\|D^{(2)} \theta\|_1 = \sum_{i=1}^{n-2} |\theta_{i+2} - 2\theta_{i+1} +
\theta_i|, \;\; 
\|D^{(3)} \theta\|_1 = \sum_{i=1}^{n-3} |\theta_{i+3} - 3\theta_{i+2} +
3\theta_{i+1} - \theta_i|.
\]
Since the $\ell_1$ norm induces sparsity, trend filtering solutions
\smash{$\hat\theta$} have the property that many elements of the differenced
vector \smash{$D^{(m+1)}\hat\theta$} will be exactly zero \citep{genlasso}, with
generally more zeros for larger $\lambda$ values. The nonzero elements
correspond to knots, which are adaptively chosen based on the data. Between the
knots, the solution traces out a degree $m$ polynomial; different segments of 
\smash{$\hat\theta$} may possess more or less smoothness, depending on how close
the knots are to each other. We refer to \citet{Tibshirani2014} for more
details, and to \citet{tibshirani2022divided} for more discussion of the broader
context and relation to splines.

The special case of $m=0$ produces a piecewise constant fit and is known as the
fused lasso \citep{og_fused_lasso} or total variation denoising
\citep{rudin1992nonlinear}, and has a longer history of study. Recently, 
\citet{fusedlasso} proposed to use a fused lasso penalty to estimate severity
rates in a regression framework with squared loss. The methods detailed below
can be understood as a generalization of their work, encompassing higher-order
trend filtering penalties, a Poisson loss which stems from approximating the
Poisson-binomial likelihood underlying secondary event generation 
(\cref{sec:stat-model}), and added tail regularization to tame real-time 
estimation (\cref{sec:real-time}, below). That said, the squared loss
used in \citet{fusedlasso} is close to the Gaussian approximation we describe in 
\cref{apx:gauss}.

\subsection{Retrospective deconvolution}\label{sec:retro}

\paragraph{Estimator.}

Our retrospective method appends a trend filtering penalty to the Poisson
likelihood approximation derived in \cref{sec:approx}. In particular,
starting from \eqref{eq:mle-pois}, we use a plug-in estimate
\smash{$\hat\pi^{(t)}$} for the delay distribution, with finite support $[0,d]$,
and append a trend filtering penalty of order $m \geq 0$, which yields 
\begin{equation}
\label{eq:tf-pois}
\minimize_{0 \preceq p \preceq 1} \; \sum_t \bigg[\bigg( \sum_{k=0}^d   
x_{t-k}\hat\pi_k^{(t-k)}p_{t-k} \bigg) - y_t\log\bigg( \sum_{k=0}^d 
x_{t-k}\hat\pi_k^{(t-k)}p_{t-k} \bigg)\bigg] + \lambda\|D^{(m+1)}p\|_1. 
\end{equation}
We denote the vector of restrospective severity rate
estimates, obtained by solving \eqref{eq:tf-pois}, by
\smash{$\hat{p}^{\text{rs}}$}. 
Problem \eqref{eq:tf-pois} is convex, and can be readily optimized with standard software. Our implementation uses CVXR \citep{cvxr}, the R interface to the disciplined convex programming framework CVX \citep{gb08}. CVXR supports a variety of back-end solvers, of which we use the Clarabel optimizer for its efficiency on cone programs \citep{clarabel}. The box constraint $\ell \leq \theta \leq u$ is equivalent to a set of linear inequalities, making this a linear cone program---an optimization problem of the form $\min_x \; c^\top x$ subject to $Ax + b \in \mathcal{K}$, where $\mathcal{K}$ is a convex cone. We discuss the computational cost of this solver in \cref{apx:extra}. Alternatively, one could adapt the specialized alternating direction method of multipliers optimizers developed for trend filtering \citep{ramdas2016fast,Jahja2022}, but we do not pursue this in the current paper.

The sum in \eqref{eq:tf-pois} is over all secondary event times. Let $N_Y$
denote the number of such event times. Note that there are $N_Y+d$ estimated
severity rates, starting at $d$ time steps before the first secondary event
time. Not all of these estimates are equally reliable. In this model, the
severity rate $p_t$ contributes to secondary events in the $d$ time steps which
follow. Away from the left or right boundary of second event times, assuming
$N_Y$ is sufficiently large relative to $d$, these secondary events are all
observed. As a result, the estimate \smash{$\hat{p}^{\text{rs}}_t$} is likely to
be more stable, as it contributes to (and is hence informed by) many elements of
the loss. 

On the other hand, severity rate estimates will be less stable when $t$ is near
the left or right boundaries of observations. In the most extreme edge cases,
the first and last severity rates only affect a single secondary count. This
data scarcity issue persists throughout the $d$ time steps of each tail. As they
contribute little to the loss, severity rate estimates at the tails are in
general subject to greater variability. To account for this, we recommend (and
implement) a burn-in and burn-out period for retrospective estimation. That is,
while $N_Y+d$ severity rates are obtained, it is prudent to ignore the values
near the boundaries.

\paragraph{Cross-validation.}

The hyperparameter $\lambda \geq 0$ in \eqref{eq:tf-pois} controls the degree of
smoothness exhibited by \smash{$\hat{p}^{\text{rs}}$}. In an extreme case 
($\lambda \to \infty$), the solution will take the form of a global polynomial
of degree $m$. At the other extreme ($\lambda \to 0$), it will be a highly
volatile piecewise polynomial, having a knot at each possible time point. We
rely on $K$-fold cross-validation for selecting $\lambda$, defining folds in a
structured way that respects the time dependence in our retrospective estimation
problem. In each of the $K$ folds, secondary incidence $y_t$ at every
$K\nth$ time step $t$ is held out from the loss, reducing the number of training
samples by a factor of $1/K$.
We use this evenly-spaced holdout scheme rather than block holdout because, per \cref{prop:corr}, the conditional dependence between successive secondary counts is negligible.
The closely related generalized cross-validation (GCV) and leave-one-out CV are standard for nonparametric smoothers such as smoothing splines \citep{wahba1990spline, mgcv}.
For fold $j \in \{1,\dots,K\}$, denoting by
\smash{$\hat{p}^{\text{rs},j}(\lambda)$} the result of optimizing the
corresponding version of \eqref{eq:tf-pois} with samples withheld, we record 
the mean absolute error (MAE) from reconvolution over the validation set $V_j$,  
\begin{equation}
\text{MAE}_j(\lambda) = \frac{1}{|V_j|} \sum_{t \in V_j} \bigg| y_{t} -
\sum_{k=0}^d x_{t-k}\hat\pi_k^{(t-k)} \hat{p}^{\text{rs},j}_{t-k}(\lambda)
\bigg|.  
\end{equation}
The cross-validation error for the given $\lambda$ is the average across the $K$
folds:
\begin{equation}
\text{MAE}(\lambda) = \frac{1}{K} \sum_{j=1}^K \text{MAE}_j(\lambda).
\end{equation}
This process is iterated over a grid of $\lambda$ values, the largest of which
should ideally produce an estimate with no knots. We denote this value by
\smash{$\lambda_\text{max}$}, and derive its form in 
\cref{apx:lambda-max}. Observe that any \smash{$\lambda \geq \lambda_\text{max}$}
will result in an estimate which is a global polynomial. From the
cross-validation error curve, we select the value $\lambda^*$ of the tuning 
parameter with the lowest error, and then use this value to estimate severity
rates using the whole time series. This is a standard strategy, sometimes called
the ``min rule.'' An alternative that is useful when the cross-validation error
curve is flat around its minimum is the ``1se rule'' \citep{hastie2009elements},
which selects the largest $\lambda$ whose cross-validation error is within one
standard error of the minimum value.  

Since fewer training samples are used when tuning $\lambda$ in cross-validation,
we make a slight adjustment to \eqref{eq:tf-pois} in order to ensure that it
balances the loss and regularization terms on a common scale. Specifically,
whenever solving \eqref{eq:tf-pois} (either in cross-validation iterations or in
full, on the whole time series), we normalize each sum by its number of
summands. 

The order of trend filtering, $m \geq 0$, may also be tuned using the data. We
can compare cross-validation errors across all $\lambda$ and $m$, and select the
pair $(\lambda^*,m^*)$ with the lowest error. As we we typically restrict our
attention to $m \in \{0,1,2\}$---resulting in piecewise constant, linear, and 
quadatic trends, respectively---such additional tuning over $m$ does not present 
much of an additional computational burden. Alternatively, the user may want 
to handpick $m \in \{0,1,2\}$ based on qualitative considerations, to reflect
the desired shape of the severity rate curve: piecewise constant, allowing for
jump discontinuities; piecewise linear, allowing for sharp turnaround points; or
piecewise quadratic, allowing for smoother evolution.

\subsection{Real-time deconvolution}\label{sec:real-time}

\paragraph{Estimator.}

Denote by $T$ the time through which data is available and at which one seeks 
to estimate the current severity rate $p_T$. To be clear, in this real-time
estimation problem, no data which becomes available after $T$ may be used, since
it does not yet exist. The most basic adaption of the retrospective approach in
the previous subsection to the current real-time case would be to solve
\eqref{eq:tf-pois}, where the sum is restricted to $t \leq T$. However, here the
severity rate $p_T$ only contributes to a single data point $y_T$ in the loss,
and this leads to a much greater degree of variability in estimating $p_T$
compared to the retrospective case. We will thus use extra regularization to
temper such tail variability. This itself introduces a potential tradeoff,
gaining stability but reducing adaptivity.

The regularization techniques we employ are inspired by \citet{Jahja2022}, who 
studied deconvolution in the context of nowcasting infections that will
eventually appear as case reports. The additional regularization comes in two 
parts. The first part is to impose a constraint on the tail values of the
severity rate curve, which prevents overfitting to the most
recent secondary counts. In particular, we constrain the differences of order
$m+1$ of the last $m$ severity rates to be zero, enforcing these rates  
to adhere to a polynomial trend of degree $m$ (preventing a knot from occuring
near the tail). This constraint differs only lightly from that in
\citet{Jahja2022}, who used a natural spline constraint. This constrains the 
tails to be a polynomial of degree $(m-1)/2$; note that our proposal applies to
all values of $m$, whereas natural splines are only defined for odd $m$.  

The second part adds a regularization term to the criterion which is a kind
of tapered smoothing penalty, penalizing the squared differences $p_{t-1}-p_t$
in adjacent severity rates with successively decreasing weight as $t$ moves away
from $T$. Such weights are chosen based on how much mass is captured by the
delay distribution between primary and secondary events. Specifically, we assign
$(p_{t-1}-p_t)^2$ a weight \smash{$w_t = 1/\hat{F}^{(T)}(T-t)$} for $t \geq
T-d$, and $w_t = 0$ for $t < T-d$, where \smash{$\hat{F}^{(T)}(k) =
  \sum_{j=0}^k \hat\pi^{(T)}_j$} is the CDF of \smash{$\hat\pi^{(T)}$} at 
$k$. Thus, the longer the tail of the delay distribution, the more we seek to
smooth out severity rates in the recent past.    

Putting these two sources of regularization together leads to the following  
optimization, which defines our real-time severity rate estimates at $T$:   
\begin{equation}
\begin{alignedat}{2}
\label{eq:tf-pois-rt}
&\minimize_{0 \preceq p \preceq 1} && \sum_{t \leq T} \bigg[\bigg( \sum_{k=0}^d    
x_{t-k}\hat\pi_k^{(t-k)}p_{t-k} \bigg) - y_t\log\bigg( \sum_{k=0}^d 
x_{t-k}\hat\pi_k^{(t-k)}p_{t-k} \bigg)\bigg] \\
& && \qquad\qquad\qquad\qquad\qquad + \lambda\|D^{(m+1)}p\|_1 + 
\gamma \|W D^{(1)}\|_2^2 \\
&\st \; && \sum_{k=0}^{m+1} (-1)^k \binom{m+1}{k} p_{T - m - 1 + k} = 0, 
\end{alignedat}
\end{equation}
where \smash{$W = \diag(0,\dots,0,\sqrt{w_{T-d}},\dots,\sqrt{w_T})$} is a
diagonal matrix whose only nonzero elements are the final $d+1$ diagonal 
entries. Denoting by \smash{$\hat{p}^{\text{rt}}$} the vector of real-time
severity rate estimates obtained by solving \eqref{eq:tf-pois-rt}, the real-time
estimate of $p_T$ is its last element \smash{$\hat{p}^{\text{rt}}_T$}.

In practice, it is common for epidemic data streams to be revised after initial
data is released \citep{reinhart2021open}. This means that, from one time $T$ to
the next, the whole sequences of primary and secondary counts used in
real-time estimation may be updated (rather than new primary and secondary
counts simply being appended). To reflect this, it helps to introduce some
additional notation: for time points $s \geq t$, denote by \smash{$x^{(s)}_t$}
the version of primary incidence at time $t$ which was available as of time $s$,
and similarly \smash{$y^{(s)}_t$} for secondary incidence. Therefore, in
practice, we would substitute each pair $x_t,y_t$ of primary and secondary
counts in \eqref{eq:tf-pois-rt} with \smash{$x^{(T)}_t,y^{(T)}_t$}, the versions
of these counts available at $T$. 

When $m=0$, the trend filtering penalty $\lambda\|D^{(1)}p\|_1$ and the tail regularization $\gamma\|WD^{(1)}p\|_2^2$ both act on first differences, so there is some redundancy.
However, their effects are not identical: the $\ell_1$ trend filtering penalty selects knots adaptively and could in principle place a knot close to the tail, allowing a steep change in the final stretch of the curve. A large $\gamma$ specifically discourages this by encouraging flatness near the boundary, weighted by the delay distribution CDF. Thus the tail regularization provides targeted control over boundary behavior that the global trend filtering penalty does not.

\paragraph{Cross-validation.}

We tune the hyperparameters $\lambda, \gamma \geq 0$ in \eqref{eq:tf-pois-rt}
with a two-stage approach, again inspired by \citet{Jahja2022}. First, setting
$\gamma=0$, tune $\lambda$ with $K$-fold validation, exactly as in the
retrospective case described previously. Then, fixing $\lambda = \lambda^*$ at
the chosen value from cross-validation, we tune $\gamma$ with a
rolling-validation (also called forward-validation) procedure. This works as
follows: for each $s = T-M,\dots,T-1$, we solve \eqref{eq:tf-pois-rt} using 
only data up through time $s$, and use the resulting severity rate estimates
\smash{$\hat{p}^{\text{rt},s}(\gamma)$} to form a prediction
\smash{$\hat{y}_{s+1}(\gamma)$} of $y_{s+1}$ by reconvolution. The 
rolling-validation error for the given  
$\gamma$ is then 
\begin{equation}
\text{MAE}(\gamma) = \frac{1}{M} \sum_{s=T-M}^{T-1} \big| \hat{y}_{s+1}(\gamma) 
- y_{s+1} \big|. 
\end{equation}
This process is iterated over a grid of $\gamma$ values. As with
cross-validation, we can select a value $\gamma^*$ according to the ``min rule''  
(the $\gamma$ with minimum MAE), or alternatively, the ``1se rule'' (the largest 
$\gamma$ whose MAE lies within one standard error of the minimum). Finally, the 
real-time problem is solved at hyperparameter values $\lambda^*$ and $\gamma^*$
on all available data, from which we extract \smash{$\hat{p}^{\text{rt}}_T$}.
This is summarized in Algorithm \cref{alg:fv}.  
\cref{fig:gammas} in \cref{apx:extra} plots deconvolution estimates at varying levels of $\gamma$, highlighting its ability to significantly alter tail predictions. 

It is worth making two further remarks. First, as in the retrospective case,
before running either cross-validation or forward-validation, the likelihood and
regularizer terms in \eqref{eq:tf-pois-rt} are placed on a common scale by
dividing each by their number of summands. Second, the order of trend filtering
$m$ can again be tuned using the two-stage approach, based on minimizing the 
forward-validation error. It can instead be chosen based on qualitative
considerations, as discussed in the retrospective subsection.  

\renewcommand{\algorithmicrequire}{\textbf{Input:}}
\renewcommand{\algorithmicensure}{\textbf{Output:}}

\begin{algorithm}[htb]
\caption{Tuning $\lambda,\gamma$ by two-stage validation}
\begin{algorithmic}[1]\label{alg:fv}
\REQUIRE Candidate sets of values $\Lambda,\Gamma$ for $\lambda,\gamma$; primary
and secondary incidence $x_t,y_t$, for $t \leq T$; number of cross-validation
folds $K$ and number of forward-validation steps $M$ 
\ENSURE Selected $\lambda^*, \gamma^*$ 

\FOR{each $\lambda \in \Lambda$}
\STATE Use $K$-fold cross-validation for problem \eqref{eq:tf-pois-rt} with
$\gamma = 0$, to record $\text{MAE}(\lambda)$ 
\ENDFOR
\STATE Select $\lambda^*$ by min or 1se rule
\FOR{each $\gamma \in \Gamma$}
\FOR{$s = T - M$ to $T - 1$}
\STATE Solve \eqref{eq:tf-pois-rt} with given $\gamma$, $\lambda = \lambda^*$,  
and $T=s$ to yield $\vphantom{\sum_{k=0}^d} \hat{p}^{\text{rt},s}(\gamma)$
\STATE Linearly extrapolate $\vphantom{\sum_{k=0}^d}
\hat{p}^{\text{rt},s}_{s+1}(\gamma)$ from
$\hat{p}^{\text{rt},s}_s(\gamma)$ and 
$\hat{p}^{\text{rt},s}_{s-1}(\gamma)$    
\STATE Compute prediction $\hat{y}_{s+1}(\gamma) = \sum_{k=0}^d x_{s+1-k} 
  \hat\pi^{(s)}_k \hat{p}^{\text{rt},s}_{s+1-k}(\gamma)$
\ENDFOR
\STATE Record $\text{MAE}(\gamma) = \frac{1}{M} \sum_{s=T-M}^{T-1}
|\hat{y}_{s+1}(\gamma) - y_{s+1}|$
\ENDFOR
\STATE Select $\gamma^*$ by min or 1se rule
\STATE \textbf{return} $\lambda^*,\gamma^*$
\end{algorithmic}
\end{algorithm}

\section{Experimental setup}\label{sec:setup2}

\subsection{Semi-synthetic data generation}
\label{sec:data-generation}

To evaluate our methods against existing benchmarks, we generated realistic
secondary incidence data using ground truth severity rates, whose computation is 
described below. These simulations are intended to mimic true COVID-19 deaths in
the \US\ over the pandemic. The severity rate we targeted throughout our
analysis was hence the hospitalization-fatality rate (HFR).   

\paragraph{Real data resources.}

Our simulations used real hospitalization counts as the primary incidence
data. In the first several years of the pandemic, COVID-19 hospital admissions
were reported daily in real time to the Department of Health and Human Services 
\citep{HHS2023}. This process was coordinated by the National Healthcare Safety
Network, thus we refer to the aggregate daily hospital admissions as NHSN 
data. 

From observed NHSN hospitalizations, we simulated COVID-19 deaths across the
\US\ and all 50 states, as well as for both the retrospective and real-time
settings. We designed separate simulation models for each of these settings; the
discrepancy here reflects the fact that different counts were available in
real-time versus in retrospect during the pandemic. True weekly death totals
were aggregated in retrospect by the National Center for Health Statistics
(NCHS) \citep{nchs}. NCHS revealed these death totals well after the week in
question. NHSN began comprehensive hospitalization reporting in summer 2020, and
NCHS stopped reporting deaths in early 2023. Our retrospective simulations cover
this 2.5-year window.

In real-time, Johns Hopkins University (JHU) published provisional daily death
counts. These data bear a subtle yet important distinction from those of
NCHS---JHU deaths reflect not the true number of deaths which occurred each day,
but rather the number of new deaths \textit{reported} on a given day. Often,
these deaths were reported days or even weeks after they occurred. As a result,
the delay distribution which underlies the JHU data has higher mean than that
for NCHS. The JHU counts were also considerably noisier, due at least in part to
reporting idiosyncrasies. Our real-time simulations cover the same 2.5-year
window, from summer 2020 to early 2023, consistent with the retrospective case. 

We use the Epidata API \citep{Epidata} to download all of the data described
above.  

\paragraph{True severity rates.}

The severity rates that we used as the ground truth in our simulations were
created using the following variant-based procedure. In each region (state or
\US\ national), we defined a single HFR associated with the four most
significant COVID-19 variants: the original strain, Alpha, Delta, and
Omicron. Using data from CoVariants \citep{Hodcroft2021}, we identified the
dominant period for each variant, based on when it accounted for over half of
cases. Within this window, we computed the HFR as the total number of NCHS
deaths divided by NHSN hospitalizations, offset by two weeks to account for
delays. Finally, we mixed these per-variant HFRs with the variant proportions in
circulation to obtain the final HFR curve: 
\begin{equation}
p_t = \sum_v c_t^v \, p^v, \;\; \text{where} \; \sum_v c_t^v = 1 \;
\text{for all $t$}.
\end{equation}
Here the sums are over variants $v$, with $p^v$ denoting the HFR for variant
$v$, and $c_t^v$ denotes the proportion of variant $v$ in circulation at time
$t$ (calculated again from CoVariants data). 

\paragraph{Delay distributions.}

While our severity rate estimators allow for nonstationary delay distributions
\smash{$\pi^{(t)}$}, for the sake of simplicity our simulations use a constant delay \smash{$\pi^{(t)} =
  \pi$}, for all $t$. 
  In reality, delay distributions are known to vary over time \citep{Ward2021}, though our evaluation of model performance under misspecified delays partially addresses this limitation.
  Moreover, this simplification is consistent with standard practice in estimating severity rates \citep{UKpaper,fusedlasso} and reproduction numbers \citep{wallinga_teunis, cori2013new, Chitwood2022}.
  
  Setting $d=60$ days, we fit
discrete gamma distributions in each region parameterized by
heuristically-chosen means and variances. The means were estimated by maximizing
the cross-correlation between NHSN hospitalizations and deaths, NCHS deaths in
the retrospective case and JHU in real-time. The standard deviations were set to
90\% of the means, motivated by empirical findings on typical
hospitalization-to-death delays \citep{UKdelay}. Denoting by $F_\gamma$ the
cumulative distribution function (CDF) of the corresponding gamma distribution,
we define the delay distribution by \smash{$\pi_k \propto F_\Gamma(k+1) -
  F_\Gamma(k)$}, where these values are normalized to sum to 1.  

\paragraph{Noise models.}

Using the NHSN hospitalizations, variant-based severity rates, and delay 
distributions as described above, we then simulate daily deaths from the
following noise models.  
\begin{itemize}
\item For the retrospective simulation, we generated deaths from a
  Poisson-binomial model; this qualitatively matches the variance of death
  counts from NCHS, as shown in \cref{fig:sim_vs_real_retro}.   
\item For the real-time simulation, we generated deaths from a beta-binomial 
  distribution, which reflects the fact that JHU deaths are more overdispersed
  than what the Poisson-binomial can accommodate. As we describe below, we first
  estimated the amount of overdispersion, then simulated deaths with the given
  variance. The results qualitatively match the dispersion of deaths counts from 
  JHU, see \cref{fig:sim_vs_real_rt}.   
\end{itemize}

To estimate the amount of dispersion in each region (50 states and \US\
overall), we fit a quasi-Poisson regression to cleaned JHU death
counts. \cref{apx:jhu-clean} details our methodology. This regression
computes a coefficient \smash{$\hbeta$} which encodes the amount of
overdispersion present in the JHU death data, in the given region. To
accommodate such overdispersion in our simulation model for deaths, we then use
a beta-binomial model, described next.  

\begin{figure}[p]
\centering
\includegraphics[width=0.95\textwidth]{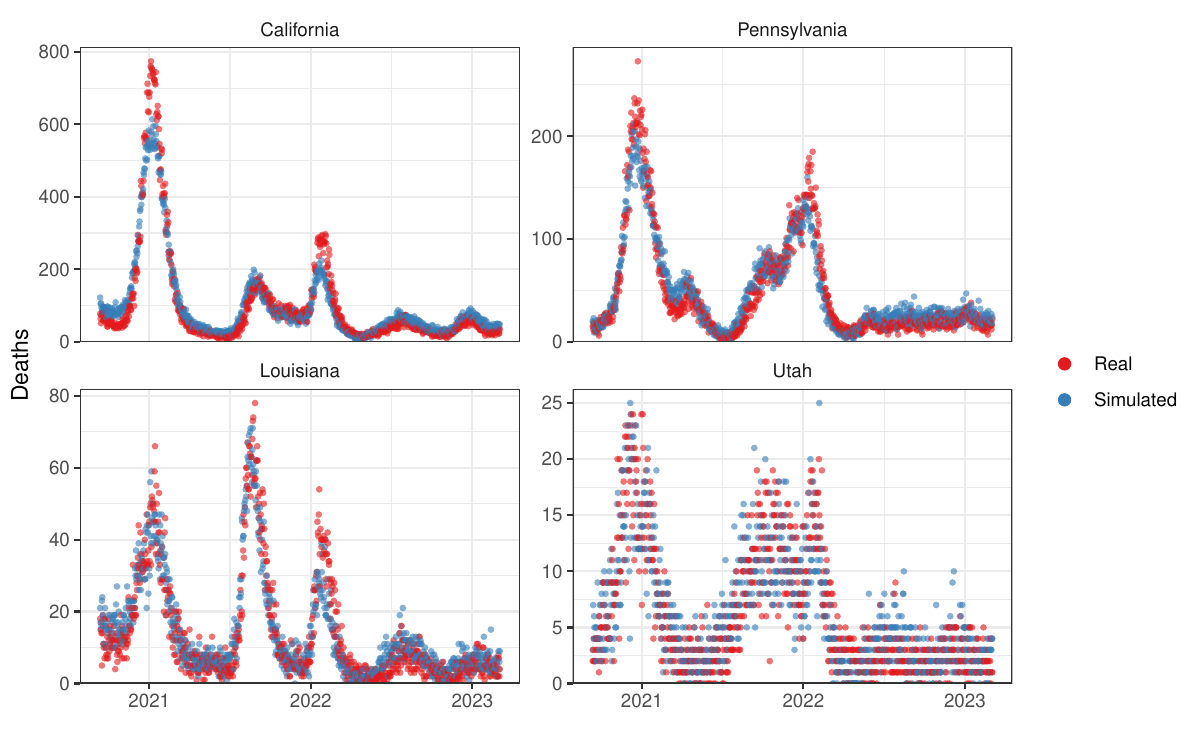}
\caption[Comparing retrospective daily death counts of COVID source and our simulation model.]{Real (NCHS) and simulated (Poisson-binomial) daily deaths. (The
  NCHS data is itself weekly, but here we have subsampled it to the daily level
  to match the time resolution of our simulated data.)}      
\label{fig:sim_vs_real_retro}

\bigskip
\includegraphics[width=0.95\textwidth]{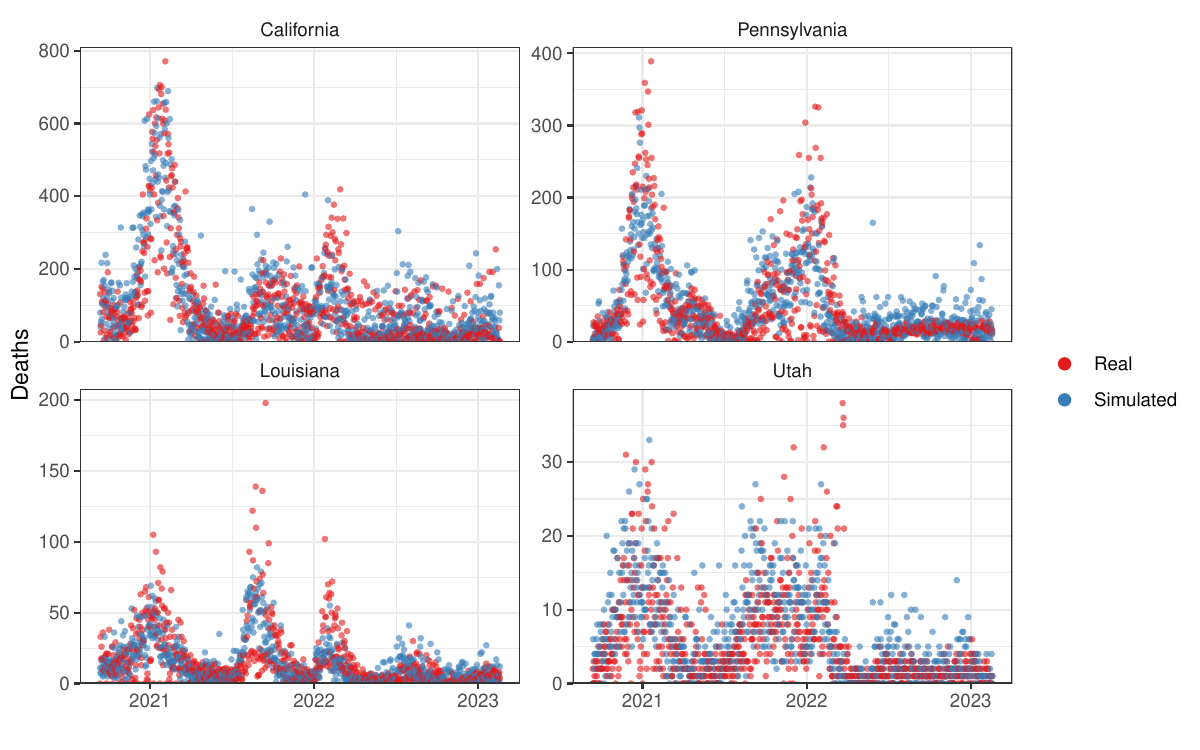} 
\caption[Comparing real-time daily death counts of COVID source and our simulation model.]{Real (JHU) and simulated (beta-binomial) daily deaths. (The JHU deaths
  displayed here are ``finalized'', meaning that the history has been revised
  after the end of the reporting period. They have been preprocessed as
  described in \cref{apx:jhu-clean}.)}     
\label{fig:sim_vs_real_rt}
\end{figure}

The beta-binomial distribution is a popular tool for modeling count data with
dispersion \citep{betabinom}. For our problem, in a given region with dispersion
coefficient \smash{$\hbeta$}, we can use the standard mean-rho parameterization
to define the beta-binomial at each time $t$, setting      
\begin{equation}
\label{eq:beta-binom}
M_t = \frac{\mu_t}{\sum_{k=0}^d x_{t-k}}, \;\; \text{and} \;\; 
\rho_t = \bigg[ \frac{\hbeta \sigma_t^2}{\mu_t\big(1-\mu_t / \sum_{k=0}^d    
x_{t-k}\big)}-1 \bigg]\bigg[ \frac{1}{\sum_{k=0}^d x_{t-k}-1} \bigg],   
\end{equation}
where $\mu_t, \sigma_t^2$ are as defined in \eqref{eq:pb_mean},
\eqref{eq:pb_var}. \cref{apx:beta-binom} gives the details behind this
calculation.   

Figures \ref{fig:sim_vs_real_retro} and \ref{fig:sim_vs_real_rt} qualitatively
assess the goodness-of-fit of our simulation models by plotting simulated 
and real data for four regions of varying sizes. While the curves do not overlap 
precisely, visual
inspection shows that the noise levels are relatively comparable.      

\subsection{Methods considered}\label{sec:methods2}

In the retrospective and real-time simulations, we ran Poisson deconvolution
with trend filtering penalties of orders $m = 0,1,2$, which deliver piecewise
constant, linear, and quadratic estimates, respectively. We further considered
tuning the order $m$ of trend filtering regularization itself, by
cross-validation in the retrospective case, and forward-validation in real-time.

We compared these deconvolution methods to the ratio-based estimators
of severity rates from \cref{ch:paper1}, which we will often refer to as ``benchmark'' methods.
Here we augment them with smoothed
primary and secondary counts over some window $W \geq 1$. 
Here, $W$ serves as a
hyperparameter, with $W=1$ corresponding to no smoothing. In the real-time
setting, we smooth counts with a trailing window; in retrospect, we use a
centered window. For the case of primary incidence, define 
\begin{equation} 
\widetilde{x}_t^{\text{rt}} = \frac{1}{W} \sum_{k = 0}^{W-1} x_{t-k},
\;\; \text{and} \;\;
\widetilde{x}_t^{\text{rs}} = \frac{1}{W} \sum_{k = -\lfloor W/2
  \rfloor}^{\lceil W/2 \rceil - 1} x_{t+k},  
\end{equation}
and analogously define \smash{$\widetilde{y}_t^\text{rt}$} and 
\smash{$\widetilde{y}_t^\text{rs}$} for secondary incidence. 

With this notation in place, we revisit the \emph{lagged ratio}
estimators from \cref{ch:paper1}, now augmented with the smoothing window $W$.
These are arguably the most widely-used estimators of severity rates in
practice \citep{yuan2020monitoring, timevar_ifr, horita2022global, lagged_chinese,
  LIU2023100350}, simply comparing primary and secondary events offset by a lag of
$\ell \geq 0$ time steps. The real-time and retrospective lagged ratios are       
\begin{equation}
\hat{p}_t^\text{rt} =
\frac{\widetilde{y}_t^\text{rt}}{\widetilde{x}_{t-\ell}^\text{rt}},
 \;\; \text{and} \;\;
\hat{p}_t^\text{rs} =
\frac{\widetilde{y}_{t+\ell}^\text{rs}}{\widetilde{x}_{t}^\text{rs}},
\end{equation}
respectively. The lag parameter $\ell$ represents the typical duration
between primary and secondary events. 

Similarly, we revisit the \emph{convolutional ratio} estimators from \cref{ch:paper1}, again with smoothing.
Originally introduced by \citet{UKpaper} (based on earlier work of \citet{nishiura}), these use an estimate of the delay distribution \smash{$\hat\pi^{(t)}$}, which convolves against trailing primary events to account for secondary incidence. The
real-time convolutional ratio is   
\begin{equation}
\label{eq:conv-rt}
\hat{p}_t^\text{rt} = \frac{\widetilde{y}_t^\text{rt}}{\sum_{k=0}^d
  \widetilde{x}_{t-\ell}^\text{rt}\hat\pi^{(t-k)}_{k}}.
\end{equation}
The retrospective version tracks the proportion of cases among primary events 
at $t$ who will eventually have secondary events. While secondary event dates  
for individual cases at $t$ are unknown, the relevant secondary count can be 
estimated with a convolutional model. This leads to the retrospective
convolutional ratio:
\begin{align}\label{eq:conv-retro2}
\hat{p}_t^\text{rs} = \sum_{k=0}^d
  \frac{\widetilde{y}_{t+k}^\text{rs}\hat\pi^{(t)}_{k}}{\sum_{j=0}^d
  \widetilde{x}_{t+k-j}^\text{rs}\hat\pi^{(t+k-j)}_{j}}. 
\end{align}
We note that when each distrbution \smash{$\hat\pi^{(t)}$} is a point mass at
$\ell$, the convolutional ratios reduce to the lagged ones. Therefore, the
former may be understood as generalizations of the latter. 

In \cref{ch:paper1}, we studied the bias of the real-time lagged
and convolutional ratios, both formally and empirically. Our analysis revealed
that the lagged ratio tends to exhibit large bias, and can both signal
nonexistent surges and fail to detect upticks in severity rate. The
convolutional ratio is generally more robust, but it still can have significant   
bias when the severity rate is changing.      

The retrospective lagged ratio is subject to roughly the same bias as in the 
real-time case: it is effectively the same estimator offset by $\ell$ time 
steps. The retrospective convolutional ratio is somewhat more challenging to
analyze than its real-time counterpart. However, its derivation is based on a
stationarity assumption, and it will still be particularly biased when the
severity rate changes around $t$.  

These ratio estimators all depend on the hyperparameter $W$, which controls the   
length of the smoothing window, and must be tuned. In the retrospective case,
this is tuned via $K$-fold cross-validation, just like $\lambda$ for the
retrospective deconvolution method. In the real-time case, it is tuned via
$M$-step forward-validation, just like $\gamma$ for the real-time deconvolution
method.  

\subsection{Experimental design}\label{sec:design}

In all cases, we examined the use of both the min rule and the 1se rule to tune
hyperparameters using cross- and forward-validation: recall this is $\lambda$
for retrospective deconvolution; $\lambda,\gamma$ for real-time deconvolution;
and $W$ for the benchmark methods, lagged and convolutional ratios, in both
retrospective and real-time settings. In each case, we report results for the
most favorable rule (min or 1se) in terms of MAE, giving each method the full
benefit of the doubt; details will be given in the next section.  We used $K=5$
cross-validation folds and $M=28$ forward-validation steps, throughout.

We evaluated deconvolution and benchmark severity rate estimators on simulated
data across 50 states as well as the \US\ nationally. The estimation window
encompassed a roughly two-year period between 2021 and 2022, omitting a burn-in
and burn-out period. With the exception of the misspecified experiments
described in the final paragraph of this subsection, the deconvolution methods
and the convolutional ratios were given access to the oracle delay distribution, 
thus using \smash{$\hat\pi^{(t)}=\pi$}, for all $t$. The lagged ratios use a lag
$\ell$ equal to the mean of the delay distribution in each region. In all but
the misspecified experiments, the reported results are averaged over 10
replications (10 draws of data from the simulation model).   

To peform deconvolution with trend filtering regularization we solved the
optimization problems \eqref{eq:tf-pois}, \eqref{eq:tf-pois-rt} with Clarabel in
CVXR \citep{clarabel}. The retrospective setting required much less computation
than real-time. For this reason, we also ran deconvolution and the benchmarks
with oracle hyperparameter tuning. This selected the hyperparameters whose
estimated severity rates had the smallest mean absolute error (MAE), in
hindsight. To lessen the computational burden in both the retrospective and
real-time cases, we only computed estimates from each method once every seven
days. In total, this amounted to about 100 estimation dates, over which we
calculated the eventual MAEs to be reported.

Lastly, we recomputed all the estimators under misspecified delay
distributions---still constant over time, but not equal to \smash{$\pi$}. We
considered six different misspecified delay distributions with varying means;
each was still a discrete gamma distribution, whose standard deviation is set to
90\% of its mean. In the retrospective case, we set the means to be 1, 2, and 3
days before and after the true value for each state. For the real-time case,
which had longer delay distributions, the means were offset by 1, 3, and 5 days,
on both sides of the true value. \cref{apx:misp} provides visualizations
of the misspecified delay distributions for a few states.

\section{Experimental results}\label{sec:results2}

We analyze results across the synthetic experiments described in the previous
section.

\subsection{Retrospective analysis}

 \begin{table}[ht]
\centering
\caption[MAE and relative performance of retrospective HFR estimation.]{MAE of methods in retrospective HFR estimation, and the
  associated percentage improvement on the convolutional ratio (CR) and lagged
  ratio (LR). Results are averaged over 51 regions and 10 replications.}    
\label{tab:retro-cv}
\begin{tabular}[t]{@{}lrrrrrr@{}}
\toprule
& Lagged & Conv & \textbf{Deconv-0} & \textbf{Deconv-1} & \textbf{Deconv-2} & \textbf{Deconv-T}\\
\midrule
MAE $\times \; 10^3$ & 16.6 ± .2 & 8.3 ± .1 & 6.7 ± .1 & 6.9 ± .1 & 7.1 ± .1 & 6.9 ± .1\\
Improv over CR (\%) & -112.5 ± 2.7 & NA & 15.2 ± .7 & 15.9 ± .7 & 12.4 ± .8 & 14.1 ± .7\\
Improv over LR (\%) & NA & 47.6 ± .6 & 56.2 ± .6 & 56.3 ± .7 & 54.3 ± .7 & 55.6 ± .6\\
\bottomrule
\end{tabular}
\end{table}

In the retrospective setting, deconvolution largely outperforms the
ratio-based estimators across the \US\ and the 50 states. 
\cref{tab:retro-cv} summarizes the performance of these methods when all of the
hyperparameters are tuned via cross-validation. As mentioned above in the
experimental design, we consider using both the min rule and 1se rule within 
cross-validation, and for each method we report the results from the rule with
the strongest performance; in the retrospective case, this ends up being the min  
rule for trend filtering, and the 1se rule for the ratio methods. The table
reports, for each method, the MAE averaged over the 51 regions and 10
replicates:      
\begin{equation}
\frac{1}{51} \sum_{r=1}^{51} \bigg (\frac{1}{10}
\sum_{i=1}^{10} \text{MAE}_{ri} \bigg),
\end{equation}
where \smash{$\text{MAE}_{ri}$} is the MAE for region $r$ and replicate $i$. It
also reports the associated Monte Carlo standard error, derived in \cref{apx:se-calc}:  
\begin{equation}
\sqrt{\frac{1}{51^2}\sum_{r=1}^{51} \frac{1}{10} \hat\sigma^2\big( 
\{\text{MAE}_{ri}\}_{i=1}^{10} \big)},   
\end{equation}
where \smash{$\hat\sigma^2(S)$} is the sample variance of elements in a set
$S$. In the table, ``Deconv-$m$'' denotes deconvolution with trend
filtering regularization of order $m$, and ``Deconv-T'' denotes deconvolution
with the order tuned by cross-validation. The same labeling is used throughout
all  tables and figures.

As we can see from \cref{tab:retro-cv}, all orders of trend filtering yield
more accurate severity rate estimates than the ratio-based methods. The MAE of
the convolutional ratio is roughly $8.3 \times 10^{-3}$ on average,
substantially outperforming the lagged ratio. For deconvolution methods, the
average MAE ranges from $6.7$ to $7.1 \times 10^{-3}$, depending on the order
$m$. This translates to a roughly 12-16\% improvement in MAE over the
convolutional ratio, and roughly 54-56\% over the lagged ratio. Tuning the order
of trend filtering also works well. Moreover, \cref{tab:retro-oracle} in
\cref{apx:extra} shows that under oracle tuning (with all tuning
parameters chosen to optimize MAE) deconvolution outperforms the benchmarks
by an even wider margin.

\begin{figure}[p]
\centering
\includegraphics[width=0.95\textwidth]{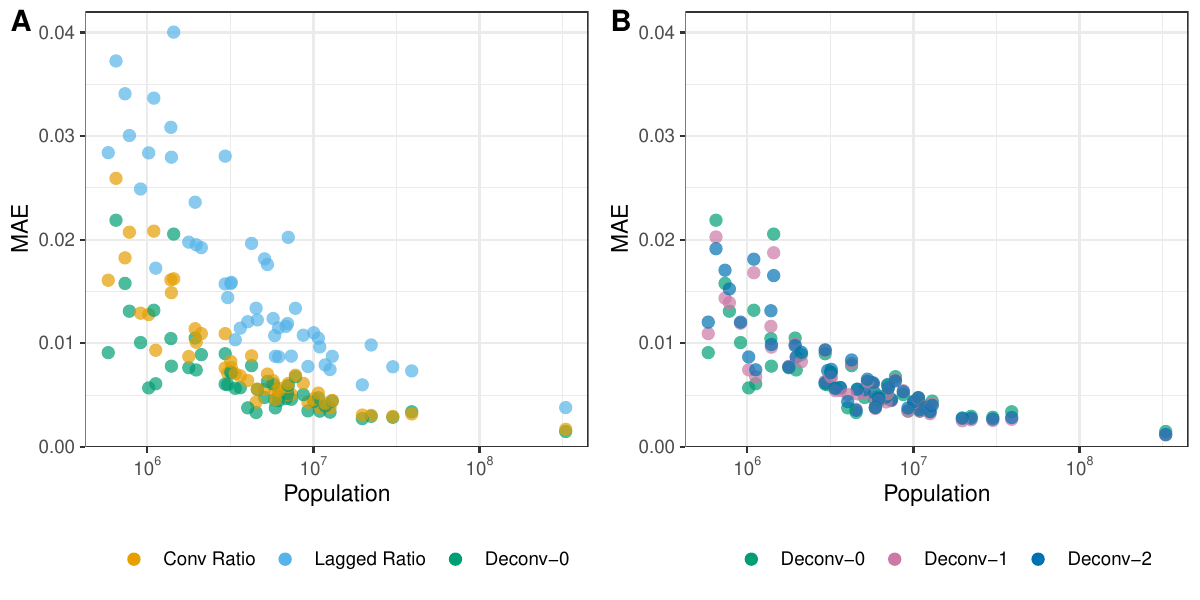}
\caption[MAE versus population size in retrospective HFR estimation.]{MAE versus population size in retrospective HFR estimation, for 51
  regions (\US\ and 50 states). Each point is averaged over 10 replications.}   
\label{fig:retrospective-comparison}

\bigskip 
\includegraphics[width=0.95\textwidth]{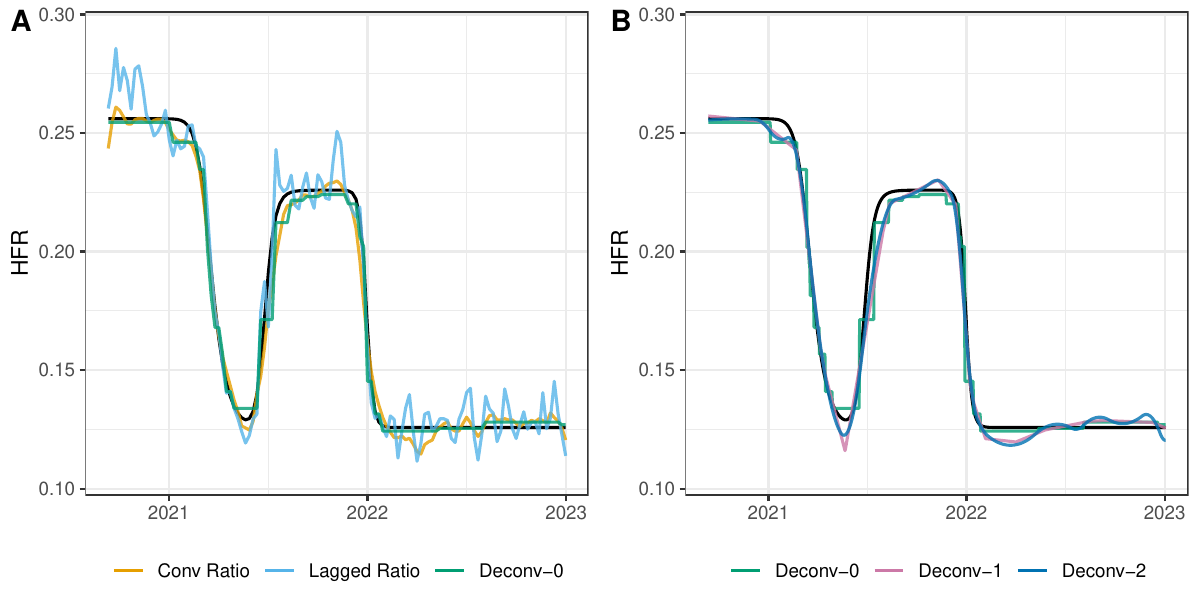}
\caption[Example retrospective HFR estimates for Pennsylvania.]{Example retrospective HFR estimates for Pennsylvania, from a single 
  replication.} 
\label{fig:retro-curves}
\end{figure}

\cref{fig:retrospective-comparison}A disaggregates the MAEs by population
size, across the 51 regions. Observe that the spread of improvement varies by
population, with deconvolution generally having stronger gains over the
benchmarks for smaller population sizes (where estimation is generally more
challenging). \cref{fig:retrospective-comparison}B confirms that this
holds for all orders of trend filtering regularization. 

\cref{fig:retro-curves} shows the HFR curves for an example region,
Pennsylvania. The ratio-based estimates oscillate around the ground truth HFR
(black line), with the lagged ratio displaying clearly greater volatility, but
the convolutional ratio still showing a nontrivial amount as well.  Broadly, the
estimates from deconvolution do not share this behavior and are qualitatively
more stable. We emphasize that this is the case even though the ratio-based
methods use cross-validation with the 1se rule to tune their hyperparameters
(which incentivizes a greater amount of regularization). The same behavior is 
generally consistent across all geographies, and the HFR curves for the
remaining states are displayed in \cref{fig:all_curves_retro} of 
\cref{apx:extra}.

\subsection{Real-time analysis}

\begin{table}[h]
\centering
\caption[MAE and relative performance of real-time HFR estimation.]{MAE of methods in real-time HFR estimation, and the associated
  percentage improvement on the convolutional ratio (CR) and lagged ratio
  (LR). Results are averaged over 51 regions and 10 replications.}     
\label{tab:rt-cv}
\begin{tabular}[t]{@{}lrrrrrr@{}}
\toprule
& Lagged & Conv & \textbf{Deconv-0} & \textbf{Deconv-1} & \textbf{Deconv-2} & \textbf{Deconv-T}\\
\midrule
MAE $\times \; 10^3$ & 27.0 ± .1 & 22.9 ± .1 & 18.5 ± .2 & 19.3 ± .1 & 19.8 ± .1 & 19.6 ± .1\\
Improv over CR (\%) & -19.3 ± .4 & NA & 16.7 ± .5 & 13.4 ± .4 & 10.9 ± .5 & 12.2 ± .5\\
Improv over LR (\%) & NA & 13.8 ± .3 & 28.6 ± .5 & 25.6 ± .5 & 23.6 ± .5 & 24.6 ± .5\\
\bottomrule
\end{tabular}
\end{table}

Deconvolution also performs strongly in the real-time setting. In the format
of \cref{tab:retro-cv} above, \cref{tab:rt-cv} reports the MAEs and 
percentage improvements over the benchmarks in the real-time case, when all 
hyperparameters are tuned with cross-validation. Again, for each method we
selected among the min and 1se rule depending on which resulted in more
favorable performance; in the real-time case, this ends up being the 1se rule
for all hyperparameters except for $\lambda$ in deconvolution, which resulted
in marginally better performance when tuned via the min rule. We note that the
convolutional ratio estimator had substantially worse performance under the min
rule, whose chosen window sizes tended to undersmooth counts. 

As we can see from the table, the deconvolution methods provide an improvement
in accuracy over the ratio-based methods, consistent with the previous results
for the retrospective case. All orders $m$ of trend filtering regularization are
around 11-17\% more accurate than the convolutional ratio, and around 24-29\%  
better than the lagged ratio. Tuning the order of trend filtering still
generally works well.        

\cref{fig:rt-comparison} breaks down the MAE by population size, and once
again we can see that deconvolution has an advantage over the ratio methods
across the full range of geographies, with the gap being generally larger for
smaller regions. \cref{fig:rt-curves} displays example HFR curves for
Pennsylvania, as before. In the current real-time case, we can clearly see the 
quality of all estimates degrade. The ratio-based estimates are especially
volatile, and their erratic behavior here can be understood from the
perspective of the analysis in \citet{goldwasser}. For example, their positive
bias in early 2022, particularly that of the lagged ratio, can be attributed to
the Omicron surge which just passed. The deconvolution estimates are
comparatively smoother and more stable. Finally, \cref{fig:all_curves_rt}
in \cref{apx:extra} displays the full set of HFR curves across all
states, where broadly the same conclusions are upheld. 

\begin{figure}[p]
\centering
\includegraphics[width=0.95\textwidth]{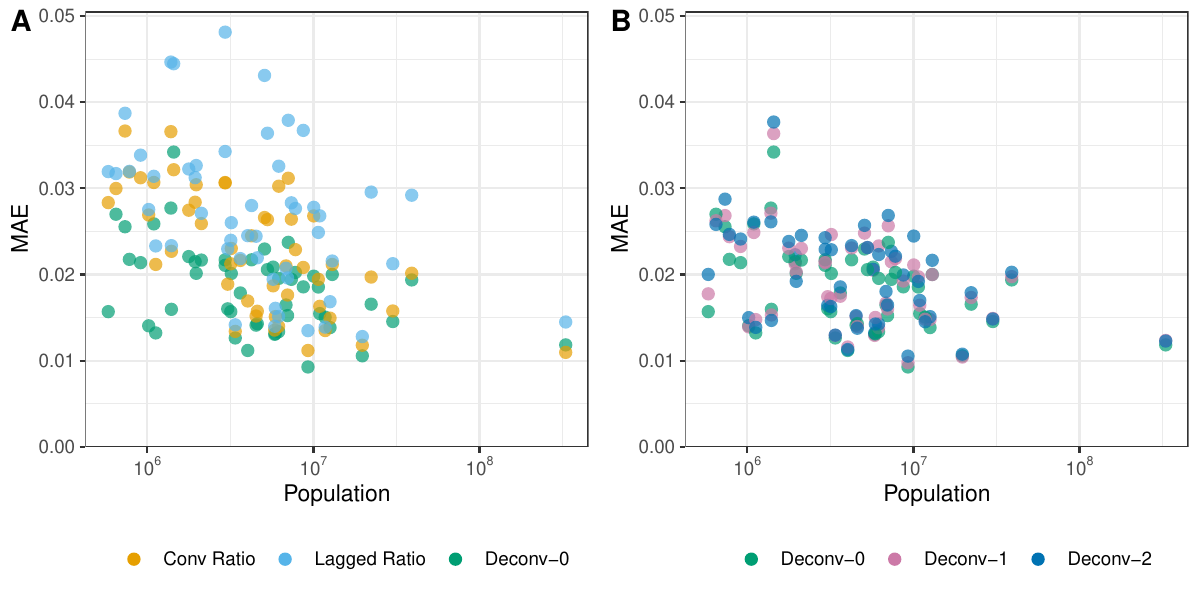}
\captionof{figure}[MAE versus population size in real-time HFR estimation.]{MAE versus population size in real-time HFR estimation, for 51
  regions (\US\ and 50 states). Each point is averaged over 10 replications.}   
\label{fig:rt-comparison}

\bigskip 
\includegraphics[width=0.95\textwidth]{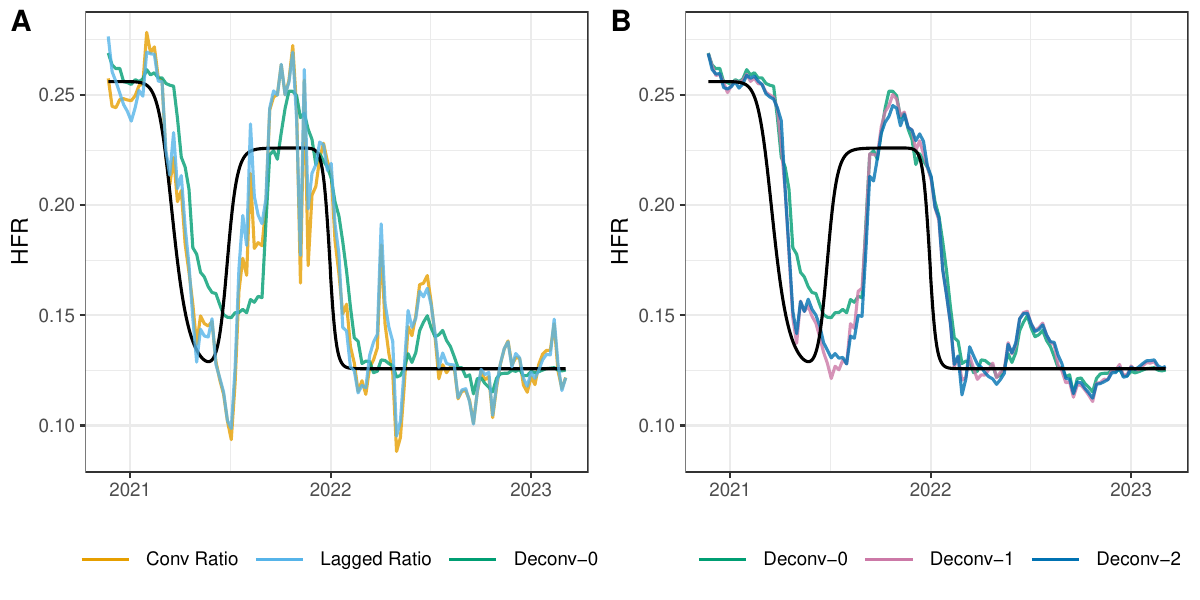}
\captionof{figure}[Example real-time HFR estimates for Pennsylvania.]{Example real-time HFR estimates for Pennsylvania, from a single 
  replication.} 
\label{fig:rt-curves}
\end{figure}

\subsection{Misspecification analysis}

The advantage of deconvolution over benchmark methods persists under various
degrees of misspecification. \cref{fig:misp-maes} displays the MAE as a
function of the offset of the mean of the working delay distribution used by
each method (or the lag used by the lagged ratio) relative to the true
delay. Unsurprisingly, we see accuracy degrade under misspecification, as
evidenced by the curves which generally slope upwards as the mean offset moves
away from zero. The gap in MAE of the lagged ratio and any of the other
estimators is clearly large, regardless of the amount of misspecification.
Comparing the convolutional ratio to deconvolution, the MAE gap narrows somewhat
as the mean offset grows, in the retrospective case; on the other hand, the gap
more or less holds steady as the mean offset varies, in the real-time case.
\cref{fig:misp-imps} in \cref{apx:extra} transforms the MAE curves
from deconvolution in \cref{fig:misp-maes} into percentage improvement
curves over the convolutional ratio. The findings are consistent and overall the
deconvolution methods display comparably strong performance in our misspecified
experiments.

\begin{figure}[t]
\centering
\includegraphics[width=0.95\textwidth]{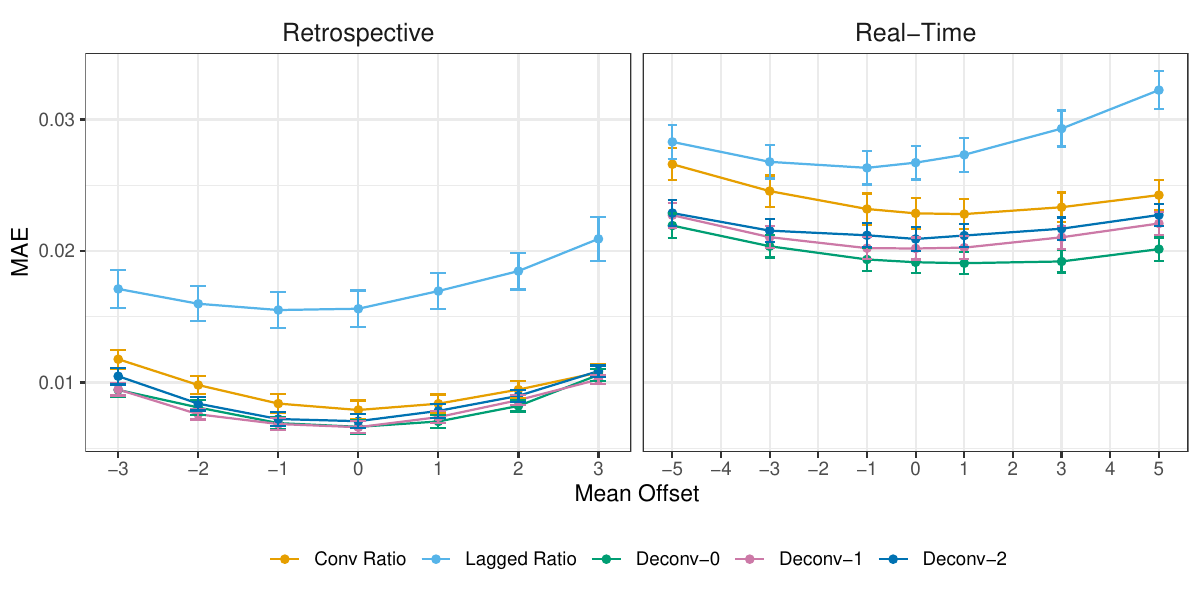}
\caption[MAE under misspecified delay distributions versus mean offset.]{MAE under misspecified delay distributions versus mean offset, with 0 offset indicating correct specification. Results are averaged over 51 regions.}  
\label{fig:misp-maes}
\end{figure}

\section{Discussion}\label{sec:discussion}

In this work, we characterize the probabilistic relationship between two time
series (primary and secondary event counts) which are related by time-varying
convolution: a delay distribution times a quantity called the severity rate. The
resulting likelihood allows us to construct new deconvolution-based estimators
of severity rates, in both the retrospective and real-time settings. To
encourage smoothness over time, we apply trend filtering regularization, and in
the real-time case, we use additional regularization to control volatility at
the boundary of observations (the right tail of the time series). In extensive
simulations, designed to mimic the relationship between COVID-19
hospitalizations and deaths in the \US, we find our deconvolution method to be
consistently more accurate than the typical ratio-based methods used in
practice, often by a significant margin. On real data, the qualitative
differences are quite similar, with deconvolution able to provide stable
estimates of the hospitalization fatality rate (HFR) under changing conditions,
while the ratio-based estimates are subject to suspicious swings and spikes, as
expected based on \citet{goldwasser}.

It is worth further reflecting on some of our results from the perspective of
public health practice. In our real-time experiments, cross-validation and the
1se rule---to choose the smoothing window in the lagged ratio and convoluational 
ratio methods---nearly always selects the maximum window size of 4 weeks. Under
smaller window sizes, the ratio-based methods perform much worse and the margin
of improvement of deconvolution only grows. This creates a tradeoff to be
navigated, as long windows may be undesirable in practice. A core purpose of
real-time estimation is to expeditiously detect changes in severity
rates. Smoothing values over an entire month seriously hampers our ability to do
so, and it seems window sizes this large are rarely used in practice. The
ability of trend filtering to perform locally adaptive smoothing---which
effectively uses a longer or shorter smoothing window as needed at different
parts of the time series---is a potentially powerful way to navigate the basic
tradeoff between detecting changes and controlling variability.

We close by mentioning several avenues for future work. A key component in
practice is the delay distribution, whose misspecification can lead to
degradation in performance. Without high-quality line list data, obtaining an
accurate plug-in estimate is challenging. Jointly estimating this distribution
along with the severity rates themselves would evade this challenge, posing an
open methodological challenge. \citet{li2023reconstructing} examine a joint
estimation strategy for a different but broadly related deconvolution problem,
based on the expectation-maximization (EM) algorithm. It would be interesting to
study an analogous approach in our problem setting. 

Another interesting direction would be to estimate severity rates disaggregated
by different demographic groups, for example, severity rates for different age
brackets. Conceivably, this could be done with the same primary or secondary
event data (not requiring counts per demographic group) by employing a
cohort-based model. This would feature a common set of severity rate curves  
(per demographic group) across all regions, and would then mix by known
demographic proportions in each region to explain the effective severity rate in
that region. 

To our knowledge, no analogous Bayesian methods exist for the general class of severity rates we consider. However, a Bayesian formulation is a natural direction for future work, where prior specification could serve as an alternative regularization mechanism. 
Bayesian methods have been developed for related epidemiological quantities, including growth rates  \citep{guzmanrincon2023bayesian} and reproduction numbers \citep{cori2013new}.
In addition, several Bayesian works on compartmental models \citep{flaxman2020, gibson2023, korolev2020identification} place a prior on the IFR and recover a posterior distribution. 
It would be interesting to compare these IFR estimates to our frequentist framework's, after first deconvolving estimates of latent infection counts.


Perhaps the most important limitation of our approach is the lack of uncertainty estimates. Fitting severity rates with splines rather than trend filtering would enable sandwich-based confidence intervals, as we derive for $R_t$ in \cref{ch4}; the Bayesian formulations discussed above would provide another avenue. For trend filtering, promising developments in data fission \citep{leiner2025data, dharamshi2025generalized} suggest a path forward. We discuss these directions further in the Conclusion.

\part{Estimating the Reproduction Number}\label{pt2}
\chapter{Unifying Perspectives on \texorpdfstring{$R_t$}{Rt}}
\label{ch:paper3}\label{ch3}

\section{Introduction}

The effective reproduction number ($R_t$) is a central metric in infectious disease epidemiology. It summarizes the transmissibility of a pathogen as the average number of secondary infections generated by a single infected individual at time $t$. 
Reliable, timely estimates of $R_t$ 
are useful for public health decision-making, enabling policymakers to adopt and adjust intervention measures and assess their effectiveness.

Translating the conceptual definition of $R_t$ into a computable quantity requires identifying a specific cohort of infectors and infectees over which to quantify average transmissions, for example based on time of acquisition or on onset of infectiousness. 
While $R_t$ can be operationalized in several ways \citep{gostic2020practical}, the standard framework for real-time estimation is instantaneous $R_t$.
This definition quantifies the ratio between new infections observed at a time $t$ to the number of contributing parent infections, with a formula called the renewal equation. 
Although described as ``instantaneous,'' this characterizes the path of the epidemic only if the conditions that produced incident infections remain unchanged in the near future. 

Although renewal equation models are the most widely-used approach to estimate $R_t$, the framework is relatively uncommon in the broader literature on infectious disease modeling because it requires a notion of generation interval that does not map straightforwardly onto observable quantities. 
By contrast, mechanistic models are more often used in papers exploring intervention design, scenario planning, and forecasting. 
These models, such as SEIR, explicitly represent ``compartments'' of infection stages across different population subsets at granular timesteps.
Progression between stages is implicitly used in some instantaneous $R_t$ estimates to, for example, deconvolve from cases to new exposures. However, the link between renewal equation estimators and standard mechanistic models is less well-established.

In this chapter, we develop an explicit crosswalk between mechanistic SEIR and renewal-equation perspectives.
We show that if both the mechanistic model and generation interval are correctly specified, the definition of $R_t$ that
arises from compartmental dynamics equals the instantaneous $R_t$ targeted by the renewal equation. 
We also derive the generation interval implicit in SEIR models, presenting the renewal equation in mechanistic terms. 

\subsection*{Notation}

Notation is summarized in \cref{tab:notation-ch3}.

\begin{table}[!ht]
\centering
\caption{Reproduction number and SEIR notation (Chapter~4).}
\label{tab:notation-ch3}
\renewcommand{\arraystretch}{1.15}
\begin{tabularx}{\linewidth}{@{}l >{\raggedright\arraybackslash}X@{}}
\toprule
\textbf{Symbol} & \textbf{Meaning} \\
\midrule
\multicolumn{2}{@{}l}{\textsc{Reproduction numbers}} \\
\addlinespace[2pt]
$R_0$ & Basic reproduction number \\
$R_t^I$ & Instantaneous reproduction number \\
$R_t^C$ & Case (cohort) reproduction number \\
$R_t^M$ & Mechanistic reproduction number \\
\midrule
\multicolumn{2}{@{}l}{\textsc{Infection counts}} \\
\addlinespace[2pt]
$x_t$ & Infections at $t$ ($E_t^*$ in SEIR model) \\
\midrule
\multicolumn{2}{@{}l}{\textsc{Renewal structure}} \\
\addlinespace[2pt]
$w(t,k)$ & Rate of secondary transmission at $t$ from an infection at $t-k$ \\
$g$ & Generation interval distribution (stationary; $g_k^{(t)}$ in general) \\
\midrule
\multicolumn{2}{@{}l}{\textsc{Compartmental structure}} \\
\addlinespace[2pt]
$N$ & Population size \\
$S_t, E_t, I_t$ & Susceptible, Exposed, Infectious prevalence at time $t$ \\
$E_t^*, I_t^*, R_t^*$ & Incidence into $E$, $I$, $R$ compartments at $t$ \\
$\beta_t$ & Effective contact rate \\
$\sigma,\ \gamma$ & Exit rates from $E$ and $I$ in basic compartmental model  \\
$ \mu^\text{EI}, \mu^\text{IR}$ & Average incubation and infectious periods; $\sigma^{-1}$ and $\gamma^{-1}$ in the basic model.\\
$W^\text{EI}, W^\text{IR}$ & Random variables for latent and infectious periods \\
$\pi^\text{EI}, \pi^\text{IR}, \pi^\text{EY}$ & Delay distributions: $E{\to}I$, $I{\to}R$, $E{\to}Y$ \\
\bottomrule
\end{tabularx}
\end{table}

\section{Background}\label{sec:background}

Literature on reproduction numbers dates back decades, especially on the basic reproduction number
\citep{dublin1925true, kermack1927contribution, macdonald1952equilibrium, diekmann1990definition, anderson1991infectious, heesterbeek2002brief}.
This metric, often referred to as $R_0$, is the expected number of secondary infections produced by one infectious individual in a fully susceptible population. 
It satisfies a sharp threshold: epidemics grow when $R_0 > 1$ and die out when $R_0 < 1$. 
Other important heuristics can be approximated from $R_0$, including short- and long-term forecasts, final size analysis, and the herd immunity threshold.

Since the assumption of full susceptibility rarely holds beyond an outbreak's earliest moments, attention often shifts to the effective reproduction number $R_t$.
Colloquially, $R_t$ is the expected secondary infections from an individual infected at time $t$ given the current population state.
\citet{fraser2007} operationalizes this definition with two related formalizations, which are conducive to estimation. 
While his ideas were originally presented in continuous time, we define them discretely for simplicity of explanation.

Let $x_t$ be the number of individuals infected at $t$. 
This indexes on the date the infection transmits, as opposed to when the person becomes infectious (capable of transmitting to others).
Further define $w(t,k)$ as the rate of secondary transmissions for a primary infection of age $k$ at time $t$:
$$w(t,k)=\mathbb{E}[\text{\# transmissions at $t$}\given\text{individual infected at $t-k$}].$$
This rate defines the mean process that underlies the transmission model:
\begin{equation}\label{eq:fraser_alt}
    \mathbb{E}[x_t\given x_{< t}] = \sum_{k> 0} x_{t-k} \cdot w(t,k),
\end{equation}
where in general $x_{<t} = (x_0,\dots,x_{t-1})$.

Some works define $w(t,k)$ to satisfy an exact transmission model, dropping the expectation in \eqref{eq:fraser_alt} \citep{fraser2007, abbott2020estimating}.
However, when infection counts $x_t$ are random, this empirical $w(t,k)$ must shift accordingly, implying $R_t$ is nonsmooth over time.
To avoid this, we treat the rate $w(t,k)$ as an expectation, per works like \citet{cori2013new}.

Before presenting $R_t$ definitions, we introduce the generation interval distribution.
The generation interval is the age of a primary infection at the time that it transmits to a second individual.
In discrete time, $g_k^{(t)}$ is the probability that a secondary transmission at time $t$ came from a primary infection $k$ timesteps earlier, given that transmission occurs at all.
\citet{fraser2007} formalizes $g_k^{(t)}$ as the $k$-step transmission rate normalized by its total mass:
\begin{align}\label{eq:gi_basic}
    g_k^{(t)}&=\frac{w(t,k)}{\sum_j w(t,j)}\\
    &=\mathbb{P}(\text{infector was infected $k$ steps ago}\given\text{transmission occurs at $t$}).
\end{align}
In general, the generation interval may vary with $t$ as epidemic conditions change. However, it is standard practice in the renewal equation literature to treat it as stationary \citep{cori2013new, fraser2007}. We adopt this assumption henceforth, writing $g_k$ in place of $g_k^{(t)}$.
\subsection{Case and Instantaneous \texorpdfstring{$R_t$}{Rt}}


The case reproduction number (also called cohort $R_t$) is defined as 
\begin{align}
    R_t^C &= \sum_k w(t+k,k) \\
    &\approx \text{Average secondary infections from cohort of primary infections at $t$.}
\end{align}

Case $R_t$ is a forward-looking quantity, analogous to how we define severity rates in Part \ref{pt1}. 
It reports conditions at time $t$ by how many secondary events (transmissions) they generate in the future.

While straightforward, case $R_t$ may be more descriptive of the near future than the present \citep{gostic2020practical}.
To analyze an intervention, $R_t^C$ has the undesirable property of changing before it occurs, since it affects individuals who were already infected.
This complicates interpretation, as it erroneously suggests transmission was already falling before the imposition of the intervention.\

While we study case $R_t$ further in \cref{apx:case_rt}, this chapter focuses more on the other definition in \citet{fraser2007}.

Case $R_t$ aggregates $k$-step transmission rates from $t$ onwards.
As an alternative, instantaneous $R_t$ shifts these rates backwards in time:
\begin{align}\label{eq:inst}
    R_t^I &= \sum_k w(t,k) \\
    &\approx \text{Average secondary infections at $t$ from a recent primary infection.}
\end{align}
Unlike case $R_t$, this rate is inherently backward-looking, in that it does not consider transmission after $t$. 
This makes it more suitable for understanding present conditions, such as the effect of interventions. 
Note that case and instantaneous $R_t$ are the same if \mbox{$w(t+k,k)=w(t,k)\;\forall k$}, meaning if $k$-step transmission rates are equivalent before and after $t$.
Thus, we can also interpret instantaneous $R_t$ as being the average number of secondary transmissions per primary infection at time $t$ assuming conditions remain the same into the near future.

Plugging \eqref{eq:inst} into the definition of generation interval \eqref{eq:gi_basic} yields \mbox{$g_k=w(t,k)/R_t^I$}.
Substituting this relation into the transmission model \eqref{eq:fraser_alt} produces the renewal equation:
\begin{equation}\label{eq:renewal}
    \mathbb{E}[x_t\given x_{< t}] = R_t^I \sum_{k> 0} x_{t-k} \cdot g_k,
\end{equation}
The renewal equation may be rearranged to solve for $R_t^I$, with $\mathbb{E}[x_t\given x_{< t}] $ in the numerator. 

Many works assume a stochastic noise model for infections $x_t$.
For example, \citet{cori2013new} assumes $x_t$ is drawn from a Poisson distribution whose rate is parameterized by $R_t^I$.
Their Bayesian method, EpiEstim, learns a posterior on $R_t$ with this likelihood and a gamma prior.
They propose using case reports as a surrogate for infections $x_t$, which are latent (unobserved).

\subsection{Mechanistic \texorpdfstring{$R_t$}{Rt}}
\label{sec:compartmental-bg}

A third definition of $R_t$ may arise from compartmental epidemic models
\citep{anderson1991infectious}.
These partition a population of size $N$ into
compartments representing disease status, with transitions between
compartments governed by specified rates.
In the original SIR model, susceptible individuals move to the infectious ($I$) compartment upon transmission, before recovering ($R$) after some duration \citep{kermack1927contribution}.
The $R$ compartment is also sometimes called Removal to include death.
We focus on the SEIR model, which extends this framework
with an Exposed compartment ($E$) for individuals who are infected
but not yet capable of transmitting the disease to others.
This compartment is biologically meaningful for pathogens like influenza and COVID-19 with non-negligible
latent periods.

The standard SEIR model is a deterministic, continuous-time process governed by a set of differential equations.
These equations contain three parameters, which are conventionally treated as stationary.
Most important for our purposes is $\beta$, the contact rate.
$\beta$ is the average number of contacts per unit time for an infectious individual, multiplied by the probability that a susceptible-infectious contact transmits the disease.
In addition, exposed individuals progress to infectious at rate $\sigma$, and infectious individuals recover at rate $\gamma$.
Implicitly, these assume waiting times are exponentially distributed, with mean latent and infectious periods of $1/\sigma$ and $1/\gamma$.
We will later relax this assumption.
Note that per-contact transmissibility is constant conditional on being in the $I$ state; the model does not allow infectiousness to vary with time since infection onset.

The basic reproduction number $R_0$ supposes the entire population is susceptible.
In a compartmental model, this is true at the onset of the epidemic, $t=0$.
There, $\beta$ is the daily number of transmissions per susceptible. 
$R_0$ is the lifetime number of secondary transmissions, so its formula merely scales $\beta$ by the average duration of infectiousness.
In simple compartmental models this is $1/\gamma$, so \mbox{$R_0 = \beta/\gamma$}.

In contrast, the effective reproduction number describes conditions at time $t$---particularly, the fact that not all individuals are susceptible.
Prevalence $S_t$ decays over the course of the epidemic, from an initial value around $N$.
Therefore, mechanistic $R_t$ in its simplest form is
\begin{equation}\label{eq:mech-rt-basic}
    R_t^M = \frac{\beta S_t }{\gamma N}.
\end{equation}

Mechanistic $R_t$ is the average number of secondary transmissions over the lifetime of an \textit{infectious} individual at time $t$, assuming conditions remain the same. 
(In reality, conditions will inevitably change as susceptibles $S_t$ deplete, but this is a small effect.) 
This notion of active infectiousness differentiates $R_t^M$ from the $R_t^C$ and $R_t^I$, which marked individuals based on their time of infection. 
We compare them in greater depth in the next section.

This definition of mechanistic $R_t$ is applicable to essentially any compartmental model.
We next present several ways in which the basic model presented above may be modified for practical purposes. 
Later, we use these modifications for $R_t$ estimation.

One important variation allows the rate parameters to be nonstationary. 
This is particularly useful for $\beta$, as behavior and policy dynamics affect the contact rate. 
We therefore depart from convention and allow $\beta_t$ to vary with time.
We call $\beta_t$ the effective contact rate, akin to the effective reproduction number $R_t$.
Without losing generality, we still assume the parameters that underlie waiting times are constant, since they are primarily biological.

Compartmental models can also be expressed in discrete-time.
Epidemic surveillance data typically arrives with daily resolution, so 
the continuous-time model is often approximated with an Euler step.
This first-order approximation uses day-to-day changes in each compartment's prevalence in place of its derivative.
These changes can be attributed to both inflow and outflow; 
we denote inflow ``incidence'' with asterisks.
For example, susceptible depletion can be written in terms of incident exposures $E_t^*$:
\begin{equation}
S_{t+1}=S_t - E_t^* \approx S_t + \frac{dS}{dt},\quad\text{where}\quad E_t^* \approx -\frac{dS}{dt} = \beta_t \frac{S_t I_t}{N}.
\end{equation}
Note that $E_t^*$ corresponds to $x_t$ in the prior section.

Mechanistic $R_t$ can also be revised for non-exponential waiting times between compartments.
This canonical assumption produces convenient differential equations: $I_t^* = \sigma E_t$ for infectious incidence and $R_t^* = \gamma I_t$ for recoveries.
However, the delay distributions they imply may have unrealistic variance, often too high.
Instead, define the exposure-to-infectious distribution $\pi^\text{EI}$ such that \mbox{$I_t^* = \sum_{s<t} E_s^* \pi_{t-s}^\text{EI}$} in a deterministic model;
Similarly define $\pi^\text{IR}$, the infectious-to-recovery distribution that relates \mbox{$R_t^* = \sum_{s<t} I_s^* \pi_{t-s}^\text{IR}$}.
This distribution has mean $\mu^\text{IR}$, which was $1/\gamma$ in the exponential case.
Combining $\mu^\text{IR}$ with $\beta_t$ yields our generalized definition of mechanistic $R_t$:
\begin{equation}\label{eq:mech-rt-cp}
    R_t^M = \beta_t\cdot \frac{S_t \mu^\text{IR}}{N}.
\end{equation}

Compartmental models may be stochastic, not deterministic.
Poisson models are the standard choice \citep{andersson2000stochastic},
as modeling exponential transitions as random implies a Poisson process for the number of events per unit time.
We impose this in our framework, though other noise models are possible.
The binomial distribution is a natural alternative, per the Reed-Frost model \citep{abbey1952examination}. 
However, Le Cam's theorem \citep{lecam1960approximation} shows that this can be well-approximated by a Poisson, when the population is large and per-timestep transition probabilities are small. 

Putting these pieces together---time-varying $\beta$, discrete-time approximation, arbitrary delay distributions, and stochastic noise---we arrive at the following mechanistic model:
\begin{subequations}\label{eq:seir-model}
\begin{alignat}{2}
    S_{t+1} &= S_t - E_t^*,         &\qquad E_t^* &\sim \text{Pois}\left(\beta_t \frac{S_t I_t}{N}\right) \label{eq:seir-s}\\
    E_{t+1} &= E_t + E_t^* - I_t^*, &\qquad I_t^* &\sim \text{Pois}\left(\sum_{s<t} E_s^* \pi_{t-s}^\text{EI}\right) \label{eq:seir-e}\\
    I_{t+1} &= I_t + I_t^* - R_t^*, &\qquad R_t^* &\sim \text{Pois}\left(\sum_{s<t} I_s^* \pi_{t-s}^\text{IR}\right) \label{eq:seir-i}\\
    R_{t+1} &= R_t + R_t^*.         &              & \label{eq:seir-r}
\end{alignat}
\end{subequations}
Henceforth, we use \cref{eq:seir-model} when discussing mechanistic $R_t$ and its estimation.

Finally, compartments can be added or dropped without changing the meaning of $R_t^M$.
The crucial piece is the SIR backbone, which encompasses $\beta_t$ and $\mu^\text{IR}$.
Nevertheless, we introduce mechanistic $R_t$ in terms of SEIR models because they are more justifiable than SIR.
A common extension is the SEIRD model, adding a fifth compartment for death ($D$).

\section{Equivalence of Instantaneous and Mechanistic \texorpdfstring{$R_t$}{Rt}}\label{sec:equivalence}


We now study the relationship between the definitions of $R_t$ introduced in \cref{sec:background}.
Our main result is as follows.

\begin{proposition}\label{prop:equivalence}
    Assume a population of size $N$ with homogenous mixing, meaning all contacts are equally likely to come into contact with one another. 
    Then instantaneous $R_t$ and mechanistic $R_t$ are equivalent: $R_t^I=R_t^M$.
\end{proposition}

Recall $R_t^I = \sum_{k> 0} w(t,k)$ depends on how the $k$-step transmission rate $w(t,k)$ is defined. 
We use the mean interpretation from \eqref{eq:fraser_alt}, as opposed to a data-driven notion that drops the expectation. 
Otherwise, the $R_t$ equivalence would only hold in expectation.

\begin{proof}

Expanding on the definition of $w(t,k)$,
\begin{align}
w(t,k) &= \mathbb{E}\left[\text{\# secondary transmissions at $t$}\mid \text{primary infection at $t-k$}\right]\\
    &=\mathbb{E}\left[\frac{1}{E_{t-k}^*}\sum_{i=1}^{E_{t-k}^*} \text{\# transmissions at $t$ of $i$'th infection at $t-k$}\right] \\
    &= \mathbb{E}\left[\frac{1}{E_{t-k}^*}\sum_{i=1}^{E_{t-k}^*} \sum_{j=1}^{S_{t}}\mathbf{1}\{\text{infection $i$ transmits to susceptible } j \text{ at } t\}\right]\\
    &= S_t \cdot \mathbb{P}(\text{$i$'th infection at $t-k$ transmits to $j$'th susceptible at $t$})
\end{align}

Three conditions must be met for transmission to occur. 
First, the $i$'th infection at $t-k$  must still be infectious by $t$. 
Since delay distributions are assumed stationary, this is equivalent to an infection being infectious $k$ days later.
Secondly, they must come into contact with the $j$'th susceptible.
Finally, secondary transmission must actually occur, as not all susceptible-infectious contacts result in infection.
Formalizing this and summing over all lags $k$  yields
\begin{align}
    R_t^I &= S_t \sum_{k> 0} {\mathbb{P}\!\left(\text{infectious at } t \mid \text{new infection at } t-k\right)} \\
    &\hspace{70pt}\cdot {\mathbb{P}(i \text{ contacts } j \text{ at } t)} \cdot {\mathbb{P} 
    (\text{transmission} \mid S\text{-}I \text{ contact})}\\
     &= S_t \cdot {\mathbb{P}(i \text{ contacts } j \text{ at } t)} \cdot {\mathbb{P} 
    (\text{transmission} \mid S\text{-}I \text{ contact})}\\
    &\hspace{70pt}\cdot\sum_{k> 0} {\mathbb{P}\!\left(\text{infectious $k$ days after infection}\right)}.
\end{align}

Next, we leverage the following two lemmas, proven below.
\begin{lemma}\label{lem:infectious_kernel}
For any delay distributions $\pi^\text{EI}$ and $\pi^\text{IR}$,
    $$\sum_{k> 0} {\mathbb{P}\!\left(\text{infectious $k$ days after infection}\right)} = \mu^\text{IR}.$$
\end{lemma}

\begin{lemma}\label{lem:homo_to_beta}
    Assuming homogeneous mixing, 
    $${\mathbb{P}(i \text{ contacts } j \text{ at } t)} \cdot {\mathbb{P} 
    (\text{transmission} \mid S\text{-}I \text{ contact})} = \beta_t / N.$$
\end{lemma}

Applying these lemmas, $R_t^I = S_t \cdot  \mu^\text{IR} \cdot \beta_t/N$.  
By definition, this is $R_t^M$  \eqref{eq:mech-rt-cp}, completing the proof. 
Again, $\mu^\text{IR} = \gamma^{-1}$ in the special case that $W^\text{IR} \sim \mathrm{Exp}(\gamma)$ or $\mathrm{Geom}(\gamma)$.
\cref{sec:compartmental-generation} provides further connections between instantaneous and mechanistic $R_t$.  

\end{proof}

\subsection{Proof of \texorpdfstring{\cref{lem:infectious_kernel}}{Lemma \ref*{lem:infectious_kernel}}}

We want to show that the probability of being infectious $k$ days after infection sums to $\mu^\text{IR}$, the mean duration of infectiousness.
We define this probability as
\begin{equation}\label{eq:infectious-kernel}
\zeta^\text{EI}_k := \Pprob(\text{infectious $k$ days after infection}).
\end{equation}
In an SIR model, this reduces to the survival function of $W^\text{IR}$. 
For SEIR models, a person is infectious $k$ days after infection if their latent period $W^\text{EI}$ was some $j\in\{1,\dots,k\}$ and their infectious duration $W^\text{IR}$ has not yet elapsed at age $k$. 
By the law of total probability and conditional independence of $W^\text{EI}$ and $W^\text{IR}$,
    \begin{align}
    \zeta^\text{EI}_k 
    &= \sum_{j=1}^{k} \Pprob(W^\text{EI}=j) \cdot \Pprob(W^\text{IR} > k-j)\\
    &= \sum_{j=1}^{k} \pi_j^\text{EI} \cdot \sum_{\ell=k-j+1} \pi_\ell^\text{IR}
\end{align}

To show that $\sum_k \zeta^\text{EI}_k=\mu^\text{IR}$, swap the order of summation:

\begin{align*}
    \sum_{k \geq 1} \zeta^\text{EI}_k &= \sum_{k \geq 1} \sum_{j=1}^{k} \Pprob(W^\text{EI} = j) \cdot \Pprob(W^\text{IR} > k-j)\\ 
    &= \sum_{j \geq 1} \Pprob(W^\text{EI} = j) \sum_{k \geq j} \Pprob(W^\text{IR} > k-j).
\end{align*}
The inner sum runs over $k = j, j+1, \dots$; substitute $a = k - j $ so that $a = 0,1, 2, \dots$:
\[
\sum_{k \geq j} \Pprob(W^\text{IR} > k-j) = \sum_{a \geq 0} \Pprob(W^\text{IR} > a).
\]
This is independent of $j$, so it factors out of the outer sum:
\[
\sum_{k \geq 1} \zeta^\text{EI}_k = \sum_{j \geq 1} \Pprob(W^\text{EI} = j) \cdot \sum_{a \geq 0} \Pprob(W^\text{IR} > a) = 1\cdot \mu^\text{IR}.
\]
The last line invokes the total mass of a distribution and the fact that the survival function of a non-negative random variable sums to its mean. 

\subsection{Proof of \texorpdfstring{\cref{lem:homo_to_beta}}{Lemma \ref*{lem:homo_to_beta}}}


Define $C_t$ as a random variable for the average number of contacts at time $t$.
Define its mean:
\[
\mu_t = \mathbb{E}[C_t] = \sum_{n=0}^N n \cdot P(C_t = n).
\]
Under the assumption of homogeneous mixing, all contacts are equally likely, and the fact that $i$ is infectious and $j$ is irrelevant. 
We decompose the contact probability using the law of total probability:
\begin{align}
     \mathbb{P}(\text{Infectious $i$ contacts Susceptible $j$ at $t$}) &=  \mathbb{P}(\text{$i$ contacts $j$ at $t$})\\
     &= \sum_{n=0}^N \mathbb{P}(\text{$i$ contacts $j$ at $t$}\given C_t=n)\Pprob(C_t=n)\\
     &= \sum_{n=0}^N \frac{n}{N}\, P(C_t = n) = \frac{\mu_t}{N}.
\end{align}

Finally, recall the effective contact rate is defined as  \mbox{$\beta_t = \mu_t \cdot \Pprob(\text{transmission} \mid S\text{-}I \text{ contact at } t)$}. Hence,
\begin{align}
     &\mathbb{P}(\text{Infectious $i$ contacts Susceptible $j$ at $t$})\cdot
\mathbb{P}(\text{transmission} \mid S\text{-}I \text{ contact}) \\&\hspace{50pt}= \mu_t/N \cdot \mathbb{P}(\text{transmission} \mid S\text{-}I \text{ contact at }t) \\
&\hspace{50pt}= \frac{\beta_t}{N}.
\end{align}


\section{The SEIR Generation Interval}\label{sec:compartmental-generation}

The equivalence result established that instantaneous and mechanistic $R_t$ agree under homogeneous mixing.
We now ask a complementary question: what generation interval distribution does the SEIR model imply?
Deriving this closes the loop between the two frameworks, expressing the renewal equation entirely in compartmental terms.
\begin{proposition}\label{prop:compartmental-generation}
    Recall the infectious kernel $\zeta^\text{EI}_k = \Pprob(\text{infectious $k$ days after exposure})$ \eqref{eq:infectious-kernel} for a discrete-time compartmental model. The generation interval distribution, which satisfies the renewal equation, has weights
    $$g_k := \Pprob(\text{transmission $k$ days after exposure $\vert$ transmission occurs}) = \zeta^\text{EI}_k/{\mu^{IR}}.$$
    Moreover, $\mu^\text{IR}$ is the total mass of the kernel $\zeta^\text{EI}$, so $g_k = \frac{\zeta_k^\text{EI}}{\sum_j \zeta_j^\text{EI}}.$
    Lastly, the mean generation time (proved in \cref{apx:mean-generation}) is
    \[
        \mathbb{E}[G] = \mu^\text{EI} + \frac{\mu^\text{IR}}{2} - \frac{1}{2} + \frac{(\sigma^\text{IR})^2}{2\, \mu^\text{IR}}.
    \]
\end{proposition}

\begin{proof}
    For now, suppose transmissions are deterministic.
In an SEIR model, $E_t^* = \beta_t \cdot \frac{S_t}{N} \cdot I_t$. 
Instantaneous $R_t$ is equivalent to mechanistic $R_t$, defined as $R_t^M = \beta_t \cdot \frac{S_t}{N} \cdot \mu^{IR}$;
plugging this in yields $E_t^* = R_t^I \cdot I_t / \mu^\text{IR}$.
Infectious prevalence $I_t$ can be defined using the kernel $\zeta^\text{EI}$:
\begin{equation}
I_t = \sum_{k \geq 1} E_{t-k}^* \cdot \zeta^\text{EI}_k.
\label{eq:prev-decomp}
\end{equation}
Substituting this expression in for $I_t$ yields the renewal equation \eqref{eq:renewal}, expressed in compartmental terms:
\begin{equation}\label{eq:mech-renewal-halfway}
    E_t^* = R_t^I/ \mu^\text{IR} \cdot \sum_{k \geq 1} E_{t-k}^* \cdot \zeta^\text{EI}_k =  R_t^I \cdot \sum_{k \geq 1} E_{t-k}^* \cdot (\zeta^\text{EI}_k / \mu^\text{IR}).
\end{equation}
Comparing with \eqref{eq:renewal}, the generation interval weights are $g_k=\zeta^\text{EI}_k / \mu^\text{IR}$:
the infectious kernel $\zeta^\text{EI}$, rescaled by the mean infectious duration $\mu^\text{IR}$.

To handle the stochastic case, 
recall that instantaneous $R_t$ and the generation interval are typically defined taking the renewal equation in expectation: \mbox{$\mathbb{E}[E_t^*\ \vert\ E_{<t}^*] = \sum_{k\geq 1} E_{t-k}^* g_k$}
To derive $\mathbb{E}[E_t^*\ \vert\ E_{<t}^*]$ in the compartmental perspective, we use the tower rule to impute $I_t$:
\begin{equation}
    \mathbb{E}[E_t^*\ \vert\ E_{<t}^*] = \mathbb{E}[\mathbb{E}[E_t^*\ \vert\ E_{<t}^*, I_t]]  = 
\frac{R_t}{\mu^\text{IR}}
\mathbb{E}[ I_t\ \vert\ E_{<t}^*] 
= \frac{R_t}{\mu^\text{IR}}
\sum_{k \geq 1} E_{t-k}^* \cdot \zeta^\text{EI}_k.
\end{equation}
Thus, we return to the renewal equation, with $g_k = \zeta_k^\text{EI}/\mu^\text{IR}$ defined probabilistically.
\end{proof}

\subsection{Discussion}

Renewal equation estimators are often presented as relatively assumption-light: given a generation interval $g$ and observed infections, one can estimate $R_t$ without specifying a full mechanistic model.
However, the results above reveal that the generation interval itself encodes compartmental structure.
The weights $g_k = \zeta_k^\text{EI}/\mu^\text{IR}$ depend on the latent and infectious period distributions $\pi^\text{EI}$ and $\pi^\text{IR}$, and the mean generation time depends on their first two moments.
In practice, choosing a generation interval for use in methods like EpiEstim implicitly assumes a model of disease progression, even if that model is never written down.
Making this dependence explicit, as we have done here, clarifies what assumptions are baked into standard $R_t$ estimates.

Our equivalence result relies on homogeneous mixing: all individuals are equally likely to contact one another.
Real populations exhibit substantial heterogeneity in contact patterns across age groups, geographic regions, and behavioral strata.
Extending the renewal--compartmental correspondence to structured populations is a natural direction, though in practice most $R_t$ estimation pipelines assume homogeneous mixing as a simplifying approximation \citep{gostic2020practical}.
\chapter{Fast, Frequentist Estimation of Epidemic Reproduction Numbers}
\label{ch4}

\section{Introduction}

Introduced in \cref{ch3}, the effective reproduction number $R_t$ is one of the most important metrics in epidemiology.
It is critical to estimate it accurately and efficiently.
Most methods use a Bayesian framework: priors simultaneously encode beliefs about $R_t$ and regularize estimates over time. 
One strategy, exemplified by EpiEstim \citep{cori2013new}, substitutes observed case reports for daily transmissions, effectively estimating $R_t$ at a lag. 
An alternative, exemplified by EpiNow2 \citep{abbott2020estimating}, explicitly models the latent infection process and the observation delay in a Bayesian generative framework. This yields more principled estimates, but at considerable computational cost and with sensitivity to prior assumptions.

We instead propose ConvRt, a frequentist estimator that replaces prior-based regularization with spline-based smoothing. Like EpiNow2, it models latent infections directly, but avoids MCMC sampling. We show that ConvRt not only matches or outperforms leading Bayesian methods in accuracy, but does so at a fraction of the computational cost.

Bayesian $R_t$ models typically conflate two jobs: expressing substantive beliefs about transmission and controlling the roughness of the estimated trajectory. Our frequentist framework separates them. A smoothness hyperparameter controls the temporal resolution of the estimate, while a separate tail parameter governs how aggressively recent trends are extrapolated. Practitioners can tune each independently, enabling the scenario analyses we demonstrate later in the chapter.

The rest of this chapter proceeds as follows. \cref{sec:rt-background} reviews existing $R_t$ methods. \cref{sec:rt-methods} develops our estimator, including retrospective and real-time versions; the latter employs the same tapered tail regularization developed for severity rates in \cref{ch:paper2}. \cref{sec:sim-study,sec:rt-experiments-real} evaluate on simulated and real surveillance data from influenza and COVID-19, demonstrating improvements in runtime, accuracy, and coverage.

\subsection*{Notation}\label{sec:notation}

Notation is summarized in \cref{tab:notation}.
\begin{table}[!ht]
\centering
\caption{Estimation notation for ConvRt (Chapter~5).}
\label{tab:notation}
\renewcommand{\arraystretch}{1.15}
\begin{tabularx}{\linewidth}{@{}l >{\raggedright\arraybackslash}X@{}}
\toprule
\textbf{Symbol} & \textbf{Meaning} \\
\midrule
\multicolumn{2}{@{}l}{\textsc{Reproduction numbers}} \\
\addlinespace[2pt]
$R_0$ & Basic reproduction number \\
$R_t$ & Effective reproduction number (instantaneous) \\
\midrule
\multicolumn{2}{@{}l}{\textsc{Infection counts}} \\
\addlinespace[2pt]
$x_t$ & Infections at $t$ \\
$y_t$ & Observed counts (e.g.\ cases, hospitalizations) at $t$ \\
$g$ & Generation interval distribution \\
$\pi$ & Infection-to-report delay distribution \\
\midrule
\multicolumn{2}{@{}l}{\textsc{Observation model}} \\
\addlinespace[2pt]
$\rho_t$ & Case ascertainment rate (proportion of infections reported at $t$) \\
$\omega_d$ & Day-of-week multiplicative effect, $d \in \{0,\ldots,6\}$ \\
$\varphi$ & Quasi-Poisson dispersion parameter \\
$\Lambda_t$ & Unscaled conditional mean of $y_t\given x_{<t}$ \\
$\mu_t$ & Conditional mean of $y_t\given x_{<t}$ \\
\midrule
\multicolumn{2}{@{}l}{\textsc{Estimation}} \\
\addlinespace[2pt]
$\theta, \omega$ & Spline coefficients and day-of-week effects \\
$S$ & Spline basis matrix; $R_t(\theta) = S(t)^\top \theta$ \\
$Z_t$ & Covariate row at $t$ from \eqref{eq:linear-mean-cases} \\
$\lambda$ & Curvature-penalty hyperparameter \\
$\gamma$ & Tail-smoothness hyperparameter (real-time setting) \\
$\Omega, \Psi$ & Curvature and tapered-smoothness penalty matrices \\
$D^{(m)}$ & Finite difference operator of order $m$ \\
\bottomrule
\end{tabularx}
\end{table}

\section{Background}\label{sec:rt-background}

\cref{ch3} introduced the concept of reproduction numbers: the average number of secondary infections generated by a single primary infection.
It compared multiple definitions of $R_t$, showing conditions under which two perspectives are equivalent.
Early lines of work aimed to estimate $R_t$ under the mechanistic definition, but these encoded strong structural assumptions that are typically unrealistic.
We focus on the instantaneous definition in this chapter, in step with most methods developed in recent decades.
While we refer the reader to \cref{ch3} for a thorough description, we restate the renewal equation here.
With the generation interval distribution $g$, infections are generated via
\[
    \mathbb{E}[x_t\given x_{< t}] = R_t \sum_{s<t} x_{s} g_{t-s}.
\]

Strategies to estimate reproduction numbers must contend with the fact that infections $x_t$ are typically latent. 
Instead, the observed data $y_t$ reflects infections at a delay---namely, positive case reports or hospitalizations.
$R_t$ estimators differ in whether they account for this structure.
Early instantaneous-$R_t$ methods ignored these delays, treating observed reports as a proxy for infections.
Later approaches, developed during the COVID-19 pandemic, handled the latency explicitly, most by modeling infections with computationally expensive MCMC sampling.
Almost all are Bayesian, with various strategies for modeling smoothness.

We now elaborate on several $R_t$ methods, which we revisit later as baselines.
All have open-source \texttt{R} implementations and are widely used in practice, spanning the major structural choices in the contemporary literature.

\textbf{EpiNow2}, the basis for the CDC's
published real-time $R_t$ estimates for seasonal influenza
\citep{cdccfa2026rt}, is a fully Bayesian generative model with three
components: a Gaussian-process prior on log-$R_t$, a renewal equation
linking $R_t$ to latent infections, and a reporting-delay convolution
mapping infections to observed counts \citep{abbott2020estimating}. Inference is by Hamiltonian
Monte Carlo via Stan.

\textbf{EpiEstim}, introduced by \citet{cori2013new}, anchors most instantaneous-$R_t$
software, and is among the most widely deployed $R_t$ tools in
epidemiological practice \citep{nash2023real}. It slides a
gamma--Poisson conjugate update along the renewal equation, treating
observed reports as the time series of new infections. The sliding
window serves as the smoother, and the gamma prior on $R_t$ is updated
to a gamma posterior in closed form.

\textbf{EstimateR}, developed at ETH Zürich
for the Swiss COVID-19 surveillance program, prepends a deconvolution
stage to EpiEstim \citep{scire2023estimateR}. Observed counts are first smoothed by LOESS, then
deconvolved against the case-to-report delay with the Richardson--Lucy
algorithm. The resulting infection series is fed into EpiEstim's
sliding-window step. Uncertainty quantification uses a block bootstrap
over the LOESS residuals, rather than EpiEstim's native posterior.

\textbf{EpiLPS} parameterizes log-$R_t$ as a
penalized B-spline (P-spline) with a second-difference roughness
penalty, and fits the renewal-equation likelihood by Laplace
approximation around the posterior mode \citep{gressani2022epilps}. The smoothing parameter is learned from the data via its hyperprior.

\textbf{Rtestim}, the most recent baseline, applies trend filtering to renewal-equation $R_t$ estimation \citep{rtestim}. 
Confidence bands are obtained by inverting a
quadratic relaxation of the objective. 
Among the baselines, it is the closest to ours: both are frequentist penalized-likelihood estimators producing piecewise-cubic $R_t$. However, it does not model latent infections, and its $\ell_1$ penalty allows locally sharper changes than our globally smooth fit.

\textbf{CovidEstim}, which produced publicly
available $R_t$ trajectories throughout the pandemic, is a Bayesian
generative model specific to COVID-19 \citep{Chitwood2022}. Modeling log-$R_t$ as a random walk, 
CovidEstim fits jointly to case and death series by Hamiltonian Monte Carlo.
Its unique contribution is to jointly learn time-varying case ascertainment rates from seroprevalence data.
This comes at a significant computational cost, often taking 10 hours to fit \citep{covidestim_summer2021}. 
We include
it only as a reference on real COVID-19 data
(\cref{sec:rt-experiments-real}), where we compare against its
publicly released $R_t$ outputs rather than refitting it.

\section{Estimating \texorpdfstring{$R_t$}{Rt} via Deconvolution}\label{sec:rt-methods}

We propose a two-stage deconvolution procedure which first estimates latent infections, and then $R_t$.
Each stage parameterizes its time series as a spline, fitting its coefficients with penalized maximum-likelihood.
This approach is fast to fit and admits standard frequentist confidence intervals.

\cref{sec:rt-seir} introduces the conditional distribution for observed cases, used in \cref{sec:inf-deconv} to deconvolve latent infections.
\cref{sec:rt-retro} expresses this distribution in terms of $R_t$, and proposes a retrospective
estimator; \cref{sec:rt-rt} extends it for the real-time setting.
\cref{sec:rt-uncertainty} introduces our approaches for uncertainty quantification.
\cref{sec:method-contrasts} situates our approach among existing
$R_t$ estimators.

\subsection{Convolutional Distribution of \texorpdfstring{$y_t\given x_{<t}$}{y\_t | x\_\{<t\}}}\label{sec:rt-seir}

In general, we observe epidemic incidence $y_t$, such as positive cases or hospitalizations.
Each infection is recorded several days after it actually occurs.
For example, COVID-19 cases were reported around 5 days after symptom onset \citep{abbott2020estimating}, which itself came 5 days after exposure \citep{kucharski2020early}.
To derive the mean report $\mu_t$, define $\pi$ as the infection-to-report delay distribution:
$$\pi_k = \mathbb{P}(\text{Infection reported after $k$ timesteps}\given\text{Infection reported}),\qquad\forall\ k>0.$$

Throughout this section, we treat $\pi$ as known. 
In practice, it is often estimated from line-list data or household contact tracing studies, and may be misspecified.
For ease of notation, we also treat $\pi$ as stationary.
This may in truth be time-varying, and could be estimated as such.

To compute the mean report count $\mu_t$, the delay distribution $\pi$ is convolved against trailing infection incidence $x_{<t}$.
If all infections are eventually observed, this is simply \mbox{$\mu_t = \sum_{s<t} x_{s} \pi_{t-s}$}.
However, this is not the case in most epidemics. 
\citet{reese2021estimated} estimated a case ascertainment rate of 13\% for COVID-19 through September 2020;
only a small fraction of influenza cases are actively recorded, generally through hospitalization pipelines like FluSurv-NET \citep{cdccfa2026rt}.
Letting the ascertainment rate $\rho_t$ be the proportion of observed infections at $t$, mean cases are thus
$$\mu_t = \rho_t\sum_{s<t} x_{s} \pi_{t-s}.$$

We impose a Poisson noise model on reported counts $y_t$.
The Poisson rate is the mean, so \mbox{$y_t\given x_{<t}\sim \text{Poisson}(\mu_t)$}.
This is a standard model in practice (e.g., \citet{cori2013new}), though some works' likelihoods allow for overdispersion (e.g., \citet{abbott2020estimating}).
The Poisson model can be relaxed to accommodate quasi-Poisson noise, characterized by a dispersion term $\varphi$.
This extension does not affect parameter estimation, so we maintain Poisson distributions in the following expressions. 
Quasi-Poisson inference simply inflates the Poisson variance, estimating $\varphi$ with the Pearson chi-squared statistic.

We are interested in the joint distribution of $y_t$ from $t=t_0$ to $t_1$. 
Conditional on the infection path $x_{<t_1}$, we assume reports are independent across time, so that
\begin{align}\label{eq:joint-likelihood}
    \mathbb{P}(y_{t_0:t_1} \given x_{<t_1}) &= \prod_{t=t_0}^{t_1} \text{Poisson}(y_t;\, \mu_t), \\
    \log\mathbb{P}(y_{t_0:t_1}\given x_{<t_1})&= \sum_{t=t_0}^{t_1} \left( y_t \log \mu_t - \mu_t \right) + \text{const}.
\end{align}
This working-likelihood assumption, standard for time series, underlies most $R_t$ estimators (e.g. \citet{abbott2020estimating}). 
Notably, it ignores any cross-time dependence in reporting noise.
To model some temporal idiosyncrasies, we optionally include day-of-week effects.
These terms scale expected counts by a factor of $\omega_{(\text{$t$ mod 7})}$ (near 1). 
Combining day-of-week effects, overdispersion, and quasi-Poisson noise, the per-timestep distribution is thus
\begin{equation}\label{eq:quasi-poisson-model}
y_t\given x_{<t} \sim \text{Quasi-Poisson}(\mu_t, \varphi), 
\qquad \mu_t = \omega_{(t \bmod 7)} \rho_t \Lambda_t, 
\qquad \Lambda_t = \sum_{s<t} x_{s} \pi_{t-s}.
\end{equation}

\subsection{Deconvolving Latent Infections}\label{sec:inf-deconv}

The joint distribution \eqref{eq:joint-likelihood} depends on incident infections $x_{<t_1}$, which are unobserved. 
Bayesian methods like EpiNow2 handle this by sampling these counts with MCMC.
We use a plug-in estimate obtained via deconvolution.
The method, presented in \cref{apx:deconvolution}, is analogous to the deconvolution problem for severity rates in \cref{ch:paper2}. 
Both solve for non-negative, smooth time series with regularized regression.

\citet{Jahja2022} proposed an estimator for this problem, using a least-squares loss and trend filtering penalty.
We instead minimize Poisson loss, motivated by the conditional distribution \eqref{eq:joint-likelihood}.
In addition, we parameterize the infection curve as a spline. Both approaches regress $x$ as a piecewise polynomial, but splines are globally smooth, whereas trend filtering can place sharp changes anywhere in the curve. The spline basis admits a closed-form penalized-likelihood solution, making it faster to fit. We compare the two further in \cref{apx:extra-methods}.
Fortunately, some degree of approximation error is tolerable, so long as errors in $x_{<t}$ are smoothed out in the convolution to $\mu_t$.

To deconvolve infections $x_t$, counts are inflated by a factor of $\rho_t^{-1}$.
If the ascertainment rate is a constant $\rho$, this cancels out with the factor in $\mu_t$ \eqref{eq:quasi-poisson-model}. 
Consequently, the rate itself is irrelevant and does not need to be specified.
That stationarity assumption is plausible in certain settings.
For example, the flu hospitalization rate should be relatively stationary, barring significant shifts in variant proportions and severity.
Otherwise, it is necessary to plug in ascertainment rates estimated \textit{a priori}. 
Jointly solving for ascertainment rates and infections is an underspecified problem, but $\rho_t$ may be estimated separately from seroprevalence data \citep{reese2021estimated, Chitwood2022}.

\subsection{\texorpdfstring{$R_t$}{Rt} Deconvolution in Retrospect}\label{sec:rt-retro}

Having estimated $\hat x_{<t_1}$, we now turn our focus to $R_t$. 
Throughout this section, we assume that the epidemic unfolds according to the renewal equation \eqref{eq:renewal}.
To reiterate, transmissions occur as \mbox{$\mathbb{E}[x_t\given x_{<t}]=R_t\sum_{s<t} x_{s} g_{t-s}$}.
As with $\pi$, we assume the generation interval distribution $g$ is known and stationary.

The Poisson rate in \eqref{eq:quasi-poisson-model} includes the convolution \mbox{$\Lambda_t = \sum_{s<t} x_{s} \pi_{t-s}$}.
To express this in terms of $R_t$, we replace $x_{s}$ with the renewal equation.
This is by nature an approximation, since we replace a quantity with its mean.
Further approximating with plug-in infections, the convolution becomes
\begin{equation}\label{eq:exp-cases-rt}
    \Lambda_t \approx \sum_{s<t} R_s \left(\sum_{u<s} \hat x_u\, g_{s-u}\right) \pi_{t-s}
\end{equation}

Uncertainty quantification is important for $R_t$, since small differences can have large epidemiological consequences.
In \cref{ch:paper2}, we performed deconvolution with trend filtering, which is excellent at curve-fitting but does not lend itself to uncertainty quantification. For that reason, we instead parameterize $R_t$ as a spline. We compare the two approaches in \cref{apx:extra-methods}.

Let $S$ be the spline basis matrix with rows $S(t)$, and $\theta$ the coefficients, so that $R_t(\theta) = S(t)^\top \theta$. In our experiments we use a natural cubic spline basis---essentially a smoothing spline with a reduced knot set for computational efficiency. This framework, however, accommodates other parameterizations (e.g., B-splines, P-splines) without modification. Substituting into \eqref{eq:exp-cases-rt} makes the expected counts linear in $\theta$:
\begin{equation}\label{eq:linear-mean-cases}
    \Lambda_t(\theta) = \sum_{s<t} R_s(\theta) \left(\sum_{u<s} \hat x_u\, g_{s-u}\right) \pi_{t-s} =
             \theta^\top \underbrace{\sum_{s<t} S(s) \left(\sum_{u<s} \hat x_u\, g_{s-u}\right) \pi_{t-s}}_{Z_t} = \theta^\top Z_t.
\end{equation}

Therefore, deconvolution amounts to fitting a penalized identity-link Poisson GLM. We minimize the negative log likelihood subject to constraints and a smoothness penalty. $\theta$ is constrained to produce non-negative $R_t$ values, and $\omega$ to produce valid day-of-week effects. Smoothness is imposed through the quadratic term $\theta^\top \Omega\, \theta$. By default, we take $\Omega$ to give the integrated second-derivative penalty,
\mbox{$\theta^\top \Omega\, \theta = \int_{t_0}^{t_1}\!\bigl[R''(u)\bigr]^2\,du$,}
familiar from the smoothing spline literature. However, any quadratic penalty $\theta^\top \Omega\, \theta$ is admissible; for example, P-splines penalize second differences of the coefficients. Alternatively, with a sparse enough knot set, the penalty can be dropped entirely.

Applying \eqref{eq:linear-mean-cases} to the joint likelihood \eqref{eq:joint-likelihood}, we solve for $\theta$ and $\omega$:
\begin{align}\label{eq:glm-retro}
\hat{\theta}, \hat\omega &= 
\argmin_{\substack{\bar{\omega}_{GM}=1 \\ S\theta \succeq 0}}\
             \sum_{t=t_0}^{t_1}\!\bigl[ \mu_t(\omega, \theta) - y_t \log \mu_t(\omega, \theta) \bigr]
             \;+\; \lambda\,\theta^\top \Omega\, \theta,\\
             &\qquad\text{where}\quad \mu_t(\omega, \theta) = \omega_{(\text{$t$ mod 7})}\,\rho_t\, \Lambda_t(\theta) 
             \quad\text{and}\quad \Lambda_t(\theta) = Z_t^\top\theta.
\end{align}
This can be fit quickly with off-the-shelf IRLS solvers.
Details on optimization are in \cref{apx:optim}.

We may tune the smoothness hyperparameter $\lambda$ with the same approach as introduced in \cref{ch:paper2}. 
To recap, $K$-fold cross-validation holds out every $K\nth$ timestep from the loss. 
The model is fit on the in-sample timesteps, normalizing each term in \eqref{eq:glm-retro} by the number of summands.
Each fold's $R_t$ values are then convolved to predict expected incidence on the held-out timesteps.
The cross-validated error curve is the Poisson deviance of these predictions \citep{rtestim}.
To avert undersmoothing, we tune $\lambda$ via the 1se rule, selecting the largest $\lambda$ whose cross-validated error is within one standard error of the minimum.

\subsection{Real-Time Methods}\label{sec:rt-rt}

Recall that in the retrospective setting, $R_t$ is estimated through time $t_1$, but the final predictions before $t_1$ are discarded. In contrast, these estimates are of utmost importance in the real-time setting, where $t_1$ denotes the present. On the days leading up to $t_1$, a growing share of infections will not surface as reports until \textit{after} $t_1$, so the $R_t$ signal in the available data is attenuated. As with severity rates in \cref{ch:paper2}, predictions at the tail are highly unstable without extra regularization.
We address this in two ways.
First, parametrizing $R_t$ as a smoothing spline constrains the tail to be linear past the boundary knots. 
This adds a degree of stability without
any explicit constraint on $\theta$.
(As an alternative, \cref{apx:optim-constraints} shows how to enforce a constant fit.)

Second, we add a tapered smoothing penalty to the loss, borrowing again from \cref{ch:paper2}. 
This progressively encourages tail predictions towards stationarity.
Squared first-order differences in $R_t$ are penalized with a weight
proportional to the inverse CDF of the exposure-to-observation delay distribution.
Here, the final timesteps correspond to early timesteps of the delay distribution, at which little mass has accumulated; therefore the inverse CDF is high, so nonstationarity is heavily penalized.
Aggregating root weights into diagonal matrix $W$, define \mbox{$\Psi = D^{(1)\top} W^\top W D^{(1)}$} such that \mbox{$\theta^\top \Psi \theta = \|W D^{(1)}\theta\|_2^2$}. 
Together, these real-time penalties modify \eqref{eq:glm-retro} to

\begin{align}\label{eq:glm-real-time}
\hat{\theta}, \hat\omega &= 
\argmin_{\substack{\bar{\omega}_{GM}=1 \\ S\theta \succeq 0}}\
     \sum_{t=t_0}^{t_1}\!\bigl[\mu_t(\omega, \theta)-y_t \log \mu_t(\omega, \theta) \bigr]
     \;+\; \lambda\,\theta^\top \Omega\, \theta 
     \;+\; \gamma \,\theta^\top \Psi \, \theta,\\
     &\qquad\text{where}\quad \mu_t(\omega, \theta) = \omega_{(\text{$t$ mod 7})}\,\rho_t\, \Lambda_t(\theta) 
     \quad\text{and}\quad \Lambda_t(\theta) = Z_t^\top\theta.
\end{align}

The tail smoothness hyperparameter, $\gamma$, may also be tuned exactly as in \cref{ch:paper2}. 
Recall we used forward-validation to estimate out-of-sample error on each candidate $\gamma$.
For each timestep in the recent past (e.g. 7 days),
we fit our estimator, then convolve the $R_t$ estimate to predict observations one day ahead.
Aggregating the absolute error over the validation window, we select $\gamma$
via the min- or 1se rule. 

Real-time surveillance counts are typically incomplete at the most recent timesteps due to reporting delays, and are revised upward as backfill accumulates. We discuss this limitation further in \cref{sec:ch4-discussion}.

\subsection{Inference}\label{sec:rt-uncertainty}

Because $R_t(\theta) = S(t)^\top\theta$ is linear in the spline coefficients, estimating $\mathrm{Cov}(\hat\theta)$ is enough to produce confidence intervals for $R_t$ via $\widehat{\mathrm{Var}}(\hat R_t) =
   S(t)^\top\,\widehat{\mathrm{Cov}}(\hat\theta)\,S(t)$. 
In \cref{apx:rt-uncertainty}, we derive $\widehat{\mathrm{Cov}}(\hat\theta)$ as a sandwich estimator for the penalized Poisson score equations of \eqref{eq:glm-retro}.
Day-of-week effects, quasi-Poisson overdispersion, and plug-in asymptotics also preserve this structure with minor modifications.
Invoking the central limit theorem, we thus produce pointwise Wald confidence intervals,
\begin{equation}\label{eq:wald}
    \hat R_t \;\pm\; z_{1-\alpha/2}\,\mathrm{SE}(\hat R_t).
\end{equation}
Simultaneous coverage bands, which guarantee all timesteps are covered with high probability, may also be attained.
To do so, $z_{1-\alpha/2}$ is replaced with a simulation-based critical value that accounts for the correlation among adjacent $\hat R_t$.

One caveat is worth flagging.
Even asymptotically, these intervals cover $\mathbb{E}[\hat\theta]^\top S(t)$ rather than the true $R_t$. 
The two coincide only when $R_t$ lies in the span of the spline basis and $\hat\theta$ is unbiased---a caveat shared with essentially all nonparametric inference, not specific to our approach.

In the real-time setting, tail regularization biases $R_t$ in order to stabilize predictions.
As a result, the variance-only intervals \eqref{eq:wald} are often too narrow near the boundary. 
To address this, we obtain real-time confidence intervals with conformal inference.
\cref{apx:rt-realtime-ci} details our approach, as well as less effective alternatives.
These uncertainty bands are not necessarily calibrated, since their assumption of exchangeable data does not hold for time series, and they rely on a surrogate ground truth.
Nevertheless, empirical results demonstrate their practical utility, especially relative to the Wald baseline.

\subsection{Comparison with Existing Methods}\label{sec:method-contrasts}

Having derived our estimator, we now situate it among the methods introduced in \cref{sec:rt-background}.
We will compare their performance in \cref{sec:sim-study,sec:rt-experiments-real}. 
\Cref{tab:method-contrast} summarizes the
contrasts.
Notably, only one of the baselines deconvolves latent infections, like ours.
All but one are Bayesian, as opposed to frequentist. 
EpiLPS and rtestim both parameterize $R_t$ as a piecewise polynomial, predating our approach.

\begin{table}[!h]
\centering
\small
\caption[Structural comparison of $R_t$ estimators.]{Structural comparison of $R_t$ estimators benchmarked in this
chapter. ``Latent infections'' indicates
how the method recovers latent infections from observed reports, if at
all. ``$R_t$ form'' is the functional class for $R_t$ (or log-$R_t$,
where applicable). All methods are run on both retrospective and weekly
real-time vintages except CovidEstim, for which we use the project's
publicly released outputs; configurations are detailed in
\cref{apx:methods-config}.}
\label{tab:method-contrast}
\begin{tabular}{lccc}
\toprule
Method      & Latent infections & $R_t$ form        & Inference   \\
\midrule
\textbf{ConvRt (ours)} & Deconvolution     & Spline             & Frequentist \\
EpiNow2     & MCMC              & Gaussian process   & Bayesian    \\
EpiEstim    & --                & Sliding window     & Bayesian    \\
estimateR   & Deconvolution     & Sliding window     & Bayesian    \\
EpiLPS      & --                & Spline             & Bayesian    \\
rtestim     & --                & Trend filter       & Frequentist \\
CovidEstim  & MCMC              & Random walk        & Bayesian    \\
\bottomrule
\end{tabular}
\end{table}

\section{Experimental Results}

We ran ConvRt with knots spaced evenly every 5 days.
To tune $\lambda$, we used $K=5$-fold cross-validation and the min rule.
For the real-time setting, we tuned $\gamma$ with the min rule on 7-day forward-validation.
All methods receive the true generation interval $g$ and delay distribution $\pi$, granting oracle knowledge of the data-generating process.

\subsection{Synthetic-Data Experiments}\label{sec:sim-study}

\subsubsection*{Benchmarking Setup}

We benchmark ConvRt against five methods detailed in \cref{sec:rt-background}---EpiEstim, estimateR, EpiNow2, EpiLPS, and rtestim---on six synthetic datasets for which the true $R_t$ is known.
CovidEstim was omitted due to its reliance on seroprevalence data.
Most of our simulation pipelines generate counts without day-of-week effects or overdispersion.

We evaluate point predictions with mean absolute error (MAE),
and uncertainty quantification with $\ell_1$ calibration error (CE).
Both metrics are defined between burn-in and burn-out periods, ensuring methods have sufficient data with which to estimate $R_t$.
For a nominal central level $\alpha \in (0,1)$, let $\hat c(\alpha)$ denote
the empirical coverage of the corresponding two-sided prediction interval.
We define calibration error
\[
\mathrm{CE} \;=\; \frac{1}{|\mathcal{A}|}\sum_{\alpha \in \mathcal{A}} \bigl|\hat c(\alpha) - \alpha\bigr|,
\qquad \mathcal{A} = \{0.5,\,0.6,\,0.7,\,0.8,\,0.9,\,0.95\},
\]
which averages the absolute gap between empirical and nominal coverage over a grid of central levels.
A perfectly calibrated method has $\mathrm{CE}=0$; values above zero penalize both over- and under-coverage symmetrically, without rewarding narrow intervals for their own sake.
While frequentist confidence intervals and Bayesian credible intervals have different interpretations, we nevertheless score them according to the same coverage criterion.

Four of our synthetic $R_t$ benchmarks are introduced in \citet{rtestim}.
That paper proposed the rtestim method, so we refer to these as the rtestim datasets.
They pose unique estimation challenges and mimic some real-world phenomena.
Three have jump discontinuities in $R_t$, which could represent sweeping policy changes like school closures or lockdowns.
The fourth is sinusoidal, modeling multiple smooth yet rapid $R_t$ shifts.
Our experiments on these benchmarks only evaluated methods' retrospective predictions, matching rtestim.

The other two simulated datasets are based on real seasonal influenza data (see \cref{fig:flu-sim-both} in \cref{apx:flu-sim}).
We simulated hospitalizations using the delay distributions described in \cref{sec:rt-experiments-real} and a 1.5\% severity rate \citep{reed2015estimating}.
For the first dataset's ground truth, we used $R_t$ predictions on the 2022/23 season, made by the method from \citet{wallinga_teunis}.
$R_t$ changes direction frequently in this season, posing an interesting challenge for estimation. 
For the second dataset, we heavily smoothed the $R_t$ curve with LOESS.
We also added a higher degree of noise to mimic real-world data.

\subsubsection*{Retrospective Performance}

  \begin{table}[ht]
  \centering
  \caption[Retrospective performance across simulation settings.]{Retrospective performance across simulation settings. Mean absolute error (MAE)
  and $\ell_1$ calibration error (CE) reported in units of $10^{-2}$; runtime is for a single
  fit on the full input series. Influenza benchmarks run from October 1 to January 31.}
  \label{tab:combined-retro}
  \setlength{\tabcolsep}{5pt}
  \begin{tabular}{l l cccccc}
  \toprule
  & & \textbf{ConvRt} & EpiNow2 & EpiEstim & estimateR & EpiLPS & rtestim \\
  \midrule
  \multicolumn{8}{l}{\textbf{\textit{rtestim benchmarks}}} \\
  \midrule
  \multirow{3}{2.2cm}{Piecewise Constant}    & MAE     & 1.58    & 2.64     & 4.86    & 1.93    & 3.96    & 3.37    \\
                                             & CE      & 9.75    & 6.89     & 11.84   & 19.69   & 15.08   & 20.03   \\
                                             & Runtime & 1.25\,s & 117.8\,m & 0.03\,s & 0.18\,s & 0.07\,s & 22.9\,s \\
  \addlinespace
  \multirow{3}{2.2cm}{Piecewise Exponential} & MAE     & 1.77    & 2.14     & 11.78   & 2.42    & 8.45    & 8.60    \\
                                             & CE      & 3.35    & 3.44     & 73.89   & 12.29   & 21.34   & 22.73   \\
                                             & Runtime & 1.21\,s & 45.1\,m  & 0.03\,s & 0.18\,s & 0.07\,s & 22.9\,s \\
  \addlinespace
  \multirow{3}{2.2cm}{Piecewise Linear}      & MAE     & 3.79    & 4.45     & 17.25   & 5.20    & 12.54   & 12.10   \\
                                             & CE      & 11.43   & 5.61     & 74.17   & 10.37   & 9.29    & 11.19   \\
                                             & Runtime & 1.18\,s & 130.5\,m & 0.03\,s & 0.18\,s & 0.07\,s & 22.9\,s \\
  \addlinespace
  \multirow{3}{2.2cm}{Periodic}& MAE     & 0.62    & 1.95     & 29.91   & 3.20    & 21.27   & 21.43   \\
                                             & CE      & 3.86    & 2.86     & 71.89   & 18.59   & 58.67   & 24.51   \\
                                             & Runtime & 1.26\,s & 38.1\,m  & 0.03\,s & 0.18\,s & 0.07\,s & 22.9\,s \\
  \midrule
  \multicolumn{8}{l}{\textbf{\textit{Influenza benchmarks}}} \\
  \midrule
  \multirow{3}{2.2cm}{Wiggly $R_t$} & MAE     & 0.84    & 0.83    & 7.40    & 1.74    & 4.64    & 4.46 \\
                                    & CE      & 7.36    & 11.47   & 52.49   & 25.66   & 66.71   & 22.58 \\
                                    & Runtime & 3.25\,s & 35.6\,m & 0.04\,s & 0.06\,s & 0.20\,s & 0.69\,s \\
  \addlinespace
  \multirow{3}{2.2cm}{Smooth $R_t$} & MAE     & 0.62    & 0.82    & 3.45    & 1.18    & 2.01    & 2.18 \\
                                    & CE      & 8.22    & 20.51   & 32.57   & 3.04    & 51.13   & 17.84 \\
                                    & Runtime & 2.38\,s & 46.1\,m & 0.04\,s & 0.07\,s & 0.17\,s & 0.61\,s \\
  \bottomrule
  \end{tabular}
  \end{table}

\begin{figure}[ht]
    \centering
    \includegraphics[width=.95\linewidth]{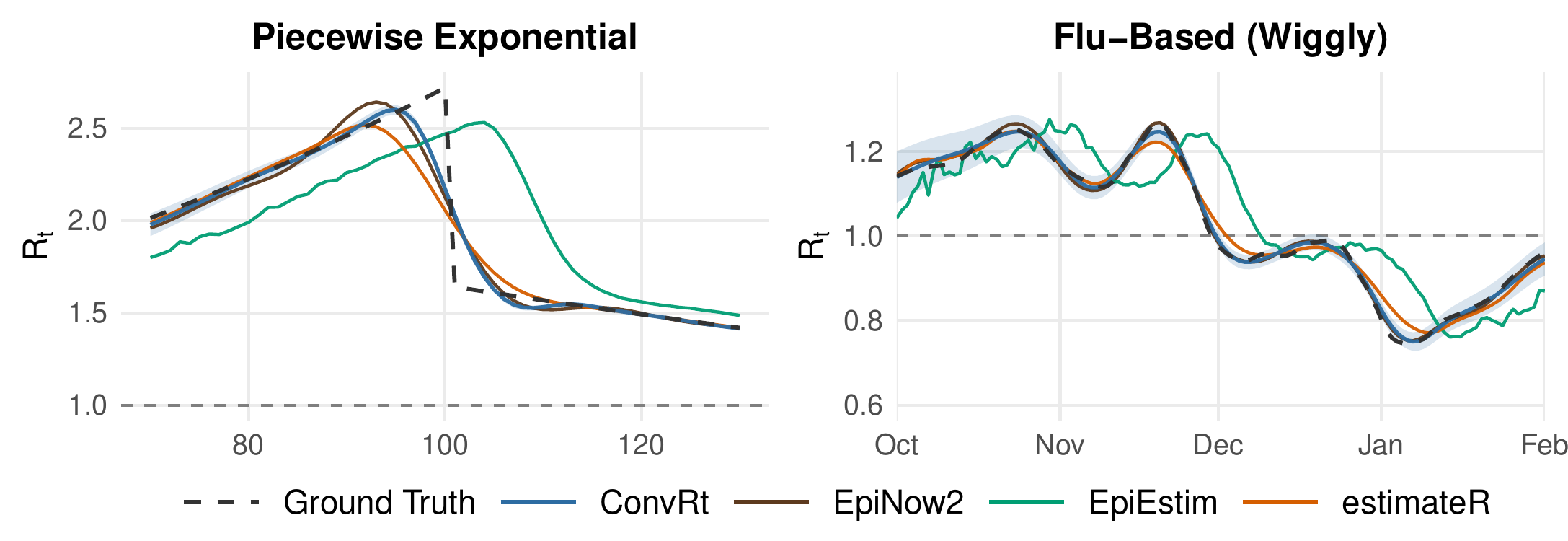}
    \caption{Comparison of retrospective $R_t$ estimates on benchmark datasets.}
    \label{fig:retro-comparison}
  \end{figure}

\Cref{tab:combined-retro} reports retrospective performance across all six benchmarks.
It reveals that EpiNow2 is ConvRt's only real rival on accuracy. 
Nevertheless, ConvRt reduces MAE by 20\% or more on 5 of 6 datasets, and is effectively tied on the last. 
In the best case---the periodic rtestim benchmark---ConvRt has nearly 70\% lower MAE than EpiNow2. 
\Cref{fig:retro-comparison} compares predictions on selected rtestim and influenza datasets.
Both ConvRt and EpiNow2 track the true $R_t$ curves closely.
Neither method has conclusively better uncertainty bands: each has lower CE in exactly half of the benchmarks.
Fortunately, both of them are relatively well-calibrated. 
With the exception of EpiNow2 on the smooth benchmark, CE always falls below 12, and on two benchmarks is below 5.

ConvRt is also several orders of magnitude faster than EpiNow2, running in seconds versus up to 2.5 hours.
Since a single fit is rarely the end of the analysis, this gap compounds: refitting under many hyperparameter settings multiplies EpiNow2's cost while leaving ConvRt's negligible.
EpiNow2 also fails to converge on many runs, lengthening the wait further.

Among methods that run in seconds, ConvRt wins on both performance metrics.
The methods that do not model latent infections track $R_t$ at a delay, and thus have much higher MAE and CE. \cref{fig:retro-comparison} omits EpiLPS and rtestim, which have similar bias to EpiEstim but are much less widely used.
estimateR fares better, since like ConvRt it deconvolves observations to recover latent infections.
However, it does not use a precise statistical model to estimate infections or $R_t$.
This results in less accurate predictions, for example oversmoothing the intervention badly in \cref{fig:retro-comparison}.
Compared to estimateR, ConvRt has over 20\% lower MAE on all datasets, and over 50\% on half.

ConvRt's main weakness is at the jump discontinuity, where its point estimates are too smooth and its confidence intervals fail to cover (\cref{fig:retro-comparison}).
Fortunately, such gaps are unlikely in practice, since interventions rarely have large instant effects.
Moreover, this has an easy fix.
Splines are not disposed to model sharp transitions between smooth periods, assuming instead that $R_t$ is piecewise cubic.
(Gaussian processes, assumed by EpiNow2, have similar smoothness constraints.)
To handle jumps in $R_t$, we augment its basis in ConvRt
with a step function that switches on at the intervention day. 
Performance then improves dramatically, with near-perfect predictions (\cref{apx:rtestim}).
The other methods could in principle be adapted to allow jumps, but we did not implement this.

\subsubsection*{Real-time Performance}

\begin{table}[ht]
\centering
\caption[Real-time performance on the influenza simulations.]{Real-time performance on the influenza simulations: last-7-day estimates
across 18 weekly vintages. MAE and $\ell_1$ calibration error (CE) in units of $10^{-2}$;
runtime is mean wall-clock time per vintage.}
\label{tab:combined-realtime}
\setlength{\tabcolsep}{5pt}
\begin{tabular}{ll cccccc}
\toprule
& & \textbf{ConvRt} & EpiNow2 & EpiEstim & estimateR & EpiLPS & rtestim \\
\midrule
\multirow{3}{*}{Flu (wiggly)} & MAE     & 4.14   & 4.29    & 7.16    & 7.17    & 6.45    & 6.70 \\
                        & CE      & 5.65   & 17.16   & 49.30   & 54.19   & 54.19   & 25.83 \\
                        & Runtime & 1.38 s & 24.3 m  & 0.01 s  & 1.36 s  & 0.11 s  & 0.44 s \\
\addlinespace
\multirow{3}{*}{Flu (smooth)} & MAE     & 2.51   & 2.14    & 3.67    & 4.24    & 2.98    & 5.09 \\
                        & CE      & 4.96   & 5.25    & 32.76   & 46.26   & 42.16   & 25.83 \\
                        & Runtime & 1.88 s & 24.9 m  & 0.02 s  & 1.36 s  & 0.12 s  & 0.42 s \\
\bottomrule
\end{tabular}
\end{table}

\Cref{tab:combined-realtime} reports real-time performance on the flu benchmark.
We refit weekly, with vintage end dates from October through January, and evaluate the last seven days of each.
ConvRt's confidence intervals use the conformalized approach of \cref{apx:rt-realtime-ci}.
Predictions are shown in \cref{apx:flu-sim}.

The retrospective picture carries over to the real-time case.
ConvRt and EpiNow2 remain the most accurate methods, with comparable point estimates.
Among methods that run in seconds, ConvRt wins on both metrics; the gap is widest on calibration, where every fast baseline has roughly $5$--$10\times$ its CE.

ConvRt also has advantages over EpiNow2 in the real-time setting.
It is much faster, taking about two seconds per vintage against EpiNow2's 25 minutes.
It has dramatically better calibration on the wiggly benchmark and comparable calibration on the smooth benchmark.
Later, we show a further advantage: its capacity for real-time scenario analysis.

\subsection{Real-Data Experiments}\label{sec:rt-experiments-real}

\subsubsection*{Setup}
Next, we demonstrate ConvRt's utility on real-world influenza data.
To run ConvRt, we tuned $\lambda$ with the 1se rule after noting a tendency for cross-validation to undersmooth given noisy counts.
Its predictions are consistent with those of established methods, while offering improved runtime.
We achieve similarly strong $R_t$ estimates for COVID-19, shown in \cref{apx:covidestim}.

The CDC posts flu $R_t$ in real-time, estimated by EpiNow2 \citep{cdccfa2026rt,abbott2020estimating}.
As in \cref{sec:setup2}, we use public hospitalization counts reported to HHS through the National Healthcare Safety Network (NHSN).
NHSN published inpatient hospitalizations for a variety of pathogens
from the beginning of the COVID-19 pandemic until May 2024.
We focus on the 2022/23 and 2023/24 flu seasons (\cref{fig:retro-flu-both}), since the pandemic flu seasons were considerably smaller.
In each season, we ran ConvRt and EpiNow2 from July 1 onwards.

We parameterized all delay distributions as discrete gamma,
choosing means and standard deviations to align with the literature. 
For $\pi^\text{EY}$, this was 5.7 and 2.3 days, respectively \citep{rousogianni2025clinical,noh2014viral}.
The generation interval had mean 3.2 and 1.6 \citep{chan2025estimating,cauchemez2009household,cowling2009estimation}.

\subsubsection*{Retrospective Results}

\begin{figure}[!h]
    \centering
    \includegraphics[width=\linewidth]{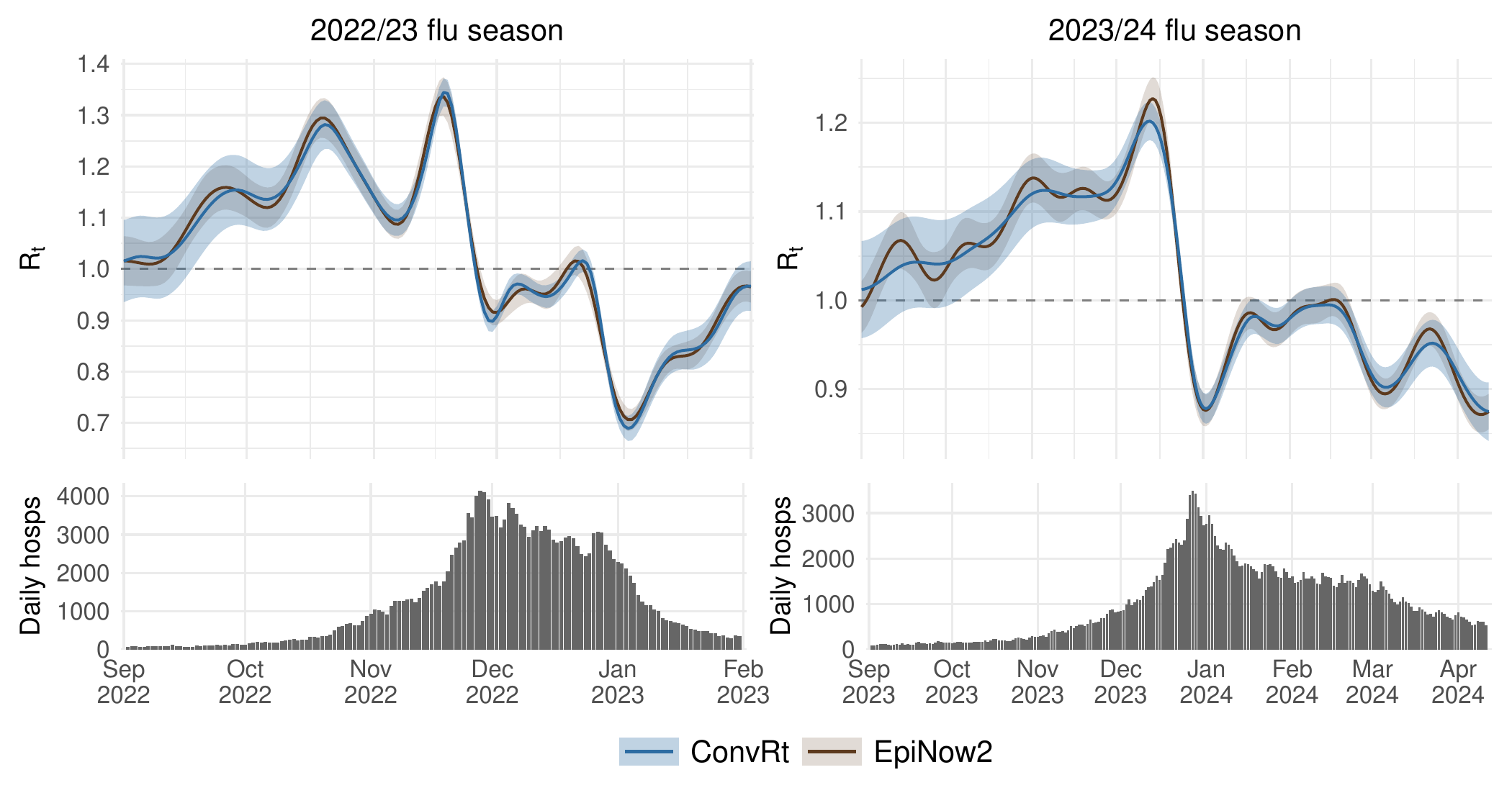}
    \caption[$R_t$ predictions for seasonal influenza.]{$R_t$ predictions for seasonal influenza. ConvRt and EpiNow2 fit retrospectively on NHSN hospitalization data, with 95\% pointwise uncertainty bands. 2023/24 was a longer flu season than 2022/23, so its plotting window is extended.}
    \label{fig:retro-flu-both}
\end{figure}

The two methods generate consistent $R_t$ curves when fit retrospectively at the end of each flu season (\cref{fig:retro-flu-both}).
Their shapes are highly similar, with two peaks in 2022 and one in 2023.
The main difference is that ConvRt predicts a smoother rise in fall 2023---a more intuitive shape, though the true $R_t$ is unknown.
EpiNow2 could match it with a tighter prior on its Gaussian-process jumps, but at the risk of oversmoothing elsewhere. ConvRt avoids this tradeoff by learning the smoothness level from the data.
Quantitatively, the two are very close. 
Both estimate $R_t$ near $1.3$ in 2022 and $1.2$ in 2023, with discrepancies around $0.02$ at the peaks.
This matches a CDC analysis that found a median peak $R_t$ of $1.28$ for seasonal influenza \citep{biggerstaff2014estimates}.

Beyond accuracy, ConvRt has much better runtime. 
It completes in 3 seconds, while EpiNow2 takes well over 30 minutes per season.
ConvRt's confidence intervals are also meaningful.
They almost always cover EpiNow2's predictions, while staying tight enough to be useful.
EpiNow2's credible intervals are narrower and sit entirely within ConvRt's, perhaps reflecting structure imposed by its priors.

\subsubsection*{Real-time Results}

\begin{figure}[!h]
    \centering
    \includegraphics[width=\linewidth]{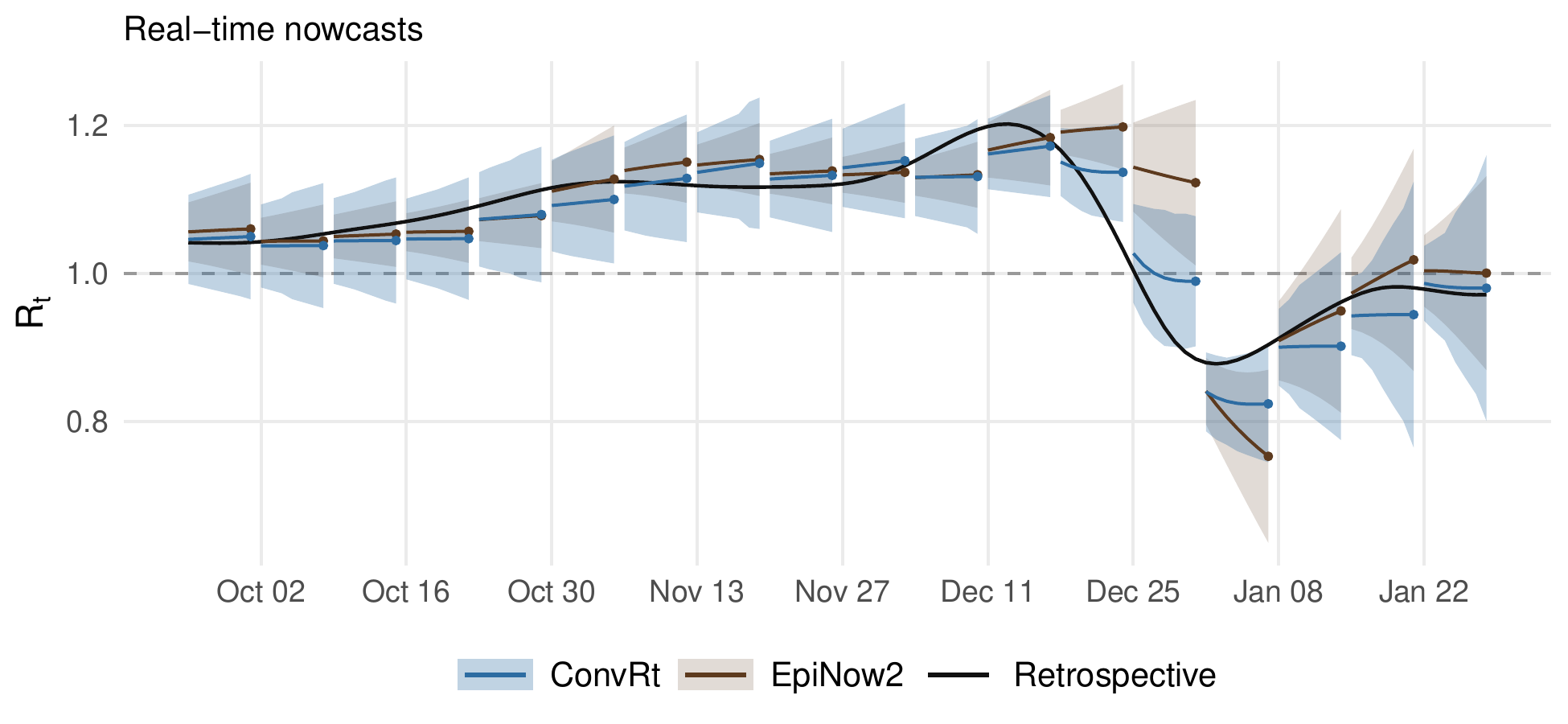}
    \caption[Real-time $R_t$ estimates and uncertainty bands for seasonal influenza.]{Real-time $R_t$ estimates and uncertainty bands for seasonal influenza in 2023/24.}
    \label{fig:UQ-ConvRt}
\end{figure}

Having established that ConvRt is competitive with EpiNow2 in retrospect, we examine its real-time performance.
\cref{fig:UQ-ConvRt} shows each method's weekly real-time estimate against the
finalized end-of-season fit.
Both closely track rising $R_t$ from October through mid-December.
However, in late-December, ConvRt predicts the decline in $R_t$ quite well;
in contrast, EpiNow2 is initially far too high, then overcorrects at the trough.
Some lag is unavoidable, since many infections take a week or more to surface in the
data, and the final nowcasts effectively extrapolate recent observations.
But ConvRt's tapered linear tail follows the prevailing trend, tracking turning
points more closely than EpiNow2.

The two methods also differ in uncertainty quantification.
ConvRt's bands come from the conformal method of \cref{apx:rt-realtime-ci} and
cover the retrospective ConvRt curve 96\% of the time, close to the 95\%
nominal level.
EpiNow2's credible intervals cover its own retrospective fit (not shown) only
76\% of the time.
The undercoverage is clearest at the late-December turning point, where its
bands widen sharply yet still miss the retrospective trough.
ConvRt's bands stay narrower through this transition and cover throughout.

\section{Scenario Analysis}\label{sec:scenario}

\begin{figure}[h]
    \centering
    \includegraphics[width=\linewidth]{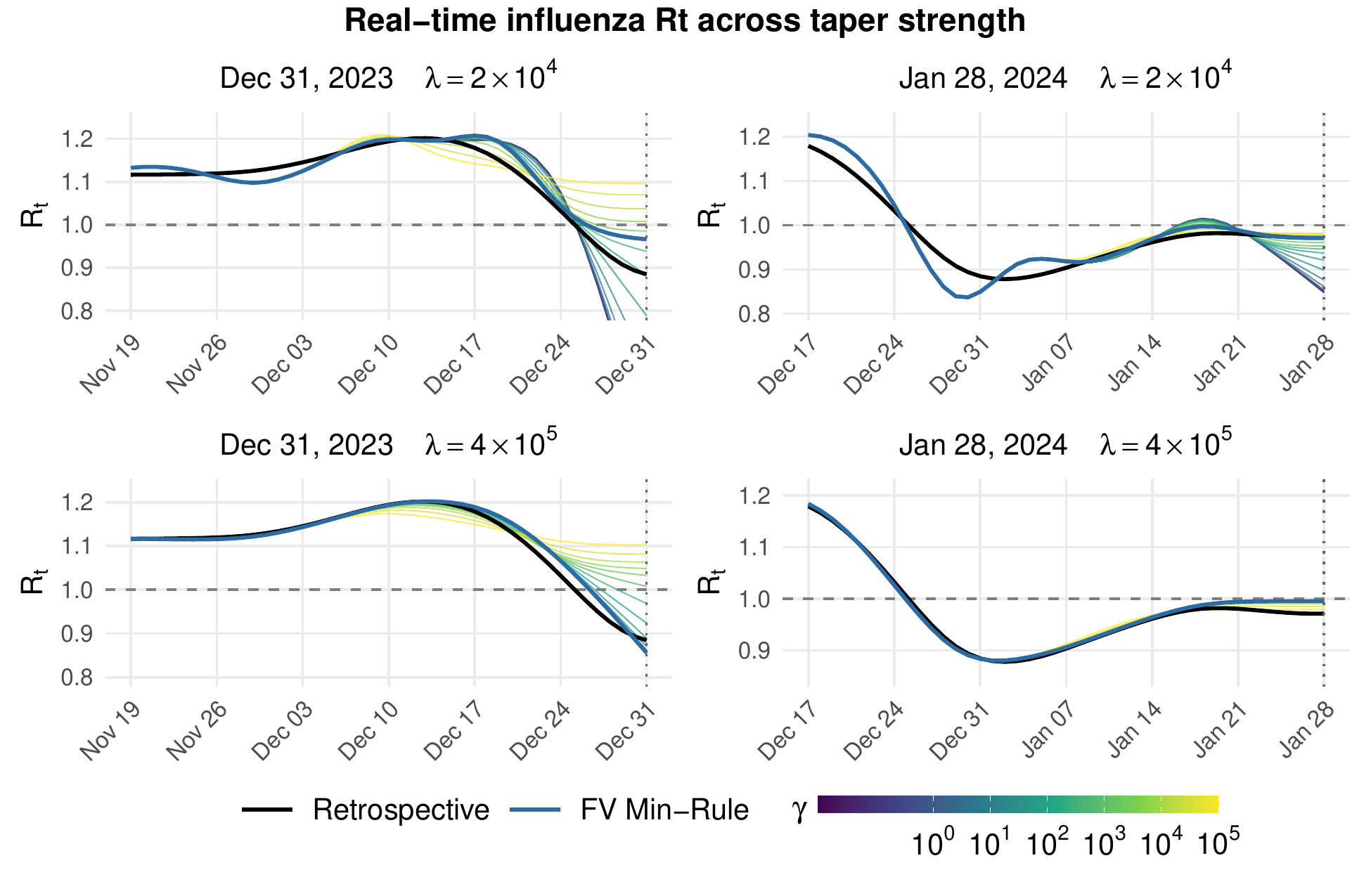}
    \caption[$R_t$ predictions at varying levels of the smoothness hyperparameters.]{$R_t$ predictions at varying levels of the smoothness hyperparameters $\lambda$ (overall) and $\gamma$ (tail).}
    \label{fig:rt-by-gamma}
\end{figure}

There is no way to know for certain which smoothness hyperparameters are optimal.
Different choices yield different stories about $R_t$, many of which may be plausible at once.
Cross-validation provides a reasonable default strategy, but distinguishing between competing narratives may be ambiguous quantitatively, as error curves are often relatively flat for deconvolution problems (e.g. \cref{fig:fv-errors-flu}).
Cross-validation may also select a fit that looks unreasonable upon visual inspection. 
Examining the range of fits across hyperparameters is therefore informative in its own right, and a practitioner's judgment matters throughout. 
This is especially important in real-time, where we have both $\lambda$ and $\gamma$.

Tuning $\lambda$ is particularly challenging towards the start of the season, as signals are hard to identify with fewer estimation dates. Our default is the 1se rule, but an acceptable option is to fix $\lambda$ at the value tuned retrospectively over the previous season, and another is to begin fitting $R_t$ from the prior year. These are specific remedies rather than a general recipe. The choice of $\gamma$ is even less clear-cut, since by its nature it extrapolates recent trends and imposes a stationarity assumption on $R_t$ near the tail. No amount of in-sample data tells us whether that assumption will hold going forward.

\cref{fig:rt-by-gamma} shows how the hyperparameters $\lambda$ and $\gamma$ affect $R_t$ predictions near the tail. Lower values of $\lambda$ (top row) correspond to more wiggly $R_t$ estimates than higher $\lambda$ (bottom row). When $\gamma$ is close to 0, $R_t$ is effectively linear after the final knot. Conversely, large values of $\gamma$ flatten the recent $R_t$ predictions, as it penalizes a weighted sum of squared first differences \eqref{eq:glm-real-time}. \cref{fig:fv-errors-flu} in \cref{apx:real-flu} visualizes the forward-validation error curves.

The two columns of \cref{fig:rt-by-gamma} illustrate contrasting scenarios. Treating retrospective estimates as ground truth, the left subplots depict $R_t$ falling sharply about 10 days after its peak. With hospitalizations averaging 5.7 days after infection, there is scarce data to express this change. Still, the linear-tail fit (with $\gamma = 10^{-1}$, effectively no penalization) predicts the steep decline. Given a smaller $\lambda$ (top left), those $R_t$ estimates fall extremely sharply, far below the true values. With larger $\lambda$ (bottom left), the sharpest fall is much less dramatic, and in fact quite accurate. $R_t$ hardly drops at all with larger $\gamma$, leveling out around the mid-December peak. Min rule tuning for $\gamma$ successfully predicts that $R_t$ is declining, in spite of the small amount of data. With smaller $\lambda$, it is spot-on. However, with larger $\lambda$ it remains above the ``true'' $R_t$.

The right subplots of \cref{fig:rt-by-gamma} show $R_t$ dropping below 0.9 on January 1. It then rises above 0.95 in mid-January, then remains there for the rest of the month. When using a smaller $\lambda$ and $\gamma$, the spline fit is not entirely steady at these dates (top right). Rather, it fits $R_t$ to noise, dropping around 0.1 in the final two weeks. Tuning $\gamma$ with the min rule averts this issue: the selected value is large enough that $R_t$ is roughly constant leading up to the estimation date. The bottom right subplot shows $R_t$ spanning a much smaller range across $\gamma$.

\section{Discussion}\label{sec:ch4-discussion}

\subsection{Contributions}

We developed ConvRt, a frequentist estimator for the effective reproduction number $R_t$.
It deconvolves latent infections from observed counts, then estimates $R_t$ via penalized quasi-Poisson regression with a spline basis.
Tapered penalties and tail structure stabilize real-time predictions.

In our experiments, ConvRt's only real competitor on accuracy is EpiNow2.
It reduces MAE on most benchmarks in retrospect, and the two are effectively tied in real time.
Every other baseline has substantially higher error.
But EpiNow2 and other accurate Bayesian methods---CovidEstim \citep{Chitwood2022} and recent MCMC-based approaches for stratified models \citep{bosse2024bayesian}---rely on sampling that can take tens of minutes to hours per fit.
In contrast, ConvRt fits in seconds.
This gap is consequential: public health agencies refit $R_t$ daily across many jurisdictions, and during surges, when timely estimates matter most, computational budgets are tightest.
Speed also enables the scenario analyses of \cref{sec:scenario}, impractical when a single refit is expensive.

ConvRt's confidence intervals do not require specifying prior distributions.
Bayesian $R_t$ methods require priors that simultaneously encode beliefs about transmission and control the roughness of the trajectory.
Prior choices always affect point estimates and credible intervals, complicating interpretation---especially when they are poorly chosen, such as when package defaults are adopted without consideration of the specific epidemic context.
Frequentist intervals avoid this entanglement: their coverage depends on the data and the model, not on prior assumptions.
Empirically, ConvRt achieves better coverage than every fast baseline and outperforms EpiNow2 in real-time calibration.

More broadly, ConvRt decouples two jobs that Bayesian methods conflate: regularizing the estimate and expressing assumptions about future transmission.
The hyperparameter $\lambda$ controls wiggliness, while $\gamma$ controls how aggressively recent trends are extrapolated.
Practitioners can choose $\lambda$ based on the temporal resolution they care about, then specify $\gamma$ based on their assumptions about the near future, rather than inheriting both from a prior chosen for in-sample regularization.
Because refitting is cheap, one can examine the full family of estimates across $(\lambda, \gamma)$ and assess which narratives the data supports.

\subsection{Extensions}

We implemented two extensions of ConvRt, detailed in \cref{apx:extra-methods}, that are not benchmarked comprehensively in this chapter.

The first handles weekly data.
Daily surveillance data were more widely available during the pandemic era, but many systems report only at weekly resolution.
FluSurv-NET has published weekly flu hospitalizations since 2005, and CDC's COVID-NET and RSV-NET follow the same cadence.
NHSN, the source used in our experiments, itself ceased daily reporting in May 2024.
We adjust the likelihood by aggregating daily means within each week, preserving the underlying daily transmission dynamics.
On the simulated flu benchmark, retrospective estimates from weekly data are nearly identical to those from daily data (\cref{apx:weekly}).
Real-time predictions are also encouraging, with forward-validation correctly selecting a tail penalty that stabilizes the final week's estimate.

The second extension replaces the spline basis with trend filtering regularization.
Trend filtering adaptively selects knot locations, whereas splines place them at predetermined positions.
This added flexibility may better capture local changes in $R_t$, such as abrupt shifts around interventions.
In \cref{apx:extra-methods}, we run ConvRt with a fourth-difference penalty on the 2022/23 flu season.
The resulting $R_t$ estimates are very similar to the spline's.
A rigorous comparison across datasets is left to future work.
Trend filtering may also be a good choice for the initial step of deconvolving latent infections, where we currently use splines but do not rely on their uncertainty bands.

We chose to focus on splines in this paper because trend filtering does not lend itself to uncertainty quantification as naturally.
This is particularly important for $R_t$, since epidemic growth is nonlinear and small differences can have large consequences. For example, an SIR model with $R_0$ near 1 ultimately infects around $2(R_0-1)\times 100\%$ of the population \citep{diekmann2013mathematical}, so $R_0 = 1.15$ versus $1.05$ implies 20\% more of the population infected. Accordingly, popular $R_t$ methods all report confidence or credible intervals \citep{Hooker2011, cori2013new, abbott2020estimating, Chitwood2022, scire2023estimateR}, and the CDC asserts whether $R_t$ is rising or falling based on EpiNow2's credible interval \citep{cdccfa2026rt}.

\subsection{Limitations and Future Directions}

We discuss limitations shared with the severity rate work of Part~\ref{pt1} in the concluding chapter.
Here, we focus on issues specific to $R_t$ estimation.

A limitation of this work is that we did not address the nowcasting problem.
Real-time counts are typically incomplete at recent dates due to reporting delays and are revised upward as backfill accumulates.
ConvRt takes observed counts as given.
In practice, real-time estimation should be paired with a nowcast \citep{mcgough2020nowcasting, abbott2020estimating}.
How uncertainty in the nowcast propagates into $R_t$ estimates deserves further study.

ConvRt's two-stage procedure also does not propagate uncertainty from the infection deconvolution step into $R_t$ inference.
The plug-in estimates $\hat{x}_t$ are treated as fixed when fitting $R_t$, so the reported confidence intervals may understate the true uncertainty.
That said, our empirical results suggest this omission may not be severe: ConvRt already achieves competitive coverage without accounting for deconvolution noise.
Still, a principled correction would be useful.
A natural approach is the parametric bootstrap: resample infections from their estimated distribution, refit $R_t$ on each draw, and combine.
Since ConvRt fits in seconds, running hundreds of bootstrap replicates is feasible.
The bootstrap variance across infection draws can be combined with the closed-form sandwich variance (conditional on a given infection curve) using Rubin's combining rules \citep{rubin1987multiple}.
This gives a total variance that accounts for both estimation noise in $R_t$ and uncertainty in the plug-in infections.

The GLM structure naturally accommodates covariates: signals like mobility indices or vaccination rates could enter the linear predictor of $R_t$ alongside the spline basis, so that \mbox{$R_t = S(t)^\top \theta + z_t^\top \alpha$}.
Care is needed to ensure that noisy covariates do not inject roughness into $R_t$, for example by penalizing $\alpha$.

Our experiments used hospitalizations and case reports, both of which depend on time-varying ascertainment rates.
Misspecification of these rates biases $R_t$ estimates, and in practice they are difficult to pin down.
Wastewater viral concentrations sidestep this problem entirely, since all infected individuals shed virus regardless of whether they seek testing.
The observation model has the same convolutional structure: viral load at time $t$ is approximately infections convolved against a shedding kernel, playing the role of $\pi$.
ConvRt could be applied directly with this substitution.

More generally, one could fuse multiple data streams---cases, hospitalizations, deaths, wastewater---by summing their log-likelihoods against the same latent infection curve, each with its own delay kernel and ascertainment model.
This would apply to both deconvolution stages: estimating infections and estimating $R_t$.
CovidEstim \citep{Chitwood2022} takes a similar approach in a Bayesian framework, jointly fitting to case and death series.
A frequentist analogue would be straightforward: stack the observation equations and optimize jointly over the shared infection curve or $R_t$ spline.
Multiple streams constraining the same underlying infections may help resolve ambiguities that arise from any single source.
They could also make estimates more robust to misspecification of any one stream's ascertainment rate or delay distribution.

\cleardoublepage
\bookmarksetup{startatroot}   
\addtocontents{toc}{\bigskip} 

\chapter{Conclusion}
\label{ch:conclusion}

This dissertation studied methods for estimating two classes of time-varying epidemic metrics: severity rates and reproduction numbers.
These problems share deep structural parallels.
Both involve a primary event count, a secondary event count, a delay distribution connecting them, and a time-varying rate parameter to be estimated.
The generation interval in reproduction numbers plays the same role as the delay distribution $\pi$ in severity rates.
Standard real-time estimators for both quantities adopt a backward-looking perspective, analyzing secondary events at $t$ in terms of past primary events. 
For reproduction numbers, this perspective defines instantaneous $R_t$; 
severity rates are defined as forward-looking \eqref{eq:severity}, but the convolutional ratio estimator \eqref{eq:conv} uses this same idea.

In Part~\ref{pt1}, we exposed systematic bias in ratio estimators for severity rates and proposed a deconvolution-based alternative regularized by trend filtering.
In Part~\ref{pt2}, we proved the equivalence of instantaneous and mechanistic $R_t$ under homogeneous mixing, derived the generation interval implied by SEIR models, and developed ConvRt, a fast frequentist $R_t$ estimator that adapts the deconvolution framework to the renewal equation.
In this concluding chapter, we discuss how ideas developed for one problem can inform the other, and identify shared challenges and future directions.

\section{From Reproduction Numbers to Severity Rates}

Several techniques developed for $R_t$ estimation in \cref{ch4} could improve the severity rate methods of \cref{ch:paper2}.

\paragraph{A backward-looking definition of severity rates.}
The severity rates in \cref{ch:paper1,ch:paper2} are forward-looking: $p_t$ is defined by the future secondary events generated by primary events at $t$.
An alternative is the backward-looking severity rate,
\begin{equation}
\label{eq:back-conclusion}
\tilde{p}_t = \sum_{k=0}^d \Pprob(\text{secondary event at $t$} \given
\text{primary event at $t-k$}),
\end{equation}
which asks how many secondary events at $t$ can be attributed to recent primary events.
This parallels the distinction between case $R_t$ (forward-looking) and instantaneous $R_t$ (backward-looking) in \cref{ch3}.

As shown in \cref{ch:paper2}, the forward- and backward-looking rates coincide when conditions are locally stationary.
When they differ, the convolutional ratio \eqref{eq:conv} is in fact unbiased for $\tilde{p}_t$, not $p_t$ (see \cref{apx:back}).
This suggests an alternative estimation strategy: rather than targeting $p_t$ with deconvolution, one could target $\tilde{p}_t$ directly.
The convolutional ratio provides an unbiased starting point, and trend filtering or splines could further improve the estimate.
Just as most $R_t$ methods target the instantaneous (backward-looking) definition, backward-looking severity rates may be a natural and practical estimand for real-time surveillance.

\paragraph{Spline-based estimation and uncertainty quantification.}
The severity rate estimator in \cref{ch:paper2} uses trend filtering, which produces locally adaptive, piecewise polynomial fits.
However, as discussed in \cref{sec:ch4-discussion}, trend filtering does not lend itself naturally to uncertainty quantification.
Fitting severity rates as splines instead would enable the same sandwich-based confidence intervals we derived for $R_t$.
The two approaches could also be compared more rigorously: trend filtering may capture abrupt changes better, while splines offer smoother fits with principled inference.

\paragraph{Deconvolving latent infections.}
The severity rates in Part~\ref{pt1} are defined between observed event types: cases to deaths (CFR), hospitalizations to deaths (HFR), and so on.
An alternative is to define severity relative to infections rather than cases, yielding quantities like the infection-fatality rate (IFR) or infection-hospitalization rate (IHR).
These are more epidemiologically fundamental, since they do not depend on case ascertainment.
Estimating them requires first deconvolving latent infections from observed counts, exactly as ConvRt does.
The two-stage approach of \cref{ch4}---deconvolve infections, then estimate the rate parameter---could be applied directly, potentially using time-varying ascertainment rates estimated from external data such as seroprevalence studies.

\paragraph{Day-of-week effects.}
ConvRt includes multiplicative day-of-week effects in its observation model (\cref{eq:quasi-poisson-model}), capturing regular reporting artifacts.
Severity rate estimation would benefit from the same treatment, since hospitalization and death reports exhibit similar weekly patterns.

\section{Shared Challenges and Future Directions}

The methods developed across this dissertation share common structure and face common challenges.

\subsection*{Tail Regularization}

Both the severity rate and $R_t$ estimators use a tapered smoothing penalty to stabilize real-time predictions.
The hyperparameters $\lambda$ and $\gamma$ have complementary roles---wiggliness and tail extrapolation---whose relative importance varies by epidemic phase.
As discussed in \cref{sec:scenario}, examining the full family of estimates across $(\lambda, \gamma)$ is informative in its own right.

Several aspects of this framework deserve further investigation.
The tapered weights were taken from \citet{Jahja2022}, but other weight functions have not been explored.
Cross-validation for $\lambda$ does not always work well: oscillatory solutions can achieve low held-out error because overestimates and underestimates cancel across folds.
Developing improved tuning heuristics is an important practical direction.

\subsection*{Plug-In Quantities}

Both severity rate and $R_t$ estimation rely on plug-in quantities---delay distributions, generation intervals, and ascertainment rates---that are treated as fixed but are estimated from external data and subject to uncertainty.
Misspecification can bias the resulting estimates.

It would be valuable to propagate uncertainty in these quantities through to the final estimates.
In \cref{sec:ch4-discussion}, we discussed a similar issue for ConvRt's plug-in infection estimates and proposed combining the sandwich variance with bootstrap variance via Rubin's rules \citep{rubin1987multiple}.
The same idea applies here: draw plug-in quantities from their estimated sampling distributions, refit the estimator on each draw, and combine the within-draw and between-draw variances.
At least one Bayesian $R_t$ method places a prior on the delay distribution and learns it jointly \citep{abbott2020estimating}, providing a point of comparison.

More ambitious approaches would learn the delay distribution jointly with the rate parameter, or allow it to vary over time with its own smoothness penalty.
\citet{li2023reconstructing} examine a joint estimation strategy to deconvolve latent infections using the EM algorithm.
Constructing analogous approaches for our two metrics poses an interesting open problem.

\subsection*{Uncertainty Quantification}

Uncertainty quantification remains an open challenge, particularly in the real-time setting.
At the tail, regularization introduces bias that standard variance-only intervals do not capture.
This is not specific to our methods; any regularized estimator faces the same tradeoff between stability and coverage.

For $R_t$, we addressed this with a conformal approach (\cref{apx:rt-realtime-ci}): at each vintage, we treat predictions from earlier vintages as approximate ground truth and take quantiles of their residuals to form prediction intervals.
More sophisticated approaches could draw on the online inference literature.
Quantile tracking methods \citep{angelopoulos2024conformal} maintain adaptive prediction intervals that respond to distribution shift, and could be applied to the stream of $R_t$ nowcasts as they arrive.

For severity rates estimated with trend filtering, uncertainty quantification is harder.
The $\ell_1$ penalty produces a non-smooth objective, and the resulting estimator does not have a standard asymptotic distribution.
Data fission \citep{leiner2025data, dharamshi2025generalized} offers a promising path forward: it splits the information in each observation to create independent copies for selection and inference.
Applying this technique to trend filtering under a Poisson inverse problem, as arises in our deconvolution setting, is an open problem.

The structural parallels between severity rates and reproduction numbers run deeper than their shared mathematical form. As the directions above suggest, methodological progress on either problem directly benefits the other. Resolving these shared challenges would make surveillance estimates more trustworthy precisely when they matter most.

\bibliographystyle{plainnat}
\bibliography{refs}

\appendix
\chapter{\texorpdfstring{Supplementary material for \cref{ch:paper1}}{Supplementary material for Chapter 1}}
\label{app:A}

\section{Proofs}
\label[appendix]{apx:proofs}
The assumption of stationary delay distribution is not necessary for either bias expression. 
In the proofs in Sections \cref{apx:OracleBias} and \cref{apx:MispBias},
the delay distributions $\pi$ may simply be replaced with $\pi^{(t)}$.

\subsection{Proof of Proposition 1}
\label[appendix]{apx:OracleBias}

Proposition 1 establishes the bias of the well-specified convolutional ratio.
Here and henceforth, we abbreviate \smash{$\E_t[\cdot] = \E[\cdot \given
  x_{\leq t}]$}. Observe that
\begin{align*}
\bias(\hat{p}_t^\pi) 
&= \frac{\E_t[y_t]}{\sum_{k=0}^d x_{t-k}\pi_k} - p_t \\
&= \frac{\sum_{k=0}^d x_{t-k}\pi_k p_{t-k}}{\sum_{k=0}^d x_{t-k}\pi_k} - 
\frac{p_t \sum_{k=0}^d x_{t-k}\pi_k}{\sum_{k=0}^d x_{t-k}\pi_k} \\
&= \sum_{k=0}^d \frac{x_{t-k}\pi_k}{\sum_{j=0}^d x_{t-j}\pi_j} (p_{t-k}-p_t).
\end{align*}


\subsection{Proof of Proposition 2}
\label[appendix]{apx:MispBias}
Proposition 2 establishes the bias of the misspecified convolutional ratio, the lagged ratio being a special case.
Observe that
\begin{align*}
\bias(\hat{p}_t^\gamma) 
&= \frac{\E_t[y_t]}{\sum_{k=0}^d x_{t-k}\gamma_k} - p_t \\
&= \frac{\sum_{k=0}^d x_{t-k}\pi_k p_{t-k}}{\sum_{k=0}^d x_{t-k}\gamma_k} -
\frac{\sum_{k=0}^d x_{t-k}\gamma_k p_t}{\sum_{k=0}^d x_{t-k}\gamma_k} \\
&= \sum_{k=0}^d \frac{x_{t-k}}{\sum_{j=0}^d x_{t-j}\gamma_j}
(\pi_k p_{t-k} - \gamma_k p_t) \\
&= \sum_{k=0}^d \frac{x_{t-k}}{\sum_{j=0}^d x_{t-j}\gamma_j}
(\pi_k p_{t-k}-(\pi_k +(\gamma_k-\pi_k)) p_t) \\
&= \frac{\sum_{j=0}^d x_{t-j}\pi_j}{\sum_{j=0}^d x_{t-j}\gamma_j}
\sum_{k=0}^d \frac{x_{t-k}\pi_k}{\sum_{j=0}^d x_{t-j}\pi_j}(p_{t-k}-p_t) -
p_t\sum_{k=0}^d \frac{x_{t-k}}{\sum_{j=0}^d x_{t-j}\gamma_j}(\gamma_k -\pi_k) \\ 
&= \frac{\sum_{j=0}^d x_{t-j}\pi_j}{\sum_{j=0}^d x_{t-j}\gamma_j} 
\bias(\hat{p}_t^\pi) + p_t\Bigg( \frac{\sum_{k=0}^d x_{t-k}\pi_k} 
{\sum_{j=0}^d x_{t-j}\gamma_j}-1\Bigg).
\end{align*}

\section{Additional analysis and data sources}
\subsection{Further analysis of bias}
\label[appendix]{apx:analysis}

We first present examples that further explain the bias for the well-specified convolutional ratio. These examples are considerably more
contrived that the ones in \cref{sec:results}. Nonetheless, their bias can be simplified to 
simple analytic formulas, isolating the three contributing factors. 

To elucidate the relationship between changing severity rates and bias, let us
consider the case where all secondary events occur after precisely $\ell$ time  
points. The well-specified convolutional and lagged ratio estimators
coincide: \smash{$\hat{p}_t^\gamma = \hat{p}_t^\ell = p_{t-\ell}$}. The bias in
this setting is simply the change in the true severity rate, $p_{t-\ell} - p_t$, 
and the ratio estimator is unbiased only if the severity rate is stationary. 
Otherwise the ratio will be 20\% too low, for example, if the true severity rate was
20\% lower $\ell$ time steps ago. 

\begin{figure}[tb]
\centering

\begin{subfigure}[b]{0.49\linewidth}
  \centering
  \includegraphics[width=\linewidth]{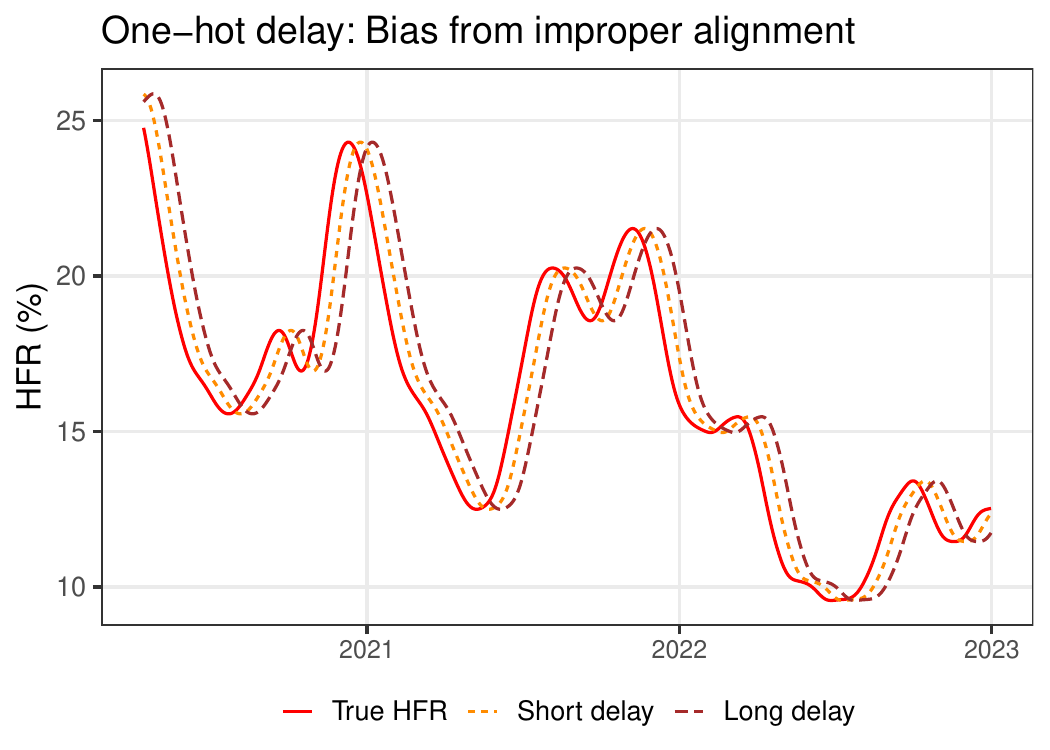}
  \caption{All deaths after $\ell$ days. HFR ratios equal;
  plotting delays of $\ell=14$ and 28 days.}
  \label{fig:onehot}
\end{subfigure}
\hfill
\begin{subfigure}[b]{0.49\linewidth}
  \centering
  \includegraphics[width=\linewidth]{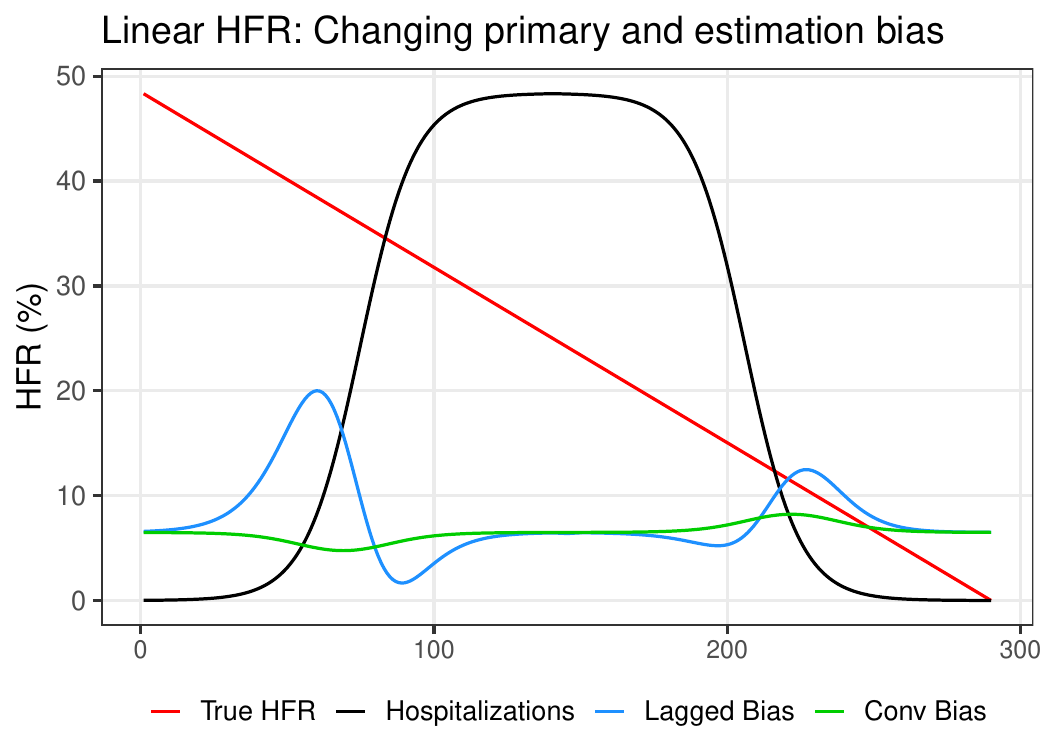}
  \caption{Changing primary incidence. Bias of lagged and
  convolutional ratios.}
  \label{fig:chging_primary}
\end{subfigure}

\vspace{0.75em}

\begin{subfigure}[b]{0.49\linewidth}
  \centering
  \includegraphics[width=\linewidth]{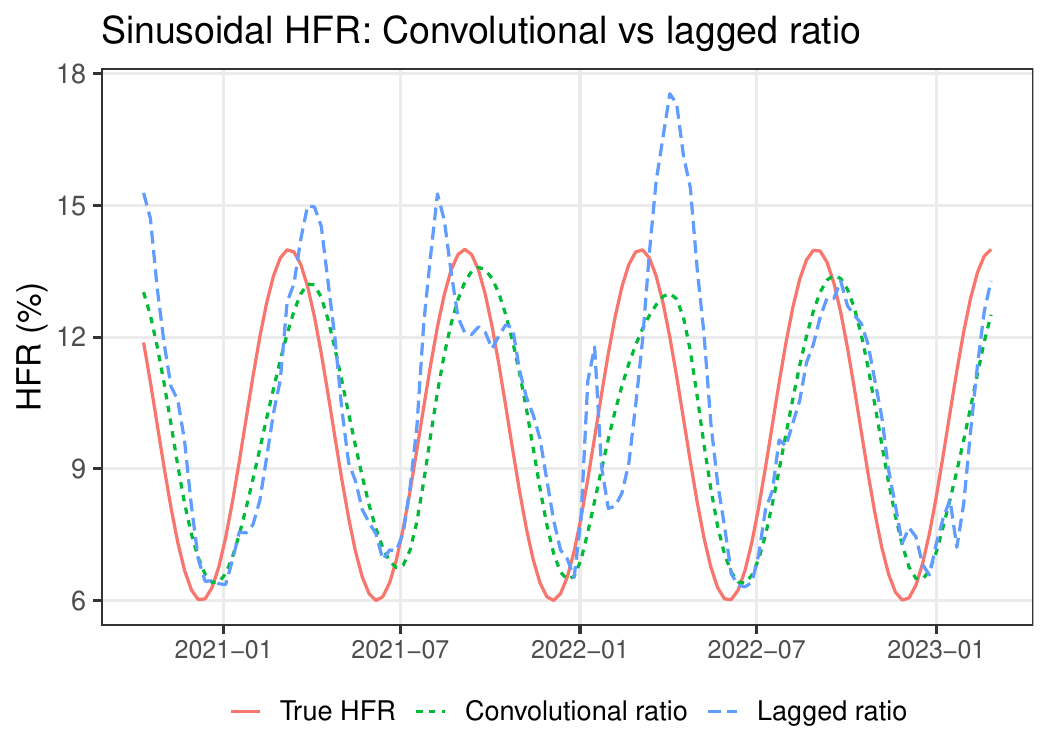}
  \caption{Sinusoidal severity rate and ratio estimates.}
  \label{fig:sinusoidal}
\end{subfigure}
\hfill
\begin{subfigure}[b]{0.49\linewidth}
  \centering
  \includegraphics[width=\linewidth]{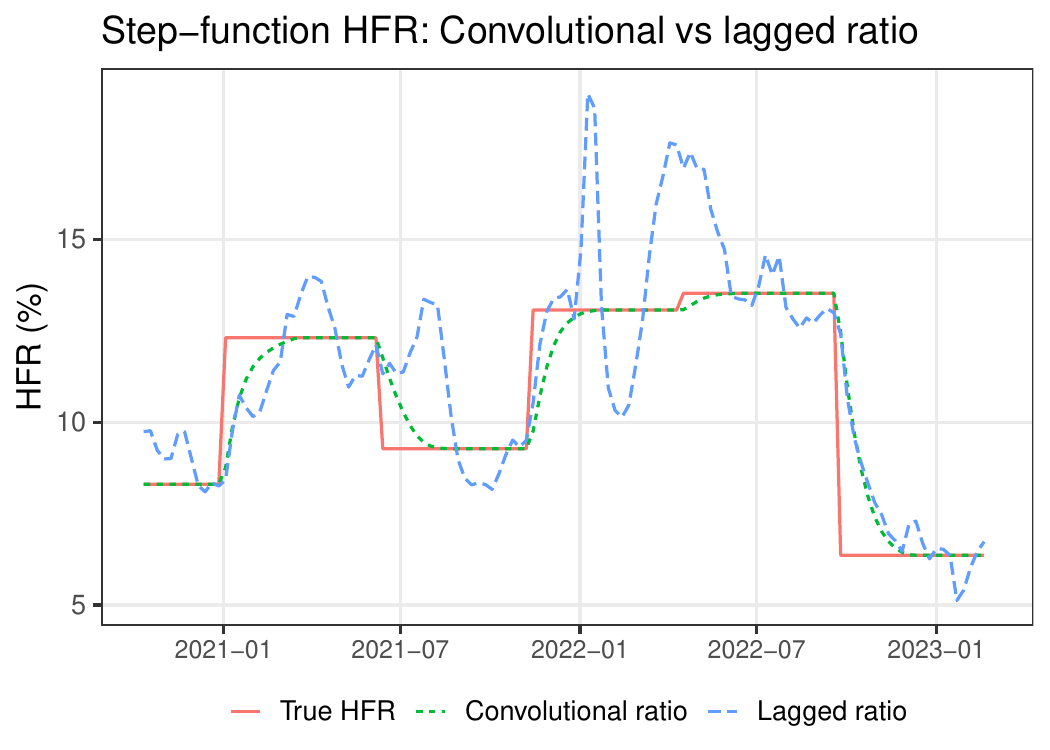}
  \caption{Severity rate shifts via step functions.}
  \label{fig:stepfun}
\end{subfigure}

\caption[Toy examples demonstrating biased severity-rate estimators.]{Toy examples demonstrating biased severity-rate estimators under different delay structures and primary-incidence patterns.}
\label{fig:bias_ex}
\end{figure}

\cref{fig:onehot} displays the results on the NHCS HFRs.
In general, severity rates will be less similar to the present value $p_t$ as we go further back in time. The bias $p_{t-\ell}-p_t$ tends to be larger when $\ell=28$ versus $\ell=14$. This supports the overarching idea
that estimates with heavier-tailed delay distributions tend to have more bias.          

Now to elucidate the relationship between primary incidence and bias, let us 
consider a delay distribution $\pi$ which places half its mass at lag 0, and the
other half at lag $q$. Then the well-specified bias has magnitude:
\[
\big| \bias(\hat{p}_t^{\pi}) \big| = \frac{\frac{1}{2} \big| x_t(p_t-p_t) +
  x_{t-q}(p_{t-q}-p_t) \big|} {\frac{1}{2}(x_t+x_{t-q})} = \frac{|p_{t-q}-p_t|}
{1 + x_t / x_{t-q}}. 
\]
In other words, the absolute bias is monotonically decreasing in
$x_t / x_{t-q}$, the proportion change in primary incidence. Rising
primary incidence ($x_t / x_{t-q} > 1)$ yields less bias, while falling 
levels yield more. 

\cref{fig:chging_primary} displays this setting with $q=10$. Rising hospitalizations are
defined as \mbox{$x = \sigma(s)*9000+1000$}, where $\sigma$ is the sigmoid function
and $s$ takes 300 evenly spaced steps from -9 to 7. These quantities are subsequently reflected to express a decline.
The true HFRs fall from
0.5 to 0 over the same number of even steps. 
Accordingly, the absolute bias of the convolutional ratio is $c_q\frac{x_{t-q}}{x_{t-q}+x_t}$, where $c_q\approx 0.0167$. Shown in red, it dips as hospitalizations rise, and rises as they fall.  

The figure also plots with lagged ratio with $\ell=\frac{d}{2}$, the mean of
the delay distribution. When daily hospitalizations are close to constant, the
two estimators converge towards the same ratio. During periods of change,
however, the lagged estimator has different bias. It first moves upwards ---
the opposite direction as the convolutional bias --- with far greater
magnitude. This can be explained by the ratio $A_t^\ell =
\frac{x_{t-2\ell}+x_t}{2x_{t-\ell}}$ from Proposition 2. As
hospitalizations begin to steeply rise, $x_{t-2\ell}$ and $x_{t-\ell}$ are
similar, but $x_t > x_{t-\ell}$. Hence, $A_t^\ell>1$, contributing positive
bias to both the oracle and misspecification terms. As hospitalizations level
out near the top, $A_t^\ell < 1$, hence the bias falling lower. The opposite
pattern occurs as hospitalizations fall.  



We consider two additional settings with contrived severity rates, beyond the linear change in \cref{fig:chging_primary}. \cref{fig:sinusoidal,fig:stepfun} generate deaths with sinusoidal and piecewise constant HFRs, respectively. They use true US hospitalization counts throughout over two years of the COVID-19 pandemic. Deaths are simulated noiselessly to highlight the estimators' bias. The delay distribution is the same as in the simulated experiments in the main text: A gamma distribution whose mean maximizes the correlation between hospitalization counts from HHS and death counts from JHU. 

\cref{fig:sinusoidal} reveals the ratio estimators have varying capacity for bias when the true severity rate is sinusoidal. In general, the ratios trail behind the true values some number of days. The convolutional ratio mimics the general sinusoidal shape, though its high and low values are about 1\% less stark than the true values. Intuitively, the ratio smooths over the trailing history, so it is incapable of precisely identifying peak and trough values. The lagged ratio, as explained in the main text, depends heavily on the primary incidence curve and its relation to the delay distribution. Consistent with our analysis in \cref{sec:results_sim}, this leads to high values of $A_t^\ell$ inflating HFR estimates in Spring 2022. We see the highest lagged HFR reaching 17\%, while the true severity rate turns over around 14\%. 

\cref{fig:stepfun} tracks the ratios as the true severity rate jumps up and down by varying degrees. The well-specified convolutional ratio gradually reaches each new plateau and stabilizes. The time for it to adjust depends on the magnitude of the shift. This is consistent with our analysis in \cref{sec:misspecified}, which showed that larger changes in severity rate correspond to more bias. 

The lagged ratio is more heavily dependent on the primary incidence curve, acting through $A_t^\ell$. 
Proposition 2 shows this term alters the well-specified bias via a multiplicative correction and, more significantly, an additive adjustment.
How $A_t^\ell$ produces the lagged ratio's high positive, negative, and positive bias with a constant severity rate in early 2022 is explained precisely in \cref{sec:results_sim}. 
Another example of biased divergence occurs in the Delta wave of summer 2021, in which rising primary incidence triggers a rise in $A_t^\ell$ and thus $\hat p^\ell$.
However, when primary incidence is roughly constant, the lagged and convolutional ratios coincide. We see this in the latter half of 2022, when hospitalization counts were low.

\begin{figure}
    \centering
    \includegraphics[width=0.9\linewidth]{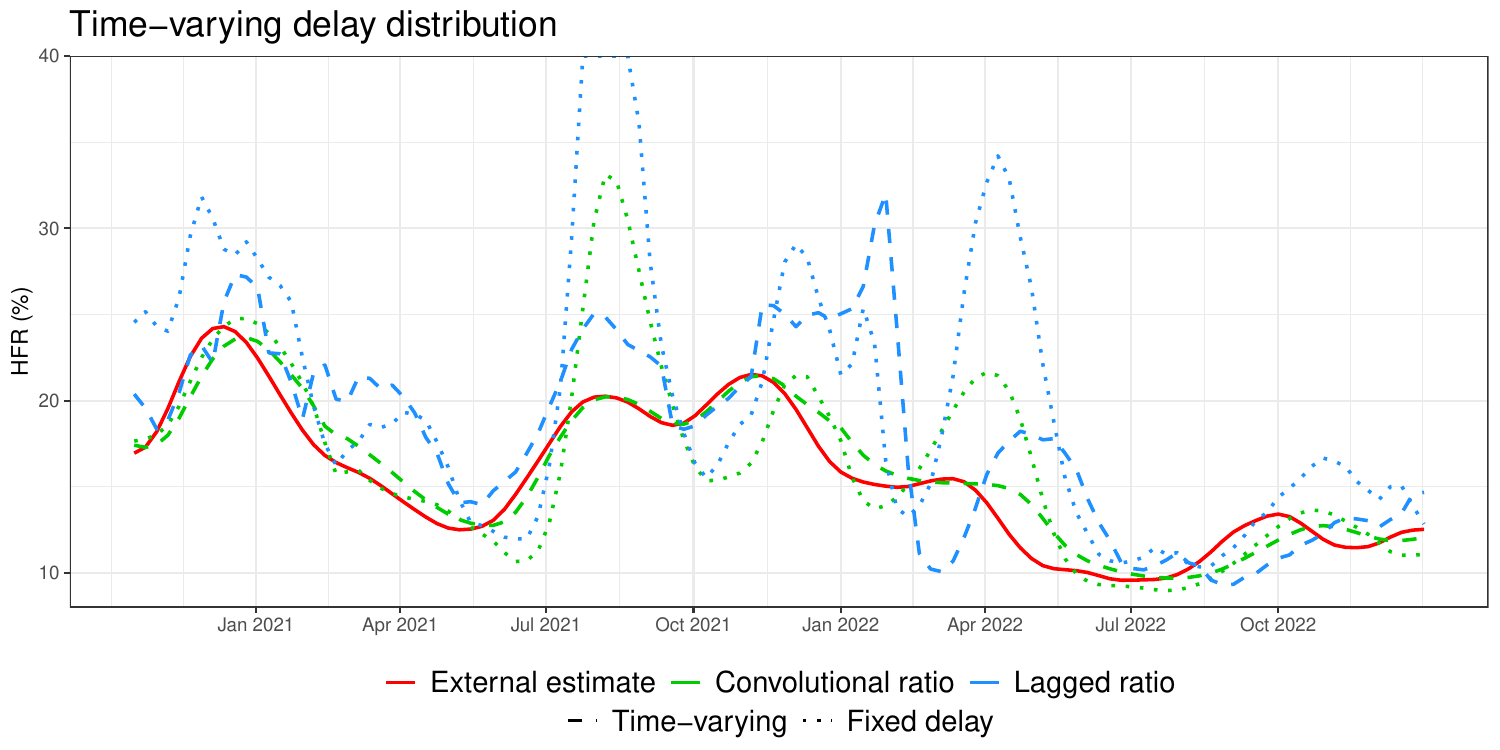}
    \caption[Severity rates where the true delay distribution is time-varying.]{Severity rates where the true delay distribution is time-varying. Dashed lines compute ratios with the true nonstationary delay distributions and their means as the lags. Dotted lines use the same constant distribution and lag -- their means over all time.}
    \label{fig:timevar}
\end{figure}

We additionally evaluated these methods when the true delay distribution is time-varying. Our analyses focus on the bias at a single point in time, so their messages should transfer translate to the time-varying setting. 
To generate realistic delay distributions, we identified the major variant periods using data described in \cref{apx:alt_gt}. 
Within these ranges, we computed the correlation-maximizing lag between HHS hospitalizations and JHU deaths. 
These lags ranged by a surprising margin: only a handful of days for the Delta wave, and over 30 during Omicron. 
We then computed the expected hospitalization-to-death delay by mixing these lags with the proportions of each variant in circulation. 
(This is akin to how we alternatively estimated HFRs in the aforementioned Section.)
Finally, we parameterized the delay distributions based on these mean delays and their corresponding standard deviations.
We simulated deaths without noise from these delay distributions and the NCHS-based HFRs.

On this simulated data, we computed the convolutional and lagged ratios in two ways. Firstly, we let them use the true time-varying delay distributions and lags. Under this formulation, the convolutional ratio is well-specified; the lagged ratio is not, as it still reduces the distribution to a point mass, but the mass is at least selected in the proper location.
The second strategy ignored the time-varying nature of the true model. 
It instead used stationary quantities: the mean hospitalization-to-death distribution and lag (18 days) over all time. 
This stationarity induces misspecification bias into both ratio estimators.

The behavior of the biases in \cref{fig:timevar} is consistent with our other analyses. 
The well-specified convolutional ratio is close to unbiased during stretches for which the delay distribution had low mean. 
During Omicron, when deaths were reported at a significant delay, it tracks the ground truth HFR at a much longer offset.
This observation in the time-varying case is identical to our analysis that the well-specified convolutional ratio is more biased under heavier-tailed delay distributions (see, for example, \cref{fig:toy_delay}).

The corresponding lagged ratio is also more biased during Omicron than Delta in this experiment. 
The light-tailed delay distribution indicates the probability mass is more condensed around the lag, so $A_t^\ell$ is not significantly greater than 1 as hospitalizations rise in Delta.
However, as Omicron surged around the start of 2022, the lag is over 30 days.
Hospitalizations were much lower at this offset, so $A_t^\ell$ has the capacity to get very large, and thus produce a high amount of bias.
The bias is greatest around the peak in late January, when the lagged ratio is more than double the true HFR of 15\%.
Interestingly, the positive bias that spring is less pronounced than in previous experiments. This is because the longer lag here captures more distant hospitalizations, which brings down $A_t^\ell$ somewhat.

The misspecified ratios with stationary parameters exhibit more bias. 
In the Delta wave, where the delay distribution and lag were too long, both estimators were much too high. They placed too much emphasis on low counts as the surge was beginning, resulting in significant positive bias. 
This was particularly pronounced for the lagged ratio, which was positively biased even under the accurate lag. 
It peaked off the figure at 60\%, triple the true HFR.

During Omicron, the ratios' delay distributions and lags were too short. 
This led to \textit{less} bias than the appropriately-specified ratios during the surge and decline, since these parameters accurately reflected current conditions. 
However, it produced high positive bias for both estimators after the surge leveled out in April 2022. 
As in the main text, the lagged ratio is extremely biased, since its denominator entirely fails to reflect the bygone surge. It reaches over 30\% as the true HFR falls towards 10\%. 

\subsection{Retrospective deaths}
\label[appendix]{apx:NCHS_deaths}

JHU aggregated daily deaths in real time, aligned by the date they were
reported. In contrast, the National Center for Health Statistics (NCHS) provided
weekly totals of deaths aligned by occurrence, which were not available in real  
time. Intuitively, the delay which relates hospitalization to death occurrence
should have a lighter tail in comparison to that relating hospitalization to
death report. Therefore, since that heavier-tailed delays generally introduce
greater bias, we would expect the ratio estimates computed from JHU deaths to
have greater bias than those from NCHS deaths. \cref{fig:jhu_vs_nchs}
shows that this is indeed the case. 
NCHS data leads to lagged HFR estimates with substantially less bias. 
However, this resource was only available in retrospect, so these ``real-time'' estimates could not actually be produced until long after the dates in question.

\begin{figure}[h]
\centering
\includegraphics[width=\linewidth]{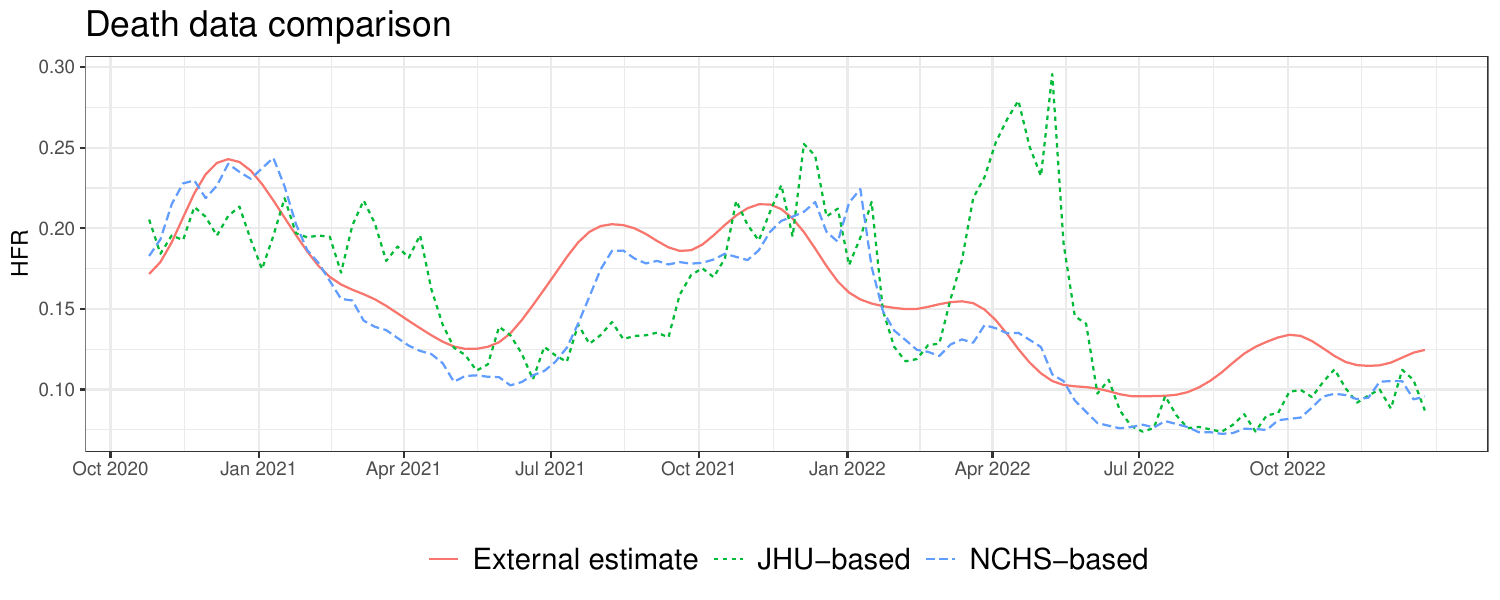}
\caption{Comparing lagged ratios based on data from JHU versus NCHS.}
\label{fig:jhu_vs_nchs}
\end{figure}

\section{Robustness checks}
\label[appendix]{apx:robustness}

\subsection{Alternative external estimates of HFR}
\label[appendix]{apx:alt_gt}

\begin{figure}[b!]
\centering
\includegraphics[width=\linewidth]{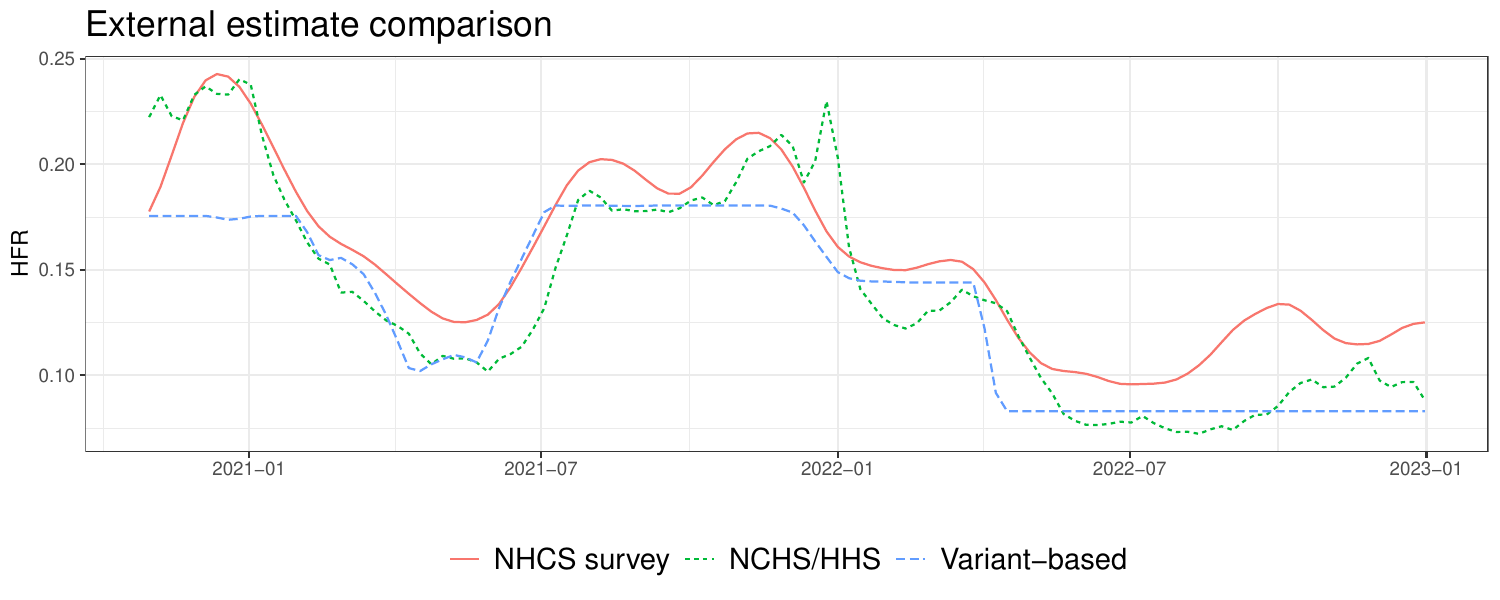}
\caption{Comparing methods that use external data to retrospectively approximate HFR.}
\label{fig:approxGT}
\end{figure}

We consider two alternative approaches to approximate the true HFR
curve, over the COVID-19 pandemic. The first is simply to use the lagged ratio
based on NCHS data. As discussed above, this benefits from a lighter-tailed
delay distribution than JHU data. 
Further, as we able to use data after $t$ to estimate $p_t$, we use the retrospective severity estimator introduced in Overton et al., 2022:
\begin{equation}\label{eq:conv-retro}
    \hat{p}_t = \sum_{k=0}^d
  \frac{{Y}_{t+k}\hat\pi^{(t)}_{k}}{\sum_{j=0}^d
  {X}_{t+k-j}\hat\pi^{(t+k-j)}_{j}}. 
\end{equation}
The second approach we consider is
to compute a single HFR by dividing total deaths by total hospitalizations in
each major variant period,  
and then create a smooth curve by mixing these per-variant HFRs by estimates
of the proportion of variants in circulation, obtained from CoVariants.org. We only consider the four largest variants: the original
strain, Alpha, Delta, and Omicron. However, because Omicron began with an
enormous surge that quickly subsided, we split it into early and late periods.

\cref{fig:approxGT} displays these two alternative HFR estimates,
alongside the curve obtained from NHCS. They have nontrivial differences in
magnitude, but reassuringly, the three curves move more or less in
conjunction. The retrospective NCHS ratio is still subject to statistical bias
similar to equation \eqref{eq:MispBias}; the variant-based HFR curve is flatter, as it 
does not account for other sources of potential variability in the underlying
severity rate.   


\subsection{Real-time versus finalized data}

Recall, the results in \cref{sec:results_real} use hospitalization and
death counts available in real time. To investigate the sensitivity of our
findings, we recompute the lagged and convolutional ratios, this time using
finalized counts. \cref{fig:rt_and_final} shows the real-time and
finalized estimates track one another very closely. Therefore, the observed bias
in \cref{fig:basic_est_vs_gt_figs} cannot be attributed to real-time
reporting quirks.  


\begin{figure}[htb]
\includegraphics[width=\linewidth]{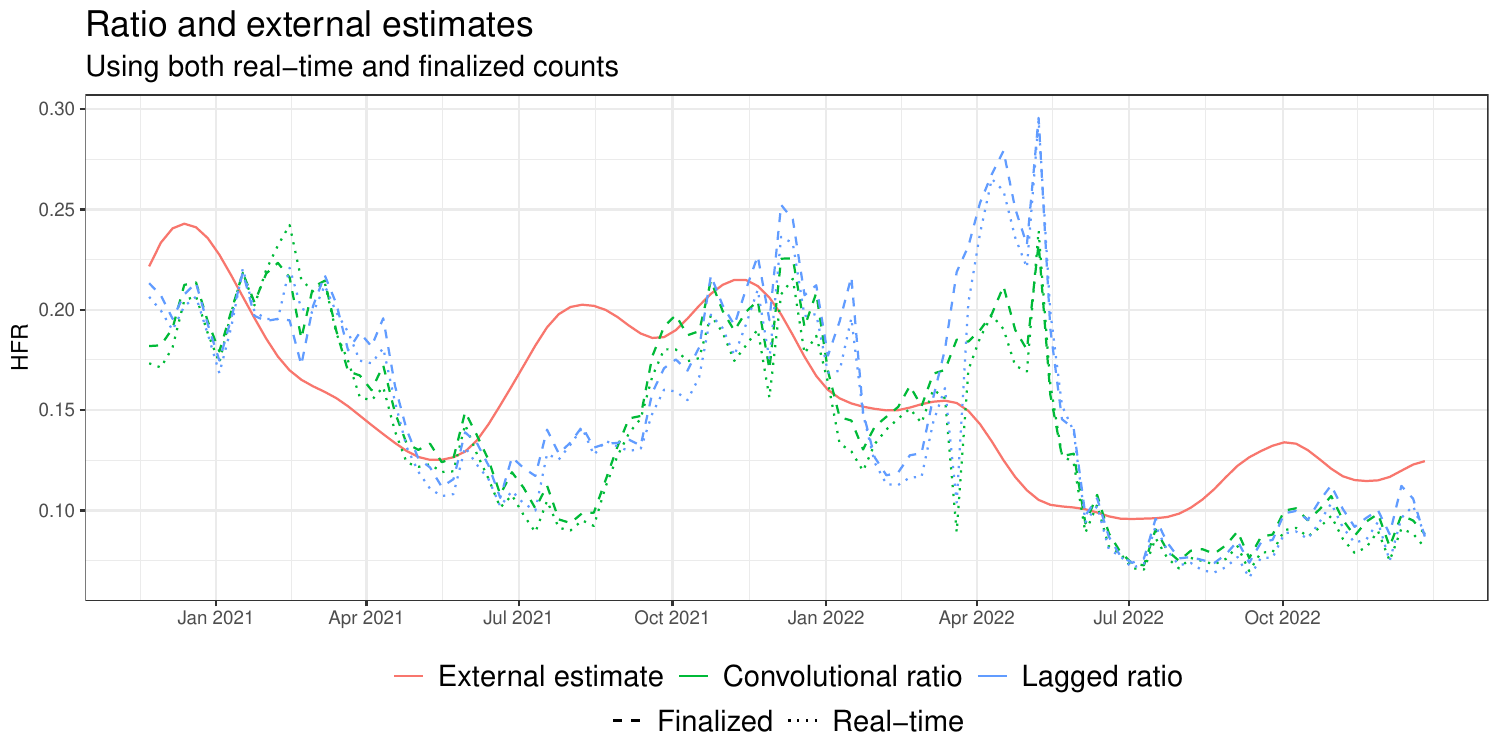}
\caption{Comparing estimates based on real-time versus finalized counts.}    
\label{fig:rt_and_final}
\end{figure}

\subsection{Hyperparameters}

We evaluate the robustness of our findings against choices of hyperparameters.
First, we analyze smoothed versions of the ratio estimators, where we smooth the 
numerator and denominator separately:
\begin{align*}
\hat{p}_t^{\ell,w} &= \frac{\sum_{s=t-w+1}^t y_s}
{\sum_{s=t-w+1}^t x_{s-\ell}}, \\
\hat{p}_t^{\gamma,w} &= \frac{\sum_{s=t-w+1}^t y_s}
{\sum_{s=t-w+1}^t \sum_{k=0}^d x_{s-\ell-k}\gamma_k}.
\end{align*} 
\cref{fig:window} shows the results for varying window lengths $w >
0$. The results are very similar, indicating the bias does not disappear when
smoothing over a longer history.  

\begin{figure}[h!]
\centering
\includegraphics[width=.95\linewidth]{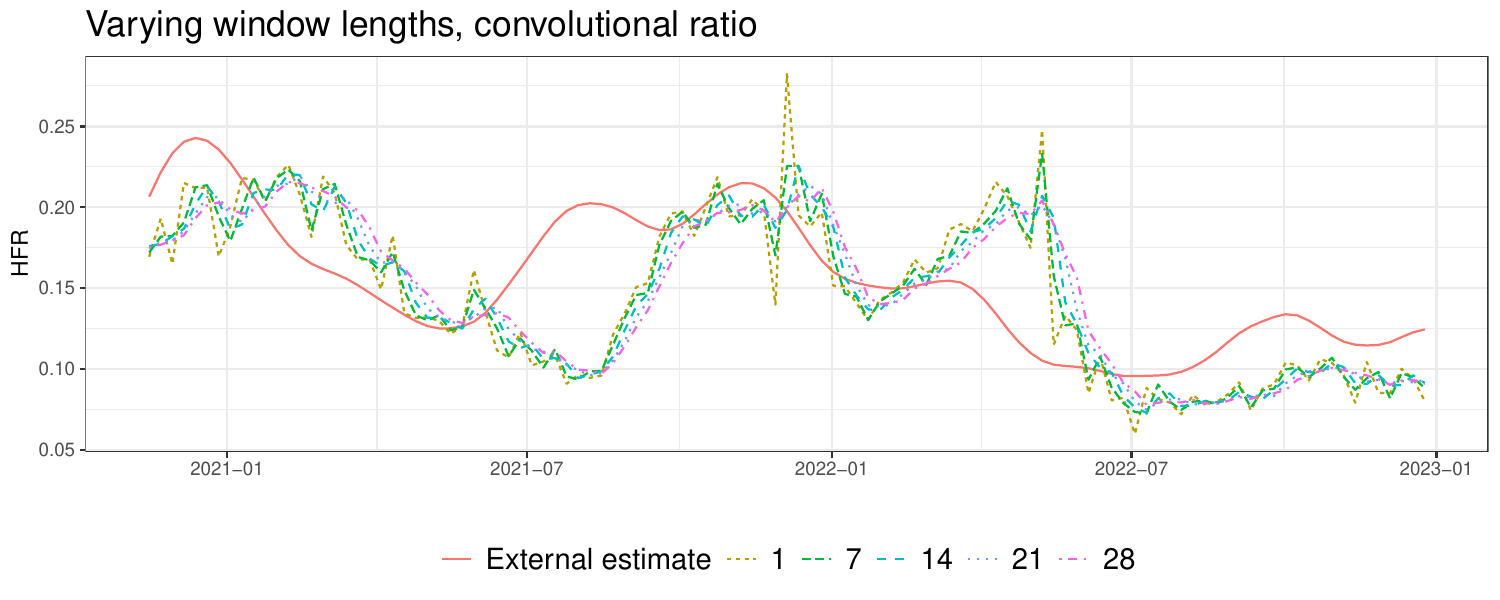}
\includegraphics[width=.95\linewidth]{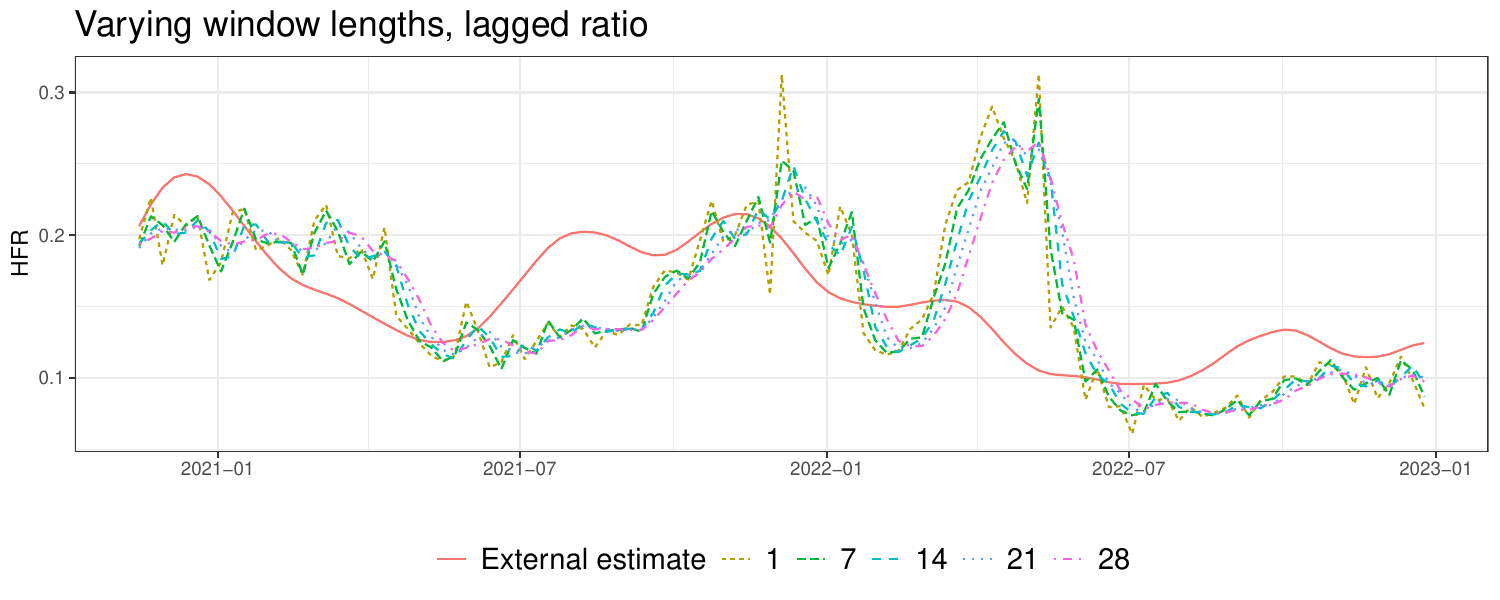}
\caption{Comparing different choices of window length in post-smoothing.}
\label{fig:window}
\end{figure}

We next examine the time-to-death hyperparameters: The lag for the lagged ratio
and delay distribution for the convolutional ratio. \cref{fig:lag}
displays lagged HFR estimates where the lag $\ell$ ranges from 2 to 5
weeks. Unlike the window size, changing this parameter leads to notably
different behavior. Some choices of lag are better than others; a 28-day lag,
for example, falls appropriately in winter 2021, and rises less slowly during
Delta. However, all are biased to varying degrees, most notably the huge
spurious surge in spring 2022.

\begin{figure}[p]
\centering
\includegraphics[width=\linewidth]{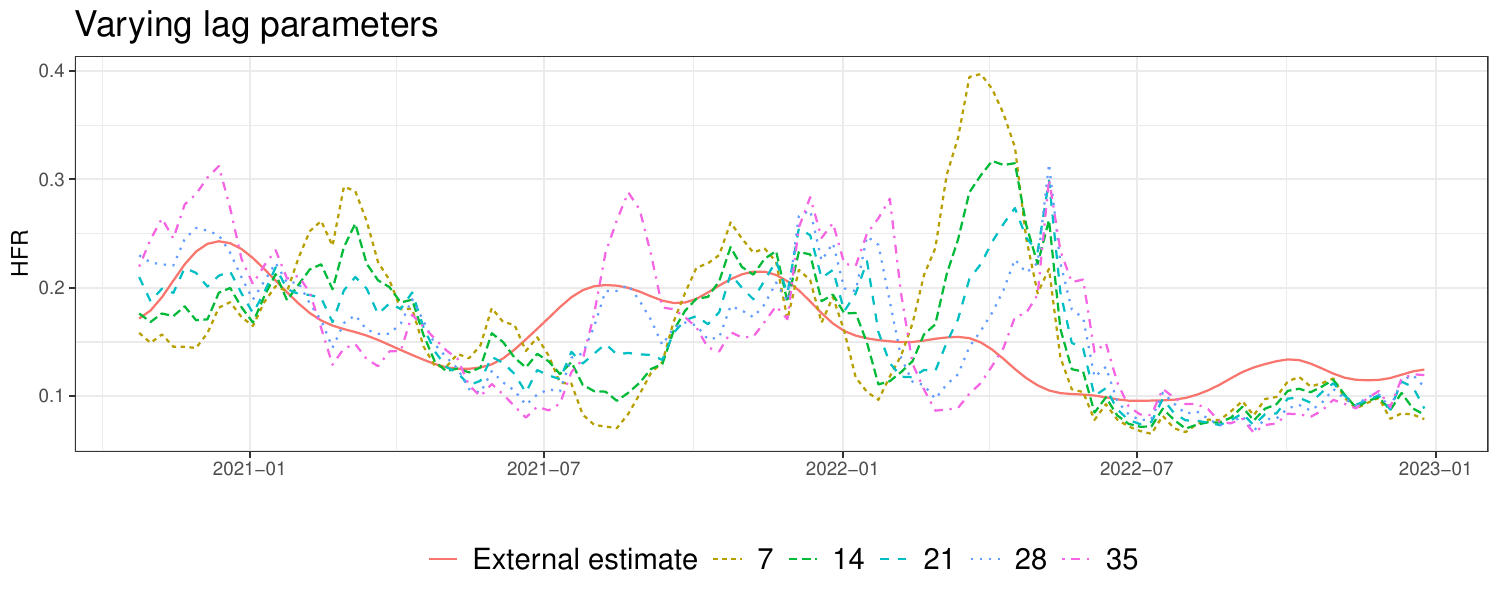}
\caption{Comparing different choices of lag parameter in the lagged ratio.}
\label{fig:lag}

\bigskip
\includegraphics[width=\linewidth]{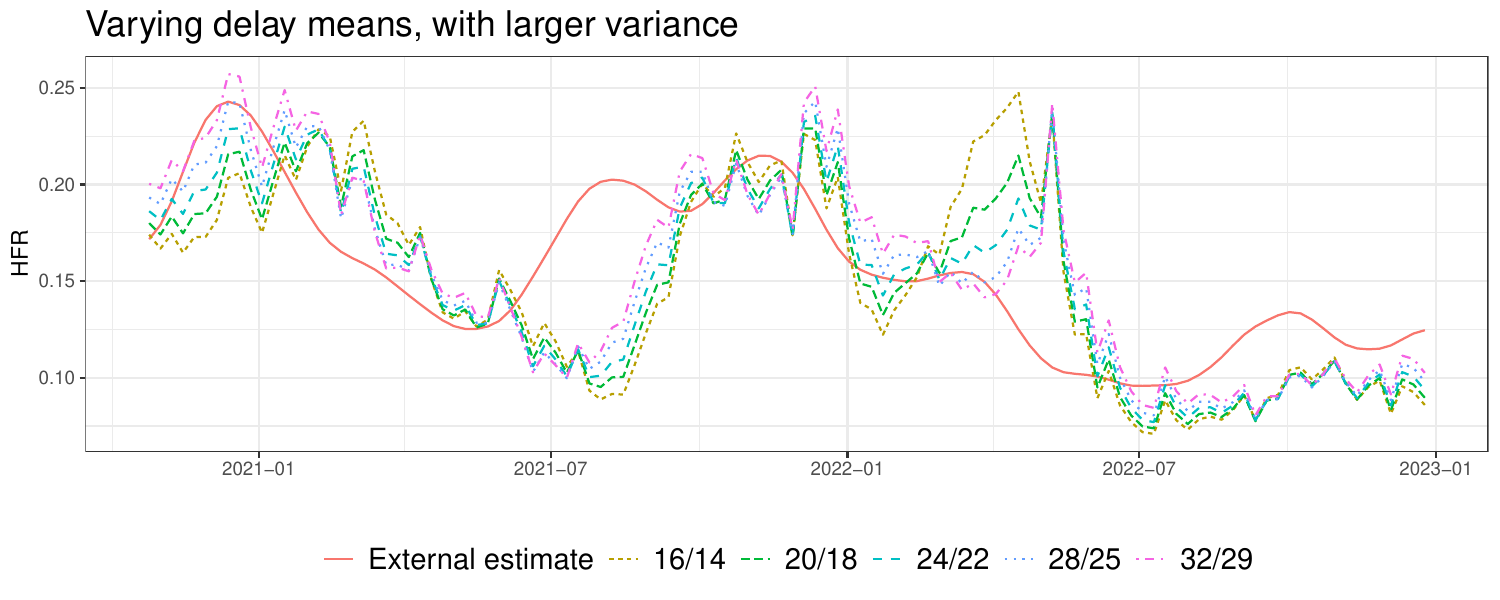}
\includegraphics[width=\linewidth]{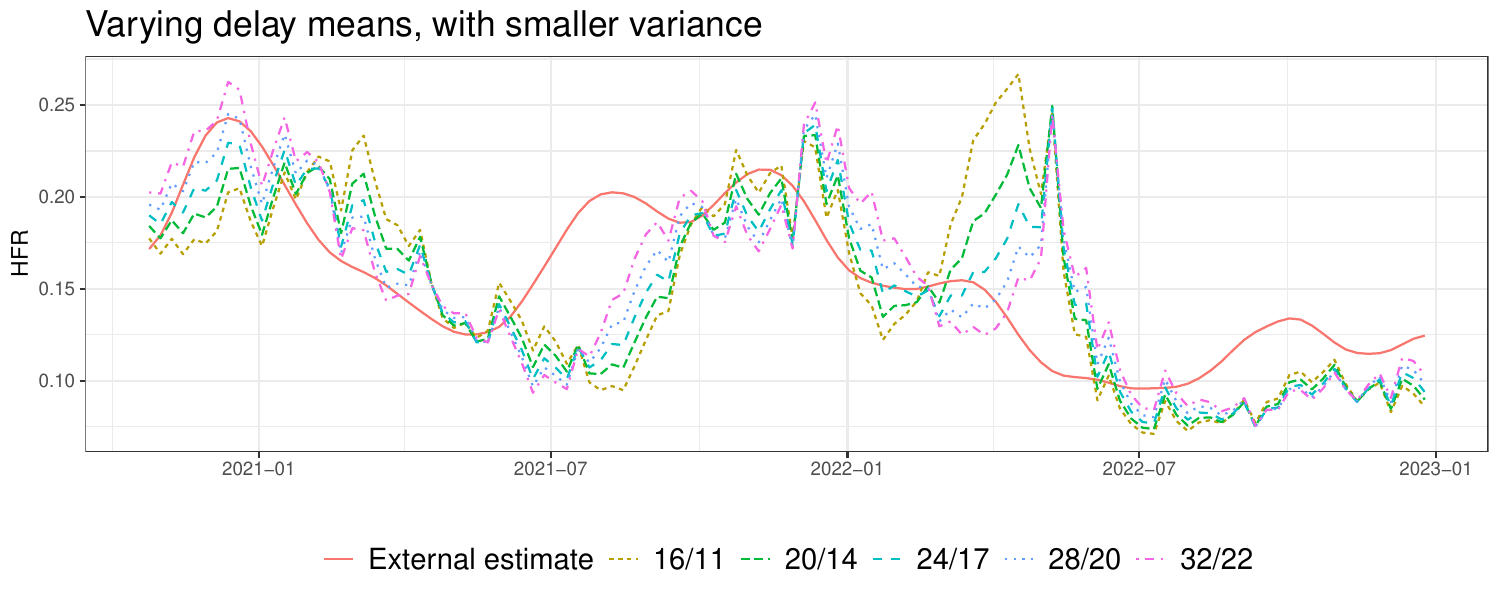}
\caption[Comparing different choices of delay distributions in the convolutional
  ratio.]{Comparing different choices of delay distributions in the convolutional
  ratio. The top panel shows gamma distributions whose standard deviation is 0.9
  times the mean, versus 0.7 times the mean the bottom panel. The legend labels
  reflect the mean/standard deviation.}  
\label{fig:delays}
\end{figure}

\cref{fig:delays} compares the performance of the convolutional ratio
across different choices of delay $\gamma$; we kept the discrete
gamma shape for each, but varied the mean and standard deviation. We
investigated a standard deviation equal to 90\% of the mean, and also a more
compact delay with standard deviation equal to 70\% of the mean. All
resulting HFR estimates are significantly biased. Regardless of delay
distribution, the ratios are negatively biased during the onset of Delta, and
surge after the peak of Omicron. This indicates the bulk of the error is
fundamental to the estimator, and cannot be attributed to model
misspecification.     


\subsection{State-level results}

We repeat our analysis the six large US states, finding similar trends. The NHCS
survey was conducted on a subset of hospitals meant to represent the US at
large, so it cannot be used to accurately estimate the HFRs within each state.
Instead, our external estimates of the state-level HFRs use the retrospective ratio discussed above (Equation \cref{eq:conv-retro}), with NCHS deaths.
\cref{fig:state-level} compares this external HFR estimate to the real-time 
convolutional and lagged ratios. 
For each state, we select the lag for the
real-time lagged ratio to maximize cross-correlation between hospitalizations and JHU deaths. 
We then use
a discrete gamma distribution with mean equal to this lag, and standard
deviation 90\% of its mean, for the convolutional ratio.

\begin{figure}[htbp]
\centering
\includegraphics[width=\linewidth]{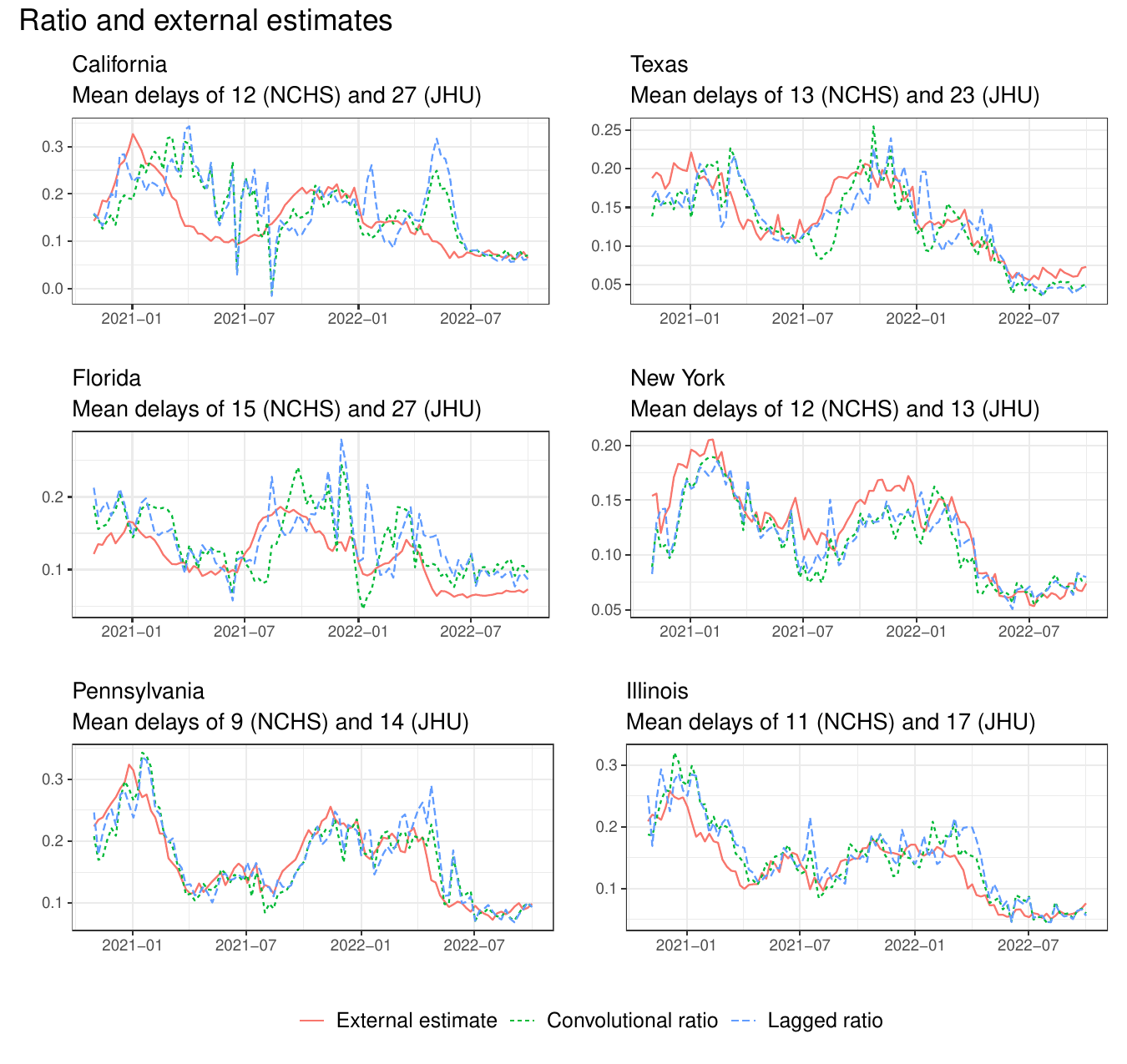}
\caption{Comparing ratio estimates for individual US states.} 
\label{fig:state-level}
 \end{figure}

Several states display similar biases to the US as whole (Figure
3). Estimates in California, Texas, and Florida are 
all slow to detect the uptick in HFR during Delta; they also spike during
Omicron in California, and to a lesser extent Florida. Note these states are the
ones with the largest cross-correlation optimal lags, an estimate of the average 
time-to-death. In contrast, New York, Pennsylvania, and Illinois have mean delays
of at most 17. Their HFR estimates are still biased, but overall less so. This
once again emphasizes the role of the delay in bias. The takeaway: fatality
ratios are generally less trustworthy in states that take longer to report
deaths. 

\section{CFR analysis}\label[appendix]{apx:CFR}

To illustrate the ubiquity of our analyses across all severity rates, we recreated our simulated experiments with CFR as opposed to HFR. 
As a recap, the case-fatality rate (CFR) measures the proportion of reported cases that ultimately die. 
It is often used as a proxy for the infection-fatality rate (IFR), since true infection counts are challenging to estimate, especially in real-time.

In this section, we recreate Figures 1, 2, and 4, using cases instead of hospitalizations. 
We restrict our focus to simulated deaths because unlike HFR, a strong external estimate for the true CFR is unavailable in real-time.
Another reason for this restriction is that CFR is a problematic surrogate for IFR. 
Not all infected individuals show up as positive case reports, and the case ascertainment rate varies over time.
Therefore, devising a schema to estimate IFR on COVID-19 data in real-time, and comparing it to notions of the ``true'' IFR and CFR is out of scope for this work.

The ground truth CFRs in these simulations follow the same NHCS-based curve in our prior experiments.
To reflect CFRs, we rescaled this curve such that the average CFR was the overall COVID-19 CFR from 2020-2022.
Deaths are simulated using the same delay distributions as in HFR. 
This is justified by our confirmation that the correlation-maximizing lag between US cases and deaths is the same as that between hospitalizations and deaths.

\begin{figure}[htb]
\centering
\begin{subfigure}[b]{0.325\linewidth}
  \centering
  \includegraphics[width=\linewidth]{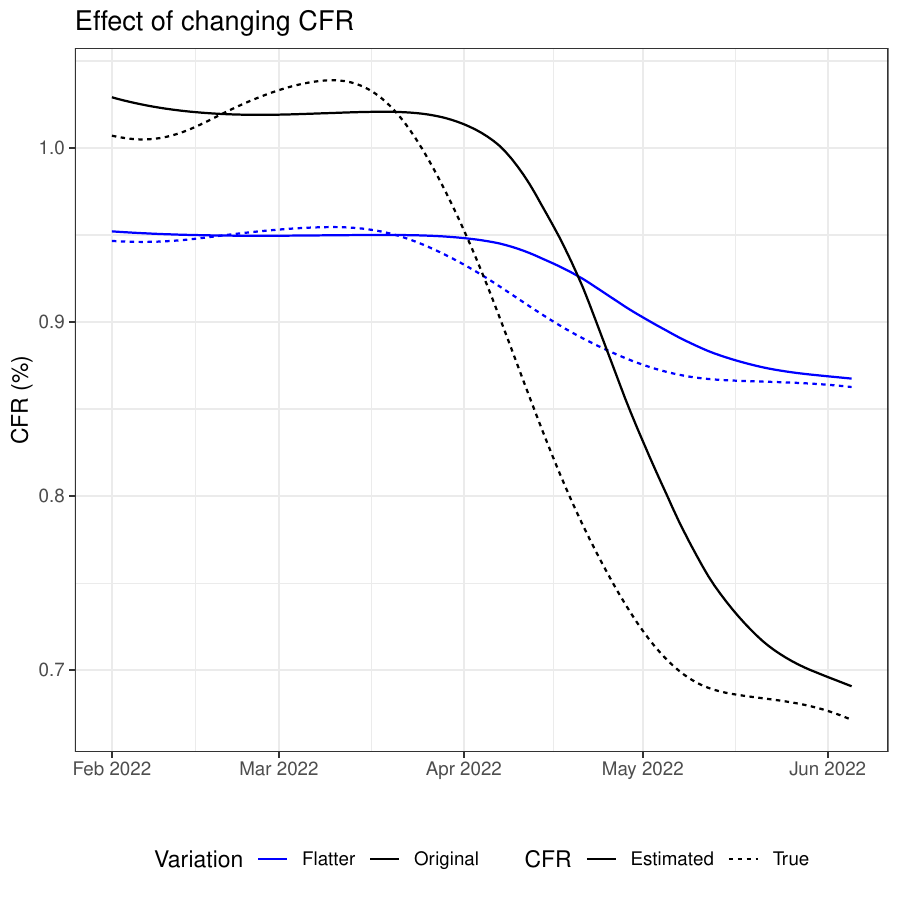} 
  \caption{}
  \label{fig:toy_cfr}
\end{subfigure}
\begin{subfigure}[b]{0.325\linewidth}
  \centering
  \includegraphics[width=\linewidth]{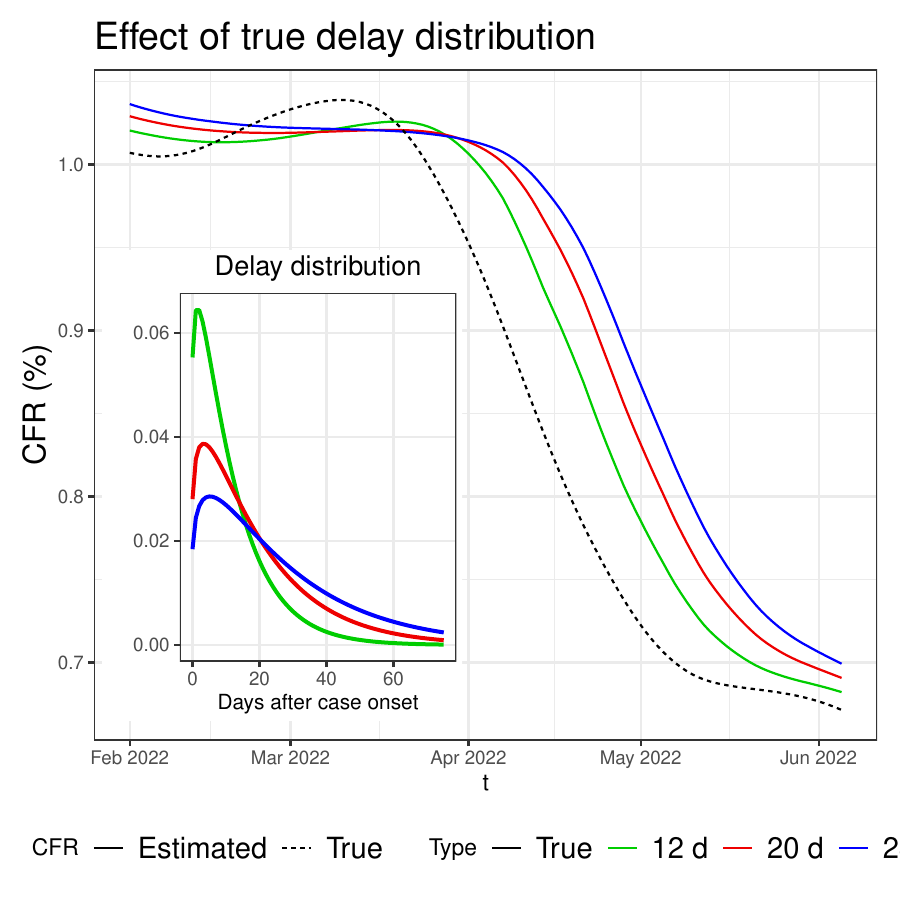}
  \caption{}
  \label{fig:toy_delay_cfr}
\end{subfigure}
\begin{subfigure}[b]{0.325\linewidth}
  \centering
  \includegraphics[width=\linewidth]{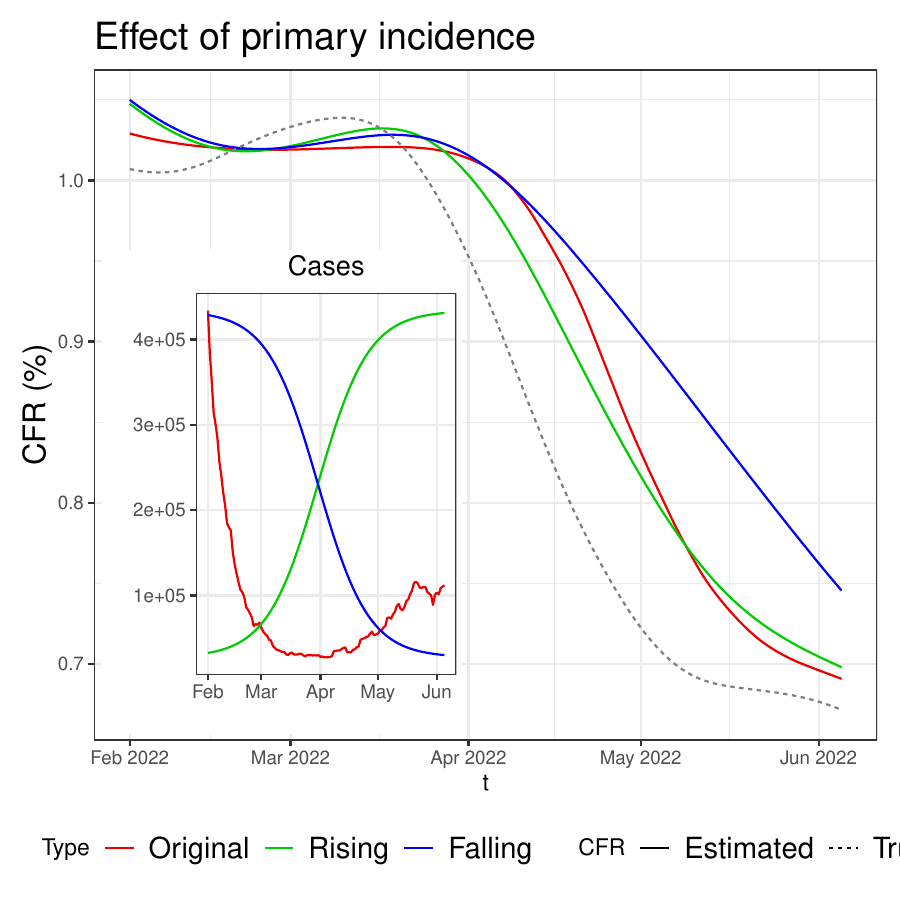} 
  \caption{}
  \label{fig:toy_primary_cfr}
\end{subfigure}
\caption[Recreating \cref{fig:wellspecified} with CFR instead of HFR.]{
  These figures illustrate the same qualitative trends as in \cref{fig:wellspecified}, 
  now using CFR instead of HFR. The same three factors—changing severity, delay distribution, 
  and primary incidence—produce similar effects on the bias of the well-specified convolutional 
  ratio.}
\label{fig:wellspecified_cfr}
\end{figure}

\cref{fig:wellspecified_cfr} presents the CFR counterpart to \cref{fig:wellspecified} in the main text. As with HFR, the same qualitative trends emerge: rapid changes in the underlying severity rate produce larger biases, longer delay distributions amplify these distortions, and variations in the primary incidence curve modulate both the direction and magnitude of bias. These parallels confirm that the mechanisms driving bias in the well-specified convolutional ratio are not specific to HFR but generalize directly to CFR.

\begin{figure}[p]
\centering
\includegraphics[width=0.9\linewidth]{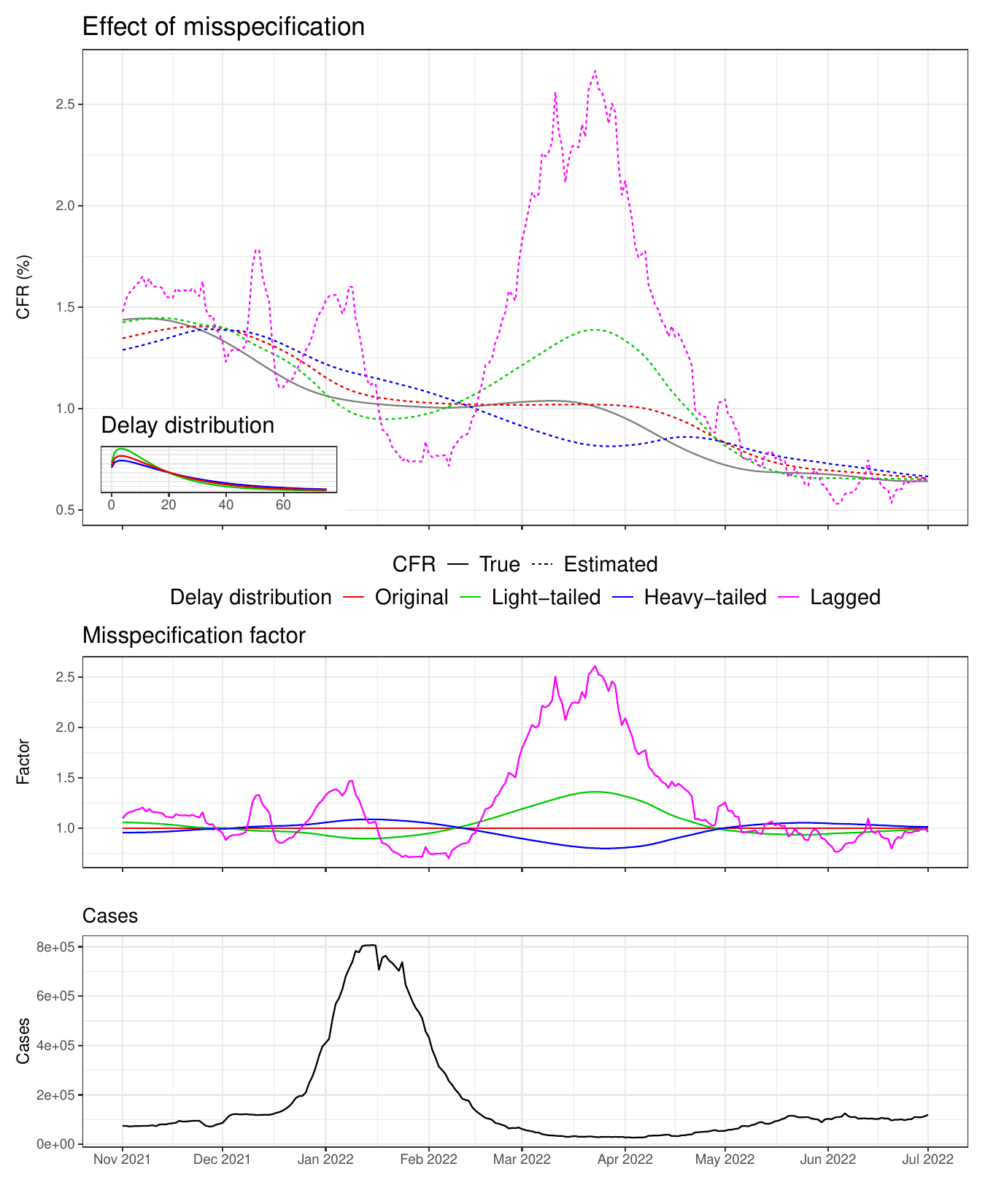}
\caption[Recreating \cref{fig:misspecified} with CFR instead of HFR.]{Analogous to \cref{fig:misspecified}, this figure illustrates
  the effects of delay-distribution misspecification when using CFR instead of HFR.
  The same qualitative behavior is observed across light-tailed, heavy-tailed,
  and lagged specifications.}
\label{fig:misspecified_cfr}
\end{figure}

\cref{fig:misspecified_cfr} shows the effects of delay‐distribution misspecification for CFR. The same qualitative patterns observed for HFR persist: the light‐tailed delay yields upward or downward deviations depending on incidence trends, the heavy‐tailed delay has smaller opposing bias, and the lagged ratio remains the most volatile. 
These results confirm that the misspecification mechanisms outlined in Proposition 2 extend directly to CFR.

In fact, the bias is even larger than in the HFR case.
The lagged ratio spikes to 2.5\%, a staggering 150\% above the true CFR of 1\%.
In contrast, lagged HFRs peaked at 25\% while the true values were 15\%
The behavior of $A_t^\ell$ explains this discrepancy.
It peaks around 2.5 for CFR, compared to 1.5 for HFR.
The larger value for CFR is driven by cases climbing higher and falling steeper than hospitalizations, in relative terms.

\begin{figure}[t!]
\centering
\includegraphics[width=\linewidth]{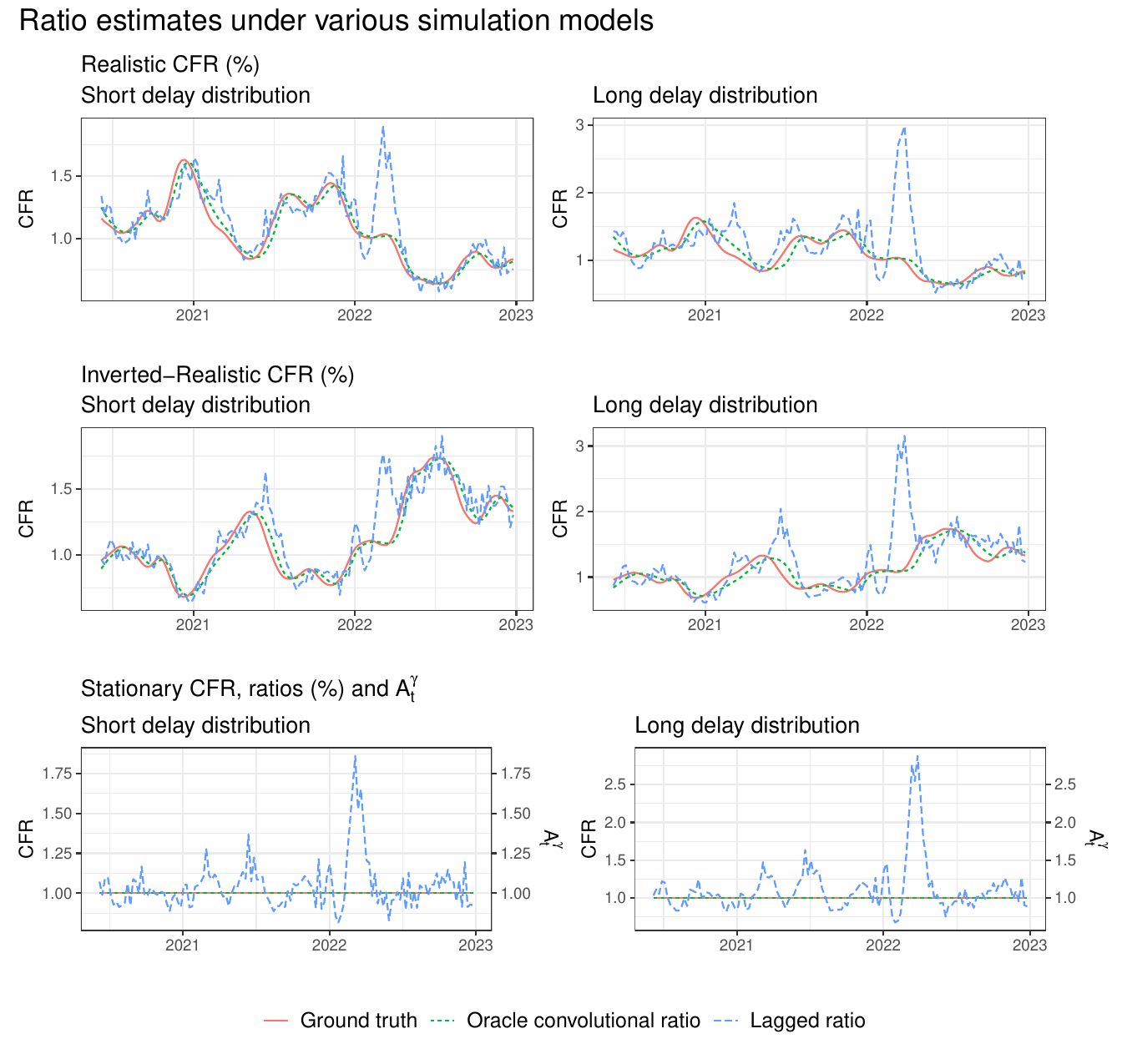}
\caption[CFR estimates under various simulation settings.]{
  Ratio estimates under various simulation settings, using CFR instead of HFR.
  As in \cref{fig:sims}, the convolutional ratio tracks the true severity curve closely,
  while the lagged ratio exhibits larger oscillations, particularly under the
  longer delay. 
}
\label{fig:sims_cfr}
\end{figure}

\cref{fig:sims_cfr} shows the simulated results for CFR, paralleling \cref{sec:results2} from the main text. The qualitative behaviors mirror those observed for HFR, confirming that the same bias mechanisms apply to CFR-based analyses.
The convolutional ratio continues to track the true curve relatively closely, while the lagged ratio exhibits exaggerated oscillations, especially under the longer delay. When incidence surges, the lagged CFR overshoots the truth; when incidence falls, it undershoots. These dynamics follow directly from the misspecification behavior of $A_t^\ell$ described in Proposition 2.

Even in the stationary-severity simulations, the lagged ratio remains heavily biased upward whenever $A_t^\ell>1$. Its peaks for CFR roughly twice as large as those seen for HFR. As discussed for \cref{fig:misspecified_cfr}, this difference is driven by the sharper swings in case incidence compared with hospitalizations. Meanwhile, the convolutional ratio’s stability reflects its correct specification of the delay distribution—its advantage holding equally for CFR as for HFR.

\chapter{\texorpdfstring{Supplementary material for \cref{ch:paper2}}{Supplementary material for Chapter 2}}
\label{app:B}

\section{Proof of \cref{prop:corr}}\label[appendix]{apx:corr}

For the sake of clarity, this proof assumes the maximal secondary event time is
$d$ time steps after a primary event. However, this constraint is not necessary
for the argument.  By definition, $\mu_t = \sum_{k=0}^d x_{t-k} \pi_k^{(t-k)}
p_{t-k}$ and $\mu_{t+1} = \sum_{k=0}^d  x_{t+1-k} \pi_k^{(t+1-k)}
p_{t+1-k}$. The first step is to bound the covariance:
\begin{equation}
\Cov(y_t, y_{t+1} \given x_{\leq t+1}) = \E[y_t y_{t+1} \given x_{\leq t+1}] - 
\mu_t\mu_{t+1}. 
\end{equation}
Define the Bernoulli random variable $D_{s,t}^{(i)}$ as the indicator that the 
$i\nth$ individual with primary event at time $s$ has secondary event at time 
$t$. With this notation, 
\begin{equation}
 y_t = \sum_{k=0}^{d} \sum_{i=1}^{x_{t-k}} D_{t-k,t}^{(i)},   
 \;\; \text{and} \;\; 
y_{t+1} = \sum_{\ell=0}^{d} \sum_{j=1}^{x_{t+1-\ell}} D_{t+1-\ell,t+1}^{(j)}.  
\end{equation}
Now partition the cross-terms of $\E[y_t y_{t+1} \given x_{\leq t+1}]$ into
three pieces: 
\begin{align}
\E[y_t y_{t+1} \given x_{\leq t+1}] 
&= \E\bigg[ \sum_{i=1}^{x_{t-d}} D_{t-d,t}^{(i)} \,\cdot\,
  \sum_{\ell=0}^{d} \sum_{j=1}^{x_{t+1-\ell}} D_{t+1-\ell,t+1}^{(j)} \bigg] 
+ \E\bigg[ \sum_{k=0}^{d} \sum_{i=1}^{x_{t-k}}D_{t-k,t}^{(i)} \,\cdot\,
  \sum_{j=1}^{x_{t+1}} D_{t+1,t+1}^{(j)} \bigg] \\
&\qquad\qquad + \E\bigg[ \sum_{k=0}^{d-1} \sum_{i=1}^{x_{t-k}} D_{t-k,t}^{(i)}  
  \,\cdot\,  \sum_{\ell=1}^{d} \sum_{j=1}^{x_{t+1-\ell}} D_{t+1-\ell,t+1}^{(j)}
  \bigg] \\ 
&= \Big( x_{t-d}\pi_d^{(t-d)}p_{t-d} \Big) \mu_{t+1} + 
\Big( x_{t+1}\pi_{0}^{(t+1)}p_{t+1}\Big) \mu_{t} \\
&\qquad\qquad + \E\bigg[ \sum_{k=0}^{d-1} \sum_{i=1}^{x_{t-k}}
  D_{t-k,t}^{(i)} \,\cdot\, \bigg( \sum_{\ell\neq k+1} \sum_{i=1}^{x_{t+1-\ell}} 
  D_{t+1-\ell,t+1}^{(i)} + \sum_{j=1}^{x_{t-k}} D_{t-k,t+1}^{(j)} \bigg) \bigg].
\end{align}
The term on the final line may be further split into pieces, rewritten as:   
\begin{align}
&\sum_{k=0}^{d-1} x_{t-k} \pi_k^{(t-k)} p_{t-k} \sum_{\ell\neq k+1}
  x_{t+1-\ell}\pi_{\ell}^{(t+1)}p_{t+1-\ell} + \E\bigg[ \sum_{k=0}^{d-1}
  \sum_{i=1}^{x_{t-k}} D_{t-k,t}^{(i)} \,\cdot\, \sum_{j=1}^{x_{t-k}} 
  D_{t-k,t+1}^{(j)} \bigg] \\
&= \sum_{k=0}^{d-1} x_{t-k} \pi_k^{(t-k)} p_{t-k} \sum_{\ell\neq k+1} 
   x_{t+1-\ell}\pi_{\ell}^{(t+1)}p_{t+1-\ell} \\ &\qquad\qquad +
  \sum_{k=0}^{d-1} \sum_{i=1}^{x_{t-k}} \E \bigg[ D_{t-k,t}^{(i)} \sum_{j\neq i}
  D_{t-k,t+1}^{(i)} \bigg] + \sum_{k=0}^{d-1} \sum_{i=1}^{x_{t-k}} \E 
  \bigg[\underbrace{D_{t-k,t}^{(i)} D_{t-k,t+1}^{(i)}}_{\text{mutually
  exclusive}} \bigg] \\  
&=\sum_{k=0}^{d-1} x_{t-k} \pi_k^{(t-k)} p_{t-k} \bigg(
 \sum_{\ell\neq k+1} x_{t+1-\ell} \pi_{\ell}^{(t+1-\ell)}p_{t+1-\ell} +
  (x_{t-k}-1) \pi_{k+1}^{(t-k)} p_{t-k} \bigg) \\ 
&= \sum_{k=0}^{d-1} \sum_{\ell=0}^{d-1} \Big( x_{t-k}\pi_k^{(t-k)}p_{t-k} 
\Big) \Big( x_{t-\ell}\pi_{\ell+1}^{(t-\ell)}p_{t-\ell} \Big) - 
\sum_{k=0}^{d-1} x_{t-k} \Big( \pi_k^{(t-k)} \pi_{k+1}^{(t-k)} \Big) p_{t-k}^2 .
\end{align}
So putting all parts together,
\begin{align}
\E[y_t y_{t+1} \given x_{\leq t+1}] &= 
\Big( x_{t-d}\pi_d^{(t-d)}p_{t-d} \Big) \mu_{t+1} + 
\Big( x_{t+1}\pi_{0}^{(t+1)}p_{t+1}\Big) \\
&\qquad\qquad + \sum_{k=0}^{d-1} \sum_{\ell=0}^{d-1} \Big( x_{t-k}
  \pi_k^{(t-k)}p_{t-k} \Big) \Big( x_{t-\ell}\pi_{\ell+1}^{(t-\ell)}p_{t-\ell}
  \Big) - \sum_{k=0}^{d-1} x_{t-k} \Big( \pi_k^{(t-k)} \pi_{k+1}^{(t-k)} \Big)
  p_{t-k}^2 \\
&= \underbrace{\bigg( \sum_{k=0}^d x_{t-k}\pi_k^{(t-k)}p_{t-k} \bigg)}_{\mu_t} 
\underbrace{\bigg( \sum_{\ell=0}^d x_{t+1-\ell}\pi_\ell^{(t+1-\ell)}p_{t+1-\ell}
\bigg)}_{\mu_{t+1}} - \sum_{k=0}^{d-1} x_{t-k} \Big( \pi_k^{(t-k)}
\pi_{k+1}^{(t-k)} \Big) p_{t-k}^2,
\end{align}
and therefore,
\begin{equation}
\Cov(y_t, y_{t+1} \given x_{\leq t+1}) = \E[y_t y_{t+1} \given x_{\leq t+1}] -
\mu_t\mu_{t+1} = -\sum_{k=0}^{d-1} x_{t-k} \Big( \pi_k^{(t-k)} \pi_{k+1}^{(t-k)}
\Big) p_{t-k}^2 .
\end{equation}
We can see that the correlation clearly must be nonpositive. To lower bound the
correlation, assume equal (conditional) variances at times $t$ and $t+1$. Then 
\begin{align}
\Cor(y_t,y_{t+1} \given x_{\leq t})
&= \frac{\Cov(y_t,y_{t+1}\given x_{\leq t})}{\Var(y_t \given x_{\leq t})} \\
&= -\frac{\sum_{k=0}^{d-1}x_{t-k}\pi_k^{(t-k)}\pi_{k+1}^{(t-k)}p_{t-k}^2} 
{\sum_{k=0}^{d}x_{t-k}\pi_k^{(t-k)}p_{t-k}(1-\pi_k^{(t-k)}p_{t-k})}\\
&\geq -\frac{(\max_{k \geq 0} \pi_k^{(t-k)}p_{t-k})
\sum_{k=0}^{d-1}x_{t-k}\pi_k^{(t-k)}p_{t-k}} 
{(1-\max_{k \geq 0} \pi_k^{(t-k)}p_{t-k})
\sum_{k=0}^{d-1}x_{t-k}\pi_k^{(t-k)}p_{t-k}}\\
&= -\frac{\max_{k \geq 0} \pi_k^{(t-k)}p_{t-k}}
{1-\max_{k \geq 0} \pi_k^{(t-k)}p_{t-k}}.
\end{align}

\section{Gaussian deconvolution filtering}\label[appendix]{apx:gauss}

The Gaussian approximation is an alternative to the Poisson approximation to the
Poisson-binomial. It uses a mean $\mu_t$ and variance $\sigma_t^2$ as defined in 
\eqref{eq:pb_mean} and \eqref{eq:pb_var}, respectively. With the Gaussian log
likelihood in place of the Poisson log likelihood, the retrospective
optimization problem \eqref{eq:tf-pois} becomes  
\begin{equation}
\minimize_{0 \preceq p \preceq 1} \; \sum_t \bigg[\bigg( 
\log \hat\sigma_t^2(p) + \frac{1}{\hat\sigma_t^2(p)} \bigg( 
y_t - \sum_{k=0}^d x_{t-k}\hat\pi_k^{(t-k)}p_{t-k} \bigg)^2\bigg] +  
\lambda\|D^{(m+1)}p\|_1,
\end{equation}
where \smash{$\hat\sigma_t^2(p) = \sum_{k=0}^\infty x_{t-k}\hat\pi_k^{(t-k)}
  p_{t-k} (1-\pi_k^{(t-k)} p_{t-k})$}. The above is a nonconvex problem, hence
to simplify computation, we drop the log term and approximate
\smash{$\hat\sigma_t^2(p)$} by an estimate \smash{$\hat\mu_t$} of the mean,
motivated by the fact that in our problem setting, $\sigma_t^2$ is close to
$\mu_t$ (recall the discussion in \cref{sec:exact-lik}). The estimate
\smash{$\hat\mu_t$} can be obtained by nonparametric smoothing of the secondary
event time series $y_t$, and can therefore be treated as fixed (not depending on
$p$) in the optimization. This results in 
\begin{equation}
\minimize_{0 \preceq p \preceq 1} \; \sum_t \frac{1}{\hat\mu_t} \bigg( y_t -
\sum_{k=0}^d x_{t-k}\hat\pi_k^{(t-k)}p_{t-k} \bigg)^2 + \lambda\|D^{(m+1)}p\|_1, 
\end{equation}
for the Gaussian retrospective deconvolution estimator. The real-time estimator
is defined similarly, simply appending the additional regularization terms as in
\eqref{eq:tf-pois-rt}. 

In experiments not shown, we found that the Gaussian model for deconvolution
performed similarly but marginally worse than the Poisson model. With constant
order ($m=0$), the average performance was nearly as strong as that of the
Poisson model. However, it was somewhat worse on higher orders, especially in 
the real-time case. We did not find significant gaps in performance when 
analyzing by region population. Finally, Gaussian deconvolution was not hindered
by its use of the plug-in estimate \smash{$\hat\mu_t$} of the variance. We also
ran experiments with the true (fixed) variance was used in lieu of
\smash{$\hat\mu_t$}, and the performance was very similar. 

\section{Derivation of \texorpdfstring{$\lambda_\text{max}$}{lambda\_max}}\label[appendix]{apx:lambda-max}

Consider first the general optimization problem, for a differentiable convex
loss $\ell$, matrix $D$, and norm $\|\cdot\|$, 
\begin{equation}
\label{eq:gen-opt}
\minimize_\theta \; \ell(\theta) + \lambda \|D\theta\|.
\end{equation}
The first-order optimality conditions determining a solution \smash{$\htheta$}
are   
\begin{equation}
\label{eq:first-order}
-\nabla \ell(\htheta) =\lambda D^\T v, \;\; \text{where $v \in \partial
  \|D\htheta\|$}, 
\end{equation}
where $\nabla\ell(\theta)$ denotes the gradient of $\ell$ at $\theta$, and
\smash{$\partial \|D\htheta\|$} denotes the subdifferential (the set of 
subgradients) of $\|\cdot\|$ at \smash{$D\htheta$}. Suppose that     
\begin{equation}  
\label{eq:lambda-max}
\lambda \geq \underbrace{\min \Big\{ \|(D^\T)^+ \nabla \ell(\theta) + 
\eta\|_* : \theta \in \nul(D), \, \eta \in \row(D) \Big\}}_{\lambda_\text{max}},     
\end{equation}
where $\nul(D)$ denotes the null space of $D$, $\row(D)$ denotes the row space
of $D$, $(D^\T)^+$ denotes the generalized inverse of $D^\T$, and $\|\cdot\|_*$
denotes the dual norm to $\|\cdot\|$. If \smash{$\htheta, \hat\eta$} denotes a 
solution to the above minimization determining \smash{$\lambda_\text{max}$},
then    
\begin{equation}
v = \frac{-(D^\T)^+ \nabla \ell(\htheta) + \hat\eta}{\lambda}
\end{equation}
is a valid subgradient of $\|\cdot\|$ at \smash{$D\htheta = 0$}, because  
\smash{$\|v\|_* = \|(D^\T)^+ \nabla \ell(\htheta) + \hat\eta\|_* / \lambda \leq
  1$} (note the subdifferential of $\|\cdot\|$ at 0 is precisely the unit ball
in the dual norm, centered at the origin). In other words, the pair
\smash{$\htheta, v$} solves the first-order conditions \eqref{eq:first-order}
and \smash{$\htheta$} solves \eqref{eq:gen-opt}.      

We now inspect \smash{$\lambda_\text{max}$} as defined in \eqref{eq:lambda-max} 
for the Poisson linear inverse problem \eqref{eq:tf-pois}. Write $A$ for the
linear operator such that 
\begin{equation}
(A p)_t = \sum_{k=0}^d x_{t-k}\hat\pi_k^{(t-k)} p_{t-k}, 
\end{equation}
and $Y$ for the vector of secondary event counts. Then the Poisson linear
inverse loss can be expressed as 
\begin{equation}
\ell(p) = 1^\T Ap - Y^\T \log(Ap),
\end{equation}
where $1$ denotes the vector of all 1s, and $\log(Ap)$ is interpreted as 
elementwise application of log. Hence
\begin{equation}
\nabla \ell(p) = A^\T (1 - Y / Ap),
\end{equation}
where $Y / Ap$ is interpreted as elementwise division. Abbreviating $D =
D^{(m+1)}$, this matrix always will always be full row rank, and thus $(D^\T)^+
=  (DD^\T)^{-1} D^\T$. Then, with $\|\cdot\| = \|\cdot\|_1$ and $\|\cdot\|_* = 
\|\cdot\|_\infty$, note that 
\begin{equation}
\label{eq:lambda-max-pois}
\lambda_\text{max} = \min \Big\{ \|(DD^\T)^{-1} D^\T A^\T (1 - Y / Ap) +
\eta\|_\infty :  p \in \nul(D), \, \eta \in \row(D), \, 0 \preceq p \preceq 1
\Big\}.   
\end{equation}
Setting $\eta = 0$ and $p = B\alpha$, where $B$ is a matrix whose columns span
$\nul(D)$---which for the trend filtering penalty matrix is the space of
polynomials of degree $m$---we obtain a simpler upper bound    
\begin{equation}
\label{eq:lambda-max-pois-ub}
\lambda_\text{max} \leq \min \Big\{ \|(DD^\T)^{-1} D^\T A^\T (1 - Y /
AB\alpha)\|_\infty :  0 \preceq B\alpha \preceq 1 \Big\}.
\end{equation} 
This reduces to rational optimization (sum of ratio of polynomials), since we
can equivalently write it as
\begin{multline}
\lambda_\text{max} \leq \min \bigg\{ \max_s \; \bigg| \sum_t [(DD^\T)^{-1} D^\T 
A^\T]_{st}  \bigg(1 - \frac{y_t}{\sum_{k=0}^d x_{t-k}\hat\pi_k^{(t-k)} g(t-k)}
\bigg) \bigg| : \\ \text{$g$ is a degree $m$ polynomial with $0 \leq g \leq 1$} 
  \bigg\} .
\end{multline}
In principle, this upper bound should be computable to arbitrary accuracy using
SDP relaxation techniques \citep{lasserre2021minimizing}, though this is likely
to be difficult in practice. Instead, we can rewrite
\eqref{eq:lambda-max-pois-ub} once more as a min-max problem,  
\begin{equation}
\lambda_\text{max} \leq \min_{0 \preceq B\alpha \preceq 1} \; 
\max_{\|u\|_1 \leq 1} \; u^\T (DD^\T)^{-1} D^\T A^\T (1 - Y / AB\alpha),
\end{equation} 
and then use alternating optimization to approximate the solution. That is, we 
iteratively take a (projected) gradient step with respect to $\alpha$, and 
perform exact minimization over $u$ (which is easy, since $u^\T x$ for a fixed  
vector $x$ is minimized over all $\|u\|_1 \leq 1$ by taking $u = \sign(x_i)
e_i$, where $|x_i| = \|x\|_\infty$ and $e_i$ is the $i\nth$ canonical basis
vector). Note that after any number of iterations of alternating minimization,
we can take the resulting $\alpha$ and plug this into
\eqref{eq:lambda-max-pois-ub} to obtain a valid upper bound. 

\section{Further experimental details}

\subsection{Dispersion and lag calculation}\label[appendix]{apx:jhu-clean}

Before estimating the dispersion for our beta-binomial simulation model, it was
necessary to preprocess the real-time JHU data. The raw data contained many
outliers. Often, large numbers of previously unreported counts were dumped on a
certain day. Low outliers were also prevalent, for example on holidays, or in the
use of negative values to revert cumulative totals after accounting for
duplicate reports. To deal with these issues, we preprocessed the data to remove
outliers. 

First, we replaced missing values with zeros.  Next, we identified clear
outliers that exceeded 3 standard deviations from a 6-week rolling mean,
replacing them with the mean itself.  We also addressed obvious data dumps:
stretches of 3 or more days with no counts, followed by a nonzero value. We
distributed this value backward over the trailing zero-count days, allocating
uniformly with a multinomial sampler.  Similarly, we redistributed negative
counts over the entire preceding history, with a multinomial sampler weighted by
the counts themselves. The negative count was then replaced with zero.

We estimated dispersion by fitting a quasi-Poisson regression to the cleaned
death counts, using a 4-week moving average as the mean curve.  This regression
included day-of-week indicator variables, to prevent bias due to systematic
under- or over-reporting. The quasi-Poisson fit estimated the dispersion, which
we used to simulate deaths.  Finally, to compute the lag between
hospitalizations and deaths, we took weekly averages of death counts to remove
day-of-week effects. Then, we computed the cross-correlation function between
the two time series, and identified the lag which maximized this function. To
ensure it fell within a reasonable range, we capped the optimal lag at a minimum
of 6 days and a maximum of 35.

\subsection{Beta-binomial distribution}\label[appendix]{apx:beta-binom}

Like the Poisson-binomial distribution, the beta-binomial is supported on
integers between 0 and the number of trials. This distribution counts the number
of successes among independent Bernoulli trials whose success probabilities are
drawn from a given beta distribution, whose parameters reflect the amount of
overdispersion. For $n$ trials, the beta-binomial distribution is often written
as $\text{BetaBinom}(n,M,\rho)$, which is called mean-rho parameterization. Here
$M\in(0,1)$ is the average probability of success, and $\rho\in[0,1)$ controls
the degree of overdispersion. (Another formulation is
$\text{BetaBinom}(n,\alpha,\beta)$, where $\alpha,\beta$ parameterize the beta
distribution; there is a one-to-one relation between the two parameterizations:
$M=\alpha / (\alpha+\beta)$, and $\rho=1/(\alpha+\beta+1)$.) The distribution
$\text{BetaBinom}(n,M,\rho)$ has mean $\mu = nM$ and variance     
\begin{equation}
\sigma^2 = nM(1-M)[1+(n-1)\rho]. 
\end{equation}
Rearranging the variance expression reveals
\begin{equation}\label{eq:rho}
\rho = \bigg( \frac{\sigma^2}{nM(1-M)} - 1 \bigg) \bigg( \frac{1}{n-1} \bigg).  
\end{equation}
In our severity rate estimation context, the number of secondary events at each
time $t$ is approximated by a beta-binomial variate. Recall that its mean
$\mu_t$ is as in \eqref{eq:pb_mean}. The number of trials $n_t$ at $t$ is the
total primary incidence over the previous $d$ days, that is, the pool of
individuals who may have a secondary event at $t$. The average success
probability at $t$ is thus  
\begin{equation}
M_t = \frac{\mu_t}{n_t} =\frac{\mu_t}{\sum_{k=0}^d x_{t-k}},
\end{equation}
which explains the first expression in \eqref{eq:beta-binom}. Moreover, given
the Poisson-binomial variance $\sigma_t^2$ in \eqref{eq:pb_var}, and a
dispersion factor \smash{$\hbeta$}, the variance of secondary incidence is
\smash{$\hbeta\sigma_t^2$}. Plugging this variance into \eqref{eq:rho} confirms
the second expression in \eqref{eq:beta-binom}.   

\subsection{Misspecified delay distributions}\label[appendix]{apx:misp}

\cref{fig:misp-delays} displays the misspecified delay distributions used
in our experiments. The black line is the original delay distribution, whereas
the colored ones have varying degrees of misspecification. The delay
distributions differ across states, since their means depend on the correlation
between primary and secondary events. We choose to visualize California,
Louisiana, and Wyoming as a representative trio because they are the largest,
roughly median, and smallest state by population, respectively. 

\begin{figure}[t]
\centering
\includegraphics[width=\textwidth]{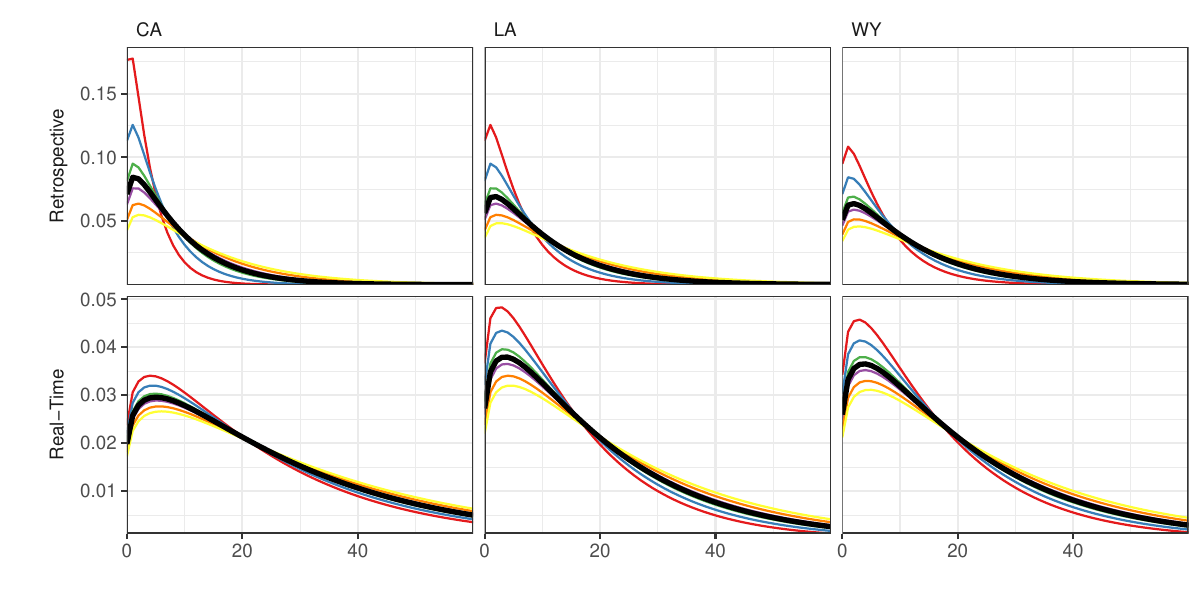}
\caption[Example misspecified delay distributions, with varying means.]{Example misspecified delay distributions, for three states, as we vary 
  the mean (which recall is tied to the variance), in both retrospective and
  real-time simulation models.} 
\label{fig:misp-delays}
\end{figure}

\subsection{Standard error of reported MAE}\label{apx:se-calc}

Let $n_r = 51$ (regions) and $n_i = 10$ (iterations). The estimand is the finite-population average of true per-region MAE,
\[
\theta = \frac{1}{n_r}\sum_{r=1}^{n_r} \mu_r, \qquad \mu_r = E[\mathrm{MAE}_{ri}],
\]
estimated by $\bar X = \frac{1}{n_r n_i}\sum_{r,i}\mathrm{MAE}_{ri}$. 
Since iterates are i.i.d.\ across $i$, $E[\bar X] = \theta$. For the variance, expand as a double sum:
\[
\mathrm{Var}(\bar X) = \frac{1}{n_r^2 n_i^2}\sum_{r,i}\sum_{r',i'} \mathrm{Cov}(\mathrm{MAE}_{ri},\, \mathrm{MAE}_{r'i'}).
\]
We simplify this in three steps.

\textit{Step 1: i.i.d.\ across iterations.} Because iterates are independent across $i$, every cross-iterate covariance ($i \neq i'$) vanishes; because they are identically distributed, the surviving $i = i'$ terms do not depend on $i$. The outer sum over $i$ thus contributes a factor of $n_i$:
\[
\mathrm{Var}(\bar X) = \frac{1}{n_r^2 n_i}\sum_{r,r'} \mathrm{Cov}(\mathrm{MAE}_{ri},\, \mathrm{MAE}_{r'i}).
\]

\textit{Step 2: split diagonal and off-diagonal.} Writing $\mathrm{Cov}(\mathrm{MAE}_{ri}, \mathrm{MAE}_{ri}) = \sigma_r^2$,
\[
\mathrm{Var}(\bar X) = \underbrace{\frac{1}{n_r^2}\sum_{r=1}^{n_r} \frac{\sigma_r^2}{n_i}}_{\text{(I)}} + \underbrace{\frac{1}{n_i n_r^2}\sum_{r \neq r'} \mathrm{Cov}(\mathrm{MAE}_{ri},\, \mathrm{MAE}_{r'i})}_{\text{(II)}}.
\]

\textit{Step 3: cross-region independence.} The only source of randomness in $\mathrm{MAE}_{ri}$ is the Monte Carlo draw used to simulate region $r$'s epidemic in iterate $i$. Since each region's simulation uses a separate random stream, $\mathrm{MAE}_{ri}$ and $\mathrm{MAE}_{ri'}$ are independent for $r \neq r'$, regardless of any similarity in regions' data-generating processes. (This is specific to the simulation; we would not expect real deaths to be independent across regions.) Therefore $\mathrm{Cov}(\mathrm{MAE}_{ri}, \mathrm{MAE}_{ri'}) = 0$ for all $r \neq r'$, term~(II) vanishes, and
\[
\mathrm{Var}(\bar X) = \frac{1}{n_r^2}\sum_{r=1}^{n_r} \frac{\sigma_r^2}{n_i}.
\]

The formula for standard error in the paper estimates $\text{Var}(\bar{X})$ by plugging in per-region sample variances $\hat\sigma_r^2 = \hat\sigma^2(\{\mathrm{MAE}_{ri}\}_{i=1}^{n_i})$. Note that this allows a distinct $\sigma_r^2$ per region, so heteroscedasticity across regions is not an issue. The reported standard error is therefore unbiased for the full $\mathrm{Var}(\bar X)$.

Another valid approach computes $M_i = \frac{1}{n_r}\sum_r \mathrm{MAE}_{ri}$ and reports $\sqrt{\frac{1}{n_i}\hat\sigma^2(\{M_i\}_{i=1}^{n_i})}$. This targets the same $\mathrm{Var}(\bar X)$ without requiring cross-region independence. To see why, note that $M_1, \ldots, M_{n_i}$ is an i.i.d. sequence, so $\bar X = \frac{1}{n_i}\sum_i M_i$ has variance 
\[
\mathrm{Var}(\bar X) = \frac{1}{n_i}\mathrm{Var}(M_i) = \frac{1}{n_i}\cdot\frac{1}{n_r^2}\sum_{r,r'}\mathrm{Cov}(\mathrm{MAE}_{ri}, \mathrm{MAE}_{ri'}) = \frac{1}{n_r^2}\sum_{r=1}^{n_r} \frac{\sigma_r^2}{n_i}.
\]
The key difference is that $\hat\sigma^2(\{M_i\}_{i=1}^{n_i})$ estimates $\mathrm{Var}(M_i)$ empirically from the $n_i$ observed values of $M_i$, thereby absorbing any cross-region covariance structure automatically rather than assuming it away. However, this robustness is unnecessary given the independence argument in Step 3, which is well-justified here. 

Moreover, this robustness comes at a steep cost in efficiency. 
Under normality, a variance estimator $\hat s^2$ with $\nu$ degrees of freedom satisfies $\nu \hat s^2 / \sigma^2 \sim \chi^2_\nu$, which gives $\mathrm{Var}(\hat s^2) = 2\sigma^4/\nu$: more degrees of freedom means a more precise variance estimate. 
The suggested estimator has only $n_i-1=9$ degrees of freedom.
Our estimator instead sums $n_r$ such terms independently.
If regions had equal variance, the degrees of freedom add to $n_r(n_i-1) = 459$.
This assumption of homoskedasticity is likely not justified, but the degrees of freedom are always greater than 9, and may be far more.

\section{Further experimental results}\label[appendix]{apx:extra}

\cref{tab:retro-oracle} displays the MAE results for retrospective
estimation, now under oracle tuning, where we select the best hyperparameters
for each method to minimize the average MAE (over all regions and 
replications). The results of deconvolution relative to the convolutional ratio
are even more favorable, compared to those from cross-validation tuning in 
\cref{tab:retro-cv}.    

\begin{table}[h]
\centering
\caption[MAE and relative performance of retrospective HFR estimation, under oracle tuning.]{MAE of methods in retrospective HFR estimation, and the
  associated percentage improvement on the convolutional ratio (CR) and lagged
  ratio (LR), under oracle tuning. Compare to  \cref{tab:retro-cv}.} 
\label{tab:retro-oracle}
\begin{tabular}[t]{@{}lrrrrrr@{}}
\toprule
 & Lagged Ratio & Conv Ratio & \textbf{Deconv-0} & \textbf{Deconv-1} & \textbf{Deconv-2} & \textbf{Deconv-Tuned}\\
\midrule
MAE $\times \;10^3$ & 11.4 ± 0.1 & 8.3 ± 0.1 & 6.0 ± 0.1 & 6.4 ± 0.1 & 6.8 ± 0.1 & 6.4 ± 0.1\\
Improv over CR (\%) & -47.3 ± 0.8 & 0.0 ± 0.0 & 22.3 ± 0.6 & 19.9 ± 0.7 & 15.4 ± 0.8 & 18.0 ± 0.7\\
Improv over LR (\%) & 0.0 ± 0.0 & 28.1 ± 0.3 & 45.0 ± 0.5 & 42.6 ± 0.6 & 39.2 ± 0.7 & 41.6 ± 0.6\\
\bottomrule
\end{tabular}
\end{table}

Figures \ref{fig:all_curves_retro} and \ref{fig:all_curves_rt} show the
retrospective and real-time HFR estimates, respectively, for all 50 states. As
in Figures \ref{fig:retro-curves} and \ref{fig:rt-curves} in the main text,
deconvolution and the convolutional ratio use the oracle delay distribution, the
lagged ratio uses the oracle lag, and the hyperparameters all are tuned via
cross-validation. For the sake of keeping the visualization simple, we show only 
constant order ($m=0$) trend filtering regularization in the retrospective case,
and quadratic order ($m=2$) in the real-time case. Overall, we see that the
deconvolution estimators are more stable than the ratio estimators; they capture
changes in HFR quickly, and for the most part, do not exhibit the same spurious
oscillations.     

\begin{figure}[p]
\centering
\includegraphics[height=0.98\textheight]{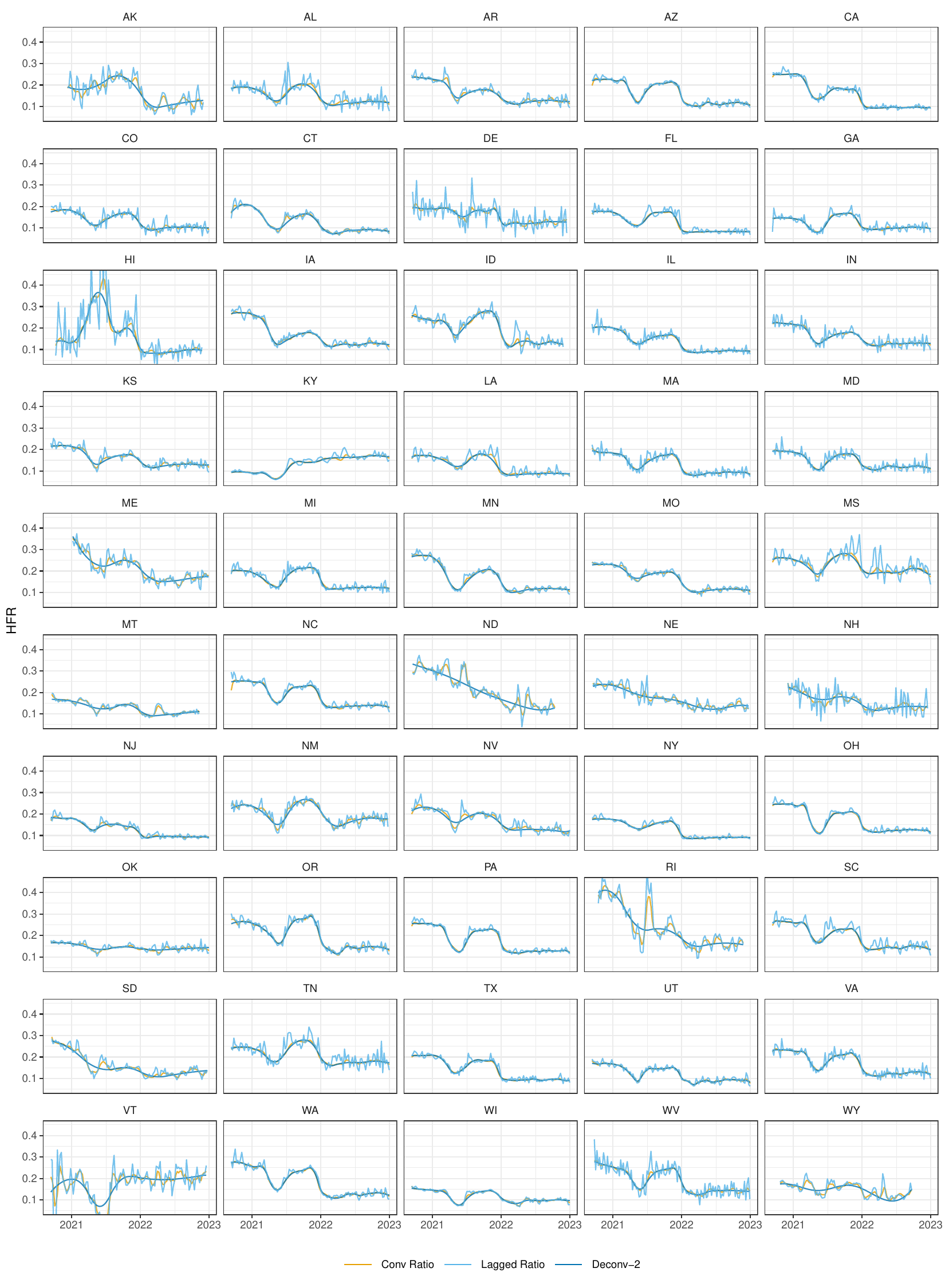}
\caption{Retrospective HFR estimates for all 50 states.}
\label{fig:all_curves_retro}
\end{figure}

\begin{figure}[p]
\centering
\includegraphics[height=0.98\textheight]{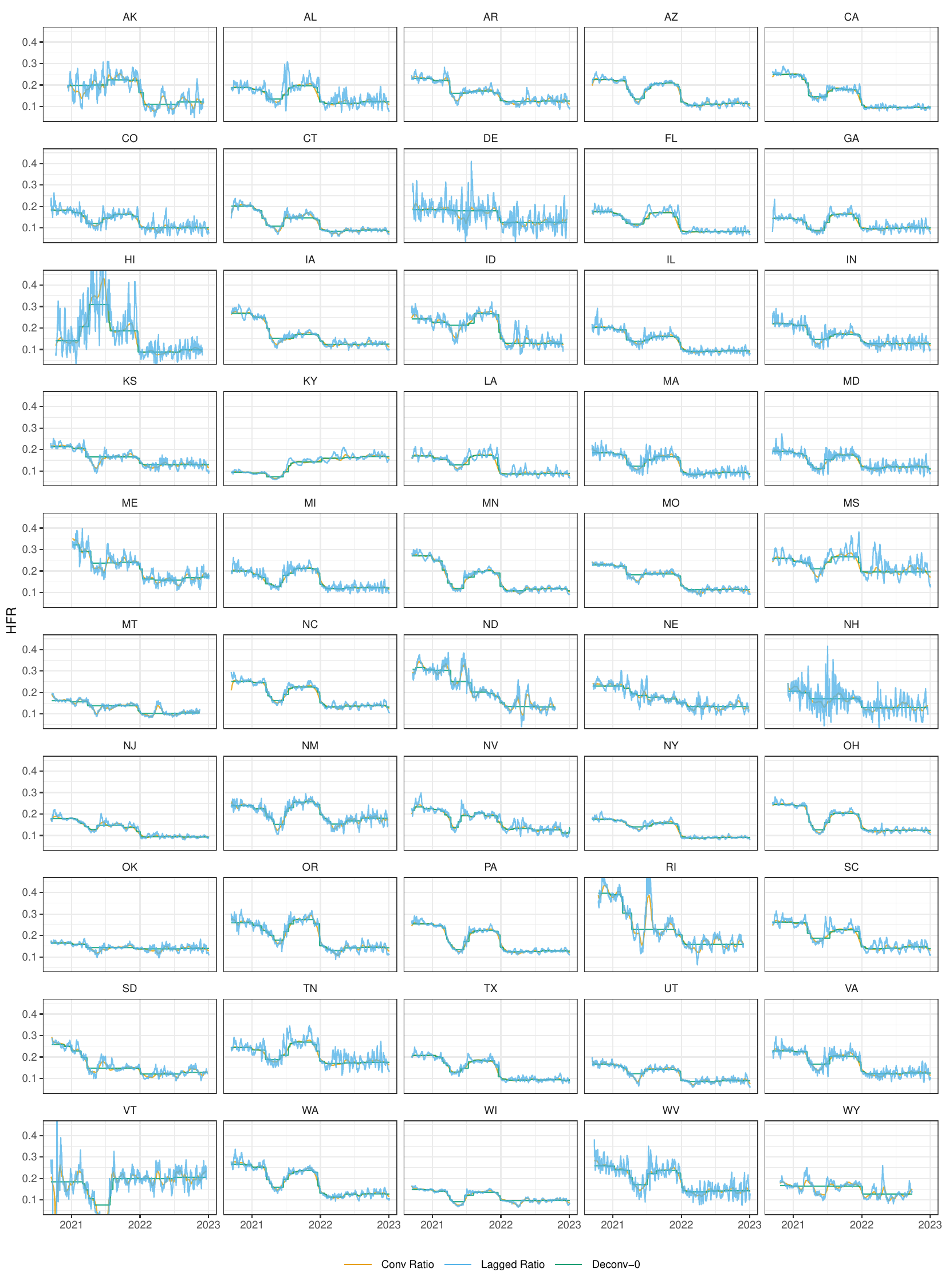}
\caption{Real-time HFR estimates for all 50 states.}
\label{fig:all_curves_rt}
\end{figure}

\cref{fig:misp-imps} displays the performance relative to the
convolutional ratio for misspecified delay distributions. This is analogous to
\cref{fig:misp-maes} in the main text. The results show that the
deconvolution approach maintains an advantage over the convolutional ratio for
varying degrees of misspecification, with the exception of large positive mean
offsets in the retrospective case, where this advantage is somewhat erased.  

\begin{figure}[t]
\centering
\includegraphics[width=0.95\textwidth]{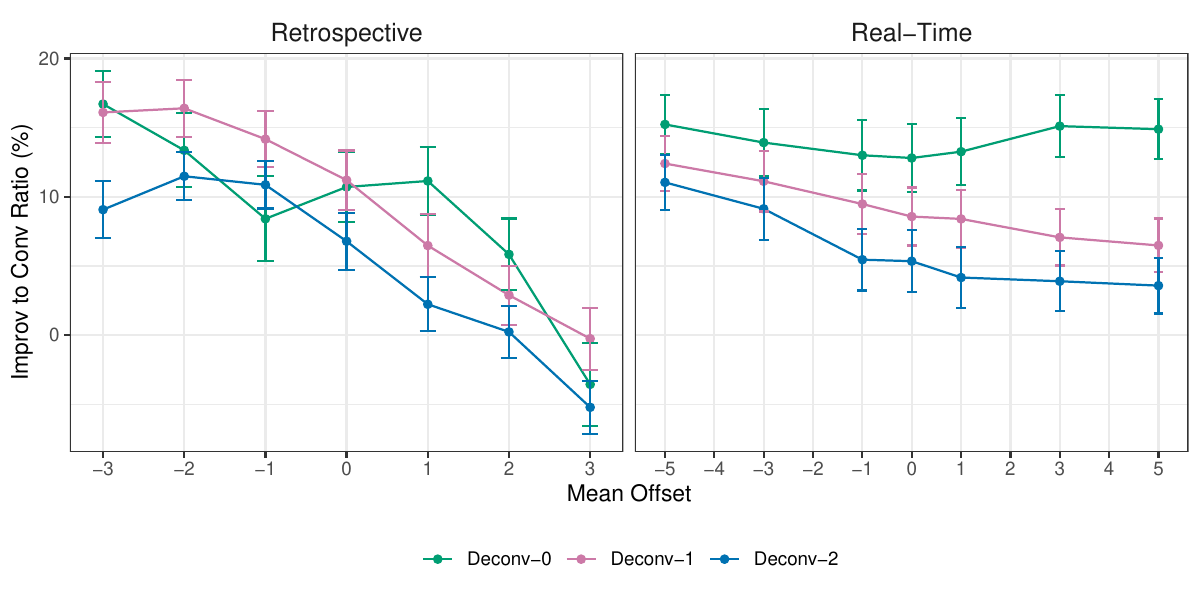}
\caption{Percentage improvement of deconvolution to the convolutional ratio as
  a function of mean offset in the misspecified case. Compare to
  \cref{fig:misp-maes}.} 
\label{fig:misp-imps}
\end{figure}

\begin{figure}[!b]
    \centering
    \includegraphics[width=0.95\textwidth]{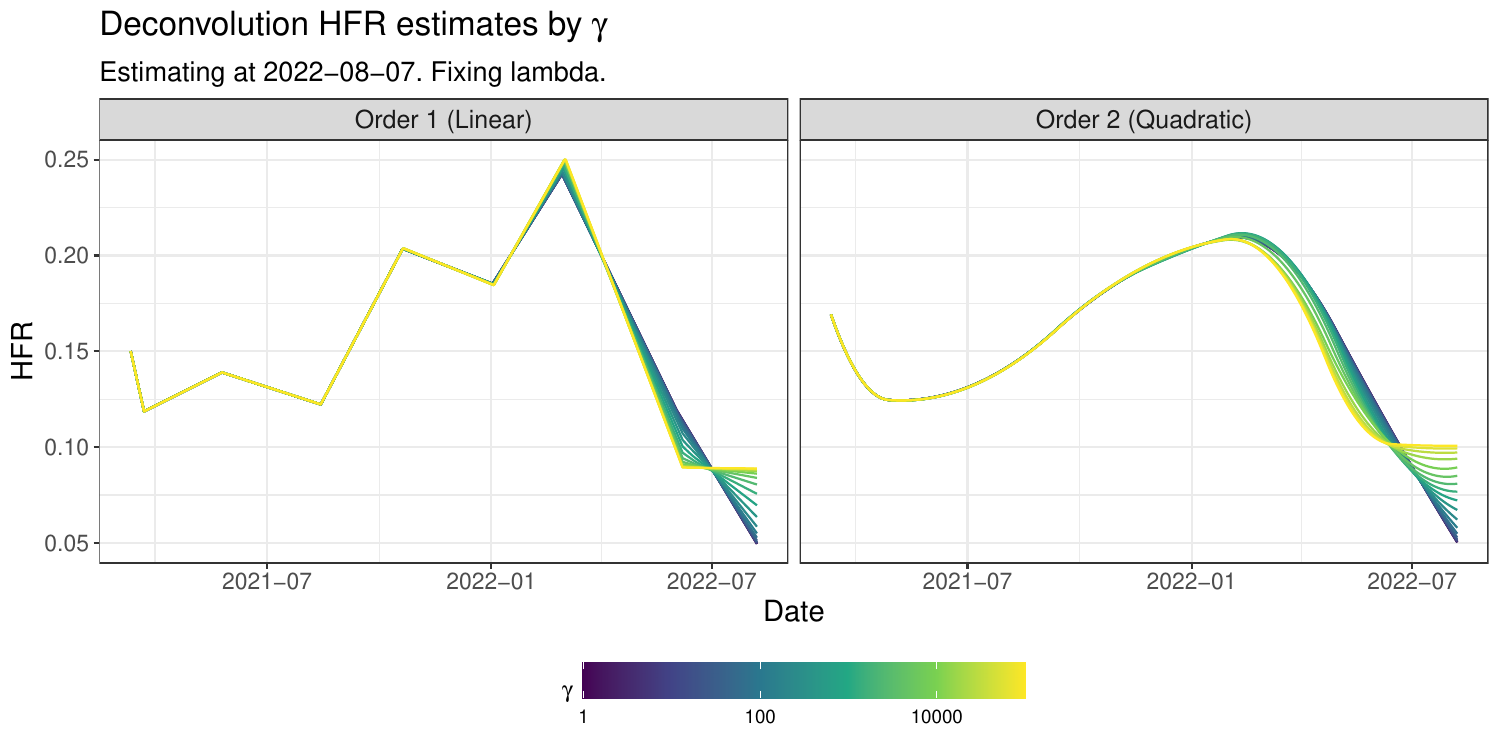}
    \caption{Deconvolution HFRs at varying levels of $\gamma$, the tail smoothness hyperparameter.}
    \label{fig:gammas}
\end{figure}

Real-time deconvolution relies heavily on the hyperparameter $\gamma$ to stabilize predictions at the tail, at least for orders above piecewise constant. \cref{fig:gammas} displays linear and quadratic deconvolution at fixed $\lambda$ and 20 values of $\gamma$. Piecewise-constant deconvolution is not shown because its tail is flat by construction, so $\gamma$ has little effect.
The HFR estimates change significantly as $\gamma$ ranges geometrically from 0.1 to 100,000. 
With close to no tail regularization, the HFRs continue their downward trend to the present date. In contrast, the largest values of $\gamma$ flatten the curve entirely, yielding final HFRs around 10\% -- double those of the minimum $\gamma$. The gap is slightly more pronounced for quadratic trend filtering, owing to its capacity for more flexible fits. 

In general, the sensitivity of deconvolution's predictions to different values of $\gamma$ highlights the importance of tuning it properly, whether by forward validation or heuristic evaluation. 
It also accounts for why on average, constant-order deconvolution outperformed quadratic deconvolution, and to a lesser extent linear. 
Future work on tuning strategies beyond the min and 1se rule could further improve these orders' predictions. 

\begin{figure}
    \centering
    \includegraphics[width=0.8\linewidth]{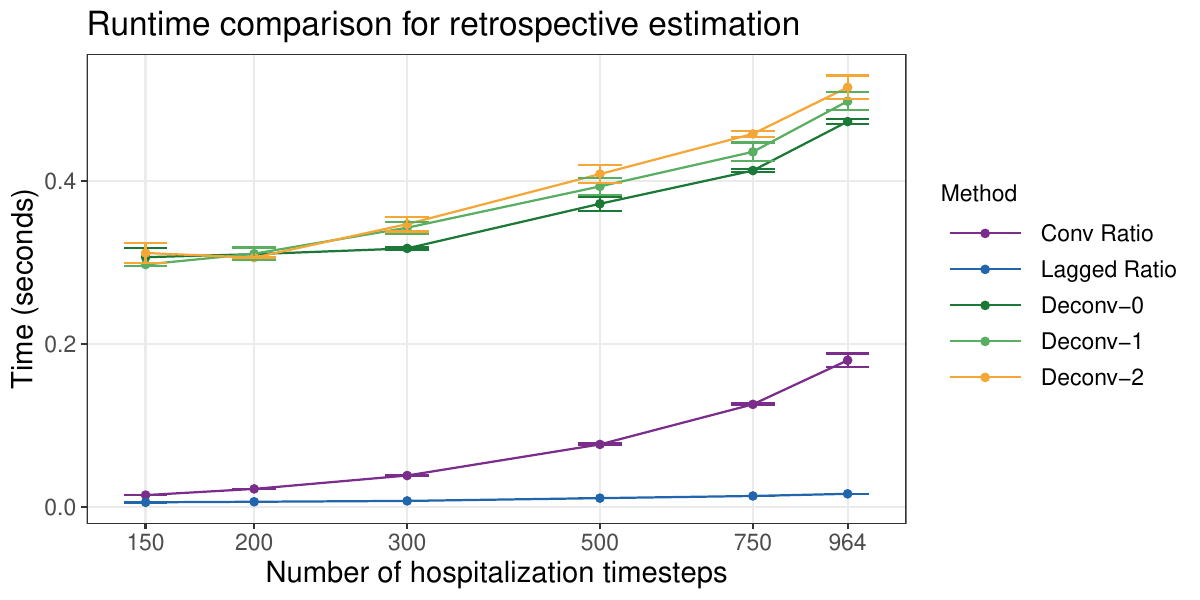}
    \caption[Wall-clock runtime of retrospective severity estimators.]{Wall-clock runtime of retrospective severity estimators. Error bars reflect $\pm 1$ standard error across 20 iterations on 4 regions.}
    \label{fig:runtime}
\end{figure}

We conducted a runtime experiment to assess the computational cost of these methods. All five estimators were timed across six values of the number of timesteps $n \in \{150,200,300,500,750,964\}$, where 
$n=964$ corresponds to the full semi-synthetic dataset described in \cref{sec:setup}. 
Each retrospective estimator was run 20 times with a single, fixed value of its hyperparameter, over 4 geographic regions (PA, UT, KS, and the \US). 
\cref{fig:runtime} displays the results; error bars reflect $\pm 1$ standard error computed across all regions for each $(n, \text{method})$ pair.

The ratio benchmarks are quite fast to run. The lagged ratio is particularly fast: its runtime remains close to zero even as $n$ scales. The convolutional ratio's runtime increases modestly with $n$, but always stays below 0.2 seconds. The deconvolution estimators exhibit nearly identical runtimes at the three orders evaluated. They also scale somewhat with sequence length, though still well under one second per evaluation at the largest $n$ considered. 

To roughly adjust these timings to account for tuning, each of the above is  multiplied by the number of folds times the number of hyperparameter values in the tuning grid. 
Extrapolating from the figure, tuning the estimators over the full data with $K=5$-fold cross-validation would take the lagged ratio approximately 1 second, the convolutional ratio 5 seconds, and deconvolution a minute. 
In total, a complete tuning run for a single estimation date finishes in a matter of seconds or minutes, making daily retuning entirely feasible as new observations arrive. That said, if epidemiological conditions appear stable, practitioners may reasonably reuse previously selected parameters or search over a coarser grid to reduce overhead further.

\section{Backward-looking severity rates}\label[appendix]{apx:back}

We provide two connections between the real-time convolutional ratio and the
backward-looking severity rate. Recall the backward-looking severity rate is
defined as
\begin{equation}
\tilde{p}_t = \sum_{k=0}^d \Pprob(\text{secondary event occurs at $t$} \given 
\text{primary event at $t-k$}) = \sum_{k=0}^d \pi_k^{(t-k)} p_{t-k},
\end{equation}
and the original forward-looking severity rate is
\begin{equation}
p_t  = \sum_{k=0}^d \Pprob(\text{secondary event occurs at $t+k$} \given
\text{primary event at $t$}) = \sum_{k=0}^d \pi_k^{(t)} p_t.
\end{equation}
For the sake of simplicity, assume throughout that the delay distributions
before $t$, occurring in any of the summands above, are equal to a constant
distribution $\pi$. In the real-time case, using the oracle delay $\pi$, the 
convolutional ratio (using raw, not smoothed counts) has conditional
expectation:    
\begin{equation}
\E[\hat{p}_t \given x_{\leq t}] = \frac{\sum_{k=0}^d x_{t-k} \pi_k p_{t-k}} 
{\sum_{k=0}^d x_{t-k} \pi_k} = \sum_{k=0}^d \frac{x_{t-k}}{\sum_{j=0}^d x_{t-j} 
  \pi_j} \pi_k p_{t-k}.  
\end{equation}
First, assume that primary counts are equal to a constant $x_0$ in the $d$ time 
points preceding $t$. Note that the convolutional ratio reduces to $y_t/x_0$, 
which is unbiased for the backwards rate:  
\begin{equation}
\E[\hat{p}_t \given x_{\leq t}] = \sum_{k=0}^d \frac{x_0}{\sum_{j=0}^d x_0
  \pi_j} \pi_k p_{t-k} = \sum_{j=0}^d \pi_k p_{t-k} = \tilde{p}_t. 
\end{equation}
Instead, assume that severity rates are equal to a constant $p_0$ in the $d$
time points preceding $t$. As the delay distributions are constant, this implies   
\begin{multline}
\Pprob(\text{secondary event occurs at $t$} \given \text{primary event at $t-k$}) \\
= \Pprob(\text{secondary event occurs at $t+k$} \given \text{primary event at $t$}), 
\end{multline}
for each $k = 0,\dots,d$. This means that the backward- and forward-looking
rates are equivalent, \smash{$\tilde{p}_t = p_t = p_0$}. Furthermore, the
convolutional ratio is unbiased for this common value:  
\begin{equation}
\E[\hat{p}_t \given x_{\leq t}] = \frac{\sum_{k=0}^d x_{t-k} \pi_k p_{t-k}} 
{\sum_{k=0}^d x_{t-k} \pi_k} = p_0 \frac{\sum_{k=0}^d x_{t-k} \pi_k}
{\sum_{k=0}^d x_{t-k} \pi_k} = p_0.
\end{equation}

\chapter{\texorpdfstring{Supplementary material for \cref{ch:paper3}}{Supplementary material for Chapter 3}}
\label{app:C}

\section{Proof of mean generation time for SEIR}\label[appendix]{apx:mean-generation}

We prove the mean generation time stated in \cref{prop:compartmental-generation}.
\begin{proposition}[Restated]
The mean generation time for an SEIR model is
    \[
        \mathbb{E}[G] = \mu^\text{EI} + \frac{\mu^\text{IR}}{2} - \frac{1}{2} + \frac{(\sigma^\text{IR})^2}{2\, \mu^\text{IR}}.
    \]
\end{proposition}

\begin{proof}

We now derive the mean generation time $G$. By definition, 
\[
\mathbb{E}[G] = \sum_{k \geq 1} k\, g_k = \frac{1}{\mu^\text{IR}} \sum_{k \geq 1} k\, \zeta^\text{EI}_k.
\]
Substituting $\zeta^\text{EI}_k = \sum_{j=1}^{k} \pi^\text{EI}_j \cdot \Pprob(W^\text{IR} > k-j)$ and swapping the order of summation,
\[
\sum_{k \geq 1} k\, \zeta^\text{EI}_k
= \sum_{k \geq 1} \sum_{j=1}^{k} k\, \pi^\text{EI}_j\, \Pprob(W^\text{IR} > k-j)
= \sum_{j \geq 1} \pi^\text{EI}_j \sum_{k \geq j} k\, \Pprob(W^\text{IR} > k-j).
\]
Reindex the inner sum with $a = k - j$, so that $a = 0, 1, 2, \dots$ and $k = j + a$:
\[
\sum_{j \geq 1} \pi^\text{EI}_j \sum_{a \geq 0} (j + a)\, \Pprob(W^\text{IR} > a)
= \sum_{j \geq 1} \pi^\text{EI}_j \left[ j \sum_{a \geq 0} \Pprob(W^\text{IR} > a) + \sum_{a \geq 0} a\, \Pprob(W^\text{IR} > a) \right].
\]
The first inner sum is $\mu^\text{IR}$, since the survival function of a non-negative integer random variable sums to its mean. The bracketed expression is therefore $j\, \mu^\text{IR} + S$, where $S := \sum_{a \geq 0} a\, \Pprob(W^\text{IR} > a)$ does not depend on $j$. Pulling $S$ out of the outer sum and using $\sum_j j\, \pi^\text{EI}_j = \mu^\text{EI}$,
\[
\sum_{k \geq 1} k\, \zeta^\text{EI}_k = \mu^\text{EI}\, \mu^\text{IR} + S.
\]

It remains to evaluate $S$. Writing the survival function as a sum over $\pi^\text{IR}$ and again swapping the order of summation,
\[
S = \sum_{a \geq 0} a \sum_{r > a} \pi^\text{IR}_r
= \sum_{r \geq 1} \pi^\text{IR}_r \sum_{a = 0}^{r-1} a.
\]
The inner sum is the standard arithmetic series $\sum_{a=0}^{r-1} a = \tfrac{r(r-1)}{2} = \tfrac{r^2 - r}{2}$. Therefore
\begin{align*}
    S &= \sum_{r \geq 1} \pi^\text{IR}_r \cdot \frac{r^2 - r}{2}
= \frac{1}{2}\left( \sum_r r^2\, \pi^\text{IR}_r - \sum_r r\, \pi^\text{IR}_r \right)
= \frac{1}{2}\left( \mathbb{E}[(W^\text{IR})^2] - \mu^\text{IR} \right)\\
&=\frac{1}{2}\left((\mu^\text{IR})^2 + (\sigma^\text{IR})^2 - \mu^\text{IR}\right).
\end{align*}
Combining the two results and dividing by $\mu^\text{IR}$,
\[
\mathbb{E}[G] = \frac{\mu^\text{EI}\, \mu^\text{IR} + S}{\mu^\text{IR}} = \mu^\text{EI} + \frac{\mu^\text{IR}}{2}  + \frac{(\sigma^\text{IR})^2}{2\, \mu^\text{IR}} - \frac{1}{2}.
\]
\end{proof}

Notably, the latent-period SD $\sigma^\text{EI}$ does not enter the expression for $\mathbb{E}[G]$. The $-\tfrac{1}{2}$ reflects that $\zeta^\text{EI}$ assigns mass at the day of the E$\to$I transition itself ($\Pprob(W^\text{IR} > 0) = 1$), so an individual is permitted to transmit on the very day they enter the infectious compartment. It vanishes in the continuous-time analogue.

\section{Case reproduction numbers}\label[appendix]{apx:case_rt}
 
Similar arguments deepen our understanding of how case $R_t$ relates to the other definitions. 
\citet{fraser2007} relates $R_t^I$ and $R_t^c$ through the generation interval distribution by substituting $w(t,k)=R_t^Ig_k$ into the definition of case $R_t$:
\begin{equation}\label{eq:case-inst-fraser}
    R_t^C = \sum_{k> 0} R_{t+k}^I g_k.
\end{equation}
This further reinforces the idea that instantaneous $R_t$ is the average number of secondary transmissions from a primary infection at $t$ if conditions remain unchanged. 
We can derive this same expression under compartmental assumptions, using the tools introduced in \cref{prop:compartmental-generation}.

By definition, $R_t^C = \sum_{k > 0} w(t+k,k)$. Expanding $w(t+k,k)$ as in the proof of Proposition~\ref{prop:equivalence} and recalling the definition of $\zeta_k^\text{EI}$,
\begin{align*}
R_t^C
  &=  \sum_{k> 0} S_{t+k}\cdot {\mathbb{P}\!\left(\text{infectious at } t+k \mid \text{new infection at } t\right)} \\
    &\hspace{70pt}\cdot {\mathbb{P}(i \text{ contacts } j \text{ at } t+k)}\cdot\mathbb{P}
    (\text{transmission} \mid S\text{-}I \text{ contact})\\
&= \sum_{k > 0} S_{t+k} \cdot \frac{\beta_{t+k}}{N} \cdot \zeta_k^\text{EI} ,
\end{align*}
\cref{lem:homo_to_beta} established that instantaneous $R_t$ and mechanistic $R_t$ both equal $\beta_t S_t \mu^\text{IR}/N$.
Plugging that in,
$$R_t^C = \sum_{k > 0} R_{t+k}^I \cdot \frac{\zeta_k^\text{EI}}{\mu^{\mathrm{IR}}}.$$
\cref{prop:compartmental-generation} established that in compartmental equations, \mbox{$g_k = \frac{\zeta_k^\text{EI}}{\mu^{\mathrm{IR}}}$}.
This completes the proof of \eqref{eq:case-inst-fraser}.

The method by \citet{wallinga_teunis} is the standard approach to estimate case $R_t$ from aggregate data. 
Given individual-level data, one could compute $R_t^C$ by simply identifying the cohort of people infected at $t$ and tracking their secondary transmissions. 
\chapter{\texorpdfstring{Supplementary material for \cref{ch4}}{Supplementary material for Chapter 4}}
\label{app:D}

\section{Optimization}\label[appendix]{apx:optim}

The objective in \eqref{eq:glm-retro} is a penalized Poisson log-likelihood with identity link $\Lambda_t(\theta)=Z_t^\top \theta$ \eqref{eq:linear-mean-cases}. This is concave in $\theta$ wherever the fitted means satisfy $\Lambda_t(\theta) > 0$.                                 
(For now, we ignore the multiplicative terms $\omega$ and $\rho$.)

Iteratively reweighted least squares (IRLS) is the standard way to optimize a GLM: each Newton step on the log-likelihood reduces algebraically to a weighted least squares problem, hence the name.
IRLS involves the score and Fisher information, which we define next.

\subsection{Score, Hessian, and Fisher information.}                                                                                                                

Consider an unregularized, identity-link Poisson GLM. Its log-likelihood contribution at time $t$ is
\[
\ell_t(\theta) \;=\; y_t\,\log \Lambda_t(\theta) \,-\, \Lambda_t(\theta), \qquad \Lambda_t(\theta) \;=\; Z_t^\top\theta,
\]
with $Z_t$ the row of the convolutional design matrix in \eqref{eq:linear-mean-cases}. Differentiating yields the individual score function and its gradient,
\[
s_t(\theta) \;=\; \frac{\partial \ell_t}{\partial \theta} \;=\; \frac{y_t - \Lambda_t}{\Lambda_t}\,Z_t, \qquad
\frac{\partial s_t}{\partial \theta^\top} \;=\; -\,\frac{y_t}{\Lambda_t^2}\,Z_t Z_t^\top.
\]
We stack timesteps into a design matrix $Z \in \mathbb{R}^{T\times p}$ with rows $Z_t^\top$, and note that $\mathbb{E}[y_t] = \Lambda_t$ under Assumption 1. Summing the scores and using $\sum_t a_t Z_t = Z^\top a$ for any vector $a = (a_1,\dots,a_T)^\top$,
\[
s(\theta) \;=\; \sum_t \frac{y_t - \Lambda_t}{\Lambda_t}\,Z_t \;=\; Z^\top W(Y - \Lambda), \qquad W \;=\; \mathrm{diag}(1/\Lambda_t).
\]
Likewise, using $\sum_t w_t Z_t Z_t^\top = Z^\top \mathrm{diag}(w)\, Z$, we define the Fisher information as
\[
\mathcal{I}(\theta) \;=\; -\,\mathbb{E}\!\left[\sum_t \frac{\partial s_t(\theta)}{\partial \theta^\top}\right] \;=\; \sum_t \frac{Z_t Z_t^\top}{\Lambda_t} \;=\; Z^\top W Z.
\]

\subsection{IRLS}

Our method uses a ridge penalty $\lambda\,\theta^\top\Omega\,\theta$ to regularize for smoothness; in the real-time setting there is a second term $\gamma\,\theta^\top\Psi\,\theta$. Differentiating, the gradients are $2\lambda\Omega\theta$ and $2\gamma\Psi\theta$ and the Hessians are $2\lambda\Omega$ and $2\gamma\Psi$. Since $\lambda$ and $\gamma$ are selected by cross-validation, the factor of two is immaterial and we absorb it into the tuning parameters henceforth, writing the gradients as $\lambda\Omega\theta$, $\gamma\Psi\theta$ and the Hessians as $\lambda\Omega$, $\gamma\Psi$. These enter the gradient and Fisher information, respectively. Our notation adheres to the retrospective case in the following derivations (just $\lambda$), though $\gamma$ can simply be added for real-time.

The penalized objective \eqref{eq:glm-retro} adds the ridge term to the negative summed log-likelihood,
\[
f_\lambda(\theta) \;=\; -\sum_t \ell_t(\theta) \,+\, \tfrac{1}{2}\lambda\,\theta^\top\Omega\,\theta.
\]
This has penalized gradient $g_\lambda(\theta) = -Z^\top W(Y - \Lambda) + \lambda\Omega\,\theta$ and Fisher information $\mathcal{I}_\lambda(\theta) = Z^\top W Z + \lambda\Omega$, the penalty contributing its Hessian $\lambda\Omega$ to the latter. 
In general, Fisher scoring updates an iterate by subtracting the information-preconditioned gradient,
\[
\theta^{(k+1)} \;=\; \theta^{(k)} \,-\, \mathcal{I}_\lambda(\theta^{(k)})^{-1}\,g_\lambda(\theta^{(k)}),
\]
which replaces the Hessian in a Newton step with its expectation $\mathcal{I}_\lambda$, guaranteeing a positive-definite, descent-directed system. Substituting our gradient and information gives
\[
\theta^{(k+1)} \;=\; \theta^{(k)} \,-\, \bigl(Z^\top W^{(k)} Z + \lambda\Omega\bigr)^{-1}\!\bigl[-Z^\top W^{(k)}(Y - \Lambda^{(k)}) + \lambda\Omega\,\theta^{(k)}\bigr],
\]
with working weights $W^{(k)} = \mathrm{diag}\bigl(1/\Lambda_t(\theta^{(k)})\bigr)$ and $\Lambda^{(k)} = Z\theta^{(k)}$.
Because $f_\lambda$ is convex (with positive-definite information $\mathcal{I}_\lambda$ once $\lambda\Omega \succ 0$), any stationary point is the global minimizer, so Fisher scoring converges to the unique penalized MLE regardless of initialization.

To simplify, write $\theta^{(k)} = (Z^\top W^{(k)} Z + \lambda\Omega)^{-1}(Z^\top W^{(k)} Z + \lambda\Omega)\theta^{(k)}$ and fold it into the bracket. The $\lambda\Omega\,\theta^{(k)}$ terms cancel, leaving
\[
\theta^{(k+1)} \;=\; \bigl(Z^\top W^{(k)} Z + \lambda\Omega\bigr)^{-1} Z^\top W^{(k)}\bigl[Z\theta^{(k)} + (Y - \Lambda^{(k)})\bigr].
\]
The bracketed term is the working response; under the identity link $\Lambda^{(k)} = Z\theta^{(k)}$, so it collapses to $Y$ and we obtain the penalized WLS normal equations
\[
\bigl(Z^\top W^{(k)} Z + \lambda\Omega\bigr)\,\theta^{(k+1)} \;=\; Z^\top W^{(k)} Y,
\]
where $Y$ appears directly rather than through a canonical working response, because the link is identity. Each iterate is thus a ridge-penalized weighted least squares fit of $Y$ on $Z$ with Poisson variance weights, reweighted as $W^{(k)}$ updates. We iterate until the relative change in $\theta$ falls below a tolerance; convergence is fast in practice, typically within 10--20 iterations from a flat initialization.

One practical wrinkle is that unlike the log link, the identity link does not automatically keep $\Lambda_t$ positive. A full Newton step can overshoot into the infeasible region where the working weights blow up. We guard against this with step-halving: if the proposed update yields any non-positive $\Lambda_t$, we halve the step length and retry, up to a small number of backtracks.

The non-negativity constraint $S\theta \succeq 0$ on $R_t$ itself is not enforced inside the solver; as argued in \cref{apx:rt-uncertainty}, it is asymptotically inactive whenever the true $R_t$ is bounded away from zero, which holds in every setting we consider.

\paragraph{Ascertainment rates}

Of the two multiplicative terms deferred above, the ascertainment rate
$\rho_t \in (0,1]$ is the simpler: it is fixed and known (e.g.\ from external
calibration), not estimated. Entering the mean multiplicatively,
$\mu_t = \rho_t\,\Lambda_t(\theta)$, it acts as a known Poisson offset that
rescales each expected count without adding free parameters. The derivations
above therefore carry through verbatim with $\Lambda_t$ replaced by
$\rho_t \Lambda_t$: the working weights become
$W = \mathrm{diag}(1/(\rho_t\Lambda_t))$ and the design rows are scaled to
$\rho_t Z_t$, leaving the IRLS recursion and its convexity unchanged. The
day-of-week multipliers $\omega$, by contrast, are estimated, and we treat them
next.

\subsection{Day-of-week effects}

When day-of-week (DoW) multipliers $\omega_{(\text{$t$ mod 7})}$ enter the mean via $\mu_t = \omega_{(\text{$t$ mod 7})}\,\Lambda_t(\theta)$ as in \eqref{eq:quasi-poisson-model}, the model is invariant to the rescaling $(\omega, \theta) \leftrightarrow (c\omega, \theta/c)$ for any $c > 0$. 
Multiplying every $\omega_j$ by $c$ and dividing $\theta$ by $c$ leaves $\mu_t$ — and hence the likelihood — unchanged.      
This is problematic, as the MLE is non-unique along this one-dimensional ray. To ensure uniqueness,
we impose the geometric-mean constraint $\prod_j \omega_j = 1$, the natural normalization for a multiplicative factor,                                                                                                                                           
Estimation proceeds by block coordinate ascent on the log-likelihood. Holding $\theta$ fixed, partition the Poisson log-likelihood by day-of-week: with $T_j = \{t : \text{$t$ mod 7} = j\}$,                  
  \[                                                                                                                                                                                                             
  \ell(\theta, \omega) \;=\; \sum_j \sum_{t \in T_j}\bigl[y_t\log(\omega_j \Lambda_t) - \omega_j \Lambda_t\bigr].
  \]                                                                                                                                                     
Differentiating reveals that the score for $\omega_j$ is $\sum_{t \in T_j}[y_t/\omega_j - \Lambda_t]$. 
This is maximized in closed form by the day's ratio of observed to expected counts:
  \begin{equation}\label{eq:dow-update}
  \hat\omega_j \;=\; \frac{\sum_{t \in T_j} y_t}{\sum_{t \in T_j} \Lambda_t(\theta)},
  \end{equation}                                                                                                                                                                                                 
To enforce $\prod_j \hat\omega_j = 1$, we divide each $\hat\omega_j$ by $\bigl(\prod_j
  \hat\omega_j\bigr)^{1/7}$. We then take one IRLS step on $\theta$ using the DoW-scaled design $\widetilde Z_t = \hat\omega_{(\text{$t$ mod 7})} Z_t$, and the two blocks alternate until the
  log-likelihood converges. Because each block is convex given the other, the iterates converge to the joint MLE---equivalently, to the maximizer of the profile likelihood with $\omega$ profiled out.

\subsection{Tail constraint}\label{apx:optim-constraints}

Our main experiments parameterize $R_t$ as a natural cubic spline. Linearity past the last knot is structural to the basis: every coefficient vector $\theta$ already yields a linear tail, so no constraint is enforced. This is the default both retrospectively and in real-time.

The most useful alternative is a \emph{constant} tail. Higher-order tail behavior (quadratic, cubic) lets the spline keep wiggling past the last knot and tends to overfit, so we focus on the constant case.

  A constant-tail constraint requires a regression-spline basis---specifically,
  a clamped cubic B-spline---because, unlike the natural basis, it imposes no
  boundary behavior of its own. The constant tail is enforced as a linear equality
  $A\theta = 0$, where $A$ collects the first three derivatives of the spline basis
  functions evaluated at a collocation point $t_c$ in the tail segment. Writing
  $S_1, \ldots, S_p$ for the components of the basis row $S(t)$, the rows of $A$ are
  \[
  A \;:=\; \begin{bmatrix}
  S_1'(t_c)   & S_2'(t_c)   & \cdots & S_p'(t_c)   \\
  S_1''(t_c)  & S_2''(t_c)  & \cdots & S_p''(t_c)  \\
  S_1'''(t_c) & S_2'''(t_c) & \cdots & S_p'''(t_c)
  \end{bmatrix}.
  \]
  Because B-splines are local---each basis function is supported on only a few
  knot spans---most columns are zero; only the basis functions whose support
  contains $t_c$ contribute nonzero entries. Since $R_t$ is a cubic polynomial on
  the tail segment, the constraint $A\theta = 0$ sets $R_t' = R_t'' = R_t''' = 0$
  there, which forces every coefficient above the constant term to vanish and
  hence flattens the entire segment.

To enforce both the penalized objective and the constraint $A\theta = 0$ exactly, we solve the constrained optimization problem via its KKT conditions. We introduce Lagrange multipliers $\nu \in \mathbb{R}^3$ (one per derivative constraint) and form the Lagrangian:
\[
L(\theta, \nu) = f_\lambda(\theta) + \nu^\top(A\theta).
\]

At optimality, the gradients with respect to both $\theta$ and $\nu$ must vanish:
\[
\frac{\partial L}{\partial \theta} = \nabla_\theta f_\lambda(\theta) + A^\top\nu = 0, \qquad
\frac{\partial L}{\partial \nu} = A\theta = 0.
\]
The first condition says the gradient of the objective is exactly canceled by the constraint forces $A^\top\nu$. Geometrically, you cannot move in any feasible direction (staying on the constraint surface) without increasing the objective. The second condition simply enforces feasibility.

Recall the gradient of the penalized objective is $\nabla_\theta f_\lambda(\theta) = -Z^\top W(Y - Z\theta) + \lambda\Omega\theta$. Substituting into the first KKT condition and rearranging:
\[
-Z^\top W(Y - Z\theta) + \lambda\Omega\theta + A^\top\nu = 0
\]
gives
\[
(Z^\top W Z + \lambda\Omega)\theta + A^\top\nu = Z^\top W Y.
\]
Combined with the feasibility condition $A\theta = 0$, we obtain the augmented linear system:
\[
\begin{bmatrix} Z^\top W Z + \lambda \, \Omega & A^\top \\ A & 0 \end{bmatrix}
\begin{bmatrix} \theta \\ \nu \end{bmatrix}
=
\begin{bmatrix} Z^\top W Y \\ 0 \end{bmatrix}.
\]
This system holds the constraint at machine precision throughout the IRLS iteration. The multipliers $\nu$ encode the strength and direction of constraint forces pulling $\theta$ into the feasible region; they are discarded after solving. Note that only basis functions near the collocation point $t_c$ feel significant constraint forces, since far-away basis functions have negligible derivatives there anyway.

This augmented system maintains exact feasibility ($A\theta = 0$ to machine precision) throughout the IRLS iteration. The multipliers $\nu$ are discarded after solving.

When a tapered penalty is also present, the augmented system can be ill-conditioned. Instead, we solve in the constraint's null space. Every feasible $\theta$ lies in $\mathrm{null}(A)$, so let the columns of $N \in \mathbb{R}^{p \times (p - \mathrm{rank}(A))}$ form an orthonormal basis for that null space:
\[
\theta = N\alpha, \quad \alpha \in \mathbb{R}^{p - \mathrm{rank}(A)}.
\]
The constraint is now self-enforcing: the problem becomes \emph{unconstrained} in $\alpha$, with effective design $ZN$ and effective penalty $N^\top \Omega N$. To obtain $N$, take the QR decomposition $A^\top = QR$ with $Q = [Q_1 \mid Q_2]$: the first $r = \mathrm{rank}(A)$ columns span the column space of $A^\top$, and the remaining $p - r$ columns span $\mathrm{null}(A)$. Set $N = Q_2$.

We do not recommend combining a constant-tail constraint with a natural-spline basis. The natural basis is already shape-constrained at the boundary, and stacking a hard equality on top was numerically fragile in our experiments. The B-spline route is cleaner.

\section{Inference in Poisson GLM}\label[appendix]{apx:rt-uncertainty}

Inference for $R_t$ rests on a Gaussian approximation to the sampling distribution of $\hat\theta$. Because $R_t(\theta) = S(t)^\top\theta$ is linear in $\theta$, once we have an estimate                    
$\widehat{\mathrm{Cov}}(\hat\theta)$, we can transport it to the $R_t$ scale by a simple quadratic form. We derive $\widehat{\mathrm{Cov}}(\hat\theta)$ from the penalized score equations for the             
retrospective estimator \eqref{eq:glm-retro} under three working assumptions: 
\begin{enumerate}
  \item The observed counts $y_t$ are conditionally independent and Poisson with mean $\Lambda_t = Z_t^\top\theta$; 
  \item The delay distributions are known; and
\item  The smoothing parameter $\lambda$ is held fixed.                                                    
\end{enumerate}

We perform inference with a standard approach: the sandwich variance for a penalized M-estimator (or quasi-likelihood). 
Without loss of generality, our notation assumes no day-of-week effects or under-reporting. 

\subsection{Confidence intervals for \texorpdfstring{$R_t$}{Rt}}
\subsubsection*{Sandwich variance of $\hat\theta$}
Stacking timesteps into a design matrix $Z \in \mathbb{R}^{T\times p}$ with rows $Z_t^\top$ and $\Lambda_t(\theta) = Z_t^\top\theta$, the total score $U(\theta) = \sum_t s_t(\theta)$ and Fisher information are
\[
U(\theta) \;=\; Z^\top W (Y - Z\theta), \qquad
\mathcal{I}(\theta) \;=\; Z^\top W Z, \qquad W \;=\; \mathrm{diag}\bigl(1/\Lambda_t(\theta)\bigr).
\]
For independent Poisson observations the information equality holds, so $\mathrm{Var}(U(\theta)) = Z^\top W Z = \mathcal{I}(\theta)$. Crucially, $U$ is the gradient of the \emph{likelihood} alone, so this variance carries no $\lambda$.

The estimator $\hat\theta$ solves the first-order condition of \eqref{eq:glm-retro},
\[
-U(\hat\theta) \,+\, \lambda\Omega\,\hat\theta \;=\; 0.
\]
Expanding to first order about the truth $\theta_0$, using $\nabla U(\theta_0) \approx -\mathcal{I}(\theta_0)$,
\[
0 \;\approx\; -U(\theta_0) \;+\; \lambda\Omega\theta_0 \;+\; \bigl(\mathcal{I}(\theta_0) + \lambda\Omega\bigr)(\hat\theta - \theta_0),
\]
where $\lambda\Omega\theta_0$ is a deterministic shift and $\lambda\Omega$ is the penalty's contribution to the curvature. Let
\[
H \;:=\; \mathcal{I}(\theta_0) + \lambda\Omega \;=\; Z^\top W Z + \lambda\Omega
\]
denote the Hessian of the penalized objective \eqref{eq:glm-retro} at $\theta_0$. Inverting,
\[
\hat\theta - \theta_0 \;\approx\; H^{-1}\bigl[U(\theta_0) \,-\, \lambda\Omega\theta_0\bigr].
\]
The term $-H^{-1}\lambda\Omega\theta_0$ is a deterministic ridge-type shrinkage bias and does not contribute to $\mathrm{Var}(\hat\theta)$. The variance of $\hat\theta$ is therefore driven entirely by the stochastic part $H^{-1}U(\theta_0)$. This gives the sandwich variance
\begin{equation}\label{eq:rt-cov-theta}
\mathrm{Cov}(\hat\theta) \;\approx\; H^{-1}\,\underbrace{(Z^\top W Z)}_{\mathrm{Var}(U(\theta_0))}\,H^{-1},
\end{equation}
where the approximation is due to the first-order Taylor expansion of the score about $\theta_0$.

\subsubsection*{Plug-in estimation}

Both $\mathcal{I}(\theta_0)$ and $W = \mathrm{diag}(1/\Lambda_t(\theta_0))$ depend on the unknown $\theta_0$. We use the standard plug-in: substitute $\hat\theta$ throughout, giving fitted means $\hat\Lambda_t = Z_t^\top\hat\theta$, working weights $\widehat W = \mathrm{diag}(1/\hat\Lambda_t)$, and
\begin{equation}\label{eq:sandwich-cov}
    \widehat H \;=\; Z^\top \widehat W Z + \lambda\Omega, \qquad \widehat{\mathrm{Cov}}(\hat\theta) \;=\; \widehat H^{-1}\,(Z^\top \widehat W Z)\,\widehat H^{-1}.
\end{equation}
Under regularity $\hat\theta$ is consistent for the penalized target it estimates, and by the continuous mapping theorem so is the plug-in covariance.

Equation \eqref{eq:rt-cov-theta} has the classical bread--meat form: bread $H^{-1}$, meat $Z^\top W Z$. The penalty sits in both slices of bread, through $H = Z^\top W Z + \lambda\Omega$, because it changes how responsive $\hat\theta$ is to a perturbation of the data. It is absent from the meat, since the meat is the variance of the unpenalized likelihood score. Setting $\lambda = 0$ collapses both breads onto the meat, recovering the usual unpenalized Poisson covariance $(Z^\top W Z)^{-1}$. Increasing $\lambda$ narrows the intervals, correctly reflecting sampling variability of $\hat\theta$ about the penalized target rather than the true $R_t$.

The real-time estimator \eqref{eq:glm-real-time} carries an additional penalty $\gamma\theta^\top\Psi\theta$, damping first differences in the sparse tail. It enters the bread identically to $\lambda\Omega$, giving
\[
\widehat H \;=\; Z^\top \widehat W Z \;+\; \lambda\Omega \;+\; \gamma\Psi,
\]
with the sandwich $\widehat{\mathrm{Cov}}(\hat\theta) = \widehat H^{-1}(Z^\top \widehat W Z)\widehat H^{-1}$ otherwise unchanged.

\subsubsection*{Wald intervals for $R_t$}

As $T \to \infty$, $U(\theta_0) = \sum_t s_t(\theta_0)$ is a sum of independent mean-zero terms and is asymptotically Gaussian by the CLT. The linearization $\hat\theta - \theta_0 \approx H^{-1}[U(\theta_0) - \lambda\Omega\theta_0]$ is an affine function of $U(\theta_0)$, with $H$ treated as a fixed constant. So $\hat\theta$ is asymptotically Gaussian with covariance \eqref{eq:rt-cov-theta}. The estimator $\hat R_t = S(t)^\top\hat\theta$ is linear in $\hat\theta$ and inherits normality with variance $S(t)^\top\mathrm{Cov}(\hat\theta)S(t)$, giving the Wald interval
\begin{equation}\label{eq:rt-wald}
\hat R_t \;\pm\; z_{1-\alpha/2}\,\sqrt{\,S(t)^\top\,\widehat{\mathrm{Cov}}(\hat\theta)\,S(t)\,}, \qquad z_{1-\alpha/2} \;=\; \Phi^{-1}(1-\alpha/2),
\end{equation}
where $\Phi$ is the standard normal CDF. We can rewrite this as \mbox{$[\hat R_t - z\cdot\mathrm{SE}_t,\, \hat R_t + z\cdot\mathrm{SE}_t]$}, where \mbox{$\mathrm{SE}_t \;=\; \sqrt{S(t)^\top\,\widehat{\mathrm{Cov}}(\hat\theta)\,S(t)}$} denotes the estimated standard error of $\hat R_t$.

\subsection{Day-of-week effects and ascertainment rates}

Having estimated day-of-week effects $\hat\omega$, we must incorporate them into the design in order to perform inference.
To do so, we apply the same sandwich machinery as the no-DoW case. Specifically, we scale $Z_t$ via $\widetilde Z_t =                    
\hat\omega_{(\text{$t$ mod 7})}\, Z_t$ and evaluate the working weights at the scaled fitted mean, $\widetilde W = \mathrm{diag}(1/\hat\mu_t)$ with $\hat\mu_t = \hat\omega_{(\text{$t$ mod 7})}\hat\Lambda_t$.
Formula \eqref{eq:rt-cov-theta} then becomes                 
  \[                                                              
  \widehat{\mathrm{Cov}}(\hat\theta) \;=\; \widetilde H^{-1}\,(\widetilde Z^\top \widetilde W \widetilde Z)\, \widetilde H^{-1}, \qquad \widetilde H \;=\; \widetilde Z^\top \widetilde W \widetilde Z + \lambda
  \Omega,          
  \]
and $R_t$ intervals are formed exactly as before via $\widehat{\mathrm{Var}}(\hat R_t) = S(t)^\top\, \widehat{\mathrm{Cov}}(\hat\theta)\, S(t)$. This is a \emph{conditional} variance: it treats $\hat\omega$ 
as fixed rather than estimated, and so omits the contribution from the joint $(\theta, \omega)$ Hessian that would reflect sampling variability in $\hat\omega$. Fortunately, this omission is typically small in      
practice, since each $\hat\omega_j$ pools across every observation from day $j$ and is far better identified than any single spline coefficient.

  
  The ascertainment rates $\rho_t$ are fixed and known. They fold into the
  sandwich just as $\hat\omega$ does: scale the design to
  $\widetilde Z_t = \rho_t Z_t$ and evaluate the weights at the fitted mean
  $\hat\mu_t = \rho_t \hat\Lambda_t$. (When both terms are present, the two
  scalings compose.) Unlike $\hat\omega$, however, $\rho_t$ carries no sampling
  variability, so no conditional-variance caveat applies: the resulting
  $\widehat{\mathrm{Cov}}(\hat\theta)$ is exact in this respect. We form $R_t$
  intervals as before, via
  $\widehat{\mathrm{Var}}(\hat R_t) = S(t)^\top \widehat{\mathrm{Cov}}(\hat\theta) S(t)$.
  Note that $\rho_t$ acts only on the observation scale; the target
  $R_t(\theta) = S(t)^\top\theta$ is unaffected.

\subsection{Quasi-Poisson adjustment}\label[appendix]{apx:quasi-poisson}

When Poisson dispersion does not hold --- as is typical for real
surveillance data --- we replace Assumption~1 with the quasi-likelihood
assumption $\mathrm{Var}(y_t) = \phi\,\Lambda_t$ for an unknown
dispersion $\phi \geq 1$. The sandwich form above is unchanged except
for the score variance, which becomes $\phi\,Z^\top W Z$ instead of
$Z^\top W Z$. Substituting back, the asymptotic covariance of $\hat\theta$
scales by $\phi$, and every standard error (and hence Wald-CI half-width)
scales by $\sqrt{\phi}$. 

We estimate $\phi$ from Pearson residuals at
the fitted values,
\[
  \hat\phi \;=\; \frac{1}{T - p_{\mathrm{eff}}}
                \sum_t \frac{(y_t - \hat\Lambda_t)^2}{\hat\Lambda_t},
\]
where $p_{\mathrm{eff}} = \mathrm{tr}\bigl((Z^\top W Z + \lambda\Omega)^{-1}
Z^\top W Z\bigr)$ is the effective degrees of freedom of the penalized
fit. All real-data confidence intervals reported in this paper use this
adjustment; most simulation results do not, where the data-generating
process is Poisson.

\subsection{Simultaneous Coverage}
The Wald intervals are \emph{pointwise}: the collection $\{[\hat R_t - z\cdot\mathrm{SE}_t,\, \hat R_t + z\cdot\mathrm{SE}_t]\}_{t=t_0}^{t_1}$ does
not jointly cover the entire $R$-curve at level $1-\alpha$, since family-wise error accumulates over $T$ timesteps. 
To get a band with \textit{simultaneous coverage}, we inflate the critical value from $z_{1-\alpha/2}$ to some larger $c_\alpha$, chosen so that the probability of \emph{any} $R_t$ escaping the band is at most $\alpha$:  \[                                                              
  \mathbb{P}\!\left(\,\max_t \frac{|\hat R_t - R_t|}{\mathrm{SE}_t} > c_\alpha\,\right) \;\leq\; \alpha,
  \]                                                                                                                                                                                                          
While no closed-form expression exists for $c_\alpha$, we can leverage the asymptotic distribution of $\hat \theta$ to estimate it empirically. 
For \mbox{$m = 1,\ldots,M$}, sample \mbox{$\delta^{(m)} \sim
\mathcal{N}(0,\,\widehat{\mathrm{Cov}}(\hat\theta))$}, as a plausible realization of the random error $\hat\theta - \theta_0$. 
Passing this through the spline basis samples deviations $\delta R_t^{(m)} = 
S(t)^\top \delta^{(m)}$, of which we record the standardized supremum \mbox{$Z_m = \max_t |\delta R_t^{(m)}|/\mathrm{SE}_t$}.
Because the $\{Z_m\}$ are i.i.d., we set $c_\alpha$ to the empirical $(1-\alpha)$ quantile to form the simultaneous band $\hat R_t \pm c_\alpha\cdot\mathrm{SE}_t$. 

A simpler alternative to attain simultaneous coverage is the Bonferroni correction.
This strategy simply adjusts the significance level by a factor of $T$, taking \mbox{$c_\alpha \leftarrow z_{1-\alpha/(2T)}$}.
While valid, the Bonferroni correction ignores the strong correlation between adjacent $R_t$ estimates, 
and thus pays a nontrivial width penalty.
In contrast, the empirical approach takes this correlation into account,
and is thus less conservative. 

\subsection{Constraints}
                                                                               
The estimator carries three constraints that restrict the feasible set. 
We ignore all of them for inference, decisions we justify here. 

First is a non-negativity constraint $S\theta \succeq 0$.
This is an inequality constraint on $R_t$, which we ignore it for inference. The justification is standard: when the true $R_t$ is bounded strictly above zero at every timestep --- as is the case in essentially   
every epidemic setting we study --- $\hat\theta$ lies in the interior of the feasible set with probability approaching one, so the constraint is asymptotically inactive and the sandwich derivation goes      
through unmodified.

Secondly, the DoW identifiability constraint $\prod_j \omega_j = 1$ is also not a concern for inference.                                   This is an \emph{identifiability} constraint: it picks a unique representative from the rescaling ray $(\omega, \theta) \leftrightarrow      
(c\omega, \theta/c)$ on which the likelihood is constant, rather than restricting the feasible set. Every statistical quantity we report --- $\hat\mu_t$, $R_t = S(t)^\top\theta$, the day-of-week ratios $\omega_j/\omega_k$ --- is invariant to that same rescaling, so none depends on which representative we pick.

The third constraint is the constant-tail restriction, which is not our default method.
Here the constraint shapes only the point estimate: $\hat\theta$ is fit subject to
$A\hat\theta = 0$ (enforced within the penalized IRLS solve via its KKT system, as
described in \cref{apx:optim-constraints}). The covariance, however, is formed in the
full coefficient space using the same sandwich as our default, unconstrained-tail fits \eqref{eq:sandwich-cov}. The working weights $\widehat W$ are evaluated at the constrained fit, but $A$ does not otherwise enter.

\subsection{Real-time inference}\label[appendix]{apx:rt-realtime-ci}

Our estimates $\hat R_t$ have error
\[
\lvert\hat R_t - R_t\rvert = \left\lvert \underbrace{\mathbb{E}[\hat R_t]-R_t}_{\text{Bias}_t}+\underbrace{\hat R_t - \mathbb{E}[\hat R_t]}_{\text{Noise}_t\ (\varepsilon_t)}\right\rvert, \qquad \varepsilon_t \overset{\cdot}{\sim} \mathcal{N}(0,\sigma_t^2).
\]
Previously, we demonstrated how to compute the sandwich variance for $\sigma_t^2$, and thus form Wald confidence intervals based on the noise $\varepsilon_t$. 
On retrospective analyses, these bands cover the true values, or nearly do.
This is because cubic splines are flexible enough to learn realistic $R_t$ curves, so bias is minimal.
In real-time, however, we impose various forms of tail regularization, which produces bias. 
Consequently, the Wald bands fail to cover the tail predictions, so we must take a different approach.

\paragraph{Split-conformal bands.}
Rather than model the edge error with a variance formula, we estimate it directly from errors ConvRt has made before.
For any past date, ConvRt's estimate eventually settles once many more weeks of data arrive.
We trust that settled value as a stand-in for the truth.
Back when the same date sat near the leading edge, the real-time estimate was worse.
The gap between that early estimate and the settled value is a concrete sample of the edge error we now face.
We collect many such gaps and read the band width off their empirical quantile.

Let $\hat R^{W}_{t}$ denote the ConvRt estimate of $R_t$ at date $t$ using data through vintage $W$, and let $b = W - t$ be its horizon behind the edge.
We index everything by $b$ because the error grows as $b \to 0$.
Write $\hat R^{W}(\tau)$ for the settled anchor of the previous paragraph, the estimate at a date $\tau \le W - \Delta$ with $\Delta = 14$ days.
Each earlier vintage $W' < W$ then contributes one calibration example per horizon: its edge-region estimate $\hat R^{W'}_{\tau}$ at $\tau = W' - b$, compared against the anchor $\hat R^{W}(\tau)$.

We score each example in one of two ways,
\begin{align}
\text{Residual:}\quad & s = \bigl\lvert \hat R^{W'}_{\tau} - \hat R^{W}(\tau)\bigr\rvert,\\
\text{Studentized:}\quad & s = \bigl\lvert \hat R^{W'}_{\tau} - \hat R^{W}(\tau)\bigr\rvert \big/ \sigma^{W'}_{\tau},
\end{align}
collect the scores separately for each horizon $b$, and take the split-conformal quantile $\hat q_b = s_{(k)}$ with $k = \lceil (1-\alpha)(n+1)\rceil$.
The band at the target date is then $\hat R^{W}_{t} \pm \hat q_b$ for the residual score, or $\hat R^{W}_{t} \pm \hat q_b\,\sigma^{W}_{t}$ for the studentized score.
Because the revisions are exchangeable across vintages at a fixed horizon, this construction carries the usual finite-sample split-conformal coverage guarantee.
Crucially, both scores absorb the tail bias directly, since the revision $\hat R^{W'}_{\tau} - \hat R^{W}(\tau)$ measures the total error $\lvert\text{Bias}_t+\text{Noise}_t\rvert$, rather than the
noise alone.

\begin{figure}
\centering
\includegraphics[width=\linewidth]{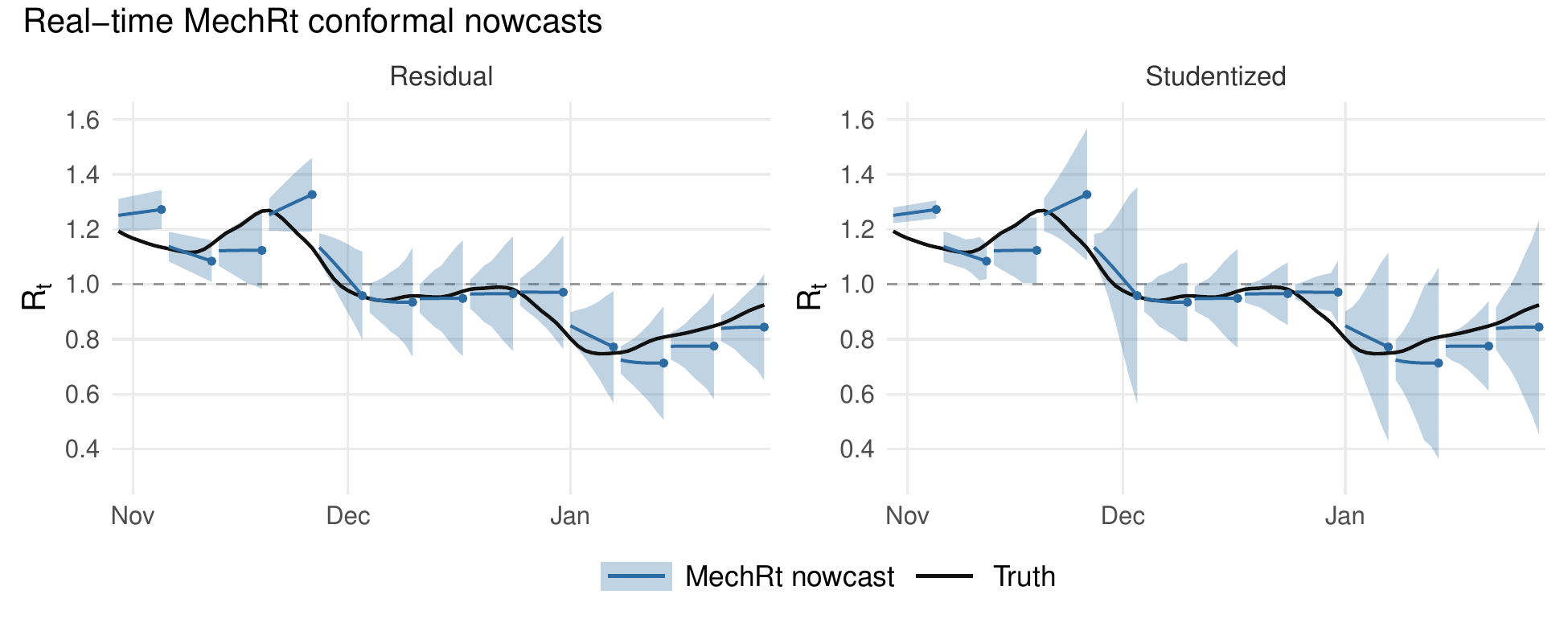}
\caption[Conformal nowcast bands on the wiggly flu simulation.]{Residual (left) and studentized (right) conformal nowcast bands for ConvRt on the wiggly flu simulation, $95\%$ level, last seven days per weekly vintage. Black is the simulation truth; blue is the
real-time mean with its conformal band. The studentized bands swing wide or tight as the edge $\sigma_t$ fluctuates; the residual bands are steadier.}
\label{fig:conf-sim}
\end{figure}

\paragraph{Residual vs. studentized.}
The two scores differ only in whether they rescale each revision by the model's reported precision $\sigma_\tau$.
The residual score adapts its band widths slowly, in response to sustained periods of (mis)coverage. 
The studentized score is designed to make more flexible bands across dates.
Dividing by $\sigma^{W'}_{\tau}$ turns each revision into a standardized residual, so the target band inherits the local scale
$\sigma^{W}_{t}$ and widens exactly where the fit reports itself to be uncertain.
This is appealing in principle, but it presumes that $\sigma$ is a trustworthy gauge of the real-time error.
At the tail, however, the dominant component of the error is the regularization bias, which $\sigma_\tau$ does not reflect. 
The sandwich covariance may also be a less reliable estimator there.

Figure~\ref{fig:conf-sim} shows the consequence on individual vintages.
The studentized bands swing wide or tight as the edge $\sigma_t$ fluctuates.
Meanwhile, the residual bands are more steady over time, yet still generally track the truth.

\begin{figure}[t]
\centering
\includegraphics[width=0.7\linewidth]{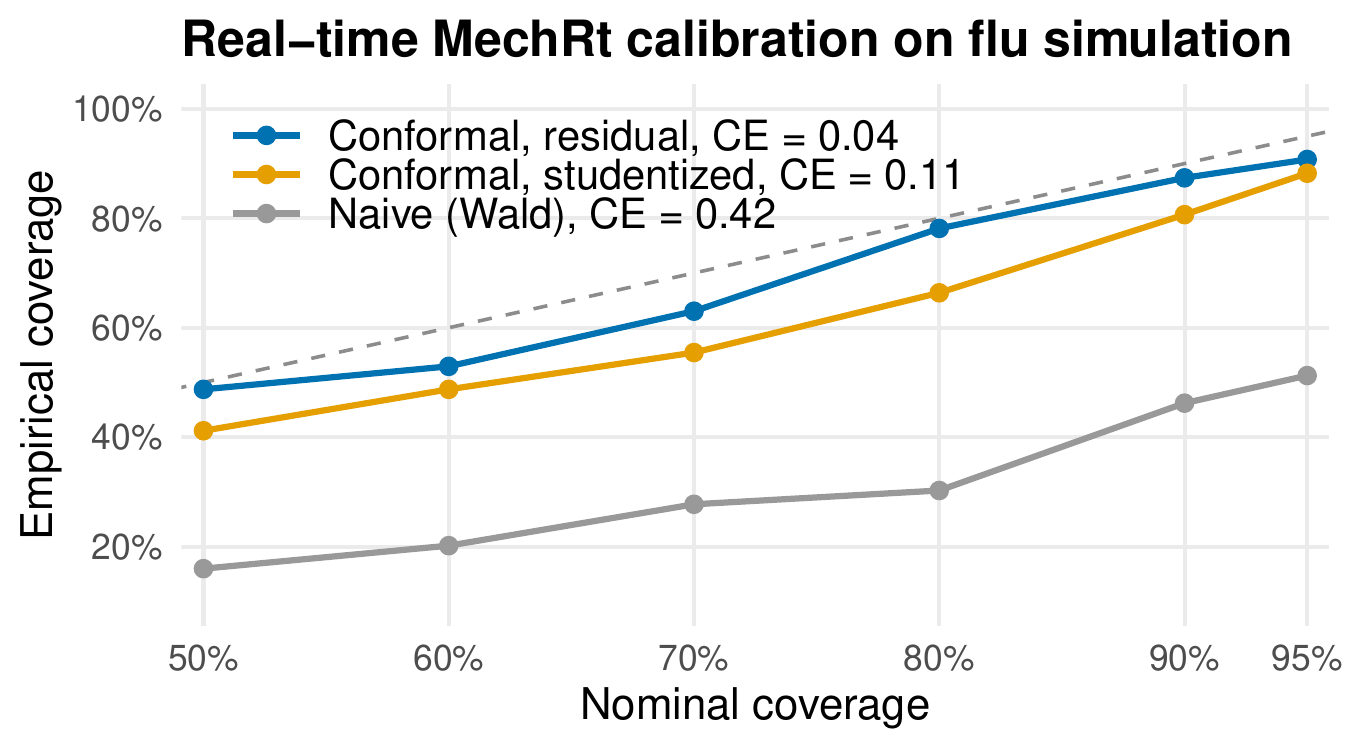}
\caption[Real-time conformal coverage calibration on the flu simulation.]{Real-time coverage calibration on the flu simulation, scored against the true simulated $R_t$.
Empirical versus nominal coverage for the residual and studentized conformal bands and the naive Wald band.
The dashed line marks perfect calibration.}
\label{fig:conf-cal}
\end{figure}

\paragraph{Empirical coverage.}
We use $\ell_1$ coverage error
\sloppy{$\text{CE} = \tfrac{1}{K}\sum_{k} \lvert \widehat{\text{cov}}(\alpha_k) - \alpha_k\rvert$}
to summarize calibration.
This is the mean gap between empirical and nominal coverage over the grid 
{$\alpha_k \in \{0.5, 0.6, 0.7, 0.8, 0.9, 0.95\}$}, 
measured on real-time $(W, t)$ rows.
On the flu simulation, where the true $R_t$ is known, the residual band is nearly perfect at $\text{CE} = 0.04$ ($n = 119$ rows; Figure~\ref{fig:conf-cal}).
The studentized band undercovers, at $\text{CE} = 0.11$.
The naive Wald band collapses to $\text{CE} = 0.42$, reaching only $51\%$ coverage at the nominal $95\%$ level.
On real flu data the ordering is the same, scored against the end-of-season retrospective: residual $\text{CE} = 0.04$, studentized $0.05$, and naive $0.27$ ($n = 175$ rows).

We therefore adopt the residual score as the default ConvRt real-time band.
It is worth noting that nothing in this construction is specific to ConvRt.
The residual band wraps any real-time point predictor, since it needs only past estimates and their revisions, not a variance model.

\section{Deconvolving latent infections}\label[appendix]{apx:deconvolution}

In this section, we share how to estimate latent infections $x$ given observations $y$.
The recovered infections $\hat x$ feed the $R_t$ estimators in the main text.

\paragraph{Spline optimization.}
This section recycles $S(t)$, $\theta$, $\Omega$, and $\lambda$ from \cref{sec:rt-retro}.
The setup is structurally the same---a spline parameterization fit by penalized Poisson regression---but here these symbols refer to the \emph{infection} curve $x_t$, not $R_t$.

Restating the observation model \eqref{eq:quasi-poisson-model}, 
\begin{equation}
y_t\given x_{<t} \sim \text{Quasi-Poisson}(\mu_t, \varphi), 
\qquad \mu_t = \omega_{(t \bmod 7)} \rho_t \Lambda_t, 
\qquad \Lambda_t = \sum_{s<t} x_{s} \pi_{t-s}.
\end{equation}

We recover daily infections $x_t$ with a natural cubic spline, $x_t = S(t)^\top \theta$.
The spline coefficients solve a penalized Poisson problem, 
\begin{equation}
\hat\theta = \argmin_\theta\;
\sum_t \bigl(\mu_t(\theta) - y_t \log \mu_t(\theta)\bigr)
+ \lambda\, \theta^\top \Omega\, \theta,
\qquad
\hat x_t = \max\bigl(S(t)^\top \hat\theta,\, 0\bigr),
\end{equation}
exactly the structure of \eqref{eq:glm-retro}.
The quasi-Poisson dispersion $\varphi$ is not necessary here, since it was only used for inference.
$\lambda$ may be selected via K-fold CV or GCV.
For computational convenience, our code uses GCV, though this could be marginally less accurate.

The smoothness penalty $\Omega$ is a magnitude-weighted integrated squared second derivative, computed by fine-grid quadrature:
\begin{equation}
\Omega = \int v(t)\, S''(t)\, S''(t)^\top\, dt,
\qquad \text{so} \qquad
\theta^\top \Omega\, \theta = \int v(t)\, \bigl(S''(t)^\top \theta\bigr)^2\, dt.
\end{equation}
The natural-spline basis pairs well with this penalty: linear tails are structural to the basis, which stabilizes the under-identified curve ends, and $\int (x{''})^2$ is the canonical smoothing-spline penalty here. (A second-difference $D^{(2)}$ penalty is the analog for B-/P-splines.)

\paragraph{Relative-curvature weights.}
$x$ ranges over several orders of magnitude---a few cases per day on the early limb, thousands at peak.
A uniform penalty ($v_t \equiv 1$) measures curvature in absolute terms, so the peak swamps the low limb in the integral, even when both are equally wiggly in proportional terms.
No single $\lambda$ works: controlling the lower limb oversmooths the peak, while tuning for the peak undersmooths the bottom.

The fix is to penalize \emph{relative} curvature, i.e.\ $\bigl((x)''/x\bigr)^2$.
Dividing by $x$ cancels the scale---$x \to c\, x$ leaves it unchanged---so a proportional wiggle counts the same on the low limb as at the peak.
This corresponds to a weighting $v_t \propto (\tilde x_t)^{-2}$.
While $x$ is unknown, a plug-in estimate $\tilde x$ works.
This makes the integrand $\bigl((x)''/\tilde x\bigr)^2$, matching the relative curvature when $\tilde x \approx x$.

The inverse blows up wherever $\tilde x$ is small, which would over-penalize the low-count tails and prevent the spline from tracking genuine early growth.
We therefore floor $\tilde x$ at a fraction $c$ of its maximum:
\begin{equation}
v_t = \left(\frac{\tilde x_{\max}}{\max(\tilde x_t,\, c\, \tilde x_{\max})}\right)^{\!2},
\qquad \tilde x_{\max} = \max_s \tilde x_s,
\quad c = 0.03.
\end{equation}
The numerator $\tilde x_{\max}$ normalizes the weights so $\lambda$ stays on a fixed, GCV-friendly scale across datasets.
The floor parameter $c$ caps how strongly low-count timesteps are penalized: any $t$ with $\tilde x_t \le c\, \tilde x_{\max}$ inherits the maximum weight $1/c^2$.
We chose $c = 0.03$ by inspection on simulated flu data; fits across the range $c \in [0.01, 0.1]$ are visually indistinguishable, so the estimator is not very sensitive to this choice.

The reference curve $\tilde x$ is deconvolution-free: a centered $7$-day average of the observations, shifted to infection time by the mean delay $\bar L = \mathrm{round}\bigl(\sum_k k\, \pi_k\bigr)$ and held constant past the data edge:
\begin{equation}
\bar y_t = \frac{1}{|W_t|} \sum_{s \in W_t} y_s^*,
\quad W_t = \{t-3, \dots, t+3\} \cap [1, n],
\qquad
\tilde x_t = \bar y_{\min(t + \bar L,\, n)}.
\end{equation}
This avoids any preliminary deconvolution and is valid in real time---the average shrinks at the edge, and the shift never reads beyond $t$.

Splines are not inherently non-negative, so this deconvolution approach risks being inaccurate when counts are low.
Therefore we also implement an alternative that uses a log-link function.
This uses $\log y_t = \mu_t(\theta)$, which is naturally non-negative at the cost of incorrect specification.

\section{Methods configuration}\label[appendix]{apx:methods-config}

\subsection{EpiNow2}\label[appendix]{apx:methods-epinow2}

\begin{figure}
  \centering
  \includegraphics[width=0.8\linewidth]{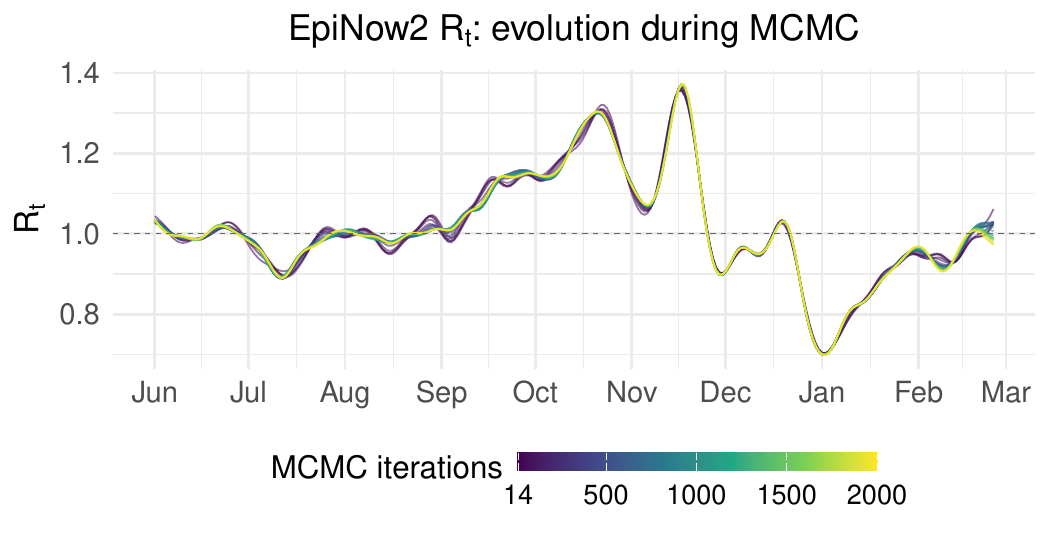}
  \caption[EpiNow2 posterior-mean $R_t$ convergence on the wiggly simulation.]{Posterior-mean $R_t$ from EpiNow2's Gaussian-process model at
  successive MCMC snapshots on the retrospective wiggly-simulation fit (4
  chains; 1500 warmup and 500 post-warmup iterations per chain). Each line is
  one snapshot; color encodes the cumulative number of post-warmup draws
  aggregated across chains, with the corresponding wallclock time in
  parentheses.}
  \label{fig:epinow2-snapshots}
\end{figure}

\begin{table}[ht]
\centering
\caption[EpiNow2 configuration for finalized runs.]{EpiNow2 configuration for the finalized runs. The random-walk (RW)
kernel uses a unit step size on $\log R_t$; the Gaussian-process (GP) kernel is a
Hilbert-space approximation with a Mat\'ern-3/2 lengthscale. `NB'' denotes negative binomial; ``--'' not applicable.}
\label{tab:epinow2-config}
\setlength{\tabcolsep}{5pt}
\begin{tabular}{l cccc cc}
\toprule
& \multicolumn{4}{c}{\textit{rtestim}} & \multicolumn{2}{c}{\textit{Influenza}} \\
\cmidrule(lr){2-5} \cmidrule(lr){6-7}
& \shortstack{Piecewise\\Constant}
& \shortstack{Piecewise\\Exponential}
& \shortstack{Piecewise\\Linear}
& Periodic
& \shortstack{Wiggly\\$R_t$}
& \shortstack{Smooth\\$R_t$} \\
\midrule
Rt prior                & \multicolumn{4}{c}{$\mathrm{LogN}(0, 1.0)$} & \multicolumn{2}{c}{$\mathrm{LogN}(0, 0.5)$} \\
Kernel                  & RW & GP & RW & GP & \multicolumn{2}{c}{GP} \\
\quad Lengthscale       & -- & $\mathrm{LogN}(10,3)$ & -- & $\mathrm{LogN}(10,3)$ & \multicolumn{2}{c}{$\mathcal{N}(21,7)$} \\
\quad Basis prop.       & -- & 0.2 & -- & 0.2 & \multicolumn{2}{c}{0.2} \\
Obs.\ model             & \multicolumn{4}{c}{Poisson} & \multicolumn{2}{c}{NB} \\
Chains                  & \multicolumn{4}{c}{4} & \multicolumn{2}{c}{2} \\
Warmup                  & 2500 & 1500 & 3000 & 1500 & \multicolumn{2}{c}{1000} \\
Samples                 & 2500 & 2000 & 2000 & 2000 & \multicolumn{2}{c}{500} \\
$\delta_{\text{adapt}}$ & 0.97 & 0.95 & 0.97 & 0.95 & \multicolumn{2}{c}{0.95} \\
Max.\ tree              & 13 & 12 & 13 & 12 & \multicolumn{2}{c}{12} \\
\bottomrule
\end{tabular}
\end{table}

We benchmark against EpiNow2 (version 1.8.0), summarizing its configuration for every finalized run in Table~\ref{tab:epinow2-config}. Our aim throughout was to run each fit as cheaply as possible, but this
was in constant tension with EpiNow2's sampler, which failed to converge far more often than not---especially on the rtestim benchmarks. In most experiments $R_t$ follows EpiNow2's default Gaussian-process
prior (Hilbert-space approximation, Mat\'ern-$3/2$ kernel). On the rtestim scenarios with jump discontinuities (piecewise constant and piecewise linear), this prior could not be salvaged at any sampling budget
we tried---fits diverged or returned $\widehat{R}$ in the thousands---so on those datasets we instead modeled $\log R_t$ as a Gaussian random walk. Even where the GP prior was ultimately retained, the
out-of-the-box settings were not enough to reach acceptable diagnostics. 

The observation model matches each dataset's noise---Poisson or negative binomial---and for the real surveillance data we additionally retain EpiNow2's day-of-week reporting effect. MCMC uses 4 chains (2 for
the real data); warmup, \texttt{adapt\_delta}, and \texttt{max\_treedepth} are all raised above EpiNow2's defaults, scenario by scenario, until each fit was free of divergent transitions and well mixed
($\widehat{R} < 1.05$). We fix a single setting per dataset rather than tuning per replicate, accepting some over-provisioning to keep the comparison fair. This makes EpiNow2 by a wide margin the slowest
method evaluated---tens of minutes to over two hours per single fit, against roughly a second for ConvRt---so the inflated sampling budgets compound an already steep runtime cost. The generation-time and
reporting-delay distributions match those used by ConvRt and the other benchmarks.



\cref{fig:epinow2-snapshots} traces EpiNow2's GP posterior mean on the wiggly
retrospective fit as the sampler accumulates draws, computed by polling
cmdstanr's per-chain CSV output every 30 seconds and re-averaging over the
post-warmup draws available at each tick. The oscillatory structure in the final
fit---the short-lived excursions from $R_t = 1$ during the fall and winter---is
present from the earliest snapshot onward and merely tightens as draws
accumulate, confirming it is genuine posterior structure rather than
under-sampled Monte Carlo noise.

The colorbar also exposes a feature of EpiNow2's runtime invisible in a single
end-of-run timing: the first snapshot sits at 1564\,s with only 14 post-warmup
draws, roughly 68\% of the total wallclock. Stan compilation is cached and costs
only seconds; the rest is the 1500 warmup iterations per chain. Because warmup is
the bulk of the iterations and cannot be skipped, EpiNow2's wallclock is bounded
below by warmup and cannot be meaningfully reduced by drawing fewer samples.

\subsection{EpiEstim}\label[appendix]{apx:methods-epiestim}

We call \texttt{EpiEstim::estimate\_R} with \texttt{method = "non\_parametric\_si"},
passing the same generation-interval pmf used by ConvRt and the other
benchmarks (prepended with $\Pr(\mathrm{SI}=0)=0$ to match EpiEstim's indexing
convention). $R_t$ is estimated on a 7-day sliding window
(\texttt{t\_start = 2:(n-6)}, \texttt{t\_end = t\_start + 6}); the gamma prior
on $R_t$ has mean 1 and SD 5, the package defaults. Per the caveat in the main
text, EpiEstim is fit directly to observed reports rather than deconvolved
infections, so its estimates are implicitly shifted by the case-reporting
delay; we do not adjust the time axis to compensate, since doing so would
require a deconvolution step that EpiEstim itself does not perform. 

Credible intervals are taken from EpiEstim's posterior quantiles at the 95\% level. For bands at other levels, we convert each interval into an implied (Gaussian) standard deviation, $\hat\sigma = (q_{0.975} - q_{0.025}) / (2\,z_{0.975})$, where $q_{0.025}$ and $q_{0.975}$ are
the lower and upper posterior quantiles and $z_{0.975} \approx 1.96$. We then form a $100(1-\alpha)\%$ band as $\hat R_t \pm z_{1-\alpha/2}\,\hat\sigma$.

\subsection{EstimateR}\label[appendix]{apx:methods-estimater}

We use \texttt{estimateR} with LOESS pre-smoothing and Richardson--Lucy
deconvolution against the same case-to-report delay distribution $\pi$
($x\!\to\!y$, discrete gamma) used elsewhere. This is followed by an
EpiEstim sliding-window $R_t$ step
(\texttt{estimation\_method = "EpiEstim sliding window"},
\texttt{estimation\_window = 3}, \texttt{mean\_Re\_prior = 1}), using the oracle generation interval distribution $g$. 

For point
estimates and timing comparisons we call
\texttt{estimate\_Re\_from\_noisy\_delayed\_incidence} (a single fit); for
uncertainty intervals we call \texttt{get\_block\_bootstrapped\_estimate} with
50 block-bootstrap replicates, the package's recommended approach for
uncertainty quantification. Because Richardson--Lucy deconvolution truncates
the tail of the incidence series (an unavoidable consequence of the
right-censored convolution), estimateR only estimates $R_t$ up to 5 days
before the vintage date, so in our real-time figures and tables we propagate
its last estimate forward through the remainder of the week.

\subsection{EpiLPS}\label[appendix]{apx:methods-epilps}

We use \texttt{EpiLPS::estimR} with the package defaults: a cubic
B-spline basis with second-order difference penalty, the default
hyperprior on the smoothing parameter, and the LPSMAP backend
(Laplace approximation around the posterior mode), rather than the
slower LPSMALA MCMC alternative. The generation interval $g$ is supplied
as the same pmf used by the other methods. Real-time fits are run on
truncated vintage data without modification.

\subsection{Rtestim}\label[appendix]{apx:methods-rtestim}

We use \texttt{rtestim::cv\_estimate\_rt} with cubic trend filtering
(\texttt{korder = 3}), 3-fold cross-validation, and an explicit
log-spaced $\lambda$ grid spanning $10^{-2}$ to $10^{5}$ with 30
values; we found the package's default auto-grid sometimes terminated
before the 1-SE optimum was located, so we supply the grid explicitly
and raise \texttt{maxiter} to $10^7$ for convergence at the
smaller-$\lambda$ end. $\lambda$ is selected by the 1-SE rule.
Confidence bands come from \texttt{confband()}, which inverts a
quadratic relaxation of the trend-filtering objective; the bands are
pointwise and do not have a frequentist coverage guarantee at the
level suggested by their nominal width. The generation interval is
supplied via \texttt{delay\_distn}, prepended with a 0 so that index
1 corresponds to $\Pr(\mathrm{lag}=0)=0$ (rtestim's convention treats
\texttt{delay\_distn[1]} as the same-day weight, differing from the
Cori convention used by EpiEstim). For series with a long all-zero
leading tail we trim to the first positive count before fitting,
since rtestim's CV objective is undefined when the weighted
past-counts denominator hits 0. Real-time fits are run on truncated
vintage data without modification.

\subsection{CovidEstim}\label[appendix]{apx:methods-covidestim}

We do not refit CovidEstim. The COVID-19 comparisons in
\cref{sec:rt-experiments-real} use $R_t$ trajectories and case
ascertainment rates from the project's public release
\citep{Chitwood2022,covidestim_summer2021}, which were produced by the
authors with their full Bayesian pipeline (random-walk log-$R_t$,
seroprevalence-informed ascertainment, HMC inference). We treat these as
fixed reference outputs; the delay distributions and generation interval
we feed to ConvRt for the same comparison are documented in
\cref{apx:covidestim}.

\section{Rtestim benchmarks}\label[appendix]{apx:rtestim}

Ideally, for both the rtestim and influenza benchmarks, we would resimulate each dataset many times.
We would run all methods on each replicate, and ultimately report average performance. 
However, we did not pursue this because EpiNow2 was slow to fit and frequently failed to converge, making repeated runs across all six datasets impractical.

  \subsection{Data-generating process}
  
   For each of the four ground-truth $R_t$ trajectories from \citet{rtestim}, we
  simulate $n = 300$ days of infections $x_t$ and case reports $y_t$ from a
  discrete-time stochastic SEIR model with explicit, discrete-gamma latent period,
  infectious period, and reporting delay. We set the latent (exposure-to-infectious)
  delay to mean $5.9$ and SD $3.3$ and the infectious period to mean $5.0$ and SD
  $2.2$, chosen so that the implied generation interval matches the mean of $8.4$
  days used by \citet{rtestim}; the implied SD is $3.95$. The exposure-to-report delay has mean $7$ and SD $3$. Each
  scenario is seeded with $10$ exposures per day over a $30$-day warm-up so that
  the infectious pool is near equilibrium at $t = 1$.


\begin{figure}
\centering
\includegraphics[width=0.9\linewidth]{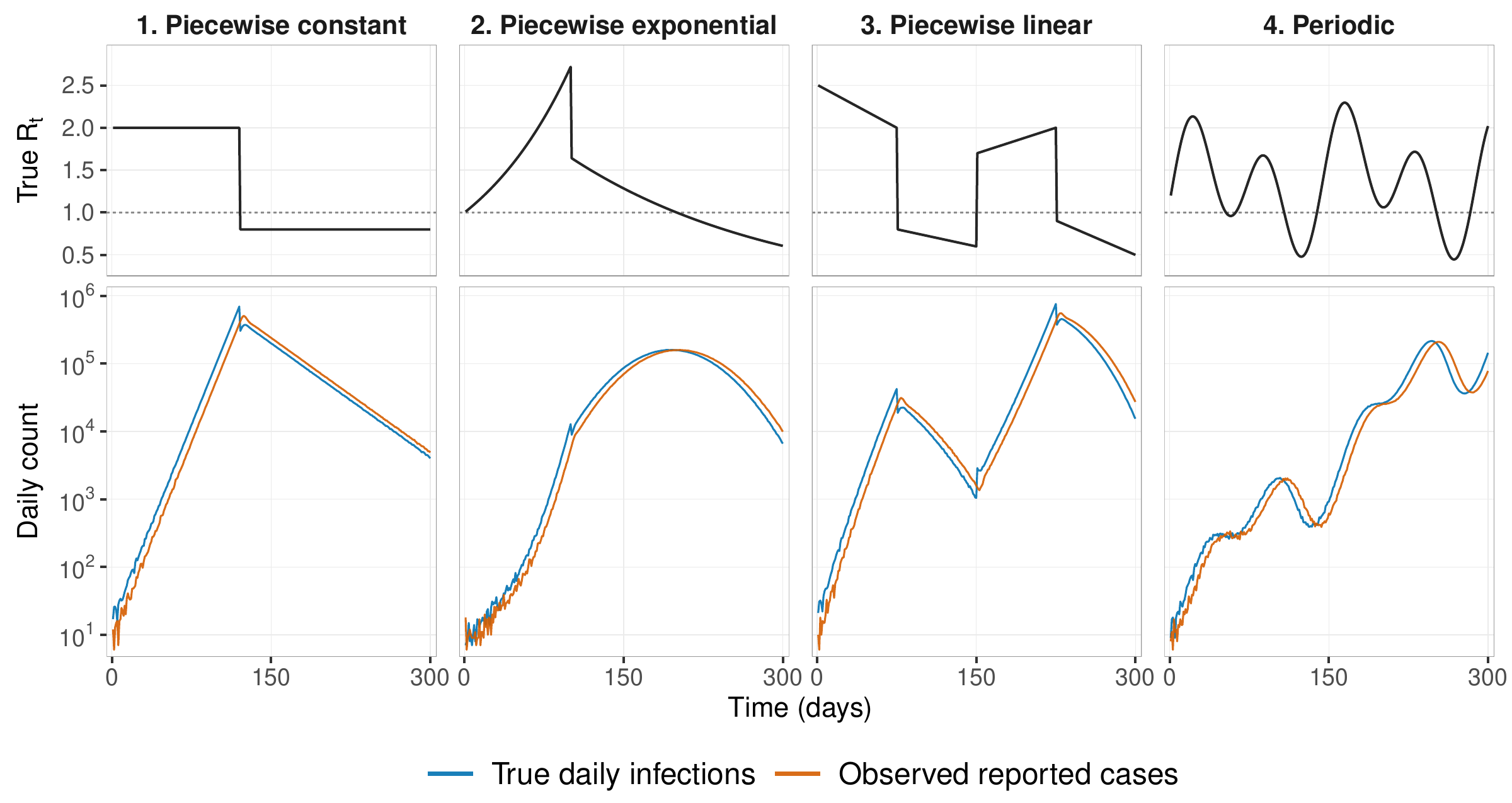}
\caption{rtestim benchmark datasets: $R_t$ curves and the cases they generate.}
\label{fig:rtestim_gt_and_cases}
\end{figure}

The four benchmark datasets from \citet{rtestim}, shown atop \cref{fig:rtestim_gt_and_cases}, are:
\begin{enumerate}
  \item \textbf{Piecewise constant.} A single sharp policy change,
  \[
    R_t =
    \begin{cases}
      2.0, & 1 \leq t \leq 120, \\
      0.8, & 120 < t \leq 300,
    \end{cases}
  \]
  representing four months of unchecked growth before an intervention drops
  transmission well below the critical threshold.

  \item \textbf{Piecewise exponential.} A fast exponential climb followed by a
  slower exponential decay, with a downward jump at the transition:
  \[
    R_t =
    \begin{cases}
      \exp(0.01\, t), & 1 \leq t \leq 100, \\
      \exp\!\bigl(0.5 - 0.005\,(t-100)\bigr), & 100 < t \leq 300.
    \end{cases}
  \]
  An exponential takeoff is partially curbed by a late-stage response that
  bends but does not break the trajectory.

  \item \textbf{Piecewise linear.} Four linear segments separated by jump
  discontinuities,
  \[
    R_t =
    \begin{cases}
      2.5 - \tfrac{0.5}{74}(t - 1),     & 1   \leq t < 76,  \\[2pt]
      0.8 - \tfrac{0.2}{74}(t - 76),    & 76  \leq t < 151, \\[2pt]
      1.7 + \tfrac{0.3}{74}(t - 151),   & 151 \leq t < 226, \\[2pt]
      0.9 - \tfrac{0.4}{74}(t - 226),   & 226 \leq t \leq 300,
    \end{cases}
  \]
  alternating between slowly drifting suppression and growth regimes, with
  abrupt regime switches at each boundary (e.g., the lifting and reimposition
  of restrictions).

  \item \textbf{Periodic.} A sum of three sinusoids on rescaled time
  $\tau = 10(t-1)/(n-1) \in [0,10]$:
  \[
    R_t = 0.2\!\left[\bigl(\sin(\pi \tau/12) + 1\bigr)
                       + \bigl(2\sin(5\pi \tau/12) + 2\bigr)
                       + \bigl(3\sin(5\pi \tau/6) + 3\bigr)\right].
  \]
  The curve oscillates smoothly between $R_t \approx 0.4$ and $\approx 2.4$
  via a slow seasonal envelope modulated by faster multi-week ripples. This is
  the only scenario without discontinuities and serves as a test of smooth,
  non-monotone tracking.
\end{enumerate}

These $R_t$ curves all generate a maximum of just below 1 million daily infections.
For context, COVID-19 crossed this threshold during the Omicron wave \citep{clarke2022seroprevalence}. 
Infections are declining by the end at all but the Periodic benchmark.

\subsection{Additional analysis}
  \begin{table}[ht]
  \centering
  \caption[ConvRt jump-discontinuity performance on the rtestim benchmarks.]{Evaluating ConvRt jump discontinuity on the rtestim benchmarks. Mean absolute error (MAE) and $\ell_1$ calibration error (CE) reported in units of $10^{-2}$.}
  \label{tab:rtestim-jumps}
  \begin{tabular}{l rr rr rr}
  \toprule
  & \multicolumn{2}{c}{S1: const} & \multicolumn{2}{c}{S2: exp}
  & \multicolumn{2}{c}{S3: linear} \\
  \cmidrule(lr){2-3} \cmidrule(lr){4-5} \cmidrule(lr){6-7}
  Method & MAE & CE & MAE & CE & MAE & CE \\
  \midrule
  ConvRt (with jump) & 0.20 &  3.27 & 0.64 & 10.74 & 0.30 &  9.29 \\
  ConvRt (no jump)   & 1.58 &  9.75 & 1.77 &  3.35 & 3.79 & 11.43 \\
  \end{tabular}
  \end{table}

\begin{figure}
    \centering
    \includegraphics[width=\linewidth]{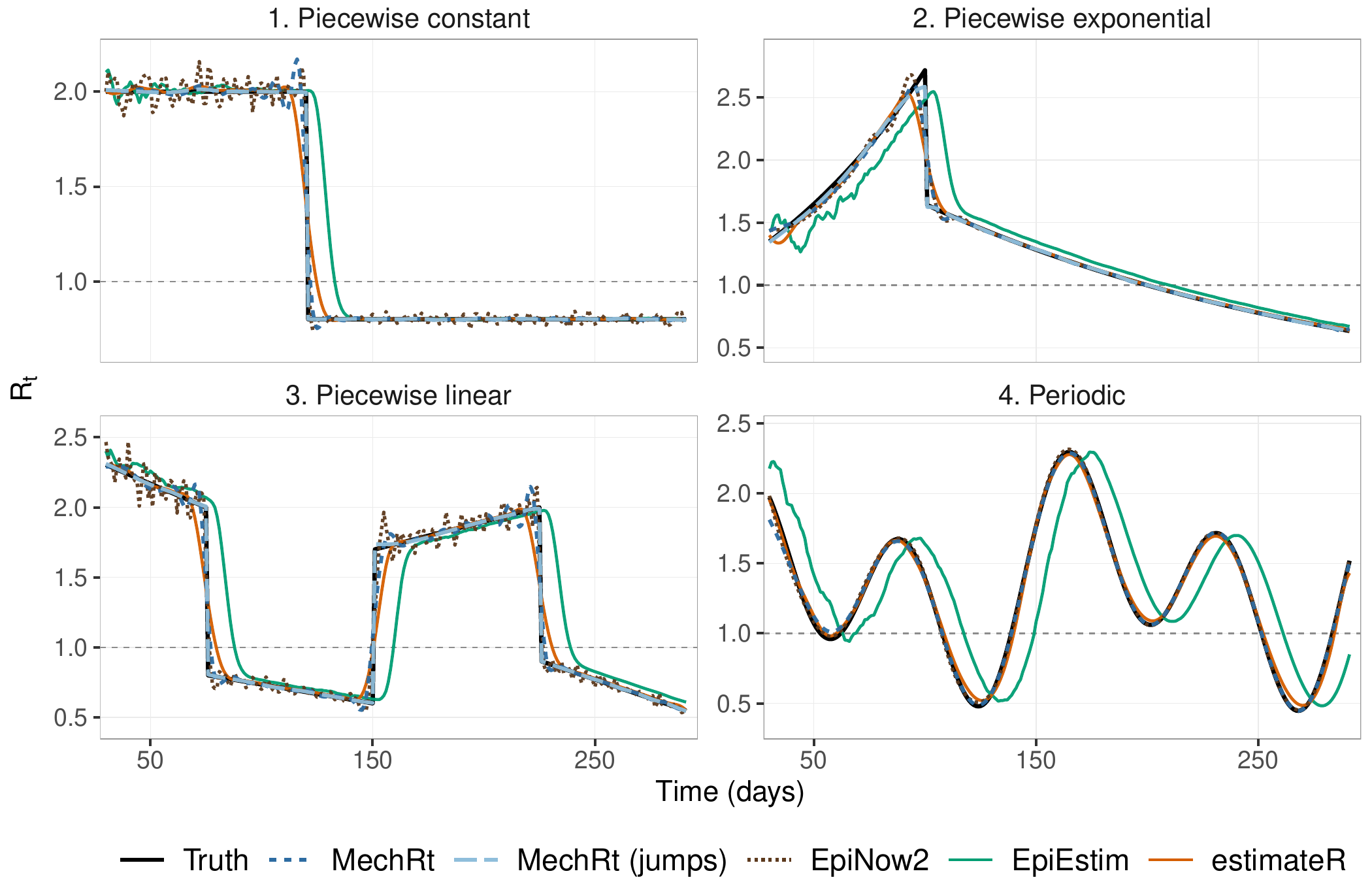}
    \caption[Retrospective $R_t$ estimates on the rtestim benchmarks.]{Retrospective $R_t$ estimates on the rtestim benchmarks, comparing ConvRt with and without the jump augmentation against EpiNow2, EpiEstim, and estimateR. The jump augmentation dramatically sharpens ConvRt's tracking at discontinuities while leaving smooth segments largely unchanged.}
    \label{fig:rtestim-all-results}
\end{figure}

\Cref{tab:rtestim-jumps} evaluates the jump augmentation on the three piecewise scenarios; \cref{fig:rtestim-all-results} provides a qualitative picture of all four benchmarks. The jump augmentation reduces MAE by factors of 8, 3, and 13 on Scenarios 1--3, respectively, with CE falling by 85\%, 76\%, and 72\%. The standard spline rounds off each discontinuity, producing a slow transition where the truth is abrupt; the step-function augmentation eliminates this artifact, yielding near-perfect recovery with tight confidence bands. The improvement on Scenario 2 is more modest in MAE because the jump occurs midway through an exponential decay, where the smooth spline partially compensates by bending aggressively. On the Periodic benchmark, the two ConvRt variants are essentially identical, as expected.

EpiNow2 tracks the broad shape of each curve but fails to recover jumps sharply. EpiEstim lags badly at every transition, substantially underestimating jump magnitudes; on the Periodic benchmark it also exhibits pronounced amplitude attenuation. estimateR oversmooths at the discontinuities but otherwise tracks reasonably well.

\section{Influenza benchmarks}\label[appendix]{apx:flu-sim} 

To prevent any method from having an unfair advantage, we chose to use the \citet{wallinga_teunis} estimator for case $R_t$.
Light smoothing was necessary on the earliest months, when low hospitalization counts led to unsteady estimates.
We analyzed the four months from October through January, when the bulk of hospitalizations occurred.

\subsection{Smooth ground truth, overdispersed observations}\label[appendix]{apx:flu-sim-smooth-od}

The ``smooth'' influenza simulation combines two modifications of its ``wiggly'' counterpart. The first replaces the wiggly Wallinga-Teunis ground truth with a heavily smoothed version, isolating method performance on a simple unimodal $R_t$ curve. The second swaps Poisson observations for a negative-binomial model calibrated to the day-to-day scatter in real flu hospitalizations.

We apply LOESS with span 0.40 to the Wallinga-Teunis $R_t$ from the main text. The resulting curve rises smoothly through October, peaks near 1.2 in late November, and decays back below 1 through January (\cref{fig:flu-sim-both}). The simulated hospitalization wave is correspondingly unimodal, peaking around 1500 daily admissions near the new year.

The observation model is negative-binomial, with $\mathrm{Var}(y_t) = \varphi\, \mu_t$ and $\varphi = 2.5$. We chose $\varphi$ to be calibrated with real flu hospitalizations. 
We fit a Poisson GAM (mgcv) to each season window, with a smooth time trend and a day-of-week factor. Using a knot count flexible enough to track the epidemic peak, the quasi-Poisson dispersion $\chi^2/(n - \text{edf})$ gives $\hat\varphi \approx 2.4$ for 2022/23 and $2.7$ for 2023/24. 
A Poisson observation model leaves 15--20\% of days outside its 95\% band around the fitted trend; $\varphi = 2.5$ leaves about 5\%, matching nominal.

\begin{figure}
    \centering
    \includegraphics[width=\linewidth]{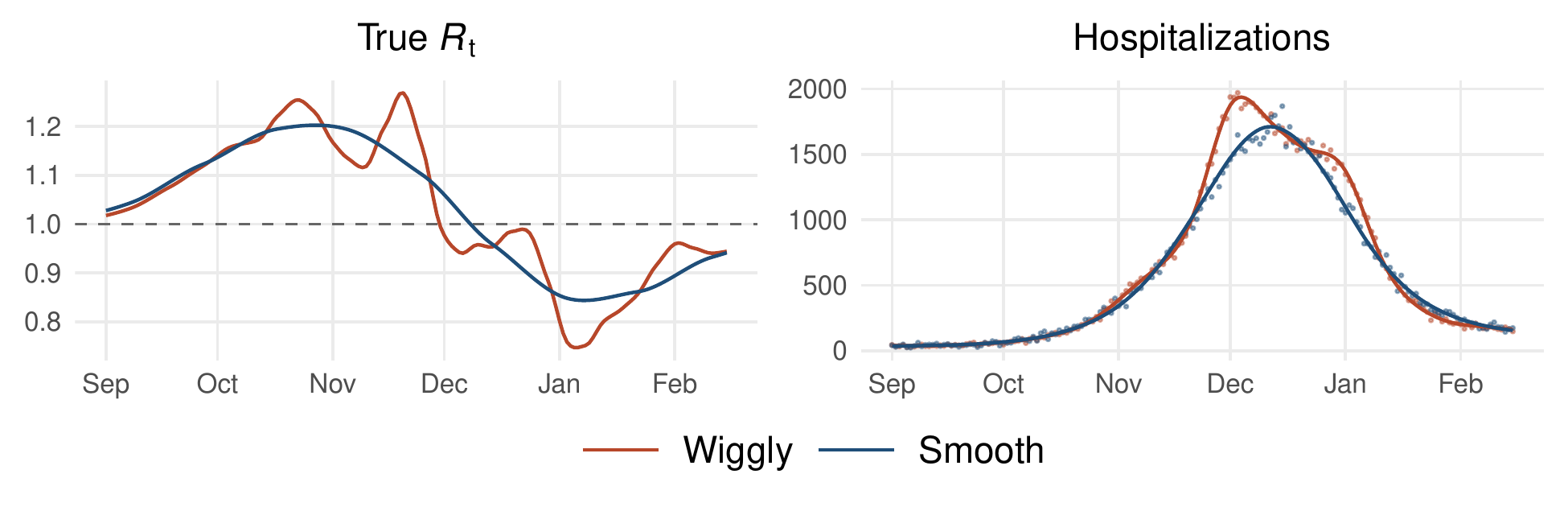}
    \caption{Ground-truth $R_t$ and daily hospitalizations for the flu simulations.}
    \label{fig:flu-sim-both}
\end{figure}

\subsection{Additional figures}

\begin{figure}
    \centering
    \includegraphics[width=\linewidth]{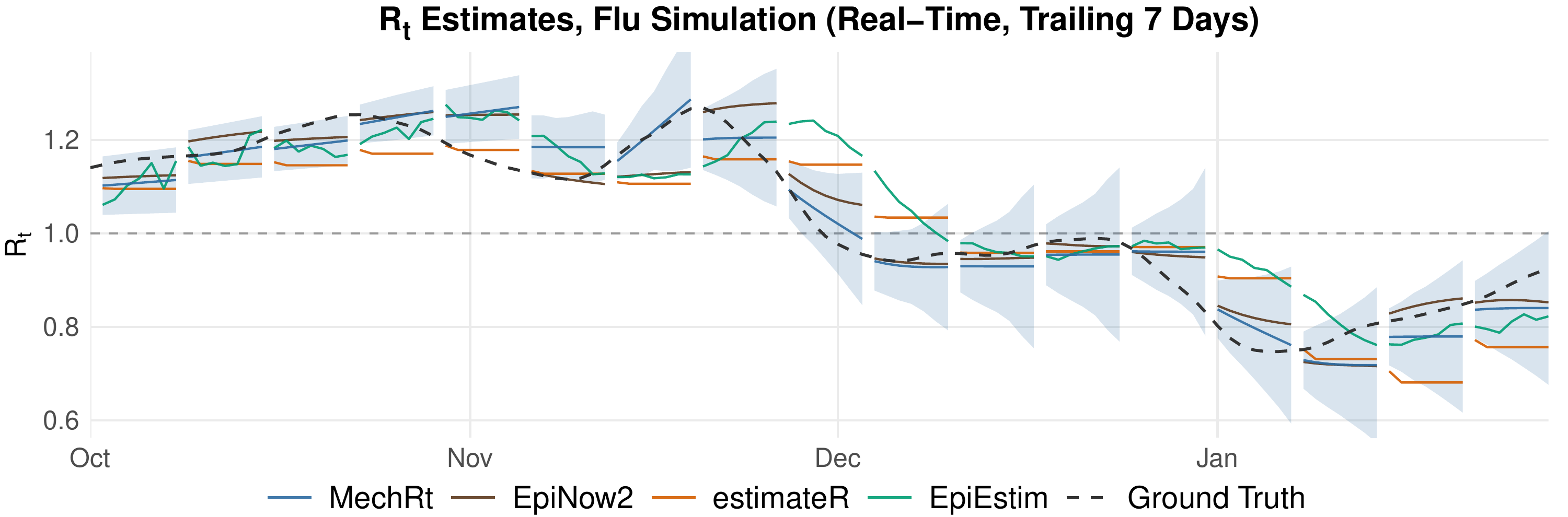}
    \caption[Real-time $R_t$ estimates on the wiggly influenza benchmark.]{Real-time $R_t$ estimates on wiggly influenza benchmark, with ConvRt's 95\% pointwise confidence bands.}
    \label{fig:realtime-rt-wiggly}
\end{figure}

We analyze two figures that compare $R_t$ estimates on the influenza benchmarks.
Matching \cref{fig:retro-comparison}, we omit EpiLPS and rtestim for clarity because they have similar bias as EpiEstim.

For the real-time setting, \cref{fig:realtime-rt-wiggly} shows last-7-day predictions on the ``wiggly'' dataset. 
ConvRt and EpiNow2 remain the most accurate methods.
As before, EpiEstim predicts roughly the true shape, but at a considerable lag. 
EstimateR is relatively unreliable in real-time, since it propagates $R_t$ estimates from a noisy deconvolution.

ConvRt gives the best predictions as $R_t$ peaks and falls in November.
It is the sole method that identifies the sharp rise in real-time. 
Similarly, it follows the descent most faithfully, dropping more quickly than EpiNow2.
Its confidence intervals are also very accurate, covering the truth across most of the season.
Intuitively, they fan out toward the end of each vintage, where the nowcast
extrapolates recent data.
They also widen over the course of the season, as $R_t$ changes direction more frequently. 
Indeed, \cref{tab:combined-realtime} reports that ConvRt has lower MAE than EpiNow2 on this dataset, and considerably lower CE.
The remaining four methods sit well behind.

\begin{figure}
    \centering
    \includegraphics[width=\linewidth]{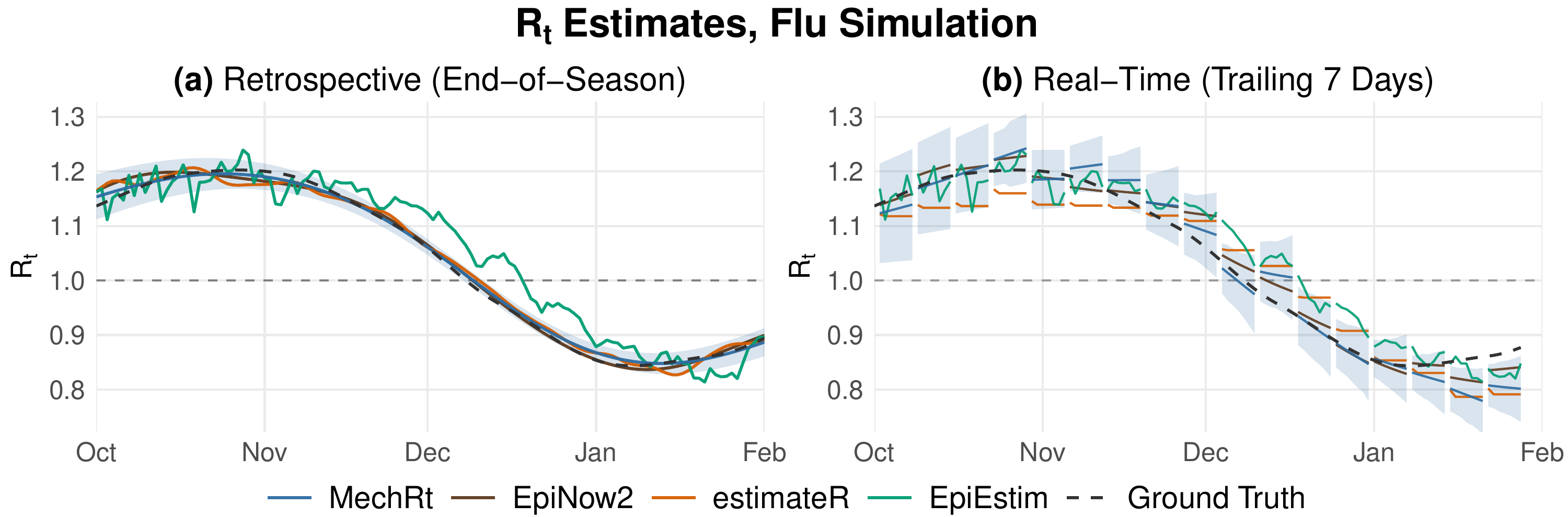}
    \caption[$R_t$ estimates on the smooth flu simulation.]{$R_t$ estimates on the smooth flu simulation. ConvRt's 95\% pointwise interval shaded.}
    \label{fig:flu4-results}
\end{figure}

\cref{fig:flu4-results} repeats both settings on the ``smooth'' benchmark, whose
true \(R_t\) rises gently to a November peak and bottoms out in January.
Here the problem is easier, and every method except EpiEstim predicts well.
In the retrospective panel (a), ConvRt, EpiNow2, and estimateR are nearly
indistinguishable from the truth. EpiEstim alone stays noisy and lags the
decline.

The real-time panel (b) tells the same story, with the three methods tracking
the truth closely and ConvRt usually getting the directional trend right.
ConvRt edges out EpiNow2 in retrospective MAE and trails it slightly in
real-time, but the gap is small, with both methods performing well (\cref{tab:combined-realtime}).
ConvRt's confidence bands remain well calibrated, with a real-time CE just below 5.
The intervals often cover even when the point estimate misreads direction. 
ConvRt incorrectly predicts a
rise in early November and a fall in early January, and the truth stays inside
the band both times.

\subsection{Results by smoothness parameter \texorpdfstring{$\lambda$}{lambda}}\label[appendix]{apx:flu-sim-lambda}

\Cref{fig:results-by-lambda}(a) shows retrospective ConvRt fits on the flu simulation across a grid of curvature penalties $\lambda$, with the min-rule fit highlighted. The min rule slightly oversmooths the peak, leaving $\hat R_t$ about $0.02$ below the truth in early January. A modestly smaller $\lambda$ closes most of this gap, lifting the peak with only minor added fluctuation elsewhere. Pushing $\lambda$ smaller still produces visibly noisy fits that chase short-run fluctuations in the cases, while larger $\lambda$ flattens the curve toward a single hump and erases the brief plateau in November. The min rule trades a small peak underestimate for stability across the rest of the season.

The spread across $\lambda$ is also far from uniform across the season. Near the start and end, where case counts are low, the fits fan out widely and the smaller $\lambda$'s oscillate. Through the high-count middle, the curves collapse onto each other regardless of $\lambda$. This reflects the likelihood: high counts pin $R_t$ down tightly, so the penalty has little room to move the fit, while at low counts the data is weak and $\lambda$ does most of the work.

Panel (b) plots the 5-fold CV Poisson deviance against $\log_{10}\lambda$. It is nearly flat for $\lambda$ below the minimizer near $\log_{10}\lambda \approx 4.5$ and rises sharply above it, so the min rule sits at the foot of a steep wall. Smaller $\lambda$'s incur almost no CV cost, consistent with the muted differences from the min-rule fit on the left.
The flat loss landscape supports the plausibility of $\lambda$ values below the min-rule, and their larger peak $R_t$.

\begin{figure}
    \centering
    \includegraphics[width=\linewidth]{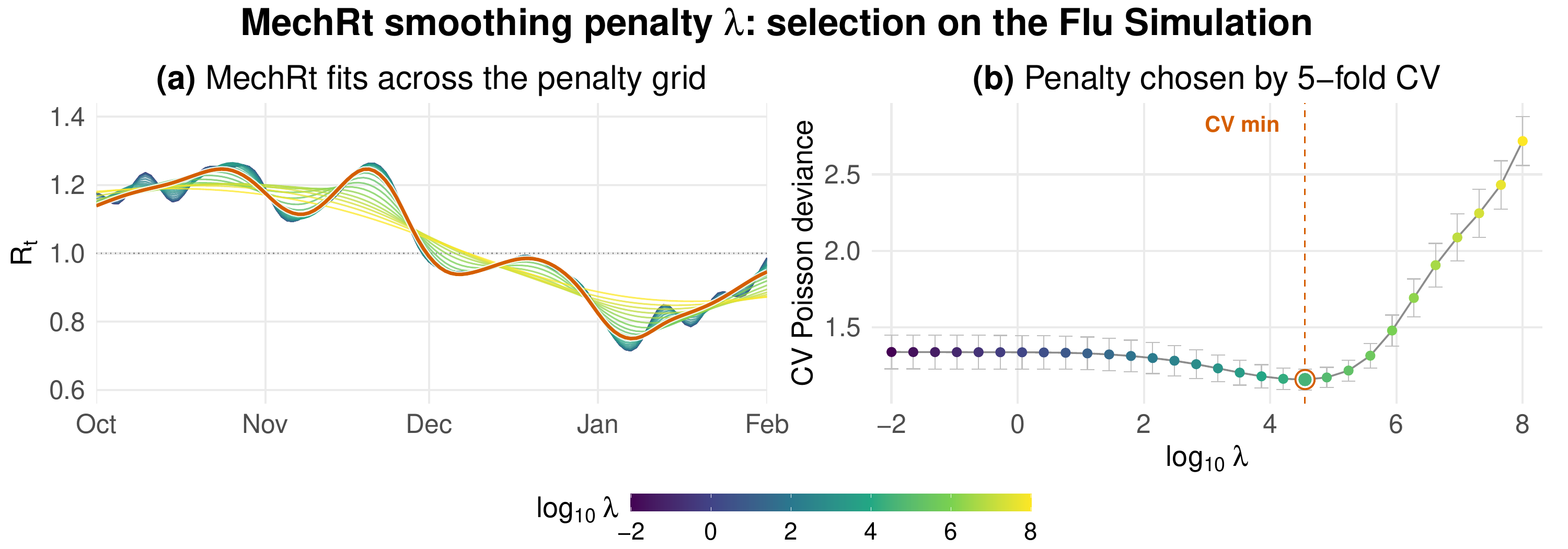}
    \caption[ConvRt predictions by smoothness parameter $\lambda$.]{ConvRt predictions by smoothness parameter $\lambda$. Most $\lambda$'s with lower cross-validation error produce reasonable $R_t$ curves, with minor qualitative differences.}
    \label{fig:results-by-lambda}
\end{figure}



\section{Influenza experiments}\label[appendix]{apx:real-flu}

To run ConvRt, we assumed a constant infection-hospitalization rate.
While a simplification, this assumption is reasonable within a season.
Moreover, estimating a time-varying IHR for a given flu season would likely require non-public data.

We also ran all experiments using finalized hospitalization counts.
These have almost certainly been backfilled, or adjusted as new counts are reported for a given date.
A properly real-time approach would have used appropriately versioned data, but we have not pursued this.

\begin{figure}
    \centering
    \includegraphics[width=\linewidth]{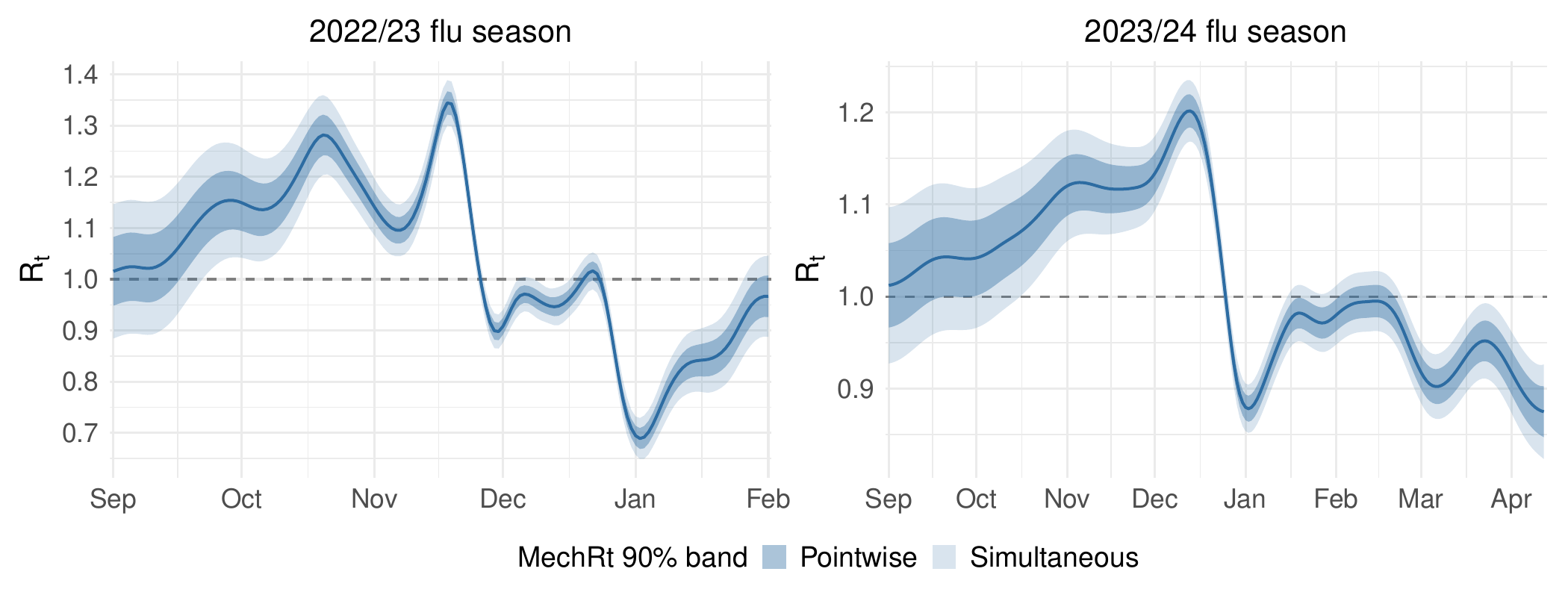}
    \caption[Retrospective ConvRt $R_t$ with pointwise and simultaneous confidence bands.]{Retrospective ConvRt $R_t$ for the 2022/23 and 2023/24 flu seasons, with 90\% pointwise (dark) and simultaneous (light) confidence bands. Pointwise intervals cover each timestep marginally at the nominal level; simultaneous bands cover the entire curve jointly.}
    \label{fig:retro-flu-simband}
\end{figure}
\cref{fig:retro-flu-simband} shows ConvRt's pointwise and simultaneous confidence bands for the two retrospective flu fits. The pointwise intervals are the ones plotted in \cref{fig:retro-flu-both} of the main text; the simultaneous bands are computed from the simulation-based critical value $c_\alpha$ described in \cref{apx:rt-uncertainty}. As expected, the simultaneous bands are uniformly wider --- by a factor of roughly $c_\alpha / z_{1-\alpha/2}$ --- since they must guard against \emph{any} of $\sim$150 timesteps escaping the band rather than a single one. The two bands answer different questions: a pointwise interval supports statements like ``$R_t$ exceeded 1 on Dec 15,'' while a simultaneous band supports statements like ``$R_t$ exceeded 1 throughout December.'' Which is appropriate depends on the inferential target.

\begin{figure}
    \centering
    \includegraphics[width=\linewidth]{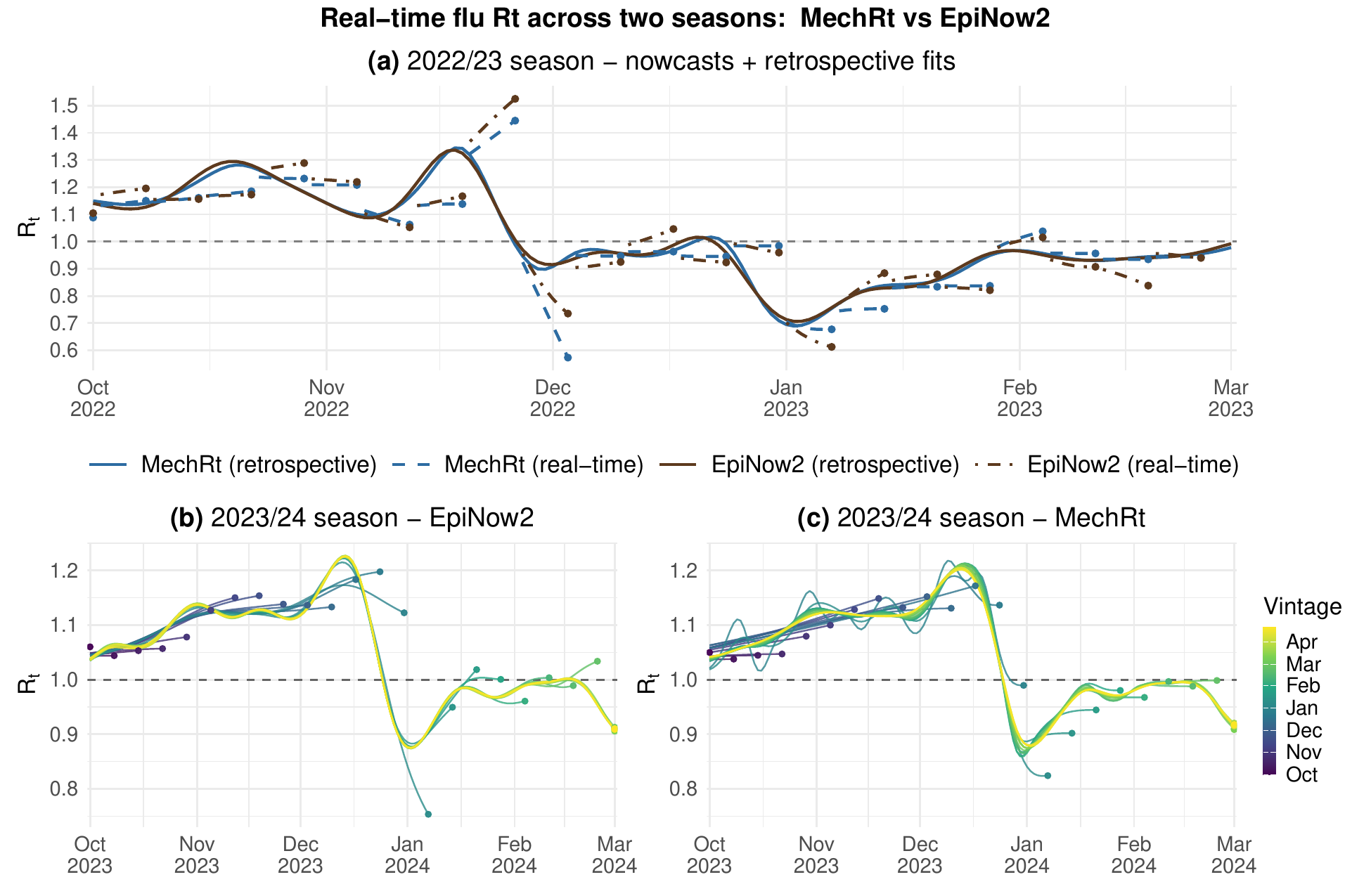}
    \caption[Real-time flu $R_t$, ConvRt vs.\ EpiNow2.]{Real-time flu $R_t$, ConvRt vs EpiNow2.
  (a)~2022/23: last-7-day nowcasts, contextualized by end-of-season fits.
  (b,c)~2023/24: weekly real-time vintages for EpiNow2 and ConvRt.}
    \label{fig:real-time-flu}
\end{figure}

\cref{fig:real-time-flu} shows ConvRt and EpiNow2's real-time predictions across both seasons, tuning $\gamma$ with the min rule.
As in the retrospective case, they are usually close to one another.
Both are quite strong when $R_t$ moves consistently.
In certain instances, ConvRt has superior performance.
For example, it correctly identifies that $R_t$ plunges in late December 2023, while EpiNow2 drops modestly.

ConvRt's nowcasts are roughly as stable as EpiNow2's.
They almost never change sharply, except when the data exhibits a longstanding trend.
This causes errors when $R_t$ changes directionality, which takes over a week to appear in the data.
For example, both methods overshoot the peak and trough in late 2022.
Those swings can be mitigated by tuning ConvRt with the 1se rule.
Alternatively, one can enforce stability by fitting ConvRt with a constant tail constraint.
Doing so produces strong predictions, visualized by \cref{fig:constant-tail} in \cref{apx:real-flu}.
However, the tapered linear tail can better estimate trends, such as rising $R_t$ in October 2023.

\begin{figure}
    \centering
    \includegraphics[width=0.9\linewidth]{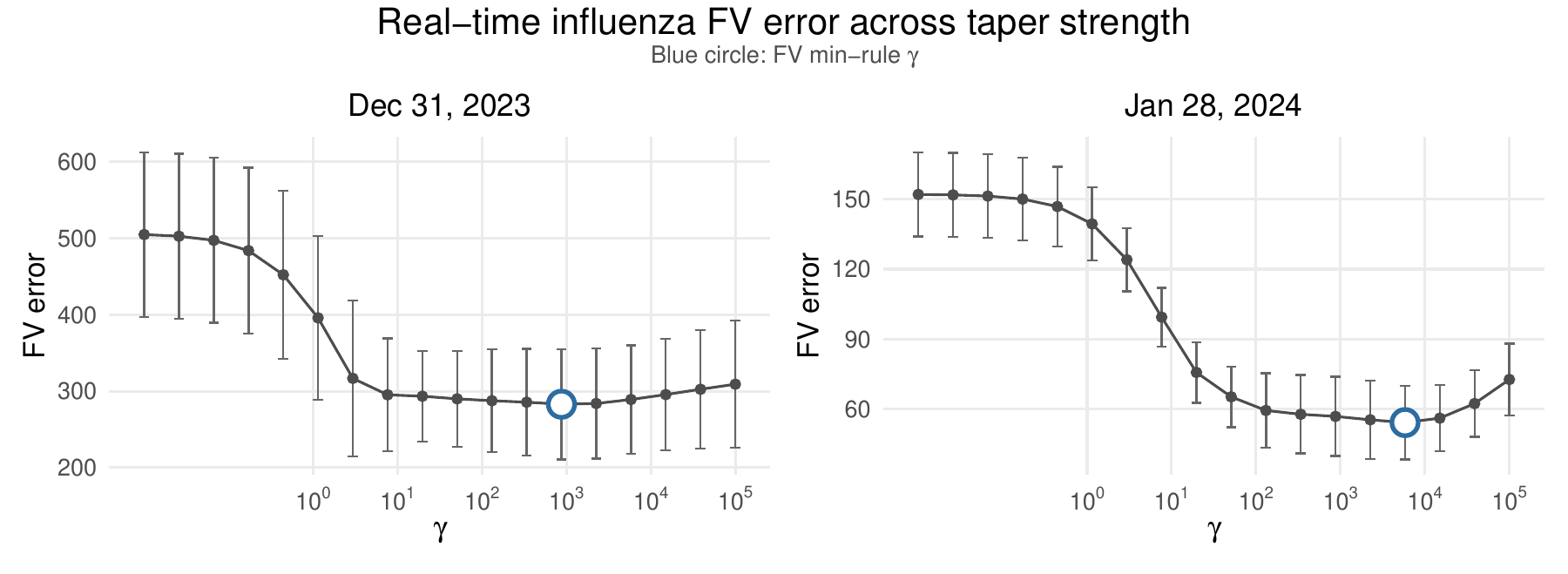}
    \caption{Forward-validation error curves (Poisson deviance) for taper hyperparameter $\gamma$.}
    \label{fig:fv-errors-flu}
\end{figure}

\cref{fig:fv-errors-flu} shows the forward-validation error curves that accompany the $R_t$ curves in \cref{fig:rt-by-gamma}. 
Notably, the error landscape is very flat. 
Barring dramatic undersmoothing (low $\gamma$),
different hyperparameter values have similar FV error.
In our experiments, we use a short 7-day tuning window to emphasize adaptivity to recent trends, 
which may contribute to the large standard errors. 
However, we find this trend persists when using a 14-day window (not shown).

This flat landscape has important consequences.
By these error metrics, hyperparameters producing very different qualitative results have comparable goodness-of-fit.
Moreover, systematic approaches like the min-rule may not be able to discern the most reasonable option.
The 1se-rule will almost always impose near-constant tail smoothing.
As a result of these findings, we find it is most prudent to use the min-rule as a best guess,
while considering predictions at all reasonable options.



\begin{figure}
    \centering
    \includegraphics[width=\linewidth]{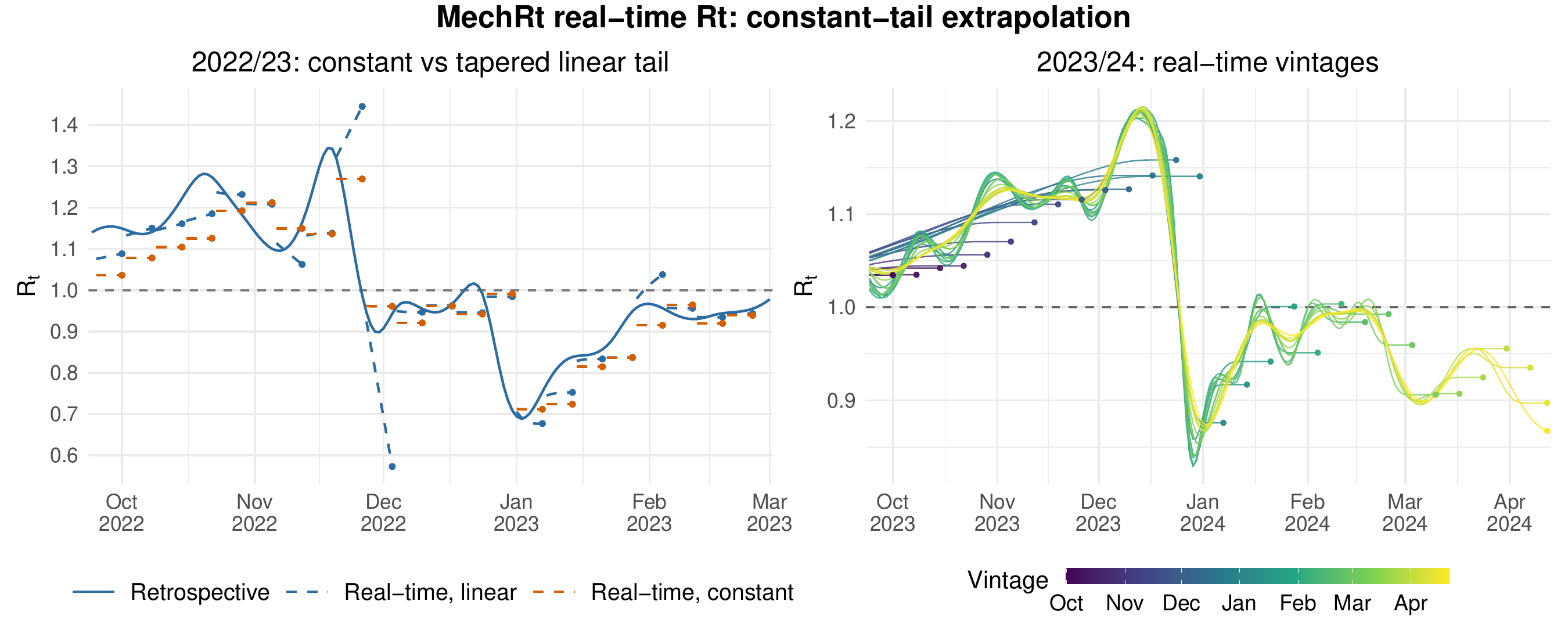}
    \caption[Real-time ConvRt nowcasts with a constant-tail constraint.]{Real-time ConvRt nowcasts with a constant-tail constraint (B-spline basis, $D^{(1)} S_\text{tail}\theta=0$; see \cref{apx:optim-constraints}). Left: 2022/23, comparing the constant-tail variant (orange) to the tapered-linear default (blue dashed) and the retrospective fit (solid blue). Right: full 2023/24 vintages under the constant-tail constraint, colored by vintage month, with dots marking each weekly nowcast.}
    \label{fig:constant-tail}
\end{figure}

The left panel of \cref{fig:constant-tail} highlights both the similarities and the qualitative differences between the two real-time variants. Across most of the season, the constant-tail and tapered-linear nowcasts agree closely. While they occasionally diverge, neither exclusively dominates the other. For example, as $R_t$ climbs in October 2022, the tapered-linear variant correctly extrapolates upward every week.
In contrast, the constant-tail variant stays flat, underestimating by a larger margin. 
But after the peak in mid-November, the true $R_t$ has begun a sharp decline that the trailing-window data does not yet reflect. The tapered-linear estimator extrapolates the still-rising trend and overshoots. The constant-tail estimator is also too high but, by refusing to extrapolate, remains closer. 

The right panel shows that the constant-tail constraint produces coherent results across vintages on the longer 2023/24 season. 
Successive nowcasts track the final retrospective trajectory smoothly.
Comparing to \cref{fig:real-time-flu}, the constant-tail fits are somewhat more stable on this season than their tapered linear counterparts, at the cost of reduced adaptivity. 

The two approaches are closely related. As the tapered penalty weight $\gamma$ grows large, the first-order differences in $R_t$ near the boundary are driven toward zero and the tapered-linear estimator approaches a constant tail. The two are not identical: the tapered penalty's weights extend a few days \emph{before} the last knot rather than acting only past it, so a high-$\gamma$ tapered-linear fit is roughly flat over a slightly wider window than a hard constant-tail constraint imposes. We default to the tapered-linear variant because it spans a richer family of tail behaviors --- approaching flat as $\gamma \to \infty$, reverting to the natural-spline linear tail as $\gamma \to 0$, and interpolating between them --- whereas a constant-tail constraint commits to a single shape. The qualitative differences above suggest this added expressivity is a real, if modest, advantage.

\section{COVID-19 experiments}\label[appendix]{apx:covidestim}

\begin{figure}[h]
    \centering
    \includegraphics[width=\linewidth]{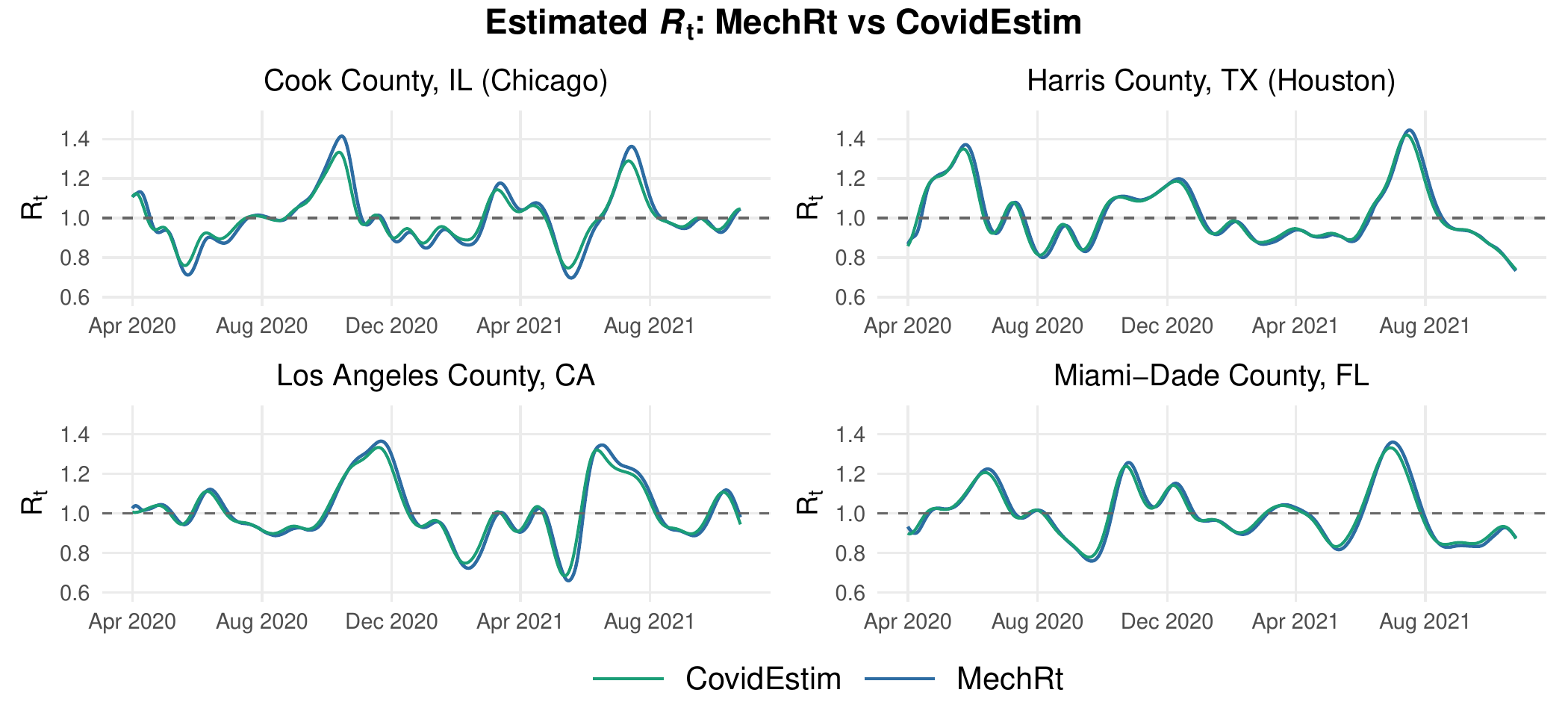}
    \caption{COVID-19 $R_t$ estimates from ConvRt and CovidEstim on four major \US\ counties.}
    \label{fig:covidestim-counties}
\end{figure}

In addition to seasonal influenza, we evaluated our $R_t$ methods on COVID-19.
Estimating $R_t$ for COVID-19 is more challenging, as it is sensitive to changing ascertainment rates over time.
The CovidEstim model addressed this by incorporating seroprevalence data, publishing case ascertainment rates along with $R_t$ \citep{Chitwood2022}.
Both quantities are modeled as splines, with prior distributions placed on their parameters.
This Bayesian method is notoriously slow to run, often taking 10 hours to fit a single state \citep{covidestim_summer2021}.
We plugged in their case ascertainment rates to conduct a retrospective analysis with our method.
The COVID-19 infections these ascertainment rates imply are shown in \cref{fig:covid-cases}, along with positive case reports.

\begin{figure}
    \centering
    \includegraphics[width=\linewidth]{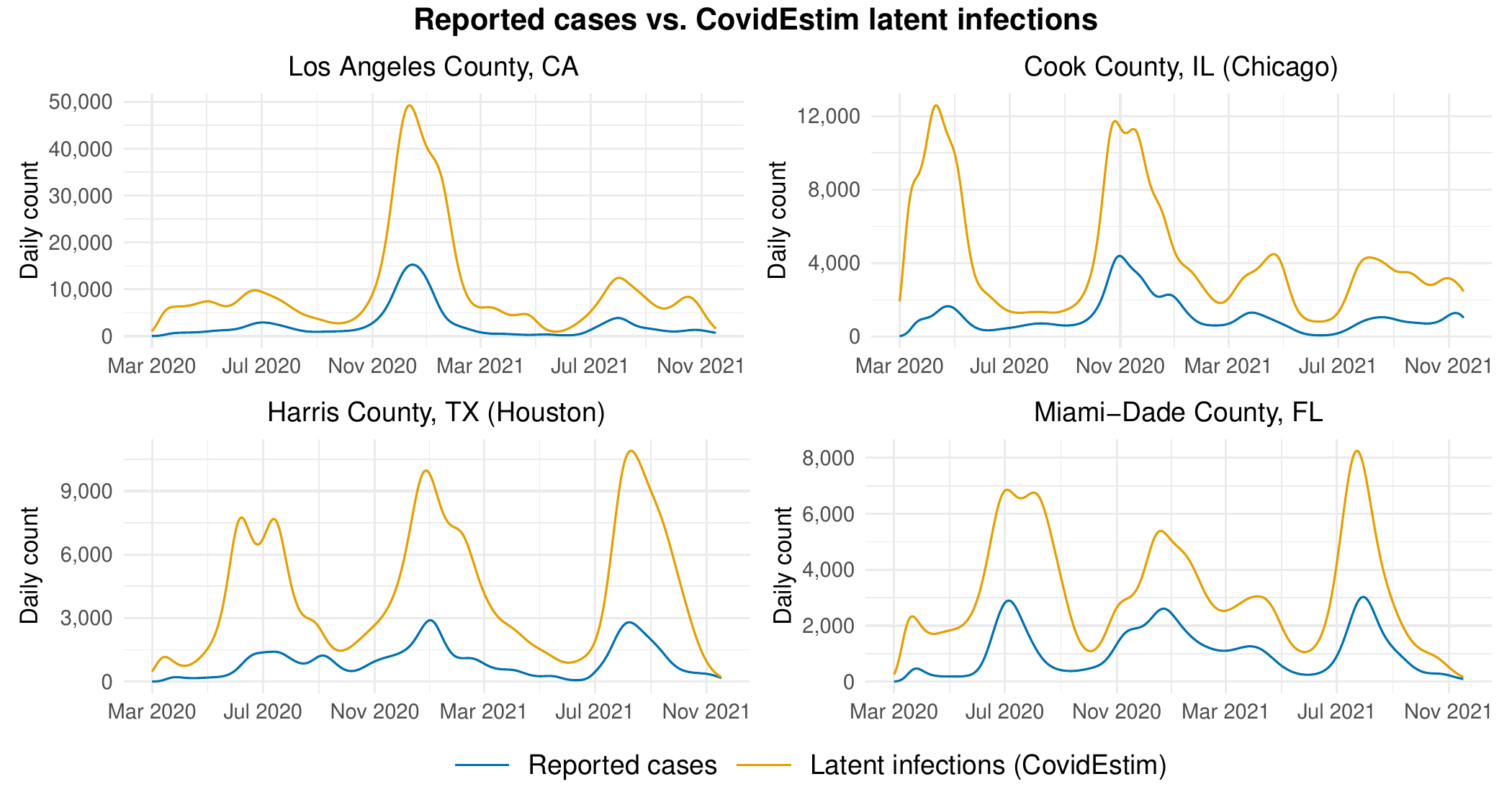}
    \caption[Cases and infection counts for pre-Omicron COVID-19.]{Cases and seroprevalence-based infection counts for pre-Omicron COVID-19 in four \US\ counties.}
    \label{fig:covid-cases}
\end{figure}


\cref{fig:covidestim-counties} compares ConvRt's $R_t$ estimates on a number of \US\ counties.
To avoid challenges with data access and model fitting, we used cached results rather than generate our own.
These results only had uncertainty bands at the state level, hence our point estimates.
Overall, predictions are almost identical, always within 0.1 of one another.
Large spikes in $R_t$ clearly correspond to surges in winter 2020, as well as the onset of the Delta variant in summer 2021.
Interestingly, surges happen at slightly different times and magnitudes, mirrored by the infection counts in \cref{fig:covid-cases}.
For example, Houston and Miami had substantial COVID-19 waves in summer 2020, driven by rises in $R_t$ well above 1.
Across the political divide, Chicago and Los Angeles generally did not, with LA's $R_t$ peaking at 1.1.

The values of $R_t$ in \cref{fig:covidestim-counties} are entirely plausible. 
A CDC analysis using EpiEstim found that most state-level peaks sit between 1.2 and 1.5, matching ours of 1.3-1.4 \citep{lopez2023covid}.
\citet{figgins2021sarscov2}, another Bayesian approach, suggests somewhat higher peaks, with a maximum of 1.66.
Gaps in $R_t$ can often be accounted for by the delay distributions,
with longer delays producing higher estimates.
It is thus easier to explore a range of possibilities with our method,
which takes seconds to run rather than hours.

\begin{figure}
    \centering
    \includegraphics[width=0.8\linewidth]{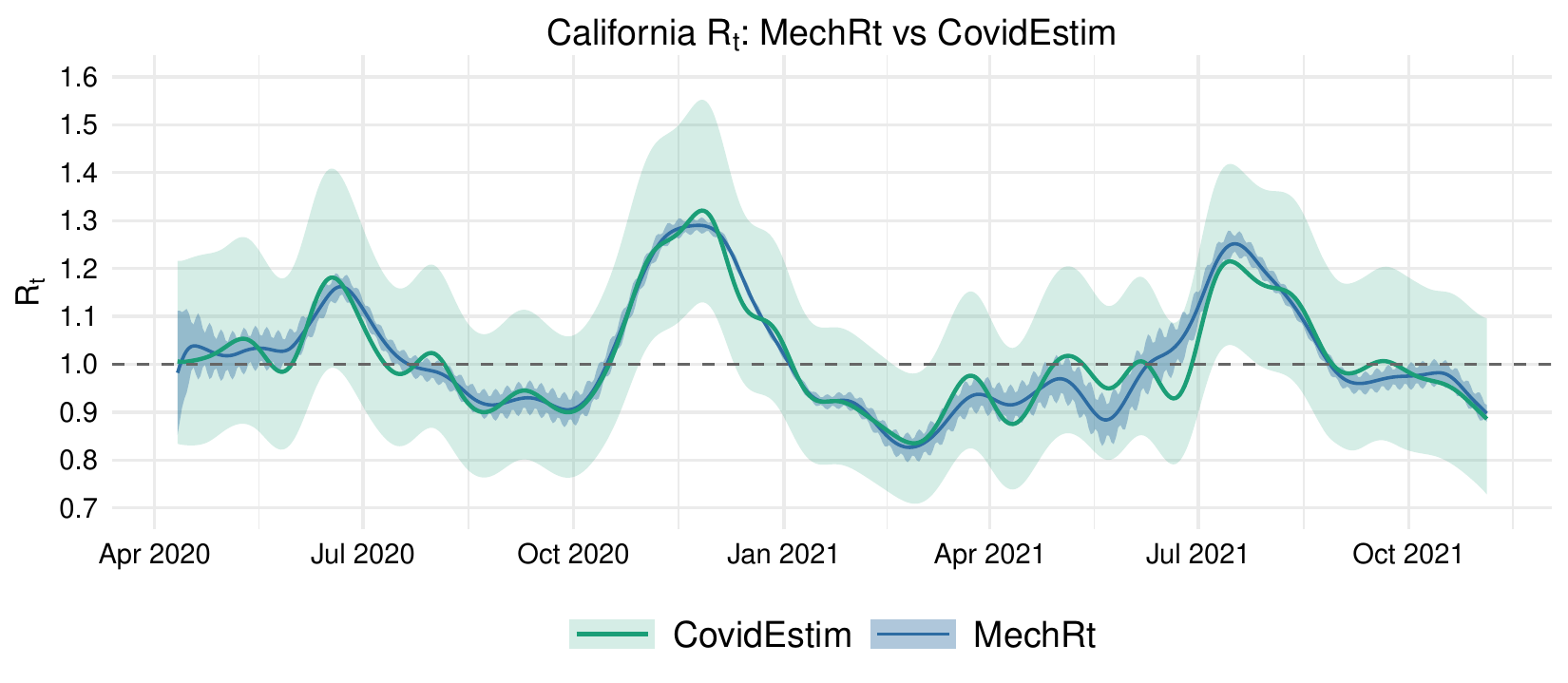}
    \caption{California $R_t$ curve, CovidEstim and ConvRt.}
    \label{fig:california-covidestim}
\end{figure}

We aggregated data to also compare $R_t$ in the state of California. 
\cref{fig:california-covidestim} shows the results, with 80\% pointwise confidence bands. 
Once again, the results are extremely similar throughout the nearly two years tracked.
At the peaks (\textasciitilde1.3 in winter 2020, \textasciitilde1.2 in Delta), the two estimators differ at most 0.05.
Moreover, CovidEstim's credible intervals are far wider than our confidence intervals.
The intervals shown are pointwise, meaning they are valid at any individual timestep but
are not expected to cover jointly. 

To produce Figures \ref{fig:covidestim-counties} and \ref{fig:california-covidestim}, we estimated generation intervals and infection-to-report distributions from the literature.
Parameterizing both as discrete gammas, we set the mean generation time to 4.5 days and an SD of 1.9 \citep{ganyani2020estimating, hart2022generation}.
For the infection-to-report delay, we set its mean to 10.2 days and its SD to 3.6 days \citep{li2020early, he2020temporal,arino2020simple, abbott2020estimating}.

\begin{figure}
    \centering
    \includegraphics[width=0.9\linewidth]{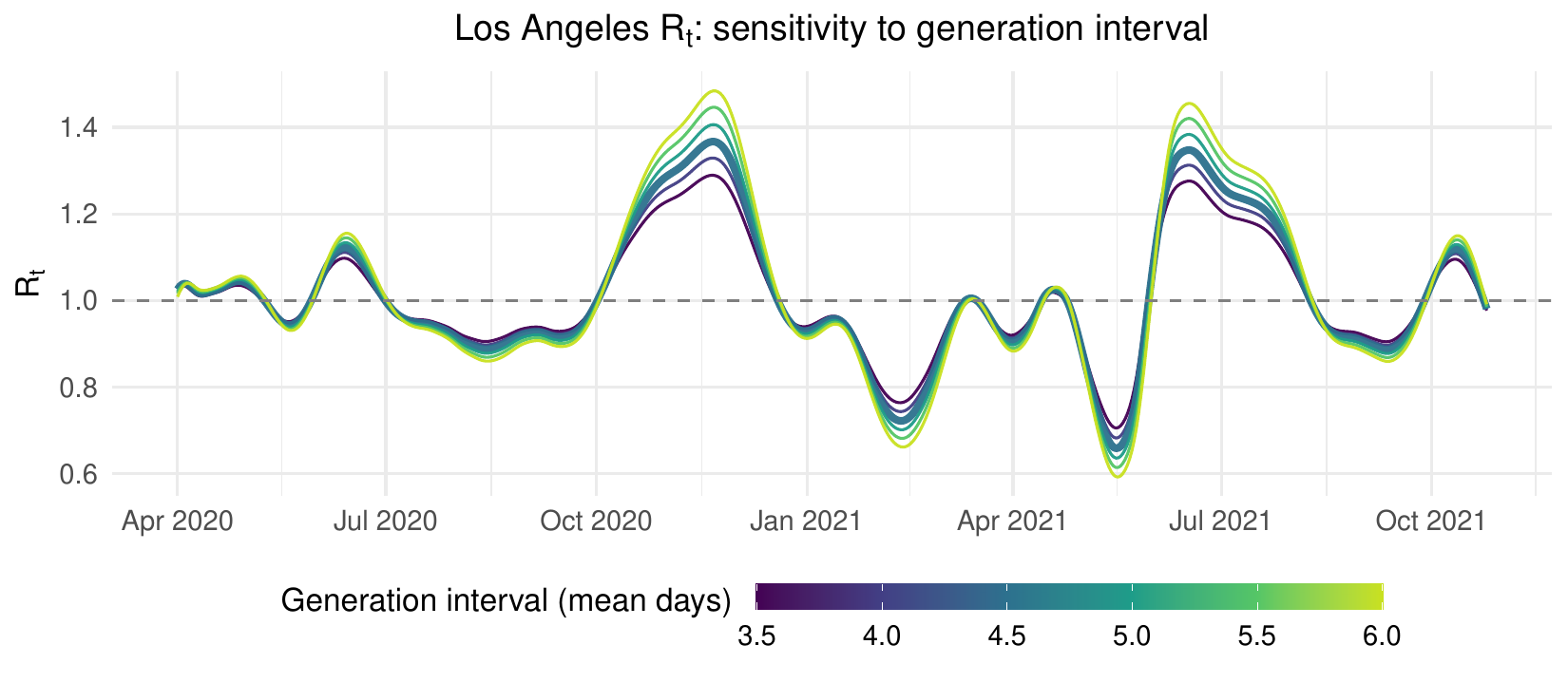}
    \caption[COVID-19 $R_t$ sensitivity to generation interval in Los Angeles.]{ConvRt's predictions of COVID-19 $R_t$ in Los Angeles under various delay distributions.}
    \label{fig:covidestim-delays}
\end{figure}



While the results matched CovidEstim nicely, neither distribution is known with high confidence. For that reason, \cref{fig:covidestim-delays} shows $R_t$ estimates on the Los Angeles CovidEstim data under varying generation intervals, sweeping mean generation times from 3.5 to 6.0 days while holding the SD and infection-to-report delay fixed. The six resulting curves share the same shape, differing only in magnitude: in the winter 2020 peak, estimates range from around 1.3 to 1.5, with longer generation intervals producing higher $R_t$ at peaks and lower at nadirs. This is expected---longer intervals compare current infections to levels further in the past, requiring a more extreme multiplier to explain the larger gap.

\section{Additional methodology}\label[appendix]{apx:extra-methods}

\subsection{Weekly data}\label[appendix]{apx:weekly}

\begin{figure}
    \centering
    \includegraphics[width=\linewidth]{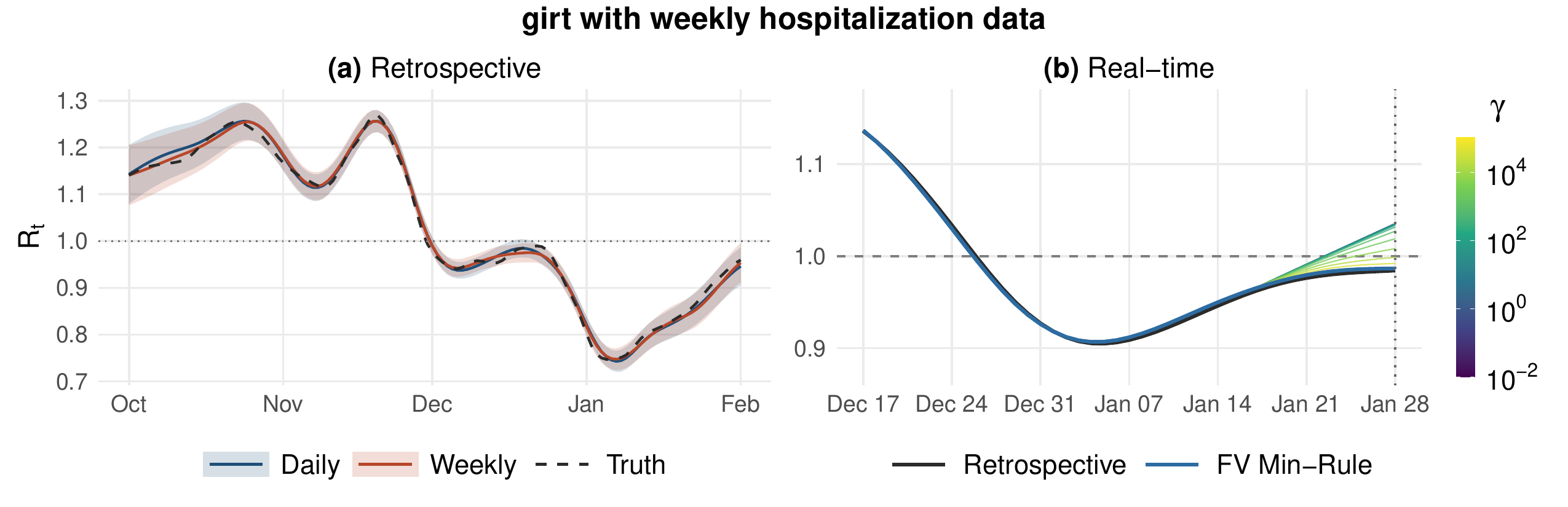}
    \caption{ConvRt predictions given data at a weekly resolution, in retrospect and real-time.}
    \label{fig:weekly}
\end{figure}

Often, data is released a weekly resolution, not daily. 
For example, the influenza experiments in this paper used hospitalization counts from NHSN.
This pandemic-era program ceased its daily reporting in May 2024.
Alternatively, Flu-Surv NET has reported weekly flu hospitalizations from 2005 to the present, based on a surveillance sample. 
In addition, COVID-NET and RSV-NET publish weekly hospitalizations for COVID-19 and RSV.

Our estimation approach must be adjusted to work with less frequent data.
One option would be to change the discrete-time transmission model to weekly resolution.
However, this risks being significantly too coarse, as most transmissions occur within the first week for many infectious diseases. 
Maintaining daily dynamics, we instead adjust the likelihood to model weekly observations $y_w$.
Equation \eqref{eq:quasi-poisson-model} related each daily total $y_t$ to its mean $\mu_t$, which had been scaled by the ascertainment rate $\rho_t$ and day-of-week effect $\omega$.
Dropping day-of-week effects for simplicity, we aggregate means over each week $w$:
\begin{equation}\label{eq:poisson-seir-weekly}
    y_w\sim \text{Quasi-Poisson}(\mu_w, \varphi); 
    \qquad \mu_w = \sum_{t\in w} \rho_t \Lambda_t;
    \qquad \Lambda_t = \sum_{s<t} x_{s} \pi_{t-s}.
\end{equation}

Once again, explicitly accounting for overdispersion $\varphi$ is optional. 
The statistical model is justified by the by the fact that the sum of independent Poissons is still Poisson.
(We assume throughout this work that counts on successive dates are independent.) 

To estimate $R_t$, everything beyond the likelihood is identical to the daily case.
Using the same spline parametrization, ConvRt performs penalized MLE in the retrospective and real-time settings.
It first deconvolves latent infections $\hat x$, as described in \cref{apx:deconvolution}.
Then, it plugs them into $\Lambda_t$ to deconvolve $R_t$ \eqref{eq:exp-cases-rt}, again replacing $x_t$ with the renewal equation \eqref{eq:renewal}.
Hyperparameters are tuned via cross-validation.
A higher degree of smoothing may be necessary, since the number of data points is reduced by a factor of 7.

We ran ConvRt on the simulated flu dataset, fitting on weekly hospitalization totals.
Hyperparameters were tuned via the min-rule.
The left panel of \cref{fig:weekly} shows the retrospective fit’s point estimates and 95\% confidence bands.
Both are almost exactly the same as those from the daily case. 

The right panel displays ConvRt’s real-time predictions, given weekly data through January 28.
The data supports a moderate but clear rise in $R_t$ from Jan 7--21. 
However, the infection-to-report delay is such that $R_t$ in the final week is barely observed.
Contrasting predictions across $\gamma$ reflect the range of possible outcomes, and the strength of the evidence supporting them.
With little tail regularization, extrapolation continues linearly and overshoots the truth by \textasciitilde0.05.
Heavier regularization, which forward-validation selects, tamps down the tail, and correctly predicts $R_t$ just under 1.

Together, these results are highly encouraging.
They suggest our method can function perfectly well given only weekly data.
While comprehensive experiments are beyond the scope of this paper, further work should evaluate its quantitative performance, and consider alternate approaches.

\subsection{Trend Filtering}

\begin{figure}
    \centering
    \includegraphics[width=.8\linewidth]{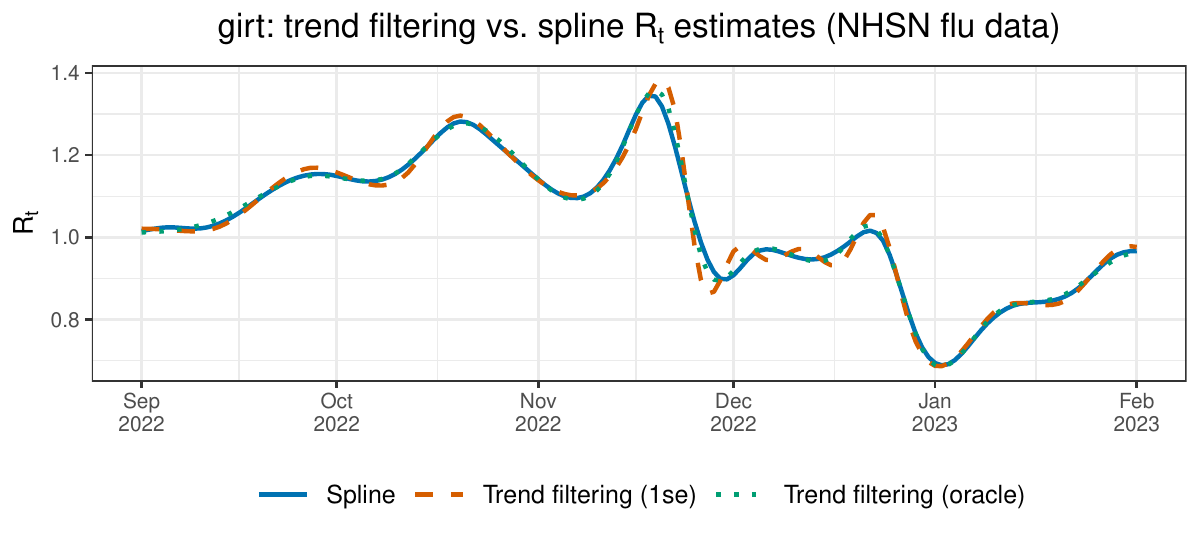}
    \caption{ConvRt predictions on NHSN flu data, modeled as a cubic spline and trend filter.}
    \label{fig:tf-rt}
\end{figure}

We ran ConvRt with trend filtering regularization, instead of as a spline. 
This nonparameteric regression technique, introduced in \cref{sec:rt-retro}, adaptively selects knot locations.
We chose a fourth-difference penalty to produce piecewise cubic functions, matching the splines in our experiments. 
We define  $\mathcal{R}$ as the vector of $R_t$ values, and ignore DoW effects for convenience.
The retrospective fit is
\begin{align}\label{eq:tf-retro}
\hat{\mathcal{R}} &= 
\argmin_{\mathcal{R} \succeq 0}
             \sum_{t=t_0}^{t_1}\!\bigl[ \mu_t(\mathcal{R}) - y_t \log \mu_t(\mathcal{R}) \bigr]
             \;+\; \lambda\lVert D^{(4)} \mathcal{R} \rVert_1 \\
             &\qquad\text{where}\quad \mu_t(\mathcal{R}) = \rho_t\, \sum_{s<t} R_s \left(\sum_{u<s} \hat x_u\, g_{s-u}\right)\pi_{t-s}.
\end{align}
The real-time fit was defined analogously with tail constraints and regularization.
CVXR optimized the objective with the ECOS solver.
For a retrospective fit, this took 30 seconds to tune over a grid of 25 lambdas.
While relatively fast, the spline took only 2 seconds to tune, since it converged in only a few iterations of IRLS.

Other design decisions also matched the spline experiments. 
The likelihood model for the data stayed the same, as did hyperparameter tuning. 
Because the objective \eqref{eq:tf-retro} is not twice-differentiable in $R_t$, we cannot derive asymptomatic confidence intervals for $R_t$.
\citet{rtestim} used a quadratic relaxation for their similar trend filtering problem, which we have not pursued here.

As with weekly data (\cref{apx:weekly}), we did not conduct thorough experiments to evaluate trend filtering. 
However, we established a meaningful proof-of-concept, this time on real data for contrast.
\cref{fig:tf-rt} displays trend filtering’s predictions on the 2022/23 flu season. 
By and large, they are very similar to the spline’s $R_t$ estimates.
Tuning $\lambda$ with the 1se rule, trend filtering is slightly wigglier, but this can be attributed to the particular value of $\lambda$ that was tuned here; selecting a larger lambda, as shown in the dotted line, yields an $R_t$ curve that is nearly identical to the spline.

In general, differences in prediction are subtle, not systematic.
This supports our choice to focus on splines due to their natural uncertainty quantification. 
Nevertheless, trend filtering remains a reasonable alternative if the true signal changes dynamically.

\end{document}